\documentclass[11pt]{article}

\usepackage{siunitx}
\usepackage{fancyhdr}
\usepackage{extramarks}
\usepackage{amsmath}
\usepackage{amsthm}
\usepackage{arydshln}
\usepackage{mathtools}
\usepackage{amsfonts}
\usepackage{enumerate}
\usepackage{graphicx}
 \graphicspath{ {figures/} }
\usepackage{algpseudocode}
\usepackage{textcomp}
\usepackage{import}
\usepackage{natbib}
\usepackage[font={footnotesize, it}]{caption}
\usepackage{bbm}
\usepackage{subcaption}
\usepackage[ruled,vlined]{algorithm2e} 
\usepackage[font={footnotesize, it}]{caption}
\usepackage{xr}
 \usepackage{hyperref}
 \hypersetup{
     linktoc=all,     %set to all if you want both sections and subsections linked
     linkcolor=black,  %choose some color if you want links to stand out
 }
 \usepackage{geometry}
 \newgeometry{vmargin={1in}, hmargin={1in,1in}}
 
 \usepackage[T1]{fontenc}
\usepackage[utf8]{inputenc}
\usepackage{authblk}
\usepackage{blindtext}
\usepackage{booktabs}
\usepackage{array, makecell}

\newcommand{\E}{\mathbb{E}}
\newcommand{\Var}{\mathrm{Var}}
\newcommand{\Cov}{\mathrm{Cov}}

\newtheorem{theorem}{Theorem}
\newtheorem{proposition}[theorem]{Proposition}
\newtheorem{lemma}[theorem]{Lemma}
\newtheorem{corollary}[theorem]{Corollary}
\newtheorem{remark}[theorem]{Remark}
\theoremstyle{definition}
\newtheorem{definition}[theorem]{Definition}
\theoremstyle{definition}

\theoremstyle{definition}
\newtheorem{example}{Example}

\begin{document}

\title{Cluster Stability Selection}
%\author{
%  For Greg Faletto qualifying exam (advisor: Jacob Bien)
%  % Vitaly Surazhsky \\
%%                Department of Computer Science\\
%%        Technion---Israel Institute of Technology\\
%%        Technion City, Haifa 32000, \underline{Israel}
%%            \and
%%        Yossi Gil\\
%%        Department of Computer Science\\
%%        Technion---Israel Institute of Technology\\
%%        Technion City, Haifa 32000, \underline{Israel}
%}
\date{\today}

\author{Gregory Faletto\thanks{Corresponding author: gregory.faletto@marshall.usc.edu} }
\author{Jacob Bien}
\affil{Department of Data Sciences and Operations \\ University of Southern California Marshall School of Business}

%\author{Gregory Faletto gregory.faletto@marshall.usc.edu \and Jacob Bien jbien@usc.edu}
%\affil{Department of Data Sciences and Operations\\
%       Marshall School of Business\\
%       University of Southern California}

%\author{Gregory Faletto \qquad Jacob Bien  \\
%%gregory.faletto@marshall.usc.edu \qquad jbien@usc.edu \\
% \\
%      Department of Data Sciences and Operations\\
%       Marshall School of Business\\
%       University of Southern California\\
%}

%\author{ Gregory Faletto \\ gregory.faletto@marshall.usc.edu \\
%      Department of Data Sciences and Operations\\
%       Marshall School of Business\\
%       University of Southern California\\
%       Los Angeles, CA 90089, USA
%       \and
%       Jacob Bien \\ jbien@usc.edu \\
%      Department of Data Sciences and Operations\\
%       Marshall School of Business\\
%       University of Southern California\\
%       Los Angeles, CA 90089, USA}
%       
%       \affil{Department of Data Sciences and Operations\\
%       Marshall School of Business\\
%       University of Southern California}
       
%\author[1]{Gregory Faletto \thanks{gregory.faletto@marshall.usc.edu}}
%\author[1]{Jacob Bien \thanks{jbien@usc.edu}}
%\affil[1]{Department of Data Sciences and Operations, University of Southern California Marshall School of Business}

%\renewcommand\Authands{ and }

%\editor{}

\maketitle

\begin{abstract}
Stability selection \citep{meinshausen-2010} makes any feature selection method more stable by returning only those features that are consistently selected across many subsamples. We prove (in what is, to our knowledge, the first result of its kind) that for data containing highly correlated proxies for an important latent variable, the lasso typically selects one proxy, yet stability selection with the lasso can fail to select any proxy, leading to worse predictive performance than the lasso alone.

We introduce {\em cluster stability selection}, which exploits the practitioner's knowledge that highly correlated clusters exist in the data, resulting in better feature rankings than stability selection in this setting. We consider several feature-combination approaches, including taking a weighted average of the features in each important cluster where weights are determined by the frequency with which cluster members are selected, which we show leads to better predictive models than previous proposals.

We present generalizations of theoretical guarantees from \citet{meinshausen-2010} and  \citet{shah_samworth_2012} to show that cluster stability selection retains the same guarantees. In summary, cluster stability selection enjoys the best of both worlds, yielding a sparse selected set that is both stable and has good predictive performance.
\end{abstract}

%\begin{keywords}
%variable selection, correlated features, lasso, latent variables
%\end{keywords}

\section{Introduction}

%\subsection{Stability Selection}

\textit{Stability}, as characterized by \citet{Yu2013}, holds when ``statistical conclusions are robust or stable to appropriate perturbations to data.'' \citet{Yu3920} call stability one of ``three core principles" necessary for ``principled inquiry to extract reliable and reproducible information from data." 

\textit{Stability selection} \citep{meinshausen-2010} adds stability to any base feature selection method. \citeauthor{meinshausen-2010} focus on the lasso \citep{Tibshirani1996} as the base procedure, and we will do the same. Even with a fixed \(\lambda\), the sets of features selected by the lasso can be unstable, particularly in the high-dimensional setting (\(p \gg n\)). 

\begin{figure}[htbp]
\begin{center}
\includegraphics[width=\textwidth]{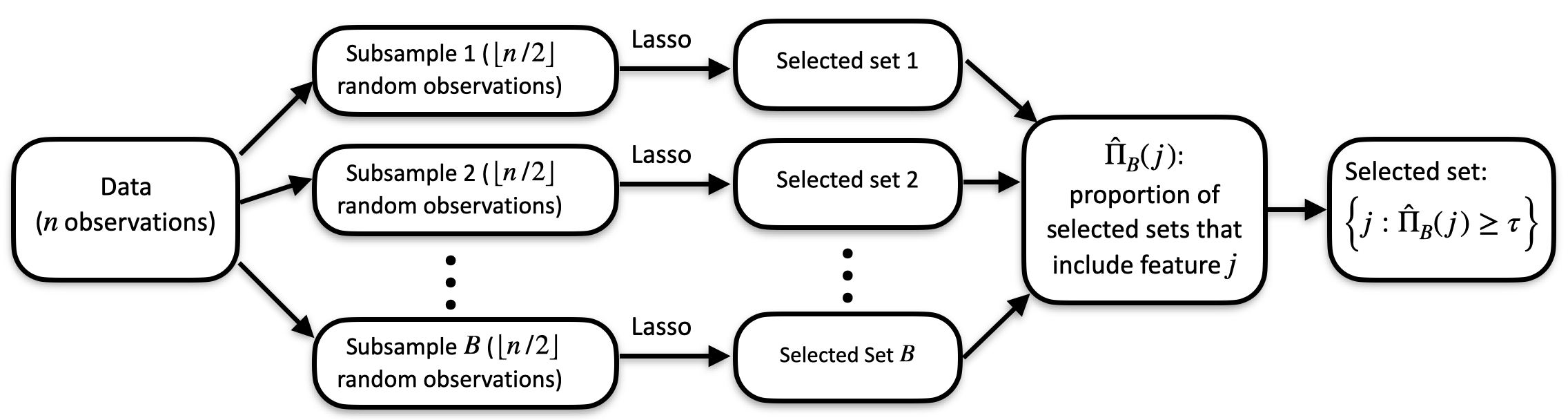}
\caption{Flow chart depicting stability selection.}
\label{stabsel_flowchart}
\end{center}
\end{figure}

Figure \ref{stabsel_flowchart} illustrates stability selection.
The data are repeatedly randomly split into subsamples of size \(\lfloor n/2 \rfloor\). On each subsample, the lasso (or any other feature selection method) is used to select a set of features. The proportion of subsamples in which each feature is selected is interpreted as an importance measure for that feature. Finally, the selected set returned by stability selection is the set of features whose selection proportion exceeds a predetermined threshold. This adds stability to the lasso and guarantees control of false discoveries under very mild assumptions \citep{meinshausen-2010}.

Stability selection has, however, a structural problem in the case where observed features are highly correlated. This is the primary motivation of our work and has been noted before.  In the words of \citet{shah_samworth_2013}, the problem is that ``highly correlated variables\ldots split the vote.'' (\citealt{kirk-2010} and \citealt{kent-2010} make very similar points.) 

As a model for why features may be correlated, we will use the \textit{errors-in-variables} framework. Suppose that a variable \(\boldsymbol{Z}\) is in the true model for \(\boldsymbol{y}\), but is not observed. Instead, \(q \geq 2\) equally good \textit{proxies}---\(\boldsymbol{Z}\) plus a little noise---are observed. Any one of these proxies would be useful for out-of-sample predictive performance because of their high correlation with \(\boldsymbol{Z}\). The lasso will tend to choose one proxy on each fit, choosing uniformly at random among the proxies, so each proxy's selection proportion tends toward \(1/q\) rather than 1. Stability selection's ranking of the importance of the features is then suboptimal for out-of-sample predictive performance. The following simulation study illustrates this problem. (We will explain this simulation in full detail in Section \ref{sim.study.sparse}; for now, we omit details for the sake of exposition.)

 \begin{example}\label{ex.stab.problem}
 
We simulate $n=200$ observations of 100 features: $q=10$ proxies that each have correlation 0.9 with the latent variable $\boldsymbol{Z}$, 10 independent ``weak signal features'' that are in the true model for \(\boldsymbol{y}\) (but with smaller coefficients than \(\boldsymbol{Z}\)), and 80 ``noise features.'' 

The left panel of Figure \ref{intro_sim} shows the top 20 selection proportions after applying stability selection with the lasso on one simulated data set. Although any one proxy for \(\boldsymbol{Z}\) would be a better selection for prediction than any other observed feature, the proxies for $\boldsymbol{Z}$ have lower selection proportions than many of the weak signal features. The right panel shows that this behavior substantially hurts the out-of-sample predictive performance of stability selection for most model sizes (the mean squared error represents the out-of-sample prediction error of a least squares model using the selected set of features).
 
\begin{figure}[htbp]
\begin{center}
\includegraphics[width=\textwidth]{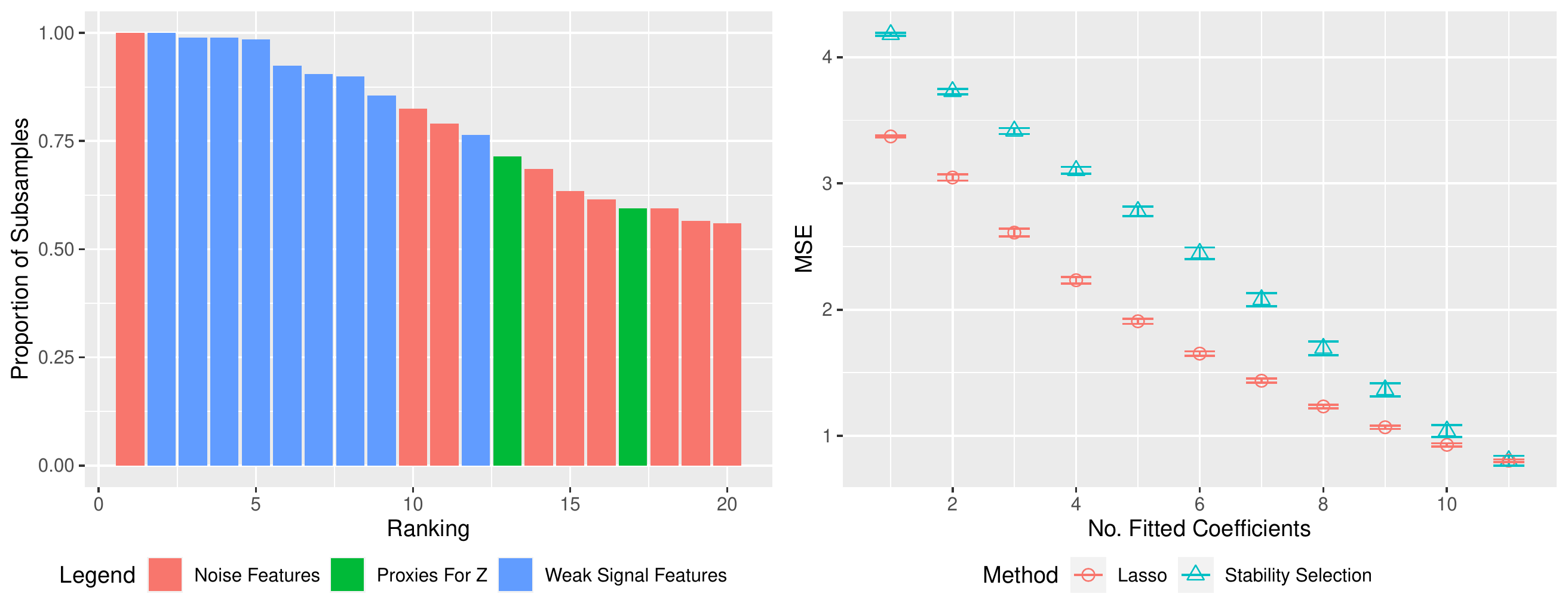}
\caption{Left: the top 20 selection proportions for stability selection with the lasso in a single simulated example. Right: the average mean squared error against model size for the lasso and stability selection across 1000 simulations from Example \ref{ex.stab.problem}.}
\label{intro_sim}
\end{center}
\end{figure}

\end{example}

This is not just a problem of theoretical concern. Highly correlated data arise in many fields of study. 

\begin{itemize}

\item In \textbf{economics}, practitioners observe \textit{repeated measurements}---multiple noisy observations of the same latent signal, like answers to survey questions \citep{Schennach2016} or the private information known to bidders in auctions \citep{Li2000, Krasnokutskaya2011}.
%wages \citep{Kennan2011}, 

\item In \textbf{education}, noisy measurements of academic ability (like test scores) may be used as features \citep{Cunha2005, davis2002statistical}.

\item In the \textbf{social sciences and humanities}, predictors may include the personal interpretations of game participants \citep{10.1145/985692.985733}, Amazon Mechanical Turk workers \citep{Mason2012}, or the general public \citep{10.1111/j.1365-2966.2008.13689.x}. Each person's assessments can be interpreted as a noisy measurement of a common underlying signal \citep{Hayes2007}.

\item In \textbf{biology}, gene expression levels are used to predict health outcomes \citep{sorlie_tibshirani_2003}. Some genes are known to share common biological pathways. Expression levels of genes in the same pathway may be nearly identical up to measurement noise \citep{Segal2004}.

\end{itemize}

In Section \ref{real.data.study2}, we will demonstrate how our proposed method can be applied to yet another setting with clustered features: genome-wide association studies (GWAS). It is common that multiple proxies for the same signal are observed in one way or another, and we should expect this to become more common in ``big data" settings where as many features are observed as possible.

To address this problem, we propose \textit{cluster stability selection}. In brief, given a set of known or estimated clusters \(\{C_k\}\), we find the proportion of subsamples in which at least one feature from cluster \(C_k\) was selected. Cluster \(C_k\) is interpreted to be important (that is, the underlying signal common to all the features in \(C_k\) is considered important) if this selection proportion is high. If \(C_k\) is important, we construct a \textit{cluster representative} by taking a weighted average of the features in \(C_k\) and using this cluster representative for downstream regression tasks (rather than regressing on the individual features separately). We propose three different ways to calculate these weights and demonstrate both theoretically and empirically the advantages each weighting scheme enjoys, depending on the context and the goals of the investigator. 

We will show that cluster stability selection dominates both the lasso and stability selection in settings where there are clusters of highly correlated features in the data. Cluster stability selection retains the benefits of stability selection---like higher stability of selected sets and fewer false selections---while avoiding ``splitting the vote" and thereby hurting prediction accuracy. Further, our proposed method works even when no clusters are present, and we show that our estimator is more powerful than previous versions of stability selection in this case.

We briefly review the literature related to our work. Investigations into the stability of learning algorithms more broadly go at least as far back as works like \citet{Devroye1979, Kearns1997}, \cite{Bousquet2002}, and \citet{Lange2003}. A few authors propose similar methods to what we discuss. \citet{kirk-2010} and \citet{Alexander2011} discuss selection proportions for groups of features in stability selection, as we propose, but they do not consider the importance of individual features in the cluster, nor do they form cluster representatives, opting instead to select all features from correlated clusters. They also use the stability selection selection probability estimator proposed by \citet{meinshausen-2010}, whereas we propose a novel estimator of the selection probabilities. Finally, they do not examine predictive performance of the resulting selected models.

\citet{Gauraha2016} also considers stability selection when features are highly correlated. \citeauthor{Gauraha2016} proposes clustering the features, forming cluster representatives, and using the lasso regressed on cluster representatives as the base procedure for stability selection.

\citet{Beinrucker2016} develop \textit{extended stability selection}. This method divides the data into subsamples of size \(n/L\) for some \(L \in \mathbb{N}\) rather than only \(\lfloor n/2 \rfloor\), and additionally chooses random subsets of features to consider for each subsample. \citeauthor{Beinrucker2016} show through simulations that their method is helpful when predictors are highly correlated.

In the presence of highly correlated features, one of the most frequently suggested ideas to practitioners is to simply drop all but one of the highly correlated features from the data; see for example \citet[Section 4.7.1]{Greene2012Econometric} and \citet[Section 3.3.3, p. 102]{james2021introduction}. This idea underlies the \textit{protolasso} \citep{Reid2015}. In this method, the feature within a cluster with the greatest marginal correlation with the response is chosen as a \textit{prototype} for the cluster. The remaining features in each cluster are discarded, and the lasso is estimated on the prototypes. Related proposals form cluster representatives via averaging \citep{Buhlmann2013, Park2007}, and other more sophisticated proposals exist as well; see  \citet{Li1998} and \citet[Section 4.1]{Schennach2016}.

Many earlier works have noted, but not proven, that the lasso tends to select only one feature from a highly correlated cluster \citep{Efron2004, Zou2005, Zhao2006, Bondell2008, Jacob2009, She2010, SortedL1, Witten2014, Anbari2014, Li2018}. Our primary theoretical contribution (Theorem \ref{thm.result.sel}) provides what is to our knowledge the first precise statement with proof of this observation.

More recent research engages with regularized estimation of models with measurement error, particularly in the high-dimensional setting \citep{Rosenbaum2010, rosenbaum, loh2012, Sorensen2015, Belloni2016, BelloniPivot2017, Belloni2017,  Zheng2018, Nghiem2019}. In contrast to this stream of literature, we are interested in stable model selection for good out-of-sample predictive performance, we focus on stabilizing existing estimation procedures rather than proposing a new one, and our method is designed to be useful when more than one noisy observation of the same lurking signal is available.

Lastly, a more general setting than the one we consider is that of \textit{latent factor models}, where observed features may be influenced by multiple latent variables rather than just one \citep{Bollen1989, izenman, Bing2019, Bing2021}.

Below, we outline our main contributions as well as the structure of our paper.

\begin{itemize}

\item In Section \ref{sec.theory}, we prove in a simple setting that the lasso tends to select one feature from a highly correlated cluster (Theorem \ref{thm.result.sel}). We show that stability selection with the lasso ranks features poorly in this setting as a result, both theoretically (Corollary \ref{cor.stab.sel2}) and through a simulation study in Section \ref{sims.data}.

\item We also show---both theoretically in Section \ref{sec.theory} (Proposition \ref{prop.gen.pred.risk}) and through simulation studies in Section \ref{sims.data}---that low-noise proxies for important latent features are not only useful predictors, they can be better predictors than directly observed signal features. We show both theoretically in Section \ref{gen.stab.sel} (Proposition \ref{proxies.risk.different.weight}) and through simulations in Sections \ref{avg.sim.study} and \ref{weight.avg.sim.study} that suitably weighted averages of proxies are better predictors than individual proxies.

\item Motivated by these results, in Section \ref{gen.stab.sel} we propose cluster stability selection. We show in the setting of Theorem \ref{thm.result.sel} that cluster stability selection ranks features optimally for predictive performance (Proposition \ref{cor.gen.stab.sel}). Further, we generalize all of stability selection's theoretical guarantees from \citet{shah_samworth_2012} (our Theorems \ref{ss.thm.1} and \ref{ss.thm.2}) and the main theoretical guarantee of \citet{meinshausen-2010}, showing that cluster stability selection enjoys analogous error control properties (even as our proposal selects more features than \citeauthor{meinshausen-2010}'s proposal in the case with no clusters).

\item We propose a novel, stability-based approach to combining cluster members for downstream regression tasks (Section \ref{gen.stab.sel.weight.avg}) motivated by Proposition \ref{proxies.risk.different.weight}. We demonstrate cluster stability selection's superior predictive performance and stability through simulation studies and a real data example in Section \ref{sims.data}.

\end{itemize}

First, in Section \ref{stabsel.sec} we review stability selection in more detail.

\section{Stability Selection}\label{stabsel.sec}

We will discuss both stability selection as proposed by \citet{meinshausen-2010} and the modification proposed by \citet{shah_samworth_2012}. It will be easier to start with the \citeauthor{shah_samworth_2012} formulation. We require a selection procedure
\begin{equation}\label{sel.proced}
\hat{S}_n^\lambda  := \hat{S}_n^\lambda \left(\boldsymbol{X}, \boldsymbol{y} \right),
\end{equation}
typically depending on a tuning parameter \(\lambda \). The only required property of \(\hat{S}_n^\lambda \) is that it maps a data set to a subset of \([p] := \{1, \ldots, p\}\) in a way that can be either random or deterministic conditional on the data. (Sometimes it will be convenient to suppress the \(n\) from the notation and write \(\hat{S}^{\lambda}\), or even suppress the \(\lambda\) if the meaning is clear.) We focus on the lasso with one pre-selected \(\lambda > 0\) as the base procedure. Similarly to \citeauthor{shah_samworth_2012}, we define the selection probability of feature \(j\) under \(\hat{S}_n^\lambda\) as
\begin{equation}\label{sel.prob}
p_{j,n,\lambda} := \mathbb{P} \left(j \in \hat{S}_n^\lambda \right) 
\end{equation}
(where the randomness is with respect to both the data and any randomness in the selection procedure). \citeauthor{shah_samworth_2012} frame these as the primary parameters of interest in stability selection.

As shown in Figure \ref{stabsel_flowchart}, in the first step of stability selection, \(B \) subsamples \(A_1, \ldots, A_B \subset [n]\) of size \(\lfloor n/2 \rfloor\) are drawn, as well as subsamples \( \overline{A}_b \subset [n] \) of the same size with \(A_b \cap \overline{A}_b = \emptyset\).\footnote{That is, if \(n\) is even then \( \overline{A}_b = [n] \setminus A_b\); if \(n\) is odd, one random index is dropped from each \([n] \setminus A_b\) so that \(|A_b| = |\overline{A}_b| = (n-1)/2 = \lfloor n/2 \rfloor\).} For convenience, we abuse notation and write
\[
\hat{S}^\lambda \left( A_b\right) :=  \hat{S}^\lambda \left(\boldsymbol{X}_{A_b \cdot}, \boldsymbol{y}_{A_b} \right),
\]
where \(\boldsymbol{X}_{A_b \cdot}\) denotes the \(\lfloor n/2 \rfloor \times p\) matrix obtained by selecting the rows from \(A_b\) from \(\boldsymbol{X}\), and similarly for \(\boldsymbol{y}_{A_b}\). Then for every \(j \in [p]\), \(p_{j,n,\lambda} \) is estimated by
\begin{equation}\label{ss.ss.est}
\hat{\Pi}_B^{\text{(SS)}}(j) = \frac{1}{2B} \sum_{b=1}^B \left[ \mathbbm{1}\left\{j \in \hat{S}^\lambda \left( A_b \right)\right\} + \mathbbm{1}\left\{j \in \hat{S}^\lambda (\overline{A}_b)\right\} \right].
\end{equation}
(Notice this estimator is unbiased for \(p_{j, \lfloor n/2 \rfloor,\lambda} \) but not necessarily for \(p_{j,n,\lambda} \)). Having computed each \(\hat{\Pi}_B(j) \), the practitioner has available an importance measure for each feature, and, by extension, a ranking of the features (perhaps with ties). Then the selected set consists of all features such that \(\hat{\Pi}_B^{\text{(SS)}}(j) \geq \tau\), with \(\tau \in (0,1)\) a pre-selected threshold.
%To extract a selected set, both \citet{meinshausen-2010} and \citet{shah_samworth_2012} advocate selecting all features such that \(\hat{\Pi}_B(j) \geq \tau\), 
%for each subsample \(A_b\), in addition to estimating \(\hat{S}^\lambda \left( A_b \right)\) they also consider a subsample 

The original proposal by \citet{meinshausen-2010} is more general in the sense that it allows for a finite set of tuning parameters \(\Lambda\). The procedure is otherwise the same except that their theoretical results require using all \(B = \binom{n}{\lfloor n/2 \rfloor}\) unique subsamples of size \(\lfloor n/2 \rfloor\), focusing on the quantity
\[
\hat{\Pi}^{\text{(MB)}}(j) :=  \max_{\lambda \in \Lambda} \left\{\binom{n}{\lfloor n/2 \rfloor}^{-1} \sum_{b=1}^{\binom{n}{\lfloor n/2 \rfloor}} \mathbbm{1}\left\{j \in  \hat{S}^\lambda \left( A_b \right)\right\} \right\}
\]
for each feature \(j\) (though in practice this is infeasible and they suggest that 100 subsamples works well).

\citet{shah_samworth_2012} call their modification that makes use of the sets \(\overline{A}_b\) (and whose theoretical results do not require using all \(\binom{n}{\lfloor n /2 \rfloor}\) subsets) \textit{complementary pairs subsampling} and prove stronger theoretical guarantees than \citeauthor{meinshausen-2010}'s when \(|\Lambda| = 1\). 

%. Then their estimator for \(p_{j,n,\lambda} \) is

%\[
%\hat{\Pi}^{\text{(MB)}}(j) :=  \max_{\lambda \in \Lambda} \left\{\frac{1}{B} \sum_{b=1}^B \mathbbm{1}\left\{j \in  \hat{S}^\lambda \left( A_b \right)\right\} \right\},
%\]
%where \(\Lambda\) is a pre-selected set of tuning parameters. \citet{shah_samworth_2012} propose a modification to this procedure: 
%
%
%
%
%
%
%Note that \citeauthor{shah_samworth_2012} do not use the notion of a set of tuning parameters \(\Lambda\) over which \(\hat{S}^{\lambda}\) is evaluated. We also focus on \(\Lambda\) with \(|\Lambda| = 1\). Nonetheless, 

We emphasize that the success of stability selection hinges on the quality of the rankings it assigns the features. We will frequently draw attention to these rankings.

\section{Theory}\label{sec.theory}

We present theoretical results that will highlight the ``vote-splitting" problem that stability selection has and thereby motivate cluster stability selection. In particular, we will show the following:

\begin{enumerate}[1.]

\item (Theorem \ref{thm.result.sel}.) When a highly correlated group of features is observed, the lasso tends to choose only one of them. In particular, we show that if two equally low-noise proxies for an important latent feature are observed, the lasso tends to select either one of them with equal probability.

\item (Corollary \ref{cor.stab.sel2}.) Because of this ``vote splitting," stability selection with the lasso tends to rank such features lower than directly observed features, even if the signal strength of these directly observed features is smaller.

\item (Corollary \ref{cor.thm.1.risk}.) These low-noise proxies would be better selections for out-of-sample predictive performance, so the ranking yielded by stability selection is detrimental to prediction.

\end{enumerate}

\subsection{The Lasso Selects Randomly Among Highly Correlated Features}\label{sec.theory.lasso.proxies}

In this section we provide theoretical support for the idea that 
\begin{equation}\label{statement.inf}
\textit{the lasso tends to select only one feature within a highly correlated cluster}.
\end{equation}

To the best of our knowledge, no existing theoretical result verifies \eqref{statement.inf}. It is not obvious how to map this informal statement into something that can be proven. For instance, we know that if the penalty parameter \(\lambda\) can be arbitrary, the lasso can typically yield selected sets of any size from 0 to \(\min\{n, p\}\). So if \(p \leq n\), there exist some lasso selected sets in which none of the features in a highly correlated cluster are selected, and some in which all of the features are selected. 

%. Nonetheless, observations like \eqref{statement.inf} have been made many times across the literature. For example, \citet{Zou2005} claim that ``if there is a group of variables among which the pairwise correlations are very high, then the lasso tends to select only one variable from the group and does not care which one is selected." (We also cited several other papers motivated by this phenomenon in the introduction.)

The common belief in statements like \eqref{statement.inf} stems back at least as early as \citet{Efron2004}. \citeauthor{Efron2004} show that the lasso is closely related to their algorithm least angle regression (LARS), which adds features one at a time according to which feature is most highly correlated with the residual between the observed response and the current model. If two features \(\boldsymbol{X}_{\cdot 1}\) and \(\boldsymbol{X}_{\cdot 2}\) are highly correlated, as one of them is added to the active set (say \(\boldsymbol{X}_{\cdot 1}\)), \(\boldsymbol{X}_{\cdot 2}\)'s correlation with the residual will tend to drop more than competitor features. Therefore \(\boldsymbol{X}_{\cdot 2}\) appears to have a disadvantage relative to competitor features.

Following this reasoning, we construct something closer to a provable mathematical statement from \eqref{statement.inf} by considering the path of selected features as the penalty \(\lambda\) decreases from infinity:
\begin{align}
& \textit{After one feature from a highly correlated cluster enters the lasso path, it is} \nonumber
\\ & \textit{(asymptotically) very unlikely that another feature from the cluster will enter until much} \nonumber
\\ & \textit{later in the lasso path (that is, after other reasonable features to select are exhausted).} \label{statement.more.f}
\end{align}

Next we construct a very simple example in which \eqref{statement.more.f} could hold, and in Theorem \ref{thm.result.sel} we prove that it does. Suppose a response \(\boldsymbol{y}\) is observed that is generated from the linear model
\begin{equation} \label{thm.alt.spec.y}
\boldsymbol{y}  = \beta_Z \boldsymbol{Z} +  \boldsymbol{X}_{\cdot 3} + \boldsymbol{\epsilon}, 
\end{equation}
where \(\boldsymbol{Z}, \boldsymbol{X}_{\cdot 3}\), and \(\boldsymbol{\epsilon}\) are independent Gaussian random variables. Also, \(\beta_Z > 1\), so \(\boldsymbol{Z}\) is a more important signal for predicting \(\boldsymbol{y}\) than \(\boldsymbol{X}_{\cdot 3}\). (We will fully specify the setup in a moment.) The practitioner observes \(\boldsymbol{y}\), \(\boldsymbol{X}_{\cdot 3}\), and two noisy proxies \(\boldsymbol{X}_{\cdot 1}\) and \(\boldsymbol{X}_{\cdot 2}\) for \(\boldsymbol{Z}\). In particular, \(\boldsymbol{X}_{\cdot 1}\) and \(\boldsymbol{X}_{\cdot 2}\) are equal to \(\boldsymbol{Z}\) plus a small amount of independent and identically distributed (i.i.d.) noise, \(\boldsymbol{\zeta}_1\) and \(\boldsymbol{\zeta}_2\); that is,
\begin{equation}\label{thm.alt.spec.x}
\boldsymbol{X}_{\cdot j} :=\boldsymbol{Z} + \boldsymbol{\zeta}_{j} , \qquad j \in [2]  .  
\end{equation}
(The amount of noise is ``small" in the sense that \(\Var(\zeta_{11})\) is small.) 

Clearly \(\boldsymbol{X}_{\cdot 1}\) and \(\boldsymbol{X}_{\cdot 2}\) form a highly correlated cluster. We also assume in this setting that \(\beta_Z\) is large enough (and \(\Var(\zeta_{11})\) is small enough) that the features \(\boldsymbol{X}_{\cdot 1}\) and \(\boldsymbol{X}_{\cdot 2}\) are more highly correlated with \(\boldsymbol{y}\) than \(\boldsymbol{X}_{\cdot 3}\) is.\footnote{It turns out it will also be important that \(\beta_Z\) is not \textit{too} large---if it is, then the selected set of size 2 that is best for out-of-sample predictive performance could be \(\boldsymbol{X}_{\cdot 1}\) and \(\boldsymbol{X}_{\cdot 2}\), rather than one of these proxies and \(\boldsymbol{X}_{\cdot 3}\). Then \eqref{statement.more.f} will not hold.} In this setting, either of \(\boldsymbol{X}_{\cdot 1}\) or \(\boldsymbol{X}_{\cdot 2}\) is a good selection if the goal of our model is a sparse selected set with good out-of-sample predictive performance (particularly since \(\boldsymbol{Z}\) is not observed), but selecting both is redundant. That is, the best selected sets of size 2 are \(\{\boldsymbol{X}_{\cdot 1}, \boldsymbol{X}_{\cdot 3}\}\) or \(\{\boldsymbol{X}_{\cdot 2}, \boldsymbol{X}_{\cdot 3}\}\). Theorem \ref{thm.result.sel} says that with high probability the lasso path matches this behavior: after selecting one of \(\boldsymbol{X}_{\cdot 1}\) or \(\boldsymbol{X}_{\cdot 2}\), the next feature to enter the lasso path is \(\boldsymbol{X}_{\cdot 3}\).

Now we specify the setting precisely. Suppose \(n\) i.i.d. draws 
\begin{equation}\label{thm.alt.spec}
(Z_i, X_{i3}, \epsilon_i, \zeta_{i1}, \zeta_{i2}), \qquad i \in [n],
\end{equation}
 are observed, with the variables having a multivariate Gaussian distribution. Each variable has mean 0 and is independent from the others, and we have \(\Var(Z_1) = \Var(X_{13}) = 1\), \(\Var(\epsilon_1) = \sigma_\epsilon^2\), and \(\Var(\zeta_{11}) = \Var(\zeta_{12}) = \sigma_\zeta^2(n) \). Denote \( \boldsymbol{Z} := (Z_1, \ldots, Z_n)^\top, \boldsymbol{X}_{\cdot 3} := (X_{13}, \ldots, X_{n3})^\top\), \(\boldsymbol{\epsilon} := (\epsilon_1, \ldots, \epsilon_n)^\top\), and \(\boldsymbol{\zeta}_j = (\zeta_{1j}, \ldots, \zeta_{nj})^\top\), \(j \in [2]\). We will require \(\sigma_\zeta^2(n)\) to vanish at a particular rate as \(n \to \infty\), so that \(\boldsymbol{X}_{\cdot 1}\) and \(\boldsymbol{X}_{\cdot 2}\) approach \(\boldsymbol{Z}\) (and the correlation of \(\boldsymbol{X}_{\cdot 1}\) and \(\boldsymbol{X}_{\cdot 2}\) approaches 1). In particular, let
\begin{equation}\label{lasso.thm.sig.zeta.def}
\sigma_\zeta^2(n) := \frac{10}{\sqrt{n \log n}}.
\end{equation}
%Lastly, the quantity
%\begin{equation}\label{def.big.delta}
%\Delta_{\text{max}} := \frac{19\left(\log 100 \right)^{3/4}}{100}  
%\end{equation}
%is part of a uniform upper bound on \(\beta_Z\) under our assumptions. 
We will consider the lasso with scaled features \citep{Tibshirani1996}:
\begin{equation}\label{lasso.solution}
\hat{\beta}(\lambda) \in \underset{\beta \in \mathbb{R}^p}{\arg \min} \left\{ \frac{1}{2n} \left\lVert \boldsymbol{y} - \sum_{j=1}^p \frac{\boldsymbol{X}_{\cdot j}}{\lVert \boldsymbol{X}_{\cdot j} \rVert_2} \beta_j \right\rVert_2^2 + \lambda \sum_{j=1}^p  |\beta_j| \right\}.
\end{equation}

\begin{theorem}\label{thm.result.sel} Let \(\boldsymbol{X}_{\cdot j}\), \(j \in [3]\) and \(\boldsymbol{y}\) be as defined in \eqref{thm.alt.spec.y} -- \eqref{thm.alt.spec}. For constants \(t_0 \in (0,1]\) (defined in Lemma \ref{lemma.ncvx.cor}) and \(c_2 \in \left(0, \frac{e-1}{8e^2}\right)\) (defined in Equation \ref{defn.c2}), assume \(n \geq 100\) is large enough to satisfy
%\begin{equation}\label{n.exp.size.req.max}
%n > \exp \left\{ \left(   \frac{100}{19 \log (100)}\sqrt{\frac{c_2}{2 + \sigma_\epsilon^2}}  +  \frac{10}{19} \sqrt{ 3 + \sigma_\epsilon^2 } \right)^4 \right\},
%\end{equation}
%%%% This is the "scariest-looking" assumption, but to see that it's reasonable in terms of the sample size required, see https://www.desmos.com/calculator/qohz3hnzju.
\begin{equation}\label{n.large.delta.cond.max}
\frac{n}{\log n} > \frac{5 + \sigma_\epsilon^2}{c_2} \cdot
\max \left\{ \frac{1}{4 t_0^2 (2 + \sigma_\epsilon^2 )^2}, 2 \left(12 + \sigma_\epsilon^2\right), 5\left(1 + \sigma_\epsilon^2 \right) \right\}
%%%%% Note: to see why we can't do without any of these conditions, examine https://www.desmos.com/calculator/piz14ehzfo.
\end{equation}
and
\begin{equation}\label{c4n.conds.max}
 \frac{n}{\left(\log n \right)^{3/2}}  > 3.61 \cdot \frac{5 + \sigma_\epsilon^2}{c_2} .
\end{equation}
Then
\begin{enumerate}[(i)]
\item the interval 
% \begin{equation}\label{samp.size.beta.min}
% n \log n  > \frac{9801}{\left(\beta_Z - 1\right)^2}.
% \end{equation}
% and
 \begin{equation}\label{cond.beta.z}
I(n) = \left( 1 + 10 \sigma_\zeta^2(n) , 1 +  \frac{19}{10}\sqrt{\frac{ 2 + \sigma_\epsilon^2 }{c_2}}   \frac{\left(\log n \right)^{3/4}}{n^{1/2}} \right)
% \left( \beta_Z - 1 \right)^2 \left( \beta_Z^2 + 1 + \sigma_\epsilon^2 \right)  < 2 + \sigma_\epsilon^2.
%\beta_Z \in \left(1 +  10 \cdot \sigma_\zeta^2(n) ,1 + \frac{19}{10} \sqrt{ \frac{2 + \sigma_\epsilon^2 }{ c_2 }}
% \frac{\left(\log n \right)^{3/4}}{n^{1/2}}\right]
% \underbrace{\frac{\sqrt{ 1 + \sigma_\zeta^2(n)} }{ 1 - 2 \delta(n) \sqrt{ 3 + \sigma_\epsilon^2 }}}_{*} < \beta_Z <  \underbrace{\frac{5}{5 -19 \left(\log n \right)^{1/4} \delta(n)}}_{**},
\end{equation}
is a nonempty subset of \((1,2)\), and
%(It turns out that these conditions imply \(\beta_Z\) is at most \( \beta_Z^{\text{max}} \) for all \(n\) satisfying our assumptions.) Then 
\item for any \(\beta_Z \in I(n)\), for the lasso path calculated as in \eqref{lasso.solution} there exists a finite constant \(c_{3} > 0\) (free of all parameters in the setup) such that the first two features to enter the lasso path are \(\boldsymbol{X}_{\cdot 1}\) followed by \(\boldsymbol{X}_{\cdot 3}\) with probability at least
\begin{align*}
 &   \frac{1}{2} -  c_{3}  \left(\underbrace{\beta_Z^2 + 1 + \sigma_\epsilon^2}_{= \Var(y_1)} \right)^{7/2} \frac{\left(\log n \right)^{3/2}}{n^{1/4}} .
 \end{align*}
%(In Lemma \ref{very.first.lemma}, we show that \eqref{cond.beta.z} is feasible in the sense that term \(**\) is strictly greater than term \(*\) for \(n \geq 100\).)
\end{enumerate}

\end{theorem}

\begin{proof}See Appendix \ref{thm.result.sel.proof}.
\end{proof}
By exchangeability, the event where the first two features to enter the lasso path are \(\boldsymbol{X}_{\cdot 2}\) followed by \(\boldsymbol{X}_{\cdot 3}\) occurs with equal probability. Therefore this result implies that for any \(\lambda\) between the second and third knot of the lasso path, with high probability the two selected features will be the weak signal feature \(\boldsymbol{X}_{\cdot 3}\) and one of the low-noise proxy features \(\boldsymbol{X}_{\cdot 1}\) or \(\boldsymbol{X}_{\cdot 2}\). 
%Therefore the probability that \(\boldsymbol{X}_{\cdot 3}\) is selected second is at least twice this probability.

%\subsubsection{Consequences for Stability Selection}

Theorem \ref{thm.result.sel} leads us to the problem with stability selection using the lasso in this setting: the probability that any one proxy will be chosen by the lasso is lower than the probability of choosing the weak signal feature. We summarize this observation in the following result.

\begin{corollary}\label{cor.stab.sel2}
%\ref{thm.result.sel}
In the setting of Theorem \ref{thm.result.sel}, consider applying stability selection with the base procedure defined as follows: the selected set on each iteration is the first two features to enter the lasso path. Assume \(\beta_Z\) satisfies the assumptions of Theorem \ref{thm.result.sel} with \(\lfloor n/2 \rfloor\) large enough to satisfy the sample size requirements of Theorem \ref{thm.result.sel}. Then there exists a constant \(c_4 > 0\) defined in \eqref{def.c14} such that
\[
\E \left[ \hat{\Pi}^{\text{(MB)}}(j)  \right] = \E \left[ \hat{\Pi}_B^{\text{(SS)}}(j)  \right]  \leq  \frac{1}{2} + c_4 \left(\beta_Z^2 + 1 + \sigma_\epsilon^2 \right)^{7/2} \frac{\left(\log n \right)^{3/2}}{n^{1/4}}
, \qquad j \in [2],
\]
and
\[
\E \left[ \hat{\Pi}^{\text{(MB)}}(3)  \right] = \E \left[ \hat{\Pi}_B^{\text{(SS)}}(3)  \right]     \geq  1 - 2 c_4 \left(\beta_Z^2 + 1 + \sigma_\epsilon^2 \right)^{7/2} \frac{\left(\log n \right)^{3/2}}{n^{1/4}}.
\]

\end{corollary}

\begin{proof} See Appendix \ref{other.other.proofs}. \end{proof}

That is, as \(n \to \infty\) the selection proportion for \(\boldsymbol{X}_{\cdot 3}\) yielded by stability selection tends to 1 and the selection proportions of \(\boldsymbol{X}_{\cdot 1}\) and \(\boldsymbol{X}_{\cdot 2}\) tend towards \(1/2\), so stability selection tends to order the features sub-optimally for out-of-sample predictive performance.

\subsection{Low-Noise Proxies Are Good Selections For Prediction}\label{sec.prox.impt}

A practitioner who is strictly interested in features that appear in the true data-generating process may be uninterested in low-noise proxies for important latent features, but we will show that they are useful for out-of-sample predictive performance. 

We will make use of a more general data-generating process than \eqref{thm.alt.spec.y}. Suppose \(q\) proxies are observed, 
\begin{equation}\label{thm.alt.spec.x.gen}
\boldsymbol{X}_{\cdot j} :=\boldsymbol{Z} + \boldsymbol{\zeta}_{j} , \qquad j \in [q]  ,
\end{equation}
with (possibly different) noise variances: \( \boldsymbol{\zeta}_{j} \sim \mathcal{N} \left(0, \sigma_{\zeta j}^2 \boldsymbol{I}_n \right), j \in [q]\). Suppose
\begin{equation}\label{thm.alt.spec.gen}
(Z_i, X_{i, q+1}, \ldots, X_{i p}, \epsilon_i, \zeta_{i1}, \ldots, \zeta_{iq}), \qquad i \in [n]
\end{equation}
are independent Gaussian random variables, and
\begin{equation}\label{thm.alt.spec.y.gen}
\boldsymbol{y}  = \beta_Z \boldsymbol{Z} + \sum_{j = q+1}^p \beta_j \boldsymbol{X}_{\cdot j} + \boldsymbol{\epsilon},
\end{equation}
with \(\beta_{q+1}, \ldots, \beta_p \in \mathbb{R}\). Now we will define a notion of prediction risk that we will use to compare features as selections.

\begin{definition}[Prediction risk of a single feature]\label{def.risk.ind}

Assume the setup of \eqref{thm.alt.spec.x.gen}, \eqref{thm.alt.spec.gen}, and \eqref{thm.alt.spec.y.gen}. Let \(\boldsymbol{\tilde{Z}}, \boldsymbol{\tilde{X}}_{\cdot q + 1}, \ldots, \boldsymbol{\tilde{X}}_{\cdot p}, \tilde{\boldsymbol{\epsilon}}\), and \(\tilde{\boldsymbol{\zeta}}_1, \ldots, \tilde{\boldsymbol{\zeta}}_q\) be i.i.d. copies of the corresponding variables, and define \(\tilde{\boldsymbol{y}}\) and \(\tilde{\boldsymbol{X}}_{\cdot 1}, \ldots, \tilde{\boldsymbol{X}}_{\cdot q}\) analogously. The \textbf{prediction risk} \(R(j)\) of a model using only feature \(j\) is the expected out-of-sample mean squared error of the predictions:
\[
R(j) := \mathbb{E} \left[ \frac{1}{n}  \left\lVert \tilde{\boldsymbol{y}} - \hat{\beta}_j \tilde{\boldsymbol{X}}_{\cdot j}  \right\rVert_2^2\right],
\]
where \(\hat{\beta}_j =  \left. \boldsymbol{X}_{\cdot j}^\top \boldsymbol{y} \middle/ \boldsymbol{X}_{\cdot j}^\top \boldsymbol{X}_{\cdot j} \right. \) is the ordinary least squares (OLS) coefficient.

\end{definition} 
We show that proxies are better selections than directly observed features (in the sense of reducing this prediction risk) if the signal strength of the latent feature is large enough and the noise added to the latent feature is not too large.
\begin{proposition}\label{prop.gen.pred.risk}
Assume\footnote{It is worth noting that this result readily generalizes further to cases of arbitrary numbers of latent signals and so on, but this becomes both notationally inconvenient and beyond the scope of what is needed in this paper.} the setup of \eqref{thm.alt.spec.x.gen}, \eqref{thm.alt.spec.gen}, and \eqref{thm.alt.spec.y.gen}. Then for any \(j \in [q]\) and any \(k \in \{q+1, \ldots, p\}\), 
\[
R(j) < R(k) \qquad \iff \qquad  \frac{\beta_Z^2}{ \beta_{k}^2 } >  1+ \sigma_{\zeta j}^2  .
\]
\end{proposition}

Proposition \ref{prop.gen.pred.risk} implies that in the setting of Theorem \ref{thm.result.sel}, \(\boldsymbol{X}_{\cdot 1}\) or \(\boldsymbol{X}_{\cdot 2}\) is a better choice for a model of size 1 than \(\boldsymbol{X}_{\cdot 3}\):
\begin{corollary}\label{cor.thm.1.risk}
Under the assumptions of Theorem \ref{thm.result.sel}, \(R(1) =R(2) < R(3)\).
\end{corollary}
(For the proofs of these results, see Appendix \ref{risk.proofs}.) Together, these results show that in the setting of Theorem \ref{thm.result.sel}, stability selection provides a worse ranking of the features for predictive performance than the lasso. In the next section we will introduce our new method and show that it performs better in this setting.
\section{Cluster Stability Selection}\label{gen.stab.sel}

In this section we detail our proposed procedure, \textit{cluster stability selection}.

\subsection{Description of Method}

 Although our procedure allows for responses in general spaces \(\mathcal{Y}\) and any base feature selection method that can be characterized as in \eqref{sel.proced}, we focus on \(\mathcal{Y} = \mathbb{R}\) and using the lasso for feature selection. 

Our procedure requires a partitioning of the features into clusters \(\mathcal{C} = \{C_1, \ldots, C_K\}\)
%\[
%\mathcal{C} = \{C_1, \ldots, C_K\}, \qquad \text{where} \qquad \bigcup_{k=1}^K C_k = [p] \qquad \text{and} \qquad C_k \cap C_{k'} = \emptyset \ \forall k \neq k'
%\]
as an input. The clusters may be known from domain knowledge, but if not, they can be estimated by methods including hierarchical clustering as well as those methods proposed by \citet{Bondell2008, She2010, Shen2010, Buhlmann2013, Sharma} and \citet{Witten2014}. 

In brief, cluster stability selection returns a ranking of clusters rather than a ranking of features. We use complementary pairs subsampling similar to the estimator \eqref{ss.ss.est} proposed by \citet{shah_samworth_2012}, except that we allow for an arbitrary finite set of parameters \(\Lambda\). We calculate the individual selection proportions for each feature \(j \in [p]\) for any \(\lambda \in \Lambda\)
\begin{equation}\label{sample.prop}
\hat{\Pi}_B(j) := \frac{1}{2B} \sum_{b=1}^B \left[ \mathbbm{1}\left\{   j \in \bigcup_{\lambda \in \Lambda} \hat{S}^\lambda \left( A_b \right) \right\} + \mathbbm{1}\left\{  j \in  \bigcup_{\lambda \in \Lambda} \hat{S}^\lambda (\overline{A}_b)\right\} \right] ,
\end{equation}
%which is equal to \(\hat{\Pi}_B^{\text{(SS)}}(j) \) when \(| \Lambda| = 1\).
and for every \(k \in [K]\) we calculate the proportion of subsets in which at least one feature from \(C_k\) is selected for at least one \(\lambda \in \Lambda\):
\begin{equation}\label{theta.hat.exp}
\hat{\Theta}_B(C_k) :=  \frac{1}{2B} \sum_{b=1}^B \left[  \mathbbm{1} \left\{    C_k \cap \bigcup_{\lambda \in \Lambda} \hat{S}^{\lambda}\left(A_b \right) \neq \emptyset   \right\} +  \mathbbm{1} \left\{   C_k \cap \bigcup_{\lambda \in \Lambda} \hat{S}^{\lambda}\left(\overline{A}_b \right) \neq \emptyset  \right\} \right]  .
\end{equation} 
Having estimated importance measures \(\hat{\Theta}_B(C_k)\) for each cluster, we construct \textit{cluster representatives} \(\boldsymbol{X}_{\cdot C_k}^{\text{rep}}\) for regression tasks by taking weighted averages of the cluster members,
\begin{equation}\label{synth.proxy.def}
\boldsymbol{X}_{\cdot C_k}^{\text{rep}}:=  \sum_{j \in C_k } w_{kj} \boldsymbol{X}_{\cdot j} ,
\end{equation}
where \(\boldsymbol{w}_k = \begin{pmatrix} w_{k j} \end{pmatrix}_{j \in C_k} \in \Delta^{|C_k| - 1}\) are weights constructed for each cluster from the individual feature selection proportions \(\hat{\Pi}_B(j)\) in one of the following ways:

\begin{itemize}

\item \textbf{Weighted averaged cluster stability selection:}
\begin{equation}\label{weight.avg.weights}
w_{kj} = \frac{\hat{\Pi}_B(j)}{\sum_{j' \in C_k} \hat{\Pi}_B(j')} \qquad \forall j \in C_k  .
\end{equation}

\item \textbf{Simple averaged cluster stability selection:}  

\begin{equation}\label{simp.avg.weights}
w_{kj} = \frac{1}{\left| C_k \right|}  \qquad \forall j \in C_k .
\end{equation}

\item \textbf{Sparse cluster stability selection:} For each \(j \in C_k \), 
\begin{equation}\label{sparse.weights}
w_{kj} = \left. \mathbbm{1} \left\{ j \in \underset{j' \in C_k}{\arg \max}  \left\{ \hat{\Pi}_B(j') \right\} \right\} \middle/ \left| \underset{j' \in C_k}{\arg \max}  \left\{ \hat{\Pi}_B(j')  \right\} \right| \right. .
\end{equation}

\end{itemize}
Briefly, the first proposal assigns weights to each feature in the cluster in proportion to how frequently they were selected, the second proposal assigns equal weight to each feature in the cluster, and the third proposal assigns equal weight to all of the most frequently selected individual cluster members, and 0 weight to the rest (in particular, if one cluster member was selected most frequently with no ties, weight 1 is assigned to that cluster member.) We interpret cluster stability selection as having rejected any features that are assigned weight 0 in whatever weighting scheme is used. In Section \ref{sec.weight.choices} we will discuss each of the weighting schemes in more detail.

Because each cluster member is a representation of the same latent signal, estimating a regression coefficient for each individual selected cluster member separately would result in extra degrees of freedom that might hurt predictive performance. Instead, we interpret \(\boldsymbol{X}_{\cdot C_k}^{\text{rep}}\) as representing our best approximation of the lurking signal and recommend using this one feature rather than the raw selected features.

Finally, to obtain a selected set, one could pre-select a threshold selection proportion \(\tau\) \citep{meinshausen-2010, shah_samworth_2012}, or the practitioner could  interpret the output as a set of at most \(\min\{K, 2B + 1\}\) candidate selected sets defined by 
\[
\bigcup_{k \in \mathcal{S}_i} C_k \qquad \text{where} \qquad \mathcal{S}_i  := \left\{k : \hat{\Theta}_B(C_k) \geq 1 -  \frac{i}{2B} \right\}  , \qquad i \in \{0, \ldots, 2B\}.
\]
In particular, one could pre-specify a desired model size \(s\) and use the features with the \(s\) highest sample proportions \citep{Kim2019}.

Note that \(\mathcal{C} = \{ \{1\}, \ldots, \{p\}\}\) is within our framework. In this special case, our proposal reduces to a a more powerful modified version of stability selection in its handling of the set $\Lambda$, in the sense that
\begin{align*}
& \frac{1}{2B} \sum_{b=1}^B \left[  \mathbbm{1} \left\{   j \in  \bigcup_{\lambda \in \Lambda}\hat{S}^{\lambda}\left(A_b \right)  \right\}  +  \mathbbm{1} \left\{   j \in \bigcup_{\lambda \in \Lambda} \hat{S}^{\lambda}\left(\overline{A}_b \right) \right\} \right] 
\\ \geq ~ & \max_{\lambda \in \Lambda} \left\{ \frac{1}{2B} \sum_{b=1}^B \left[  \mathbbm{1} \left\{   j \in \hat{S}^{\lambda}\left(A_b \right)  \right\} +  \mathbbm{1} \left\{ j \in \hat{S}^{\lambda}\left(\overline{A}_b \right) \right\} \right]  \right\},
\end{align*}
so for a fixed threshold \(\tau\), more features will exceed \(\tau\) in our proposal than in a complementary pairs subsampling version of the \citet{meinshausen-2010} proposal. 

Algorithm \ref{css.algorithm} spells out the procedure.

\begin{algorithm}[H]\label{css.algorithm}
 \SetKwInOut{Input}{input}\SetKwInOut{Output}{output}
\Input{Data \((\boldsymbol{X}, \boldsymbol{y}) \in \mathbb{R}^{n \times p} \times \mathcal{Y}^n\);
 \\ a selection method \(\hat{S}^{\lambda}\), as in \eqref{sel.proced}; 
\\  a set of tuning parameters \(\Lambda\);
\\ a number of subsamples \(B \in \mathbb{N}\); 
\\ clusters \(\mathcal{C} = \{C_1, \ldots, C_K\}\) partitioning \([p]\); 
\\ a selected weighting scheme; one of \eqref{weight.avg.weights}, \eqref{simp.avg.weights}, or \eqref{sparse.weights};
\\  (optionally) a selection threshold \(\tau\);
%\\  (optionally) a desired number of selected clusters \(s\).
}

 initialization\;
 \For{$b\leftarrow 1$ \KwTo $B$}{
\(A_b \leftarrow \text{a random subsample of } [n] \text{ of size } \lfloor n/2 \rfloor  \) \;
\(\overline{A}_b \leftarrow \text{a random subsample of } [n] \setminus A_b \text{ of size } \lfloor n/2 \rfloor  \) \;
compute \(\hat{S}^{\lambda}(A_b)\) and \(\hat{S}^{\lambda}(\overline{A}_b)\) \;
}
 \For{$j\leftarrow 1$ \KwTo $p$}{
calculate \(\hat{\Pi}_B(j)\) as in \eqref{sample.prop} \;
}
 \For{$k\leftarrow 1$ \KwTo $K$}{
calculate \(\hat{\Theta}_B(C_k)\) as in \eqref{theta.hat.exp} \;
calculate the weights \(\boldsymbol{w}_k\) using the chosen weighting scheme \;
compute the cluster representative \(\boldsymbol{X}_{\cdot C_k}^{\text{rep}}\) as in \eqref{synth.proxy.def} \;
\(\hat{C}_k \leftarrow \left\{j \in C_k: w_{kj} \neq 0 \right\} \) \;
}
\uIf{\(\tau\) is provided}{
\( \mathcal{S}  \leftarrow \{k : \hat{\Theta}_B(C_k) \geq \tau \}  \) \;
%\uElseIf{\(s\) is provided}{
% \For{$i\leftarrow 0$ \KwTo $2B$}{
%\( \mathcal{S}  \leftarrow \left\{k : \hat{\Theta}_B(C_k) \geq 1 -  \frac{i}{2B} \right\}  \) \;
%\If{\( \left| \mathcal{S} \right| \geq s\) }{\textbf{break} \;}
%}
%%\( \mathcal{S}  \leftarrow \{k : k \in [K], \hat{\Theta}_B(C_k) \geq \tau \}  \) \;
%}
\Else{
\(\mathcal{S} \leftarrow [K]\) \;
}
}
\Output{\( \left\{ \hat{C}_k,  \hat{\Theta}_B(C_k), \boldsymbol{X}_{\cdot C_k}^{\text{rep}} \right\}_{k \in \mathcal{S}}  \) } 
 \caption{Cluster Stability Selection}
\end{algorithm}

We show next that cluster stability selection has more desirable behavior in the context of Theorem \ref{thm.result.sel} than stability selection. In particular, we will show that we pay no price in guaranteed error control for the increased power our procedure enjoys, and we will compare our method to \citeauthor{meinshausen-2010} and \citeauthor{shah_samworth_2012}'s proposals in more detail.

\subsection{Properties of Cluster Stability Selection}\label{method.props.sec}

Recall that the primary motivation of cluster stability selection is to improve on stability selection in the case of highly correlated features. Corollary \ref{cor.thm.1.risk} suggests that in the setting of Theorem \ref{thm.result.sel}, features 1 and 2 are both better first selections than feature 3. Corollary \ref{cor.stab.sel2} shows that stability selection unfortunately tends to rank feature 3 ahead of feature 1 or 2. In contrast, we provide evidence that cluster stability selection tends to rank the features correctly in this setting:
\begin{proposition}\label{cor.gen.stab.sel}

In the setting of Theorem \ref{thm.result.sel}, consider applying cluster stability selection with the cluster assignments
\[
C_1 =\{1, 2\}, \qquad C_2 = \{3\}.
\]
Define the base procedure \(\hat{S}\) as follows: the selected set on each iteration is the first two features to enter the lasso path. Assume \(\beta_Z\) satisfies the assumptions of Theorem \ref{thm.result.sel} with \(\lfloor n/2 \rfloor\) large enough to satisfy the sample size requirements of Theorem \ref{thm.result.sel}. Then
\[
\E \left[ \hat{\Theta}_{B} \left(C_1 \right)  \right]  \geq \E \left[ \hat{\Theta}_{B} \left(C_2 \right)  \right] 
\]
for \(n\) sufficiently large. 

The result also holds if the base procedure is defined to select the first feature to enter the lasso path.

\end{proposition}

\begin{proof} See Appendix \ref{other.other.proofs}.
  \end{proof}

Additionally, we will show in Theorem \ref{ss.thm.1} that cluster stability selection retains analogous error control guarantees to those provided by both \citet{meinshausen-2010} and \citet{shah_samworth_2012}. For a finite set of tuning parameters \(\Lambda\), define the set of selected clusters (that is, the set of clusters from which at least one feature was selected by \(\hat{S}_n^\lambda\) for at least one \(\lambda \in \Lambda\))
\[
\hat{S}_n^{\Lambda; \mathcal{C}}   : = \left\{C \in \mathcal{C} : C \cap  \bigcup_{\lambda \in \Lambda} \hat{S}_n^{\lambda} \neq \emptyset \right\} ,
\]
where \(\mathcal{C} = \{C_1, \ldots, C_K\}\) is a partitioning of \([p]\). We generalize the selection probability for a single feature \eqref{sel.prob} and define 
\[
p_{C, n, \Lambda} := 
%\mathbb{P} \left( \bigcup_{\lambda \in \Lambda} \left\{ C \cap \hat{S}_{n}^\lambda \neq \emptyset \right\} \right) = 
\mathbb{P} \left( C \in \hat{S}_{n}^{\Lambda; \mathcal{C}} \right).
\]
% \(p_{C_k, n, \lambda}\) be the probability any feature in \(C_k\) is selected by \(\hat{S}_n ^\lambda \).
 For a fixed \(\theta \in [0,1]\), let \(L_\theta := \left\{C_k: p_{C_k, \lfloor n/2 \rfloor, \Lambda} \leq \theta \right\}\) denote the set of clusters that have low selection probability under \(\hat{S}_{\lfloor n/2 \rfloor}^{\Lambda; \mathcal{C}}\), and let \(H_\theta := \left\{C_k: p_{C_k, \lfloor n/2 \rfloor, \Lambda} > \theta \right\}\) denote the set of clusters with high selection probability. Denote by \( \hat{S}_{n, \tau}^{\text{CSS}; \Lambda, \mathcal{C}} \subseteq \mathcal{C} \) the set of clusters selected by cluster stability selection in this setting, using a pre-selected \(\tau\).
%First we generalize all of the theoretical guarantees provided by \citet{shah_samworth_2012} for stability selection to the case of cluster stability selection (when using a pre-selected threshold \(\tau\) to select clusters):
%Generalization of Theorem 1 of \citealt{shah_samworth_2012}
\begin{theorem}\label{ss.thm.1}

\begin{enumerate}[(i)]

\item  For \(\tau \in (1/2, 1]\),
\[
\E \left| \hat{S}_{n, \tau}^{\text{CSS}; \Lambda, \mathcal{C}} \cap L_\theta  \right| \leq 
\frac{\theta}{2 \tau -1} \E \left|  \hat{S}_{\lfloor n/2 \rfloor}^{\Lambda; \mathcal{C}} \cap L_\theta \right| .
\]

\item Let \(\hat{N}_{n, \tau}^{\text{CSS}; \Lambda, \mathcal{C}}  = \mathcal{C} \setminus \hat{S}_{n, \tau}^{\text{CSS}; \Lambda, \mathcal{C}}\) and let \(\hat{N}_n^{\Lambda; \mathcal{C}} :=  \mathcal{C}  \setminus \hat{S}_n^{\Lambda; \mathcal{C}}\). For \(\tau \in [0, 1/2)\),
%the expected number of high-probability clusters not selected by cluster stability selection is at most
\[
\E \left| \hat{N}_{n, \tau}^{\text{CSS}; \Lambda, \mathcal{C}}  \cap H_\theta  \right| \leq 
\frac{1 - \theta}{1 - 2 \tau} \E \left|  \hat{N}_{\lfloor n/2 \rfloor}^{\Lambda; \mathcal{C}} \cap H_\theta \right| .
\]

\end{enumerate}

\end{theorem}

%\begin{proof}
%
%See Appendix \ref{proof.ss.thm.1} in the online supplement.
%
%\end{proof}
We see from part (i) of Theorem \ref{ss.thm.1} that for suitable choices of \(\theta\) and \(\tau\) such that \(\theta/(2 \tau -1) < 1\), the expected number of features selected from clusters in \(L_\theta\) (that is, low-selection-probability clusters) is strictly smaller under cluster stability selection than under the base procedure \( \hat{S}_{\lfloor n/2 \rfloor}^{\Lambda; \mathcal{C}}\). Similarly, part (ii) shows that cluster stability selection controls the number of high-selection-probability clusters that are missed by the base procedure.

In Appendix \ref{proof.ss.thm.1} we also present Theorem \ref{ss.thm.2}, which similarly generalizes Theorem 2 and Equations 7 and 8 of \citet{shah_samworth_2012} to our setting. These results provide tighter bounds under stronger assumptions on the distributions of the selection proportions yielded by cluster stability selection. We omit proofs for Theorems \ref{ss.thm.1} and \ref{ss.thm.2} because they can be proven using identical arguments used to prove the corresponding results from \citet{shah_samworth_2012} by replacing indicators of features being selected with indicators of clusters being selected (we spell this point out in more detail in Appendix \ref{proof.ss.thm.1}).

Note that Theorems \ref{ss.thm.1} and \ref{ss.thm.2} are valid for \textit{any} partitioning of \([p]\) into clusters, though the results are more interesting if the clusters are meaningful. Also, these results generalize the results from the original papers in the sense that choosing \((C_1, \ldots, C_p ) = (\{1\}, \ldots, \{p\})\) recovers the original results.

\begin{remark}

%We can achieve a bound that is analogous to the one provided by Theorem 1 of \citet{meinshausen-2010} by applying Theorem \ref{ss.thm.1}(i) with \(\theta = q_\Lambda/K\):
%\[
%\E \left| \hat{S}_{n, \tau}^{\text{CSS}; \Lambda, \mathcal{C}} \cap L_{q_\Lambda/K}  \right| 
%\leq  \frac{1}{2 \tau -1}  \frac{q_\Lambda}{K} \E \left|  \hat{S}_{\lfloor n/2 \rfloor}^\lambda \cap L_{q_\Lambda/K} \right|
% \leq \frac{1}{2 \tau -1}  \frac{q_\Lambda}{K} \E \left|  \hat{S}_{\lfloor n/2 \rfloor}^\lambda  \right| 
% = \frac{1}{2 \tau - 1} \frac{q_{\Lambda}^2}{K}.
%\]
The discussion in the last paragraph of Section 3.1 of \citet{shah_samworth_2012} about how the guarantee of their Theorem 1(a) is at least as good as the \citet{meinshausen-2010} Theorem 1 guarantee, even when \(B = 1\), applies to our Theorem \ref{ss.thm.1}(a) as well in the special case \(\mathcal{C} = \{ \{1\}, \ldots, \{p\}\}\). We also point out that this does not require the same assumptions that (i) the base procedure is `no worse than random guessing" (in the sense that the average selection probability of the ``signal features" \(S \subset [p]\) under \(\hat{S}_{\lfloor n/2 \rfloor}^{\Lambda; \mathcal{C}}\) is at least as large as the average selection probability of the ``noise features" \(N = [p] \setminus S\)) and (ii) the distribution of the noise features is exchangeable; we only require the weaker assumption that \(N \subseteq L_{\theta}\) for \(\theta = \E \left| \hat{S}_{\lfloor n/2 \rfloor}^{\Lambda; \mathcal{C}} \right|/p \).

\end{remark}

%Next we briefly discuss some of the properties of each weighting proposal.

\subsection{Choice of Weights}\label{sec.weight.choices}

The intuition behind regressing against weighted averages of cluster members is that averaging several noisy proxies might improve our approximation of the latent signal. Proposition \ref{proxies.risk.different.weight} below sheds some light on how this might be done optimally. First we will need a notion of prediction risk for cluster representatives.
\begin{definition}[Prediction risk of cluster representatives]\label{def.risk.synth}
Assume the same setup as Definition \ref{def.risk.ind}. For a cluster \(C \subset [p]\) and arbitrary weights \(\boldsymbol{w} \in \Delta^{|C| - 1}\), recall the definition of the cluster representative \eqref{synth.proxy.def}, and define \(\tilde{\boldsymbol{X}}_{\cdot C}^{\text{rep}}\) as an out-of-sample draw (using the same weights). 
Define the prediction risk of a model containing only a cluster representative constructed using these weights as
\begin{equation}\label{pred.risk.synth.proxy.def}
R(C; \boldsymbol{w} ) := \mathbb{E} \left[  \left\lVert \tilde{\boldsymbol{y}} - \hat{\beta} \tilde{\boldsymbol{X}}_{\cdot C}^{\text{rep}}  \right\rVert_2^2\right],
\end{equation}
where \(\hat{\beta} =  \left. \left( \boldsymbol{X}_{\cdot C}^{\text{rep}} \right)^\top
  \boldsymbol{y} \middle/  \left( \boldsymbol{X}_{\cdot C}^{\text{rep}}\right)^\top \boldsymbol{X}_{\cdot C}^{\text{rep}}  \right.\).
\end{definition}
\begin{proposition}\label{proxies.risk.different.weight}

Assume the setup of \eqref{thm.alt.spec.x.gen}, \eqref{thm.alt.spec.gen}, and \eqref{thm.alt.spec.y.gen}. Consider regressing \(\boldsymbol{y}\) against a weighted average of the \(q\) proxies, as in \eqref{synth.proxy.def}. Then
\begin{enumerate}[(i)]
\item
the weights that minimize the prediction risk defined in \eqref{pred.risk.synth.proxy.def} are
\[
w_j^*  = \left.  \frac{1}{ \sigma_{\zeta j}^2}  \middle/ \sum_{j' = 1}^q \frac{1}{ \sigma_{\zeta j'}^2} \right., \qquad j \in [q],
\]
\item the minimum risk (using the optimal weights \(\boldsymbol{w}^* = \{w_j^*\}_{j \in [q]}\)) is
\[
\mathcal{E}_{\text{ideal}} + \frac{n - 1 }{n-2} \cdot \frac{\beta_Z^2   }{1+  \sum_{j=1}^q   \frac{1}{ \sigma_{\zeta j}^2} }    ,
\]
where \( \mathcal{E}_{\text{ideal}}\), defined in \eqref{dist.cond.z.e.mse.result}, is the prediction risk of regressing against \(\boldsymbol{Z}\) directly, and
\item 
for any \(k \in \{q+1, \ldots, p\}\),
\begin{equation}\label{mult.risk.threshold}
R([q]; \boldsymbol{w}^* )  < R(k) \qquad \iff \qquad \frac{ \beta_Z^2 }{ \beta_k^2}
 >    \frac{1+  \sum_{j=1}^q   \frac{1}{ \sigma_{\zeta j}^2} }{ \sum_{j=1}^q   \frac{1}{ \sigma_{\zeta j}^2}    }.
\end{equation}
\end{enumerate}
\end{proposition}
\begin{proof}See Appendix \ref{risk.proofs}.
 \end{proof}

Proposition \ref{proxies.risk.different.weight}(i) yields the intuitive result that the optimal weights are higher for less noisy proxies and lower for higher noise proxies. Also, note that if any one proxy has noise variance tending towards 0, the optimal weight on that proxy tends toward 1, the prediction risk tends toward \( \mathcal{E}_{\text{ideal}}\), and the quantity on the right side of \eqref{mult.risk.threshold} tends toward 1.

From Proposition \ref{proxies.risk.different.weight}(ii) we see that adding one more proxy with finite noise variance \(\sigma_{\zeta j}^2\) always reduces the prediction risk when optimal weights are used. Also, by comparison of Proposition \ref{proxies.risk.different.weight}(iii) to Proposition \ref{prop.gen.pred.risk}, we see that the signal strength of \(\boldsymbol{Z}\) does not have to be as high for the prediction risk of \(\boldsymbol{X}_{\cdot [q]}^{\text{rep}}\) to be less than any \(R(k), k \in \{q+1, \ldots, p\}\) than it does for any one \(R(j), j \in [q]\) to be less than \(R(k)\), since if \(0 < \sigma_{\zeta j}^2 < \infty\) for all \(j\) it holds that
\[
 \frac{1+  \sum_{j'=1}^q   \frac{1}{ \sigma_{\zeta j'}^2} }{ \sum_{j'=1}^q   \frac{1}{ \sigma_{\zeta j'}^2}    } <  1+ \sigma_{\zeta j}^2 \qquad \forall q > 1, \forall j \in [q].
\]
See also \citet{Park2007} and \citet{Buhlmann2013} (particularly Propositions 4.2 and 4.3) for other theoretical results about the benefits of averaging correlated features (under different assumptions from ours).

This result will be informative in discussing the virtues of each of our proposed weighting schemes \eqref{weight.avg.weights} -- \eqref{sparse.weights}.

\subsubsection{Weighted Averaged Cluster Stability Selection}\label{gen.stab.sel.weight.avg}

Since by Proposition \ref{proxies.risk.different.weight}(i) the optimal weights are higher when the noise level is lower, if the underlying selection procedure is more likely to select lower noise proxies, the weights from weighted averaged cluster stability selection may resemble the optimal weights. We show empirical evidence in Section \ref{weight.avg.sim.study} that weighted averaged cluster stability selection with the lasso does seem to outperform either other weighting scheme when proxy noise levels vary significantly.

It is worth noting that weighted averaging allows features to be dropped from a cluster with nonzero probability. This is useful particularly if the clusters are estimated and might include some irrelevant features. In Section \ref{real.data.study2} we present evidence that this occurs in practice in a real-data example.

\subsubsection{Simple Averaged Cluster Stability Selection}

Weighting scheme \eqref{simp.avg.weights} aligns with an intuitive idea for a practitioner who knows she has several noisy observations of the same quantity: just average them. It is analogous to the cluster representative lasso \citep{Buhlmann2013, Park2007} in which all of the features in a cluster are averaged, except that we run the underlying feature selection method on the original features, forming the cluster representatives after the selection proportions have been determined. 

By Proposition \ref{proxies.risk.different.weight}, simple averaging is optimal if the noise levels of each proxy are equal. If the practitioner is confident in the clustering used and the noise levels within clusters are unknown but believed to be equal, or even close to equal, simple averaging could result in a more favorable bias-variance tradeoff than estimating weights. This is particularly true if all of the proxies have roughly equal noise levels that are also high, making estimation of the individual weights noisy.

\subsubsection{Sparse Cluster Stability Selection}

Sparse cluster stability selection \eqref{sparse.weights} removes all but the most frequently selected individual features from each cluster. Its output resembles stability selection in the sense that sparse cluster stability selection often returns only clusters of size 1 (if the most frequently selected feature within a cluster is not tied with another feature in the cluster). This weighting scheme is in a way analogous to the protolasso \citep{Reid2015}, but the cluster prototype is the most frequently selected feature across all subsamples rather than the cluster member with the greatest marginal correlation with the response. (One advantage cluster stability selection enjoys over the protolasso is that cluster stability selection does not require a notion of correlation with the response.) 

This weighting scheme has the virtue of representing the important signals with as few features as possible. This could be particularly advantageous if the clusters are estimated, rather than known, and might be too large---sparse cluster stability selection may reject irrelevant noise features in a cluster. This is also useful if the clusters are known in advance, but sparsity is desired to aid interpretability, because individual features are expensive to measure, or for any number of other reasons. However, if all of the cluster members are genuine proxies as in our model, the weights in sparse cluster stability selection are always suboptimal by Proposition \ref{proxies.risk.different.weight}(i) (except for the very special case where all the noise levels are equal and every feature happens to tie in selection proportion). Our simulations and real data examples in Section \ref{sims.data} suggest that in practice, sparsity does indeed come at a price to both out-of-sample predictive performance and stability.

\section{Simulations and Data Application}\label{sims.data}

In this section, we discuss two data simulations and a real data application demonstrating cluster stability selection with each weighting scheme.

\subsection{Simulation Study: Sparse Cluster Stability Selection}\label{sim.study.sparse}

In this section, we describe the simulation study from Example \ref{ex.stab.problem} in more detail. We repeat the following procedure 1000 times:

\begin{itemize}

\item

The design matrix \(\boldsymbol{X} \in \mathbb{R}^{200 \times 100}\) has rows \(X_{i \cdot}\) where
\[
 \begin{pmatrix} Z_i
 \\ \boldsymbol{X}_{i \cdot}
 \end{pmatrix} \sim \mathcal{N}_{101} \left( \begin{pmatrix} 0
 \\ \boldsymbol{0}_{100} 
 \end{pmatrix},
 \begin{pmatrix}
 1 & 0.9 \boldsymbol{1}_{10}^\top & \boldsymbol{0}_{90}^\top 
\\ 0.9 \boldsymbol{1}_{10} & 0.9 \boldsymbol{1}_{10}  \boldsymbol{1}_{10}^\top + 0.1 \boldsymbol{I}_{10} & \boldsymbol{0}_{10}  \boldsymbol{0}_{90}^\top 
\\ \boldsymbol{0}_{90} & \boldsymbol{0}_{90} \boldsymbol{0}_{10}  ^\top   & \boldsymbol{I}_{90}
 \end{pmatrix} \right).
\]

\item The response \(\boldsymbol{y}\) is generated by
\begin{align*}
\boldsymbol{\mu} & = 1.5 \boldsymbol{Z} +  \sum_{j=1}^{10} \beta_j \boldsymbol{X}_{\cdot (j + 10)}  \qquad \text{and} \nonumber
\\ \boldsymbol{y} & = \boldsymbol{\mu} + \boldsymbol{\epsilon}, 
\end{align*}
where \(\beta_j = 1/\sqrt{j}\) and \(\boldsymbol{\epsilon} \sim \mathcal{N} \left( \boldsymbol{0}, \sigma_\epsilon^2 \boldsymbol{I}_{200}\right)\), with \(\sigma_\epsilon^2\) determined so that the signal-to-noise ratio is 3:
\[
\sigma_\epsilon^2  = \frac{\lVert \boldsymbol{\mu} \rVert_2^2/200}{3 }.
\]

\item We obtain selected sets using the lasso, stability selection, sparse cluster stability selection, and the protolasso (which we described in the introduction). The protolasso and sparse cluster stability selection are provided with the correct clusters \(\{[10], \{11\}, \ldots, \{100\}\}\). Both stability selection and sparse cluster stability selection use complementary pairs subsampling, \(B = 100\) subsamples\footnote{The results are similar but noisier using the smaller number of subsamples recommended by \citet{shah_samworth_2012}.} of size \(\lfloor n/2 \rfloor = 100\), and the lasso with penalty chosen in advance (separately for each simulation) by cross-validation.

\item For each method, selected sets of each size \(s \in [11]\) are obtained in the following ways:

\begin{itemize}

\item Lasso and protolasso: the first feature set of size \(s\) to appear in the lasso path. 

\item Stability selection: the \(s\) features with the greatest selection proportions. 

\item Sparse cluster stability selection: similar to stability selection, but \(\hat{\Theta}_B(j)\) is used to select \(s\) clusters instead of using \(\hat{\Pi}_B(j)\) to select \(s\) features.

\end{itemize}

(For stability selection and sparse cluster stability selection, sets are not always defined for every \(s\) due to ties in selection proportions.)

\item Finally, a training set of 10,000 observations is generated in the same way that the original 200 observations were generated. For each method and for every model size defined in that method, \(\boldsymbol{y}\) is regressed against the selected features from \(\boldsymbol{X}\) using OLS, and the mean squared error (MSE) of the resulting training set predictions \(\boldsymbol{\hat{y}}\) compared to \(\boldsymbol{\mu}\) is calculated.

\end{itemize}

\begin{figure}[htbp]
\begin{center}
\includegraphics[width=\textwidth]{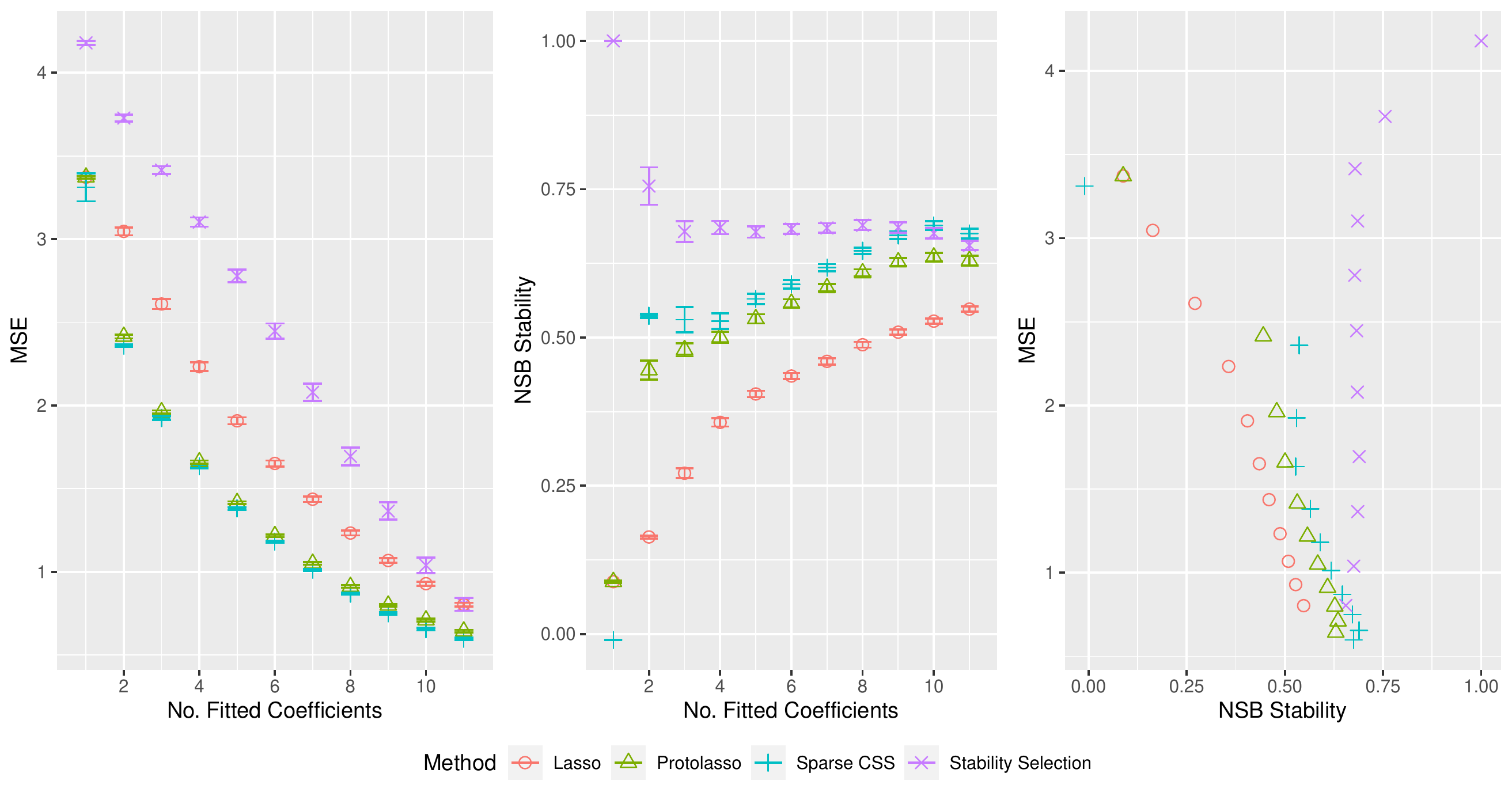}
\caption{Left and center panels: plots of average mean squared error and stability against model size (measured by number of fitted coefficients) across all 1000 simulations for the simulation study from Section \ref{sim.study.sparse}. The error bars show asymptotic \(95\%\) confidence intervals; the error bars for MSE are those implied by the Central Limit Theorem, and the error bars for the stability metric are provided by \citet{Nogueira2018}. Right panel: average MSE plotted against stability. Each point corresponds to one model size.}
\label{3_figures_sparse}
\end{center}
\end{figure}

After completing this procedure 1000 times, the average MSE is calculated for each method for every model size \(s\) (including in the average for each method only those iterations where the model of size \(s\) for that method is defined). We then evaluate the stability of each method across the 1000 simulations using the metric proposed by \citet{Nogueira2018}, which equals 1 if the selected set yielded by a given method is identical across every simulation and has expected value 0 for a ``null" feature selection method that selects features at random.

The results are shown in the left panel of Figure \ref{3_figures_sparse}.\footnote{In the left panel, the error bars for the sparse cluster stability selection model of size 1 are quite wide, and in the center panel, the stability of this model is particularly low, and has a confidence interval of width 0. These oddities are due to the fact that in 1000 simulations, there were only two instances where a sparse CSS model of size 1 was defined---in most of the remaining instances, both the cluster of proxies and the cluster containing only the strongest weak signal feature had selection proportions equal to 1. This is perhaps not shocking because the lasso penalty \(\lambda\) was chosen by cross-validation, which is known to result in larger than optimal model sizes; see, for example, \citet[Prop. 1]{Buhlmann2006} and \citet[Section 2.5.1]{MR2807761}. As a result, the strongest signal features are likely to be selected with very high probability. 
%The fact that only two cluster stability selected models of size 1 were yielded explains the width of the confidence interval in the left panel, and why the stability estimate is close to zero in the second panel.
} Sparse cluster stability selection has better predictive performance than stability selection because stability selection fails to select any proxy feature for \(\boldsymbol{Z}\), and selecting at least one proxy is important because the coefficient on \(\boldsymbol{Z}\) in the true model is large. Sparse cluster stability selection's superior predictive performance over the lasso appears to be due to both the benefits of stability and the fact that the lasso seems to often predict more than one proxy feature (this explains the gap in predictive performance between the protolasso and the lasso). Note that sparse cluster stability selection enjoys a modest, but still statistically significant, improvement over the protolasso in both MSE and stability.

Stability selection is the most stable overall for most model sizes because it tends to ignore the proxies for \(\boldsymbol{Z}\) due to the ``vote-splitting" problem. Instead, the \(j^{\text{th}}\) feature chosen by stability selection tends to be the weak signal feature with the \(j^{\text{th}}\) largest coefficient. Meanwhile, the methods other than stability selection tend to select proxies for \(\boldsymbol{Z}\) early on, and because they select randomly among those proxies, the resulting selected sets are less stable. For this reason, sparse cluster stability selection is not the most stable method overall, but it beats both the lasso and protolasso in stability by statistically significant margins, and even beats stability selection by statistically significant margins for models of size 10 and 11.

In general, practitioners may be interested in simultaneous predictive performance and stability, so we also consider what tradeoff exists between the two across model sizes and methods. The right panel of Figure \ref{3_figures_sparse} plots out-of-sample MSE against stability for each method and each model size. Sparse cluster stability selection clearly dominates the lasso and the protolasso. A couple of the stability selected models have better stability than the most stable sparse cluster stability selection model, but this stability comes at a steep price to MSE.

\subsection{Simple Averaged Cluster Stability Selection}\label{avg.sim.study}

To show the benefits of simple averaged cluster stability selection, we evaluate it in the same simulation study as above. We compare simple averaged cluster stability selection to sparse cluster stability selection, and we also consider the cluster representative lasso \citep{Buhlmann2013, Park2007}, which takes a simple average of the features in the known cluster and fits the lasso on that cluster representative and the remaining original features\footnote{Note that we are using an oracle version of the cluster representative lasso; \citet{Buhlmann2013} estimate the clusters.}.

How to measure the model size of the methods that rely on averaging is ambiguous. When a cluster representative constructed by averaging the original features is added to the model, we could say that the model size increases by the size of the cluster (that is, count the size of the model by the number of included features in the original feature space), or we could say the model size increases by one (that is, count the number of fitted coefficients in the model). We choose the latter approach, but note that this gives the averaging methods an ``unfair advantage" in predictive performance at a fixed model size in the sense that a cluster representative resulting from averaging a number of noisy features is a better predictor than any one noisy feature (Proposition \ref{proxies.risk.different.weight}). (This is another reason to include the cluster representative lasso in this comparison---unlike cluster stability selection, it also enjoys the benefits of averaging features, so the comparison is in that sense more fair.)

\begin{figure}[htbp]
\begin{center}
\includegraphics[width=\textwidth]{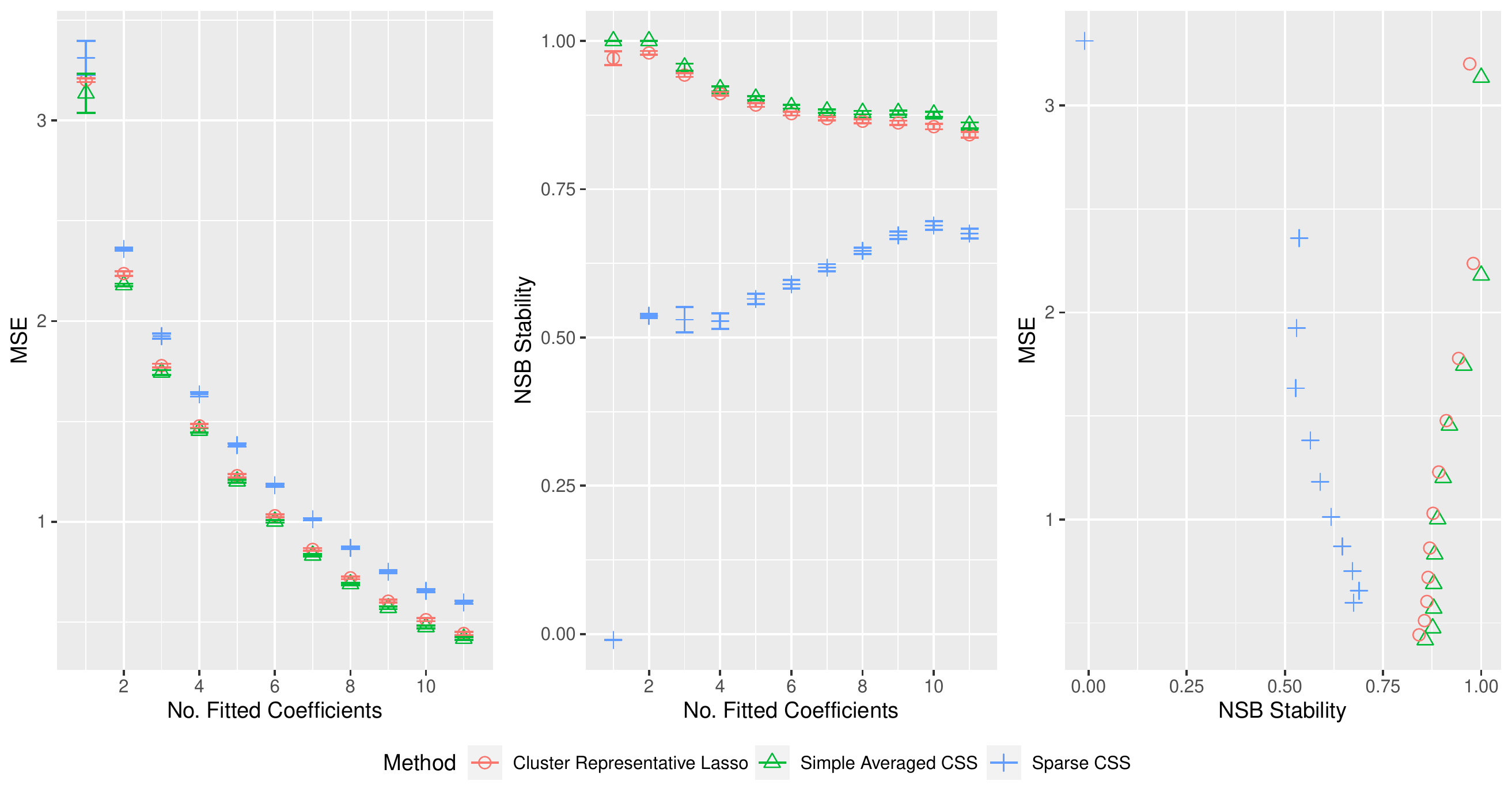}
%[scale=0.65]
\caption{The left two panels are plots of average MSE and stability against model size (measured by number of fitted coefficients) for cluster stability selection, simple averaged cluster stability selection, and the cluster representative lasso across all 1000 simulations for the same simulation study as the one in Section \ref{sim.study.sparse}. The error bars show asymptotic \(95\%\) confidence intervals; the error bars for MSE are those implied by the Central Limit Theorem, and the error bars for the stability metric are provided by \citet{Nogueira2018}. The right panel plots average MSE against model size.}
\label{ranking_avg_mse_stab}
\end{center}
\end{figure}

Figure \ref{ranking_avg_mse_stab} shows the results. (We omit the methods from Figure \ref{3_figures_sparse} other than sparse cluster stability selection for visual clarity, and because we already know that sparse cluster stability selection dominates the other methods.) We see that both of the averaging methods have better predictive performance than sparse cluster stability selection. They also have better stability because they select entire clusters rather than choosing among the cluster members at random. Simple averaged cluster stability selection has better stability and predictive performance (because stability leads to fewer false selections) than the cluster representative lasso. In the right panel, we see that the averaging methods all dominate sparse cluster stability selection. and simple averaged cluster stability selection dominates the cluster representative lasso.

\subsection{Simulation Study: Weighted Averaged Cluster Stability Selection}\label{weight.avg.sim.study}

Lastly, we conduct one more simulation study designed to illustrate the benefits of weighted averaging. The simulation is the same as the simulation study from above, except that the design matrix \(\boldsymbol{X}\) is constructed slightly differently. Rather than observing 10 proxies that all have a correlation of \(0.9\) with \(\boldsymbol{Z}\), \(\boldsymbol{X}\) contains 5 \textit{strong proxies} with a correlation of \(0.9\) with \(\boldsymbol{Z}\) and 10 \textit{weak proxies} with a correlation of \(0.5\) with \(\boldsymbol{Z}\). By Proposition \ref{proxies.risk.different.weight} it is suboptimal to discard the weak proxies altogether, but it is also suboptimal to weight them equally with the strong proxies. Again 10 weak signal features with coefficients \(1/\sqrt{j}\) are observed, along with 75 noise features to yield a total of 100 features.

\begin{figure}[htbp]
\begin{center}
\includegraphics[width=\textwidth]{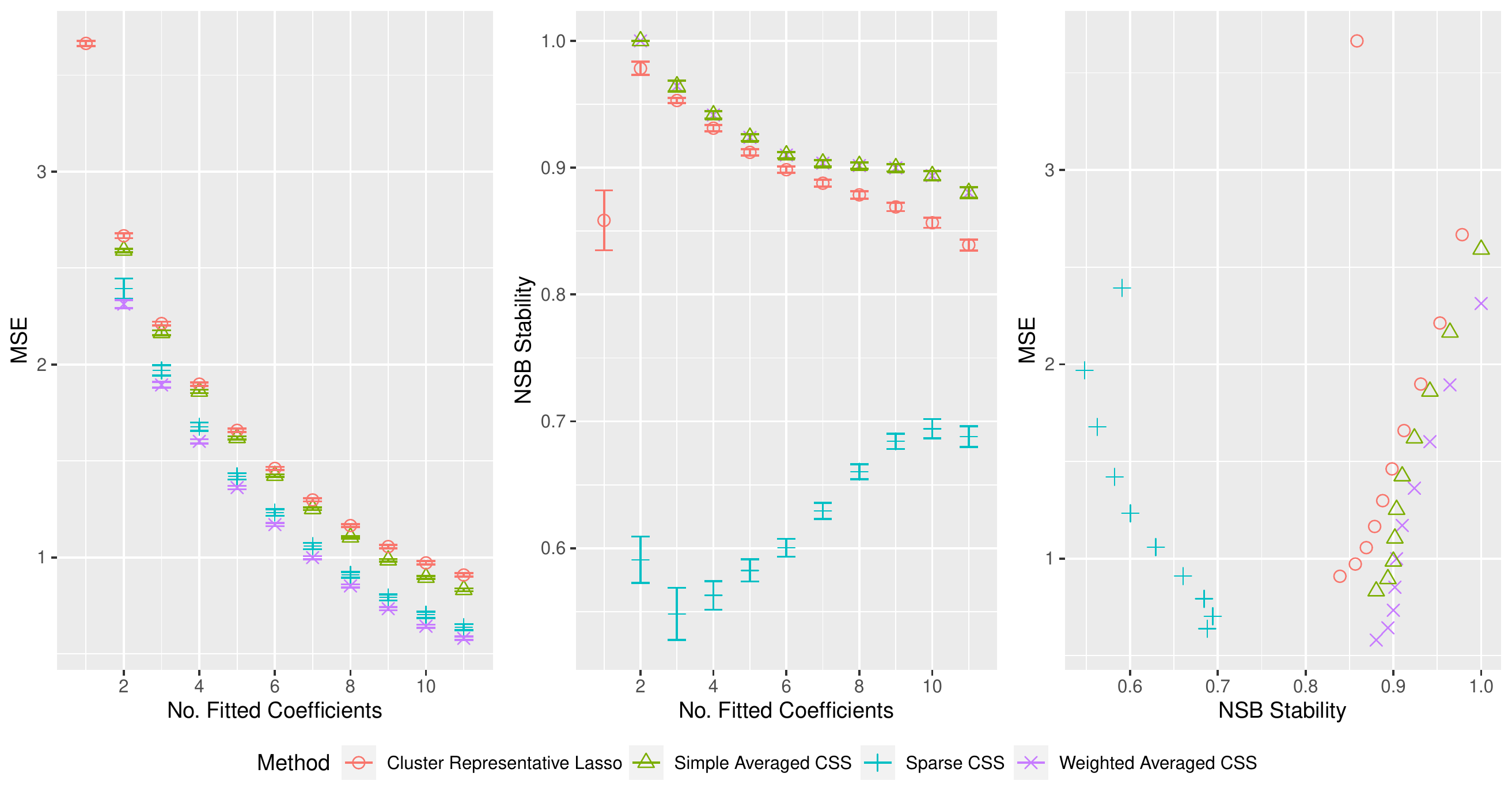}
%[scale=0.65]
\caption{The left two panels are plots of average MSE and stability against model size (measured by number of fitted coefficients) for the methods from the simulation study in Section \ref{weight.avg.sim.study} across all 1000 simulations. The error bars show asymptotic \(95\%\) confidence intervals; the error bars for MSE are those implied by the Central Limit Theorem, and the error bars for the stability metric are provided by \citet{Nogueira2018}. The right panel plots average MSE against model size.}
\label{weighted_avg_mse_stab}
\end{center}
\end{figure}

Figure \ref{weighted_avg_mse_stab} shows the results. We see that weighted averaged cluster stability selection does dominate the other methods, suggesting that the weights it chooses are closer to optimal than either simple averaging or sparse cluster stability selection. Weighted averaged cluster stability selection and simple averaged cluster stability selection have roughly identical stability because they both select every feature in the cluster when they select a cluster. The cluster representative lasso lags behind all cluster stability selection methods in predictive performance both because equal weighting is suboptimal and because it makes more false selections.

\subsection{Data Application}\label{real.data.study2}

%In genome-wide association studies (GWAS), DNA is collected from a random sample of individuals. Then single-nucleotide polymorphism (SNP) arrays are used to extract genetic variants (primarily SNPs) from each individual. An organism's DNA typically has millions of base positions, and at each base position there is typically one more common SNP allele and one less common one. For \textit{diploid} species with two SNPs at each base position, the SNP data for organism \(i\) at base position \(j\) takes on one of the values \(X_{ij} \in \{0, 1, 2\}\) equalling the number of less common alleles at that position.
%
%In a GWAS, DNA samples are taken from \(n\) organisms, yielding \(p\) SNPs. These are stored in a matrix \(\boldsymbol{X}\). A vector of phenotypes \(\boldsymbol{y}\)---one for each organism---is also observed. Examples of phenotypes could include physical traits like adult height or weight, or the presence or absence of a given disease at a fixed age. The goal of GWAS is to identify which SNPs are associated with variation in a given phenotype, typically using regression methods. Identifying SNPs associated with health outcomes or physical traits can be useful for identifying individuals at risk of developing a disease or developing treatments for diseases, among other applications. That is, feature selection methods are of great practical importance in GWAS.

We demonstrate cluster stability selection on an open genome-wide association study data set \citep{Togninalli2017, Alonso-Blanco2016} collected from \(n = 1,058\) accessions of \textit{Arabidopsis thaliana}, a small flowering plant that has been widely studied. The feature matrix \(\boldsymbol{X}\) contains SNP data---\(X_{ij} = 0\) if in accession \(i\) both alleles at base position \(j\) take on the more common value, and \(X_{ij} = 1\) if both alleles take on the more rare value. GWAS data is a natural application of cluster stability selection because nearby SNPs tend to be highly correlated due to \textit{linkage disequilibrium}, which is caused by a variety of mechanisms; see \citet{Nordborg2002} and \citet{Kim2007} for discussion of linkage disequilibrium specifically in \textit{Arabidopsis thaliana}. As a result, clusters of highly correlated SNPs can be identified. For the response, we use the logarithm of the measured flowering time (in days) at \(10^\circ\) C \citep{Alonso-Blanco2016}.\footnote{The imputed genotype was downloaded from \url{https://aragwas.1001genomes.org/} and the phenotypes were downloaded via \url{https://arapheno.1001genomes.org/phenotype/261/}.}

Prior to evaluating our methods on the data, we pre-process the data using standard methods \citep{Candes2018, Sesia2019}. We screen out SNPs where the less common allele appears in fewer than \(1\%\) of observations (that is, the minor allele frequency is less than \(1\%\)). We do not have to screen for missing values or incorrect SNP position labeling because the data set as it is available online has already been cleaned and imputed using standard methods. The Hardy-Weinberg equilibrium test is commonly used for screening SNPs, but it is not applicable for our data because \textit{Arabidopsis thaliana} is almost always homozygous. Finally, for computational speed we retained only the first 1000 SNPs that remained after screening.

We repeat the following procedure (similar to the procedure from the simulation section) 100 times. We randomly divide the data into feature selection and model estimation sets of 423 observations (about \(40\%\) of the data for each) and a test set of the remaining 212 observations. In each iteration we cluster the features using hierarchical clustering on the non-test-set data. For the distance metric, we use one minus the absolute value of the correlation between the SNPs. We use a single-linkage cutoff of \(0.5\), following \citet{Candes2018}. Then for every feature selection method, we use the feature selection set to obtain selected sets of sizes \(\{1, \ldots, 100\}\), providing the estimated clusters to those methods that make use of them. Next, for every method and every model size, we use the model estimation set to estimate linear models by OLS using the selected features. Finally, we use these models to generate predictions for the test set, and evaluate the MSE of the predictions against the actual values. After all simulations are complete, we also evaluate the stability of each method across all simulations, again using the metric proposed by \citet{Nogueira2018}.

\begin{figure}[htbp]
\begin{center}
\includegraphics[width=\textwidth]{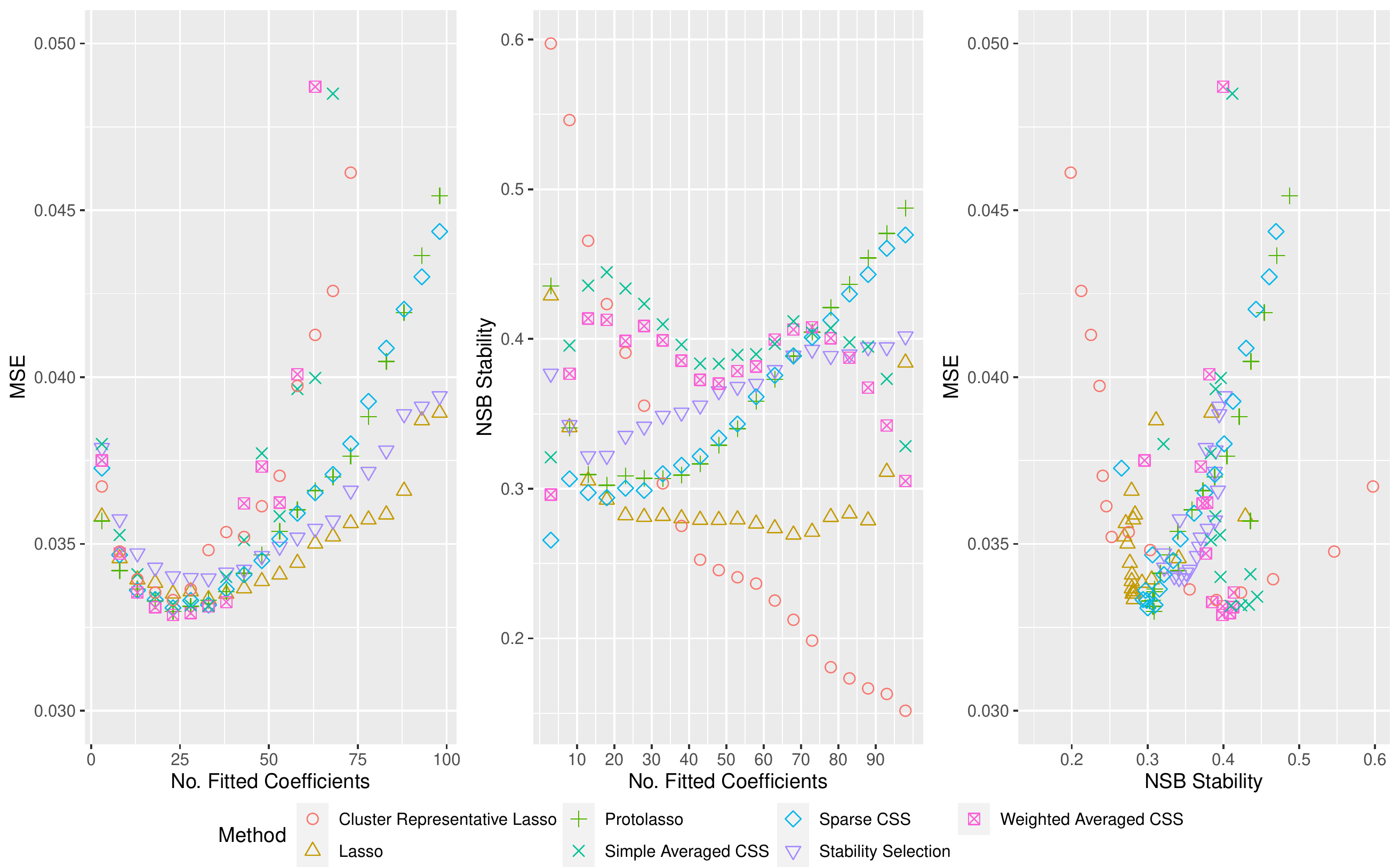}
%[scale=0.65]
\caption{Results from the simulation study in Section \ref{real.data.study2}. The left two panels are plots of average MSE and stability against model size (measured by number of fitted coefficients) for the methods from the simulation study in Section \ref{weight.avg.sim.study} across all 1000 simulations. The right panel plots average MSE against model size. In all cases, points on the plot represent averages across 5 model sizes to smooth the plot (for example, the leftmost points in the left plot are aggregated over model sizes 1 through 5). No error bars are provided in these plots because these represent statistics calculated from subsamples of a fixed data set, not i.i.d. random variables.}
\label{real_data_3_side_plot}
\end{center}
\end{figure}

Figure \ref{real_data_3_side_plot} shows the results. In the left plot, we see that the best-performing model is yielded by weighted averaged cluster stability selection. The remaining cluster stability selection methods are also among the best-performing models. 

The middle plot of Figure \ref{real_data_3_side_plot} shows the stability of each method. At most model sizes (and in particular at the model sizes that are best for predictive performance, in the size range of 20 -- 35), the cluster stability selection with averaging methods are the most stable. (Note that unlike in the simulation study, weighted averaged cluster stability selection has appreciably different stability that simple averaged cluster stability selection, because in this setting weighted averaged cluster stability selection assigns weight 0 to some cluster members with appreciably high probability.)

Examining the third plot in Figure \ref{real_data_3_side_plot}, we see that the cluster stability selection with averaging methods are among the best models for predictive performance, and are more stable than competing methods with similar predictive performance. The models closest to the bottom right corner of the graph (the Pareto frontier of the tradeoff between stability and predictive performance) are mostly yielded by averaging cluster stability selection methods. The three smallest model sizes for the cluster representative lasso are on the Pareto frontier, with higher stability than the CSS methods and impressive predictive performance, but their predictive performance still lags behind cluster stability selection with averaging. Also, these represent smaller models with fewer discoveries.

\section{Conclusion}\label{conclusion}

Stability is a fundamental property in practice and for replicability in science. Practitioners might reasonably hope that if they were to receive another random sample from the same distribution, the main findings of their statistical analysis would be minimally altered.

Stability selection adds stability to the lasso, leading to fewer false selections. However, we have shown that stability selection can miss important features leading to poorly predicting models when there are highly correlated clusters of features. 
Cluster stability selection allows practitioners to exploit knowledge of clustered features in the data to enjoy the benefits of stability selection while still maintaining the lasso's ability to select important clustered features.

\section*{Acknowledgements}

We thank Saharon Rosset for a helpful discussion. This work was supported in part by NSF CAREER Award DMS-1653017.

\bibliographystyle{abbrvnat}
\bibliography{mybib2fin}

\begin{thebibliography}{70}
\providecommand{\natexlab}[1]{#1}
\providecommand{\url}[1]{\texttt{#1}}
\expandafter\ifx\csname urlstyle\endcsname\relax
  \providecommand{\doi}[1]{doi: #1}\else
  \providecommand{\doi}{doi: \begingroup \urlstyle{rm}\Url}\fi

\bibitem[Alexander and Lange(2011)]{Alexander2011}
D.~H. Alexander and K.~Lange.
\newblock {Stability selection for genome-wide association}.
\newblock \emph{Genetic Epidemiology}, 35\penalty0 (7):\penalty0 722--728,
  2011.
\newblock ISSN 07410395.
\newblock \doi{10.1002/gepi.20623}.

\bibitem[Alonso-Blanco et~al.(2016)Alonso-Blanco, Andrade, Becker, Bemm,
  Bergelson, Borgwardt, Cao, Chae, Dezwaan, Ding, Ecker, Exposito-Alonso,
  Farlow, Fitz, Gan, Grimm, Hancock, Henz, Holm, Horton, Jarsulic, Kerstetter,
  Korte, Korte, Lanz, Lee, Meng, Michael, Mott, Muliyati, N{\"{a}}gele, Nagler,
  Nizhynska, Nordborg, Novikova, Pic{\'{o}}, Platzer, Rabanal, Rodriguez,
  Rowan, Salom{\'{e}}, Schmid, Schmitz, Seren, Sperone, Sudkamp, Svardal,
  Tanzer, Todd, Volchenboum, Wang, Wang, Wang, Weckwerth, Weigel, and
  Zhou]{Alonso-Blanco2016}
C.~Alonso-Blanco, J.~Andrade, C.~Becker, F.~Bemm, J.~Bergelson, K.~M.~M.
  Borgwardt, J.~Cao, E.~Chae, T.~M.~M. Dezwaan, W.~Ding, J.~R.~R. Ecker,
  M.~Exposito-Alonso, A.~Farlow, J.~Fitz, X.~Gan, D.~G.~G. Grimm, A.~M.~M.
  Hancock, S.~R.~R. Henz, S.~Holm, M.~Horton, M.~Jarsulic, R.~A.~A. Kerstetter,
  A.~Korte, P.~Korte, C.~Lanz, C.~R. Lee, D.~Meng, T.~P.~P. Michael, R.~Mott,
  N.~W.~W. Muliyati, T.~N{\"{a}}gele, M.~Nagler, V.~Nizhynska, M.~Nordborg,
  P.~Y.~Y. Novikova, F.~X. Pic{\'{o}}, A.~Platzer, F.~A.~A. Rabanal,
  A.~Rodriguez, B.~A.~A. Rowan, P.~A.~A. Salom{\'{e}}, K.~J.~J. Schmid,
  R.~J.~J. Schmitz, {\"{U}}.~Seren, F.~G.~G. Sperone, M.~Sudkamp, H.~Svardal,
  M.~M.~M. Tanzer, D.~Todd, S.~L.~L. Volchenboum, C.~Wang, G.~Wang, X.~Wang,
  W.~Weckwerth, D.~Weigel, and X.~Zhou.
\newblock {1,135 Genomes Reveal the Global Pattern of Polymorphism in
  Arabidopsis thaliana}.
\newblock \emph{Cell}, 166\penalty0 (2):\penalty0 481--491, jul 2016.
\newblock ISSN 10974172.
\newblock \doi{10.1016/j.cell.2016.05.063}.

\bibitem[Anbari and Mkhadri(2014)]{Anbari2014}
M.~E. Anbari and A.~Mkhadri.
\newblock {Penalized regression combining the L1 norm and a correlation based
  penalty}.
\newblock \emph{Sankhya B}, 76\penalty0 (1):\penalty0 82--102, 2014.
\newblock ISSN 09768394.
\newblock \doi{10.1007/s13571-013-0065-4}.

\bibitem[Beinrucker et~al.(2016)Beinrucker, Dogan, and
  Blanchard]{Beinrucker2016}
A.~Beinrucker, {\"{U}}.~Dogan, and G.~Blanchard.
\newblock {Extensions of stability selection using subsamples of observations
  and covariates}.
\newblock \emph{Statistics and Computing}, 26\penalty0 (5):\penalty0
  1059--1077, 2016.
\newblock ISSN 15731375.
\newblock \doi{10.1007/s11222-015-9589-y}.

\bibitem[Belloni et~al.(2016)Belloni, Rosenbaum, and Tsybakov]{Belloni2016}
A.~Belloni, M.~Rosenbaum, and A.~B. Tsybakov.
\newblock {An {l1, l2, l infinity}-regularization approach to high-dimensional
  errors-in-variables models}.
\newblock \emph{Electronic Journal of Statistics}, 10\penalty0 (2):\penalty0
  1729--1750, 2016.
\newblock ISSN 19357524.
\newblock \doi{10.1214/15-EJS1095}.

\bibitem[Belloni et~al.(2017{\natexlab{a}})Belloni, Chernozhukov, Kaul,
  Rosenbaum, and Tsybakov]{BelloniPivot2017}
A.~Belloni, V.~Chernozhukov, A.~Kaul, M.~Rosenbaum, and A.~B. Tsybakov.
\newblock {Pivotal Estimation Via Self-Normalization for High-Dimensional
  Linear Models With Errors In Variables}, 2017{\natexlab{a}}.
\newblock ISSN 23318422.
\newblock URL \url{https://arxiv.org/abs/1708.08353}.

\bibitem[Belloni et~al.(2017{\natexlab{b}})Belloni, Rosenbaum, and
  Tsybakov]{Belloni2017}
A.~Belloni, M.~Rosenbaum, and A.~B. Tsybakov.
\newblock {Linear and conic programming estimators in high dimensional
  errors-in-variables models}.
\newblock \emph{Journal of the Royal Statistical Society. Series B: Statistical
  Methodology}, 79\penalty0 (3):\penalty0 939--956, 2017{\natexlab{b}}.
\newblock ISSN 14679868.
\newblock \doi{10.1111/rssb.12196}.
\newblock URL
  \url{https://rss-onlinelibrary-wiley-com.libproxy2.usc.edu/doi/pdf/10.1111/rssb.12196}.

\bibitem[Bing et~al.(2020)Bing, Bunea, Ning, and Wegkamp]{Bing2019}
X.~Bing, F.~Bunea, Y.~Ning, and M.~Wegkamp.
\newblock {Adaptive estimation in structured factor models with applications to
  overlapping clustering}.
\newblock \emph{Annals of Statistics}, 48\penalty0 (4):\penalty0 2055--2081,
  2020.
\newblock ISSN 21688966.
\newblock \doi{10.1214/19-AOS1877}.
\newblock URL \url{https://arxiv.org/pdf/1704.06977.pdf}.

\bibitem[Bing et~al.(2021)Bing, Bunea, and Wegkamp]{Bing2021}
X.~Bing, F.~Bunea, and M.~Wegkamp.
\newblock {Inference in latent factor regression with clusterable features}.
\newblock Technical report, 2021.
\newblock URL \url{https://arxiv.org/pdf/1905.12696.pdf}.

\bibitem[Bogdan et~al.(2013)Bogdan, Berg, Su, and Cand{{\`e}}s]{SortedL1}
M.~Bogdan, E.~v.~d. Berg, W.~Su, and E.~J. Cand{{\`e}}s.
\newblock Statistical estimation and testing via the sorted l1 norm.
\newblock \emph{arXiv preprint arXiv:1310.1969}, 2013.

\bibitem[Bollen(1989)]{Bollen1989}
K.~A. Bollen.
\newblock \emph{{Structural equations with latent variables}}.
\newblock Wiley series in probability and mathematical statistics. Applied
  probability and statistics section. John Wiley {\&} Sons, Oxford, England,
  1989.
\newblock ISBN 0-471-01171-1 (Hardcover).
\newblock \doi{10.1002/9781118619179}.

\bibitem[Bondell and Reich(2008)]{Bondell2008}
H.~D. Bondell and B.~J. Reich.
\newblock {Simultaneous Regression Shrinkage, Variable Selection, and
  Supervised Clustering of Predictors with OSCAR}.
\newblock \emph{Biometrics}, 64:\penalty0 115--123, 2008.
\newblock \doi{10.1111/j.1541-0420.2007.00843.x}.
\newblock URL
  \url{https://onlinelibrary-wiley-com.libproxy1.usc.edu/doi/pdf/10.1111/j.1541-0420.2007.00843.x}.

\bibitem[Bousquet and Elisseeff(2002)]{Bousquet2002}
O.~Bousquet and A.~Elisseeff.
\newblock {Stability and Generalization}.
\newblock \emph{Journal of Machine Learning Research}, 2\penalty0 (3):\penalty0
  499--526, 2002.
\newblock ISSN 15324435.
\newblock \doi{10.1162/153244302760200704}.

\bibitem[B{\"{u}}hlmann and Meinshausen(2006)]{Buhlmann2006}
P.~B{\"{u}}hlmann and N.~Meinshausen.
\newblock {High-Dimensional Graphs and Variable Selection With the Lasso}.
\newblock \emph{The Annals of Statistics}, 34\penalty0 (3):\penalty0
  1436--1462, 2006.
\newblock \doi{10.1214/009053606000000281}.
\newblock URL
  \url{https://projecteuclid-org.libproxy2.usc.edu/download/pdfview_1/euclid.aos/1152540754}.

\bibitem[B{\"u}hlmann and van~de Geer(2011)]{MR2807761}
P.~B{\"u}hlmann and S.~van~de Geer.
\newblock \emph{Statistics for High-Dimensional Data}.
\newblock Springer Series in Statistics. Springer, Heidelberg, 2011.
\newblock ISBN 978-3-642-20191-2.
\newblock \doi{10.1007/978-3-642-20192-9}.
\newblock URL \url{http://dx.doi.org/10.1007/978-3-642-20192-9}.
\newblock Methods, theory and applications.

\bibitem[B{\"{u}}hlmann et~al.(2013)B{\"{u}}hlmann, R{\"{u}}timann, van~de
  Geer, and Zhang]{Buhlmann2013}
P.~B{\"{u}}hlmann, P.~R{\"{u}}timann, S.~van~de Geer, and C.~H. Zhang.
\newblock {Correlated variables in regression: Clustering and sparse
  estimation}.
\newblock \emph{Journal of Statistical Planning and Inference}, 143\penalty0
  (11):\penalty0 1835--1858, 2013.
\newblock ISSN 03783758.
\newblock \doi{10.1016/j.jspi.2013.05.019}.
\newblock URL \url{http://dx.doi.org/10.1016/j.jspi.2013.05.019}.

\bibitem[Cand{\`{e}}s et~al.(2018)Cand{\`{e}}s, Fan, Janson, and
  Lv]{Candes2018}
E.~Cand{\`{e}}s, Y.~Fan, L.~Janson, and J.~Lv.
\newblock {Panning for gold: `model-X' knockoffs for high dimensional
  controlled variable selection}.
\newblock \emph{Journal of the Royal Statistical Society. Series B: Statistical
  Methodology}, 80\penalty0 (3):\penalty0 551--577, 2018.
\newblock ISSN 14679868.
\newblock \doi{10.1111/rssb.12265}.
\newblock URL
  \url{https://rss-onlinelibrary-wiley-com.libproxy1.usc.edu/doi/pdf/10.1111/rssb.12265}.

\bibitem[Cunha et~al.(2005)Cunha, Heckman, and Navarro]{Cunha2005}
F.~Cunha, J.~Heckman, and S.~Navarro.
\newblock {Separating uncertainty from heterogeneity in life cycle earnings}.
\newblock \emph{Oxford Economic Papers}, 57\penalty0 (2):\penalty0 191--261,
  2005.
\newblock ISSN 00307653.
\newblock \doi{10.1093/oep/gpi019}.
\newblock URL \url{http://jenni.uchicago.edu/Hicks2004/}.

\bibitem[Davis(2002)]{davis2002statistical}
C.~Davis.
\newblock \emph{Statistical Methods for the Analysis of Repeated Measurements}.
\newblock Online access with purchase: Springer. Springer, 2002.
\newblock ISBN 9780387953700.
\newblock URL \url{https://books.google.com/books?id=MIwSjX4UW60C}.

\bibitem[Devroye and Wagner(1979)]{Devroye1979}
L.~Devroye and T.~Wagner.
\newblock Distribution-free performance bounds for potential function rules.
\newblock \emph{IEEE Transactions on Information Theory}, 25\penalty0
  (5):\penalty0 601--604, 1979.
\newblock \doi{10.1109/TIT.1979.1056087}.

\bibitem[Efron et~al.(2004)Efron, Hastie, Johnstone, and Tibshirani]{Efron2004}
B.~Efron, T.~Hastie, I.~Johnstone, and R.~Tibshirani.
\newblock {Least Angle Regression}.
\newblock \emph{The Annals of Statistics}, 32\penalty0 (2):\penalty0 407--499,
  2004.
\newblock ISSN 0090-5364.
\newblock \doi{10.1214/009053604000000067}.
\newblock URL
  \url{http://ieeexplore.ieee.org/lpdocs/epic03/wrapper.htm?arnumber=25879}.

\bibitem[Gauraha(2016)]{Gauraha2016}
N.~Gauraha.
\newblock {Stability Feature Selection using Cluster Representative LASSO}.
\newblock In \emph{Proceedings of the 5th International Conference on Pattern
  Recognition Applications and Methods (ICPRAM 2016)}, pages 381--386.
  Scitepress, 2016.
\newblock \doi{10.5220/0005827003810386}.

\bibitem[Greene(2012)]{Greene2012Econometric}
W.~H. Greene.
\newblock \emph{Econometric Analysis}.
\newblock Pearson Education, 7th edition, 2012.

\bibitem[Hayes and Krippendorff(2007)]{Hayes2007}
A.~F. Hayes and K.~Krippendorff.
\newblock {Answering the Call for a Standard Reliability Measure for Coding
  Data}.
\newblock \emph{Communication Methods and Measures}, 1\penalty0 (1):\penalty0
  77--89, 2007.
\newblock ISSN 1931-2458.
\newblock \doi{10.1080/19312450709336664}.
\newblock URL
  \url{https://www.tandfonline.com/action/journalInformation?journalCode=hcms20}.

\bibitem[Izenman(2008)]{izenman}
A.~J. Izenman.
\newblock \emph{Modern Multivariate Statistical Techniques: Regression,
  Classification, and Manifold Learning}.
\newblock Springer-Verlag New York, 1 edition, 2008.

\bibitem[Jacob et~al.(2009)Jacob, Obozinski, and Vert]{Jacob2009}
L.~Jacob, G.~Obozinski, and J.-P. Vert.
\newblock Group lasso with overlap and graph lasso.
\newblock In \emph{Proceedings of the 26th Annual International Conference on
  Machine Learning}, ICML '09, pages 433--440, New York, NY, USA, 2009.
  Association for Computing Machinery.
\newblock ISBN 9781605585161.
\newblock \doi{10.1145/1553374.1553431}.
\newblock URL \url{https://doi.org/10.1145/1553374.1553431}.

\bibitem[James et~al.(2021)James, Witten, Hastie, and
  Tibshirani]{james2021introduction}
G.~James, D.~Witten, T.~Hastie, and R.~Tibshirani.
\newblock \emph{An Introduction to Statistical Learning: with Applications in
  R}.
\newblock Springer Texts in Statistics. Springer US, 2021.
\newblock ISBN 9781071614174.
\newblock URL \url{https://books.google.com/books?id=g5gezgEACAAJ}.

\bibitem[Kearns and Ron(1997)]{Kearns1997}
M.~Kearns and D.~Ron.
\newblock Algorithmic stability and sanity-check bounds for leave-one-out
  cross-validation.
\newblock In \emph{Proceedings of the Tenth Annual Conference on Computational
  Learning Theory}, COLT '97, pages 152--162, New York, NY, USA, 1997.
  Association for Computing Machinery.
\newblock ISBN 0897918916.
\newblock \doi{10.1145/267460.267491}.
\newblock URL \url{https://doi.org/10.1145/267460.267491}.

\bibitem[Kent(2010)]{kent-2010}
J.~T. Kent.
\newblock Discussion of ``stability selection".
\newblock \emph{Journal of the Royal Statistical Society: Series B (Statistical
  Methodology)}, 72\penalty0 (4), 2010.

\bibitem[Kim and Sun(2019)]{Kim2019}
K.~Kim and H.~Sun.
\newblock {Incorporating genetic networks into case-control association studies
  with high-dimensional DNA methylation data}.
\newblock \emph{BMC Bioinformatics}, 20\penalty0 (1), 2019.
\newblock ISSN 14712105.
\newblock \doi{10.1186/s12859-019-3040-x}.
\newblock URL \url{https://doi.org/10.1186/s12859-019-3040-x}.

\bibitem[Kim et~al.(2007)Kim, Plagnol, Hu, Toomajian, Clark, Ossowski, Ecker,
  Weigel, and Nordborg]{Kim2007}
S.~Kim, V.~Plagnol, T.~T. Hu, C.~Toomajian, R.~M. Clark, S.~Ossowski, J.~R.
  Ecker, D.~Weigel, and M.~Nordborg.
\newblock {Recombination and linkage disequilibrium in Arabidopsis thaliana}.
\newblock \emph{Nature Genetics}, 39\penalty0 (9):\penalty0 1151--1155, 2007.
\newblock ISSN 10614036.
\newblock \doi{10.1038/ng2115}.

\bibitem[Kirk et~al.(2010)Kirk, Lewin, and Stumpf]{kirk-2010}
P.~D. Kirk, A.~M. Lewin, and M.~P. Stumpf.
\newblock Discussion of ``stability selection".
\newblock \emph{Journal of the Royal Statistical Society: Series B (Statistical
  Methodology)}, 72\penalty0 (4), 2010.

\bibitem[Krasnokutskaya(2011)]{Krasnokutskaya2011}
E.~Krasnokutskaya.
\newblock {Identification and estimation of auction models with unobserved
  heterogeneity}.
\newblock \emph{Review of Economic Studies}, 78\penalty0 (1):\penalty0
  293--327, 2011.
\newblock ISSN 1467937X.
\newblock \doi{10.1093/restud/rdq004}.
\newblock URL \url{https://academic.oup.com/restud/article/78/1/293/1534722}.

\bibitem[Lange et~al.(2003)Lange, Braun, Roth, and Buhmann]{Lange2003}
T.~Lange, M.~Braun, V.~Roth, and J.~Buhmann.
\newblock Stability-based model selection.
\newblock In S.~Becker, S.~Thrun, and K.~Obermayer, editors, \emph{Advances in
  Neural Information Processing Systems}, volume~15. MIT Press, 2003.
\newblock URL
  \url{https://proceedings.neurips.cc/paper/2002/file/37d097caf1299d9aa79c2c2b843d2d78-Paper.pdf}.

\bibitem[Li and Vuong(1998)]{Li1998}
T.~Li and Q.~Vuong.
\newblock {Nonparametric Estimation of the Measurement Error Model Using
  Multiple Indicators}.
\newblock \emph{Journal of Multivariate Analysis}, 65\penalty0 (2):\penalty0
  139--165, 1998.
\newblock ISSN 0047259X.
\newblock \doi{10.1006/jmva.1998.1741}.

\bibitem[Li et~al.(2000)Li, Perrigne, and Vuong]{Li2000}
T.~Li, I.~Perrigne, and Q.~Vuong.
\newblock {Conditionally independent private information in OCS wildcat
  auctions}.
\newblock \emph{Journal of Econometrics}, 98\penalty0 (1):\penalty0 129--161,
  2000.
\newblock ISSN 03044076.
\newblock \doi{10.1016/S0304-4076(99)00081-0}.

\bibitem[Li et~al.(2020)Li, Mark, Raskutti, and Willett]{Li2018}
Y.~Li, B.~Mark, G.~Raskutti, and R.~Willett.
\newblock {Graph-based regularization for regression problems with
  highly-correlated designs}.
\newblock \emph{SIAM Journal on Mathematics of Data Science}, 2\penalty0
  (2):\penalty0 480--504, 2020.
\newblock \doi{10.1137/19M1287365}.
\newblock URL \url{http://arxiv.org/abs/1803.07658}.

\bibitem[Lintott et~al.(2008)Lintott, Schawinski, Slosar, Land, Bamford,
  Thomas, Raddick, Nichol, Szalay, Andreescu, Murray, and
  Vandenberg]{10.1111/j.1365-2966.2008.13689.x}
C.~J. Lintott, K.~Schawinski, A.~Slosar, K.~Land, S.~Bamford, D.~Thomas, M.~J.
  Raddick, R.~C. Nichol, A.~Szalay, D.~Andreescu, P.~Murray, and J.~Vandenberg.
\newblock {Galaxy Zoo: morphologies derived from visual inspection of galaxies
  from the Sloan Digital Sky Survey}.
\newblock \emph{Monthly Notices of the Royal Astronomical Society},
  389\penalty0 (3):\penalty0 1179--1189, 09 2008.
\newblock ISSN 0035-8711.
\newblock \doi{10.1111/j.1365-2966.2008.13689.x}.
\newblock URL \url{https://doi.org/10.1111/j.1365-2966.2008.13689.x}.

\bibitem[Loh and Wainwright(2012)]{loh2012}
P.-L. Loh and M.~J. Wainwright.
\newblock High-dimensional regression with noisy and missing data: Provable
  guarantees with nonconvexity.
\newblock \emph{Ann. Statist.}, 40\penalty0 (3):\penalty0 1637--1664, 06 2012.
\newblock \doi{10.1214/12-AOS1018}.
\newblock URL \url{https://doi.org/10.1214/12-AOS1018}.

\bibitem[Mason and Suri(2012)]{Mason2012}
W.~Mason and S.~Suri.
\newblock {Conducting behavioral research on Amazon's Mechanical Turk}.
\newblock \emph{Behavior Research Methods}, 44\penalty0 (1):\penalty0 1--23,
  2012.
\newblock ISSN 1554351X.
\newblock \doi{10.3758/s13428-011-0124-6}.

\bibitem[Meinshausen and B{\"u}hlmann(2010)]{meinshausen-2010}
N.~Meinshausen and P.~B{\"u}hlmann.
\newblock Stability selection.
\newblock \emph{Journal of the Royal Statistical Society: Series B (Statistical
  Methodology)}, 72\penalty0 (4):\penalty0 417--473, May 2010.

\bibitem[Nghiem and Potgieter(2019)]{Nghiem2019}
L.~Nghiem and C.~Potgieter.
\newblock {Simulation-selection-extrapolation: Estimation in high-dimensional
  errors-in-variables models}.
\newblock \emph{Biometrics}, 75\penalty0 (4):\penalty0 1133--1144, 2019.
\newblock ISSN 15410420.
\newblock \doi{10.1111/biom.13112}.

\bibitem[Nogueira et~al.(2018)Nogueira, Sechidis, and Brown]{Nogueira2018}
S.~Nogueira, K.~Sechidis, and G.~Brown.
\newblock {On the stability of feature selection algorithms}.
\newblock \emph{Journal of Machine Learning Research}, 18:\penalty0 1--54,
  2018.
\newblock ISSN 15337928.
\newblock URL \url{http://jmlr.org/papers/v18/17-514.html.}

\bibitem[Nordborg et~al.(2002)Nordborg, Borevitz, Bergelson, Berry, Chory,
  Hagenblad, Kreitman, Maloof, Noyes, Oefner, Stahl, and Weigel]{Nordborg2002}
M.~Nordborg, J.~O. Borevitz, J.~Bergelson, C.~C. Berry, J.~Chory, J.~Hagenblad,
  M.~Kreitman, J.~N. Maloof, T.~Noyes, P.~J. Oefner, E.~A. Stahl, and
  D.~Weigel.
\newblock {The extent of linkage disequilibrium in Arabidopsis thaliana}.
\newblock \emph{Nature Genetics}, 30\penalty0 (2):\penalty0 190--193, 2002.
\newblock ISSN 10614036.
\newblock \doi{10.1038/ng813}.

\bibitem[Park et~al.(2007)Park, Hastie, and Tibshirani]{Park2007}
M.~Y. Park, T.~Hastie, and R.~Tibshirani.
\newblock {Averaged gene expressions for regression}.
\newblock \emph{Biostatistics}, 8\penalty0 (2):\penalty0 212--227, apr 2007.
\newblock ISSN 1465-4644.
\newblock \doi{10.1093/biostatistics/kxl002}.
\newblock URL
  \url{https://academic.oup.com/biostatistics/article-lookup/doi/10.1093/biostatistics/kxl002}.

\bibitem[Pinelis and Molzon(2016)]{pinelis2016}
I.~Pinelis and R.~Molzon.
\newblock Optimal-order bounds on the rate of convergence to normality in the
  multivariate delta method.
\newblock \emph{Electronic Journal of Statistics}, 10\penalty0 (1):\penalty0
  1001--1063, 2016.
\newblock \doi{10.1214/16-EJS1133}.
\newblock URL \url{https://doi.org/10.1214/16-EJS1133}.

\bibitem[Reid and Tibshirani(2016)]{Reid2015}
S.~Reid and R.~Tibshirani.
\newblock {Sparse regression and marginal testing using cluster prototypes}.
\newblock \emph{Biostatistics}, 17\penalty0 (2):\penalty0 364--376, 11 2016.
\newblock ISSN 1465-4644.
\newblock \doi{10.1093/biostatistics/kxv049}.
\newblock URL \url{https://doi.org/10.1093/biostatistics/kxv049}.

\bibitem[Rosenbaum and Tsybakov(2010)]{Rosenbaum2010}
M.~Rosenbaum and A.~B. Tsybakov.
\newblock {Sparse recovery under matrix uncertainty}.
\newblock \emph{Annals of Statistics}, 38\penalty0 (5):\penalty0 2620--2651,
  2010.
\newblock ISSN 00905364.
\newblock \doi{10.1214/10-AOS793}.
\newblock URL
  \url{https://projecteuclid-org.libproxy1.usc.edu/download/pdfview{\_}1/euclid.aos/1278861455}.

\bibitem[{Rosenbaum} and {Tsybakov}(2013)]{rosenbaum}
M.~{Rosenbaum} and A.~B. {Tsybakov}.
\newblock {Improved matrix uncertainty selector}.
\newblock In \emph{From probability to statistics and back: high-dimensional
  models and processes. A Festschrift in honor of Jon A. Wellner. Including
  papers from the conference, Seattle, WA, USA, July 28--31, 2010}, pages
  276--290. Beachwood, OH: IMS, Institute of Mathematical Statistics, 2013.
\newblock ISBN 978-0-940600-83-6.
\newblock \doi{10.1214/12-IMSCOLL920}.

\bibitem[Schennach(2016)]{Schennach2016}
S.~M. Schennach.
\newblock {Recent Advances in the Measurement Error Literature}.
\newblock \emph{Annual Review of Economics}, 8:\penalty0 341--377, 2016.
\newblock \doi{10.1146/annurev-economics-080315-015058}.
\newblock URL \url{www.annualreviews.org}.

\bibitem[Segal et~al.(2004)Segal, Dahlquist, and Conklin]{Segal2004}
M.~R. Segal, K.~D. Dahlquist, and B.~R. Conklin.
\newblock {Regression Approaches for Microarray Data Analysis}.
\newblock \emph{Journal of Computational Biology}, 10\penalty0 (6):\penalty0
  961--980, 2004.
\newblock ISSN 1066-5277.
\newblock \doi{10.1089/106652703322756177}.

\bibitem[Sesia et~al.(2019)Sesia, Sabatti, and Cand{\`{e}}s]{Sesia2019}
M.~Sesia, C.~Sabatti, and E.~J. Cand{\`{e}}s.
\newblock {Gene hunting with hidden Markov model knockoffs}.
\newblock \emph{Biometrika}, 106\penalty0 (1):\penalty0 1--18, 2019.
\newblock ISSN 14643510.
\newblock \doi{10.1093/biomet/asy033}.
\newblock URL
  \url{https://academic.oup.com/biomet/article-abstract/106/1/1/5066539}.

\bibitem[Shah and Samworth(2012)]{shah_samworth_2012}
R.~D. Shah and R.~J. Samworth.
\newblock Variable selection with error control: another look at stability
  selection.
\newblock \emph{Journal of the Royal Statistical Society: Series B (Statistical
  Methodology)}, 75\penalty0 (1):\penalty0 55--80, 2012.
\newblock \doi{10.1111/j.1467-9868.2011.01034.x}.

\bibitem[Shah and Samworth(2013)]{shah_samworth_2013}
R.~D. Shah and R.~J. Samworth.
\newblock Discussion of `correlated variables in regression: Clustering and
  sparse estimation'.
\newblock \emph{Journal of Statistical Planning and Inference}, 143\penalty0
  (11):\penalty0 1866--1868, 2013.
\newblock \doi{10.1016/j.jspi.2013.05.022}.

\bibitem[Sharma et~al.(2013)Sharma, Bondell, and {Helen Zhang}]{Sharma}
D.~B. Sharma, H.~D. Bondell, and H.~{Helen Zhang}.
\newblock {Consistent Group Identification and Variable Selection in Regression
  With Correlated Predictors}.
\newblock \emph{Journal of Computational and Graphical Statistics}, 22\penalty0
  (2):\penalty0 319--340, 2013.
\newblock \doi{10.1080/15533174.2012.707849}.
\newblock URL
  \url{https://www.tandfonline.com/action/journalInformation?journalCode=ucgs20}.

\bibitem[She(2010)]{She2010}
Y.~She.
\newblock {Sparse regression with exact clustering}.
\newblock \emph{Electronic Journal of Statistics}, 4:\penalty0 1055--1096,
  2010.
\newblock \doi{10.1214/10-EJS578}.
\newblock URL
  \url{https://projecteuclid.org/download/pdfview{\_}1/euclid.ejs/1286889184}.

\bibitem[Shen and Huang(2010)]{Shen2010}
X.~Shen and H.~C. Huang.
\newblock {Grouping pursuit through a regularization solution surface}.
\newblock \emph{Journal of the American Statistical Association}, 105\penalty0
  (490):\penalty0 727--739, 2010.
\newblock ISSN 01621459.
\newblock \doi{10.1198/jasa.2010.tm09380}.
\newblock URL
  \url{https://www.tandfonline.com/action/journalInformation?journalCode=uasa20}.

\bibitem[S{\o}rensen et~al.(2015)S{\o}rensen, Frigessi, and
  Thoresen]{Sorensen2015}
{\O}.~S{\o}rensen, A.~Frigessi, and M.~Thoresen.
\newblock {Measurement error in Lasso: Impact and likelihood bias correction}.
\newblock \emph{Statistica Sinica}, 25\penalty0 (2):\penalty0 809--829, 2015.
\newblock ISSN 10170405.
\newblock \doi{10.5705/ss.2013.180}.
\newblock URL \url{http://dx.doi.org/10.5705/ss.2013.180}.

\bibitem[S{\o}rlie et~al.(2003)S{\o}rlie, Tibshirani, Parker, Hastie, Marron,
  Nobel, Deng, Johnsen, Pesich, Geisler, and et~al.]{sorlie_tibshirani_2003}
T.~S{\o}rlie, R.~Tibshirani, J.~Parker, T.~Hastie, J.~S. Marron, A.~Nobel,
  S.~Deng, H.~Johnsen, R.~Pesich, S.~Geisler, and et~al.
\newblock Repeated observation of breast tumor subtypes in independent gene
  expression data sets.
\newblock \emph{Proceedings of the National Academy of Sciences}, 100\penalty0
  (14):\penalty0 8418--8423, 2003.
\newblock \doi{10.1073/pnas.0932692100}.

\bibitem[Sun et~al.(2018)Sun, Tan, Liu, and Zhang]{pmlr-v80-sun18c}
Q.~Sun, K.~M. Tan, H.~Liu, and T.~Zhang.
\newblock Graphical nonconvex optimization via an adaptive convex relaxation.
\newblock In J.~Dy and A.~Krause, editors, \emph{Proceedings of the 35th
  International Conference on Machine Learning}, volume~80 of \emph{Proceedings
  of Machine Learning Research}, pages 4810--4817, Stockholmsm{\"a}ssan,
  Stockholm Sweden, 10--15 Jul 2018. PMLR.
\newblock URL \url{http://proceedings.mlr.press/v80/sun18c.html}.

\bibitem[Tibshirani(1996)]{Tibshirani1996}
R.~Tibshirani.
\newblock {Regression Shrinkage and Selection via the Lasso}.
\newblock \emph{Journal of the Royal Statistical Society. Series B: Statistical
  Methodology}, 58\penalty0 (1):\penalty0 267--288, 1996.

\bibitem[Tibshirani(2013)]{Tibshirani2013}
R.~J. Tibshirani.
\newblock {The lasso problem and uniqueness}.
\newblock \emph{Electronic Journal of Statistics}, 7\penalty0 (1):\penalty0
  1456--1490, 2013.
\newblock ISSN 19357524.
\newblock \doi{10.1214/13-EJS815}.

\bibitem[Togninalli et~al.(2017)Togninalli, Seren, Meng, Fitz, Nordborg,
  Weigel, Borgwardt, Korte, and Grimm]{Togninalli2017}
M.~Togninalli, {\"U}.~Seren, D.~Meng, J.~Fitz, M.~Nordborg, D.~Weigel,
  K.~Borgwardt, A.~Korte, and D.~G. Grimm.
\newblock {The AraGWAS Catalog: a curated and standardized Arabidopsis thaliana
  GWAS catalog}.
\newblock \emph{Nucleic Acids Research}, 46\penalty0 (D1):\penalty0
  D1150--D1156, 10 2017.
\newblock ISSN 0305-1048.
\newblock \doi{10.1093/nar/gkx954}.
\newblock URL \url{https://doi.org/10.1093/nar/gkx954}.

\bibitem[von Ahn and Dabbish(2004)]{10.1145/985692.985733}
L.~von Ahn and L.~Dabbish.
\newblock Labeling images with a computer game.
\newblock In \emph{Proceedings of the SIGCHI Conference on Human Factors in
  Computing Systems}, CHI '04, pages 319--326, New York, NY, USA, 2004.
  Association for Computing Machinery.
\newblock ISBN 1581137028.
\newblock \doi{10.1145/985692.985733}.
\newblock URL \url{https://doi.org/10.1145/985692.985733}.

\bibitem[Witten et~al.(2014)Witten, Shojaie, and Zhang]{Witten2014}
D.~M. Witten, A.~Shojaie, and F.~Zhang.
\newblock {The cluster elastic net for high-dimensional regression with unknown
  variable grouping}.
\newblock \emph{Technometrics}, 56\penalty0 (1):\penalty0 112--122, 2014.
\newblock ISSN 00401706.
\newblock \doi{10.1080/00401706.2013.810174}.
\newblock URL
  \url{https://www.tandfonline.com/action/journalInformation?journalCode=utch20http://www.tandfonline.com/r/TECH}.

\bibitem[Yu(2013)]{Yu2013}
B.~Yu.
\newblock {Stability}.
\newblock \emph{Bernoulli}, 19\penalty0 (4):\penalty0 1484--1500, 2013.
\newblock \doi{10.3150/13-BEJSP14}.
\newblock URL
  \url{https://projecteuclid.org/download/pdfview{\_}1/euclid.bj/1377612862}.

\bibitem[Yu and Kumbier(2020)]{Yu3920}
B.~Yu and K.~Kumbier.
\newblock Veridical data science.
\newblock \emph{Proceedings of the National Academy of Sciences}, 117\penalty0
  (8):\penalty0 3920--3929, 2020.
\newblock ISSN 0027-8424.
\newblock \doi{10.1073/pnas.1901326117}.
\newblock URL \url{https://www.pnas.org/content/117/8/3920}.

\bibitem[Zhao and Yu(2006)]{Zhao2006}
P.~Zhao and B.~Yu.
\newblock {On Model Selection Consistency of Lasso}.
\newblock \emph{Journal of Machine Learning Research}, 7:\penalty0 2541--2563,
  2006.

\bibitem[Zheng et~al.(2018)Zheng, Li, Yu, and Li]{Zheng2018}
Z.~Zheng, Y.~Li, C.~Yu, and G.~Li.
\newblock {Balanced estimation for high-dimensional measurement error models}.
\newblock \emph{Computational Statistics and Data Analysis}, 126:\penalty0
  78--91, oct 2018.
\newblock ISSN 01679473.
\newblock \doi{10.1016/j.csda.2018.04.009}.

\bibitem[Zou and Hastie(2005)]{Zou2005}
H.~Zou and T.~Hastie.
\newblock {Regularization and variable selection via the elastic net}.
\newblock \emph{Journal of the Royal Statistical Society. Series B: Statistical
  Methodology}, 67\penalty0 (2):\penalty0 301--320, 2005.
\newblock ISSN 00426989.
\newblock \doi{10.1016/S0042-6989(99)00110-8}.

\end{thebibliography}

\newpage

\appendix

\section{Proof of Theorem \ref{thm.result.sel}}\label{thm.result.sel.proof}

\subsection{Proof of Statement (i)} 

Note that
\begin{align*}
 \frac{100}{\sqrt{n \log n}} & < \frac{19}{5} \sqrt{ \frac{2 + \sigma_\epsilon^2 }{4 c_2 }} \frac{\left(\log n \right)^{3/4}}{n^{1/2}}
 \\ \iff \qquad  \frac{500}{19} \sqrt{ \frac{4 c_2 }{2 + \sigma_\epsilon^2 }}  & <\left(\log n \right)^{5/4}
  \\ \iff \qquad \exp \left\{ \left( \frac{1000}{19} \sqrt{ \frac{ c_2 }{2 + \sigma_\epsilon^2 }} \right)^{4/5} \right\} & < n .
\end{align*}
Because \(c_2 < (e-1)/(8e^2)\) and this quantity is decreasing in \(\sigma_\epsilon^2 \geq 0\), a sufficient condition for this is
\[
n >  81 > \exp \left\{ \left( \frac{1000}{19} \sqrt{ \frac{ e-1 }{8e^2 \cdot 2  }}  \right)^{4/5} \right\}  ,
\]
so 
\[
10 \sigma_\zeta^2(n) =  \frac{100}{\sqrt{n \log n}} < \frac{19}{10}\sqrt{\frac{ 2 + \sigma_\epsilon^2 }{c_2}}   \frac{\left(\log n \right)^{3/4}}{n^{1/2}}
\]
holds for \(n \geq 100\). It remains to show that \(I(n) \subseteq (1,2)\) under our assumptions. It is clear that \(1 + 10 \sigma_\zeta^2(n) > 1\). To see that the upper bound of \(I(n)\) is less than 2, note that since
\[
 \frac{\left(\log n \right)^{3/4}}{n^{1/2}}  < \sqrt{\frac{c_2}{3.61 \left( 5 + \sigma_\epsilon^2 \right)}} =  \frac{10}{19}\sqrt{\frac{c_2}{ 5 + \sigma_\epsilon^2 }} 
 \]
 from \eqref{c4n.conds.max}, we have
 \begin{align*}
 \frac{19}{10}\sqrt{\frac{ 2 + \sigma_\epsilon^2 }{c_2}}   \frac{\left(\log n \right)^{3/4}}{n^{1/2}} & < \sqrt{\frac{2 + \sigma_\epsilon^2}{5  + \sigma_\epsilon^2}} < 1,
 \end{align*} 
so we have that \(\beta_Z \in (1, 2)\) for all \(n\) satisfying the assumptions of Theorem \ref{thm.result.sel}. 
\subsection{Proof of Statement (ii)} Our proof strategy will be to walk through the lasso path as \(\lambda\) decreases from \(\infty\) and the first two features enter to show that a few events are sufficient for \(\boldsymbol{X}_{\cdot 3}\) to be the second feature to enter the lasso path. We then show that one of these events holds with probability tending towards \(1/2\) and the rest hold with probability tending towards 1. Then the final result will come from a union bound. 

Throughout this proof we will refer to the Karush-Kuhn-Tucker (KKT) conditions
\begin{equation}\label{lasso.gen.simplest.kkt.conds}
-\frac{1}{n  \lVert \boldsymbol{X}_{\cdot j} \rVert_2} \boldsymbol{X}_{\cdot j}^\top\left(\boldsymbol{y} - \sum_{\ell =1}^3 \frac{\boldsymbol{X}_{\cdot \ell}}{\lVert \boldsymbol{X}_{\cdot \ell} \rVert_2} \hat{\beta}_\ell(\lambda) \right) + \lambda s_j = 0, \qquad \forall j \in [3]
\end{equation}
where \(\hat{\beta}_\ell(\lambda)\) is the lasso estimated coefficient for feature \(\ell\) at \(\lambda\) and 
\[
s_j \in \begin{cases}
\left\{ \operatorname{sgn} \left(\hat{\beta}_j(\lambda) \right) \right\}, & \hat{\beta}_j(\lambda) \neq 0 \\
[-1, 1], & \hat{\beta}_j(\lambda) = 0
\end{cases}, \qquad \forall j \in [3].
\]
Consider
\begin{equation}\label{def.lambda.1}
\lambda_1
=   \max_j   \left\{ \frac{\left| \boldsymbol{X}_{\cdot j}^\top \boldsymbol{y} \right| }{n \lVert \boldsymbol{X}_{\cdot j} \rVert_2 } \right\} .
\end{equation}
For \(\lambda \geq \lambda_1\), \(\hat{\beta}(\lambda) = 0\) is a solution to \eqref{lasso.solution} because it satisfies \eqref{lasso.gen.simplest.kkt.conds} with
\[
s_j = \frac{\boldsymbol{X}_{\cdot j}^T \boldsymbol{y}}{n \lambda \lVert \boldsymbol{X}_{\cdot j} \rVert_2},
\]
and \(s_j \in [-1, 1]\) as long as \(\lambda \geq \lambda_1\). 
The first feature enters the active set for \(\lambda < \lambda_1\), and is the feature \(j\) attaining the maximum in \eqref{def.lambda.1}. That is, if all of these sample correlations are positive, the first feature to enter the active set is the one with the largest correlation with \(\boldsymbol{y}\). Define the uncentered sample correlations 
%, we know \(\hat{\beta}(\infty) = 0\); in particular, for sufficiently large \(\lambda\), \(\hat{\beta}(\lambda) =0\) satisfies the KKT conditions \eqref{lasso.gen.simplest.kkt.conds}. That is, for sufficiently large \(\lambda\), for every \(j \in [3]\) there exists some \(s_j \in [-1, 1]\) such that
%the c
%\begin{align}
%& \frac{1}{n  \lVert \boldsymbol{X}_{\cdot j} \rVert_2} \boldsymbol{X}_{\cdot j}^\top \boldsymbol{y} = \lambda s_j .\nonumber
%%\\ \iff \qquad &  X_j^\top y = \lambda \sqrt{n} \lVert X_j \rVert_2 s_j \nonumber
%\end{align}
%
%This implies
%
%\begin{align*}
%& \left| \frac{\boldsymbol{X}_{\cdot j}^\top \boldsymbol{y} }{n  \lVert \boldsymbol{X}_{\cdot j} \rVert_2} \right| \leq \lambda .
%%\label{lasso.gen.simplest.ktt.2}
%%\\ \iff \qquad &    \lambda s_j =  \frac{\boldsymbol{X}_{\cdot j}^\top \boldsymbol{y}}{n \lVert \boldsymbol{X}_{\cdot j} \rVert_2 }  
%\end{align*}
%\begin{equation}\label{lasso.gen.lambda.1.exp}
%\end{equation}
%In fact, \(\hat{\beta}(\lambda) =0\) for all \(\lambda \in (\lambda_1, \infty]\), where
%
%
%\lambda_1 = \max_j \left\{ \frac{n^{-1} | X_j^\topy|}{n^{-1/2} \lVert X_j \rVert_2} \right\},
\begin{equation}\label{lemma.unc.samp.corr.def}
\hat{R}_{jy} = \frac{\boldsymbol{X}_{\cdot j}^\top \boldsymbol{y}}{\lVert \boldsymbol{X}_{\cdot j} \rVert_2 \lVert \boldsymbol{y} \rVert_2} 
%= \frac{\tilde{x}_j^\topy}{\lVert \boldsymbol{y} \rVert_2}
, \qquad j \in [3] , \qquad \hat{R}_{12} = \frac{\boldsymbol{X}_{\cdot 2}^\top \boldsymbol{X}_{\cdot 1}}{\lVert \boldsymbol{X}_{\cdot 2} \rVert_2 \lVert \boldsymbol{X}_{\cdot 1} \rVert_2} 
.
\end{equation}
Define the events
\begin{center}
\noindent\fbox{
\parbox{0.9\textwidth}{
\begin{align}
\mathcal{A}_{12} := & \left\{\hat{R}_{1y} -  \hat{R}_{2y}  >   0 \right\},  \nonumber % \label{lasso.gen.def.a1.eta}
\\ \mathcal{A}_{13} := &  \left\{ \hat{R}_{1y} - \hat{R}_{3y} > 0 \right\}, \qquad \text{and} \nonumber % \label{lasso.gen.def.a13.eta}
\\
S_1 :=  & \left\{ \hat{R}_{1y} > 0, \hat{R}_{2y} > 0, \hat{R}_{3y} > 0,  \hat{R}_{12} > 0 \right\}. \nonumber % \label{lasso.gen.def.a1.s1}
\end{align}
}
}
\end{center}
Note that under \(\mathcal{A}_{12} \cap \mathcal{A}_{13}  \cap S_1\), \(\boldsymbol{X}_{\cdot 1}\) is the first feature to enter the active set and \(s_1 = \operatorname{sgn}(\boldsymbol{X}_{\cdot 1}^\top \boldsymbol{y}) = 1\).

Next we will consider the second feature to enter the active set. Denote by \(\operatorname{supp} \left( \hat{\beta}(\lambda) \right) \subset [3]\) the active set at \(\lambda\). Let \(\lambda_2\) be the first (greatest) \(\lambda < \lambda_1\) where \( \operatorname{supp} \left( \hat{\beta}(\lambda) \right) \neq \operatorname{supp} \left(  \hat{\beta}(\lambda_1) \right) \). At each knot in the lasso path, a feature may either enter or leave the active set. We show that \(\boldsymbol{X}_{\cdot 1}\) cannot leave the active set before another feature enters:
\begin{lemma}\label{lemma.x1.not.leave}

The first feature to enter the lasso path cannot leave at the second knot. (That is, at the second knot in the lasso path, a second feature enters the model with probability one.)

\end{lemma}
(The proofs of all lemmas stated in the proof of this theorem are provided in Appendix \ref{sec.lemma.prob}.) That is, at \(\lambda_2\), \(\boldsymbol{X}_{\cdot 1}\) is never removed from the active set; instead, either \(\boldsymbol{X}_{\cdot 2}\) or \(\boldsymbol{X}_{\cdot 3}\) enters. In particular, for \(\boldsymbol{X}_{\cdot 2}\) and \(\boldsymbol{X}_{\cdot 3}\) there exist knots \(\lambda_2^{(2)}\) and \(\lambda_2^{(3)}\) determined by the KKT conditions \eqref{lasso.gen.simplest.kkt.conds} such that the next feature to enter the active set is \(\underset{i \in \{2, 3\}}{\arg \max} \left\{\lambda_2^{(i)}:  \lambda_2^{(i)} < \lambda_1 \right\}\) \citep{Tibshirani2013}. Therefore to show that \(\boldsymbol{X}_{\cdot 3}\) enters before \(\boldsymbol{X}_{\cdot 2}\), it is enough to show

\begin{equation}\label{lambda.ineq.order}
\lambda_2^{(2)} < \lambda_2^{(3)} < \lambda_1.
\end{equation}

We will calculate \(\lambda_2^{(2)}\) and \(\lambda_2^{(3)}\) to determine which feature enters next. We will show in Lemma \ref{lemma.e.event} that on the event
\begin{center}
\noindent\fbox{
\parbox{0.9\textwidth}{
% \begin{equation}\label{lasso.gen.def.a3.prime}
 \[
% \boxed{
% \mathcal{A}_2 := \{ s_3 = 1\}.
 \mathcal{A}_3 := \left\{ \hat{R}_{3y} -   \hat{R}_{13}  \hat{R}_{1y} 
 \geq 0 \right\},
% }
\]
% \end{equation}
 }
 }
 \end{center}
we have
\begin{equation} \label{lasso.gen.lambda.3.exp}
 \lambda_2^{(3)} =  \frac{\lVert \boldsymbol{y} \rVert_2}{n} \frac{ \hat{R}_{3y}  - \hat{R}_{13} \hat{R}_{1y} }{1 -  \hat{R}_{13}}.
 \end{equation}
Define the events
\begin{center}
\noindent\fbox{
\parbox{0.9\textwidth}{
%\begin{equation}\label{lasso.gen.cond.sel.order2}
\[
\mathcal{E}_1 := 
\left\{
\frac{\hat{R}_{1y} - \hat{R}_{2y}}{1  - \hat{R}_{12}} 
> 
\frac{\hat{R}_{1y} - \hat{R}_{3y}}{1  - \hat{R}_{13}} 
\right\}  , \qquad  \mathcal{E}_2 := \left\{
\frac{\hat{R}_{1y} + \hat{R}_{2y}}{1  + \hat{R}_{12}} 
> 
\frac{\hat{R}_{1y} - \hat{R}_{3y}}{1  - \hat{R}_{13}} 
\right\} .
\]
%\end{equation}
}
}
\end{center}
We show that \(\mathcal{A}_3  \cap \mathcal{E}_1  \cap \mathcal{E}_2\), along with the other events so far, is sufficient to ensure that \(\boldsymbol{X}_{\cdot 3}\) enters the active set before \(\boldsymbol{X}_{\cdot 2}\):
 \begin{lemma}\label{lemma.e.event}
 
Under the event \( \mathcal{A}_{12} \cap \mathcal{A}_{13}  \cap S_1  \cap \mathcal{A}_3 \cap \mathcal{E}_1  \cap \mathcal{E}_2 \), \eqref{lambda.ineq.order} and \eqref{lasso.gen.lambda.3.exp} hold, and the first two features in the lasso selection path are \(\boldsymbol{X}_{\cdot 1}\) followed by \(\boldsymbol{X}_{\cdot 3}\).
 
 \end{lemma}

We will be almost done if we can show that \( \mathcal{A}_{12}\) occurs with probability tending towards \(1/2\) and the rest of the events occur with probability tending towards 1. It is straightforward that \(\mathbb{P}(\mathcal{A}_{12}) = 1/2\) for all \(n\) by exchangeability of \(\boldsymbol{X}_{\cdot 1}\) and \(\boldsymbol{X}_{\cdot 2}\). However, the event \(\mathcal{E}_1\) complicates our analysis. Specifically, since \(\E[\hat{R}_{1y} - \hat{R}_{3y}] = \rho_{1y}(n) - \rho_{3y} > 0\) and asymptotically \(1 - \hat{R}_{13} \to 1 - \rho_{13} = 1\) and \(\E[1 - \hat{R}_{12}] = 1 - \rho_{12}(n) \to 0\), for \(\mathcal{E}_1\) to hold we need \(\hat{R}_{1y} - \hat{R}_{2y}\) to vanish slowly enough in \(n\). To sort this out, we define
%it turns out that showing this will be complicated by the fact that \(\hat{R}_{1y} - \hat{R}_{2y}\) appears in both \(\mathcal{A}_{12}(n)\) and the left side of \(\mathcal{E}_1\). 
\[
\mathcal{A}_{12}(n) := \left\{\hat{R}_{1y} -  \hat{R}_{2y}  >   \eta(n) \right\}
\]
and
\begin{equation}\label{lasso.gen.cond.sel.order3}
\tilde{\mathcal{E}}_1(n) := \left\{ \frac{\eta(n)}{1 - \hat{R}_{12}} > \frac{\hat{R}_{1y} - \hat{R}_{3y}}{1 - \hat{R}_{13}} \right\}
\end{equation}
for some function \(\eta: \mathbb{N} \to \mathbb{R}_{++}\). Note that \(\mathcal{A}_{12}(n)  \cap \tilde{\mathcal{E}}_1(n) \) implies \( \mathcal{A}_{12} \cap \mathcal{E}_1\), so we can change our focus to bounding the probability of the event \(\mathcal{A}_{12}(n) \cap \mathcal{A}_{13} \cap S_1  \cap \mathcal{A}_3  \cap \tilde{\mathcal{E}}_1(n) \cap \mathcal{E}_2\). As discussed, if \(\eta(n)\) goes to 0 too quickly, \(\tilde{\mathcal{E}}_1(n)\) could fail to hold with high probability. But as \(n \to \infty\), \(\hat{R}_{1y} - \hat{R}_{2y}\) will concentrate around its expectation, 0, with high probability, so \( \eta(n)\) will need to tend towards 0 quickly enough for \(\mathcal{A}_{12}(n)\) to hold with high probability, 

%In particular, under our assumptions the probability of \(\tilde{\mathcal{E}}_1(\eta)\) will tend towards 1. Examining the right side of \eqref{lasso.gen.cond.sel.order3}, \(\hat{R}_{13}\) is an estimator for \(\rho_{13} = 0\), so the right side is roughly \(\rho_{1y}(n) - \rho_{3y} > 0\). In the denominator of the left side, \( \hat{R}_{12}\) tends to 1 because \(\sigma_\zeta^2(n)\) (the variance of the noise added to \(\boldsymbol{Z}\) to get \(\boldsymbol{X}_{\cdot 1}\) and \(\boldsymbol{X}_{\cdot 2}\)) tends to 0. Specifically, the rate of convergence of the denominator to 0 will be driven by the quantity
%\[
% 1 - \rho_{12}(n)=   \frac{\sigma_\zeta^2(n)}{1 + \sigma_\zeta^2(n)}   = \frac{10}{\sqrt{n \log n} + 10} 
%\]
%where \(\sigma_\zeta^2(n) = 10/\sqrt{n \log n} \) as in \eqref{lasso.thm.sig.zeta.def} and the expression for \(\rho_{12}(n)\) is calculated in Lemma \ref{lemma.dist.bounds}. 
Soon we will concern ourselves with a good choice of \(\eta(n)\), but for now, the following result allows us to bound \(\mathbb{P} \left( \mathcal{A}_{12}(n) \right)\) for an arbitrary \(\eta\).
\begin{proposition}\label{lemma.prob.a1.eta} Suppose
\begin{equation}\label{lemma.cov.matrix}
\begin{pmatrix}
y_i \\
X_{i1} \\
X_{i2} 
\end{pmatrix} \sim \mathcal{N} \left( \boldsymbol{0}, \begin{bmatrix}
\Sigma_{yy} & \Sigma_{1y} & \Sigma_{2y} \\
\Sigma_{1y} & \Sigma_{11} & \Sigma_{12}  \\
\Sigma_{2y} & \Sigma_{12} & \Sigma_{22} 
\end{bmatrix} \right), \qquad \forall i \in [n],
\end{equation}
are \(n\) i.i.d. draws with \(\Sigma_{yy} > \Sigma_{1y} = \Sigma_{2y} > \Sigma_{11} = \Sigma_{22} > \Sigma_{12}  \geq 1\). Let \( \boldsymbol{y} := (y_1, \ldots, y_n)^\top\) and \(\boldsymbol{X}_{\cdot j} := (X_{1j}, \ldots, X_{nj})^\top, j \in [3]\). Assume \(n \geq 100\). Define the uncentered sample correlations as in \eqref{lemma.unc.samp.corr.def}. Then for any \(\eta > 0\),
 \begin{align}
 \mathbb{P} \left(  \hat{R}_{1y} - \hat{R}_{2y} \leq \eta \right)   \leq  & \Phi\left( \frac{ \eta \sqrt{n}}{\tilde{\sigma}}\right) +   \left(    \frac{1463}{ \tilde{\sigma}} + 14 \right)  \frac{\Sigma_{yy}^3}{n^{1/2}}   , \label{lemma.prob.ineq.1}
 \end{align}
where \(\tilde{\sigma}\) is defined in \eqref{lemma.def.tilde.sigma} and \(\Phi\) is the distribution function for a standard Gaussian random variable.

\end{proposition}
\begin{proof} Provided in Appendix \ref{proof.prop}.
\end{proof}
To briefly summarize how we prove Proposition \ref{lemma.prob.a1.eta}, note that \( \hat{R}_{1y} - \hat{R}_{2y} \) is a well-behaved functional of the jointly Gaussian data which, when scaled by \(\sqrt{n}\), converges to a normal distribution by the delta method. Our proof relies on a Berry-Esseen-type result for the delta method due to \citet{pinelis2016}. Because in our setting the correlation between \(\boldsymbol{X}_{\cdot 1}\) and \(\boldsymbol{X}_{\cdot 2}\) varies with \(n\), in our setting \(1/\tilde{\sigma} = \mathcal{O}( (n \log n)^{1/4})\), so examining the argument of \(\Phi(\cdot)\) in \eqref{lemma.prob.ineq.1} we see we will require a choice of \( \eta = \eta(n)\) that goes to 0 quickly enough that \(\eta(n) \cdot n^{3/4} (\log n)^{1/4} \to 0\); then \( \mathbb{P} \left( \mathcal{A}_{12}(n) \right) \to 1/2\).

Meanwhile, all of the events besides \(\mathcal{A}_{12}(n)\) can be shown to hold if the sample correlations concentrate around their expectations due to the following lemma:

\begin{lemma}\label{lem.ineq.assum}

Under the assumptions of Theorem \ref{thm.result.sel}, the following identities and inequalities hold:
%From the setup and assumptions of Theorem \ref{thm.result.sel}, 
\begin{align}
\rho_{12}(n) & > \delta (n) , \label{itm:3}
\\ \rho_{1y}(n) & = \rho_{2y}(n) > \rho_{3y}  >3 \delta (n) > 0,  \label{itm:1} 
 \\ \rho_{1y}(n) -  \rho_{3y}  & \geq 2\delta (n) ,  \text{ and}  \label{itm:7}
\\ \rho_{1y}(n) -  \rho_{3y}  & \leq \frac{19}{5} \left(\log n \right)^{1/4} \delta (n)  , \label{itm:8a}
%\\ \rho_{1y}(n) \left( 1 - \rho_{12}(n)  \right) &\leq 3 \delta (n) , \label{itm:4}
\end{align}
where
\begin{equation}\label{def.delta.n}
\delta(n) := \sqrt{\frac{(\beta_Z^2 + 1 + \sigma_\epsilon^2) \log n}{4 c_2 n}}  =\frac{1}{20}\sqrt{\frac{\beta_Z^2 + 1 + \sigma_\epsilon^2}{c_2}} \log(n) \sigma_\zeta^2(n) ,
\end{equation}
where \(c_2 \in \left(0, \frac{e-1}{8e^2}\right)\) is defined in \eqref{defn.c2}.
\end{lemma}

Examining the definitions of events \(\mathcal{A}_{13} \), \( S_1 \), \( \mathcal{A}_3  \), \( \tilde{\mathcal{E}}_1(\eta) \), and \( \mathcal{E}_2\), one can show that for a suitably chosen \(\eta(n)\) (that is, \(\eta(n)\) not vanishing too quickly), Lemma \ref{lem.ineq.assum} implies that all of these events will hold if the relevant sample correlations are within \(\delta(n)\) of their expectations. Denote this event by
\begin{align*}
\mathcal{F}_n := &   \left\{ \left| \hat{R}_{1y} - \rho_{1y}(n) \right| < \delta(n) \right\}  \cap \left\{ \left| \hat{R}_{2y} - \rho_{2y}(n) \right| < \delta(n)  \right\} \cap \left\{   \left| \hat{R}_{3y} - \rho_{3y} \right| < \delta(n)   \right\}
\\ & \cap \left\{   \left| \hat{R}_{12} - \rho_{12}(n) \right| < \delta(n)   \right\}  \cap \left\{   \left| \hat{R}_{13}  \right| < \delta(n)   \right\} .
\end{align*}
Later we prove that \(\mathcal{F}_n\) holds with high probability under our assumptions using a concentration inequality on the sample correlations from \citet{pmlr-v80-sun18c}.
\begin{lemma}\label{lem.conc} For \(\eta(n)\) as defined in \eqref{lemma.def.eta.alpha}, \(\mathbb{P}  ( \mathcal{F}_n ) \geq 1 - 30/n^{1/4}\).
\end{lemma}
We also prove our claim that \(\mathcal{A}_{13} \), \( S_1 \), \( \mathcal{A}_3  \), \( \tilde{\mathcal{E}}_1(\eta) \), and \( \mathcal{E}_2\) hold under \(\mathcal{F}_n\).
\begin{lemma}\label{lemma.bound.probs}

Under the assumptions of Theorem \ref{thm.result.sel},
\[
\mathcal{F}_n  \subseteq \left( \mathcal{A}_{13}  \cap S_1 \cap \mathcal{A}_3 \cap \tilde{\mathcal{E}}_1(n) \cap \mathcal{E}_2  \right)   
\]
for \(\eta(n)\) defined in \eqref{lemma.def.eta.alpha} below.
%Let the events \(\mathcal{A}_1(\eta(n)), S_1, \mathcal{A}_3\), \(\tilde{\mathcal{E}}_1(n)\), and \(\mathcal{E}_2\) be defined as in \eqref{lasso.gen.def.a1.s1}, \eqref{lasso.gen.def.a3.prime},  \eqref{lasso.gen.cond.sel.order2}, and \eqref{lasso.gen.cond.sel.order3} (with \(\eta(n)\) as defined in \eqref{lemma.def.eta.alpha}). Then (under the assumptions of Theorem \ref{thm.result.sel}) there exists a finite constant \(c_5 >0\) such that
%\begin{align}
%\mathbb{P}(\mathcal{A}_1(\eta(n))^C) 
%& \leq \frac{1}{2} + c_5 \left(\beta_Z^2 + 1 + \sigma_\epsilon^2 \right)^{7/2} \frac{\left(\log n \right)^{3/2}}{n^{1/4}} + \frac{30}{n^{1/4}} , \label{prob.a1}
%\\  \mathbb{P}(S_1^C) &  \leq  \frac{30}{n^{1/4}} ,\label{prob.s1}
%\\ \mathbb{P}(\mathcal{A}_3^C) & \leq  \frac{30}{n^{1/4}},  \label{prob.a3}
%\\ \mathbb{P}(\tilde{\mathcal{E}}_1(n)^C) & \leq  \frac{30}{n^{1/4}},  \text{ and} \label{prob.e}
%\\ \mathbb{P}\left( \mathcal{E}_2^C  \right) & \leq \frac{30}{n^{1/4}}. \label{prob.e.2}
%\end{align}

\end{lemma}

It turns out that for \(\mathcal{F}_n   \subseteq \tilde{\mathcal{E}}_1(n)\) to hold, we require \(\eta(n)\) to be no less than (that is, go to 0 no faster than)
\begin{equation}\label{lemma.def.eta.alpha}
\eta(n) :=   2 \left(2+ \frac{19}{5} \left(\log n \right)^{1/4} \right) \delta(n) [\delta(n) +  1 - \rho_{12}(n) ].
\end{equation}
Due to the decaying variance of the noise on \(\boldsymbol{X}_{\cdot 1}\) and \(\boldsymbol{X}_{\cdot 2}\) defined in \eqref{lasso.thm.sig.zeta.def}, under our assumptions \(1 - \rho_{12}(n)\) tends to 0 at the same rate (up to log terms) as \(\delta(n)\), so \(\eta(n)\) as defined in \eqref{lemma.def.eta.alpha} goes to 0 at a rate equal to \(\delta^2(n)\) (again, up to log terms). Since for \(\delta(n)\) defined in \eqref{def.delta.n} we have \(\delta^2(n) = \mathcal{O} \left(\log n /n \right)\), this definition of \(\eta(n)\) also allows \(\eta(n) \cdot n^{3/4} (\log n)^{1/4} \to 0\) as required earlier for \( \mathbb{P} \left( \mathcal{A}_{12}(n) \right) \to 1/2\).

In particular, in the following lemma, we work out the bound on \(\mathbb{P}(\mathcal{A}_{12}(n))\) from Proposition \ref{lemma.prob.a1.eta} for \(\eta(n)\) as defined in \eqref{lemma.def.eta.alpha}.
 \begin{lemma}\label{lemma.jacob}
Under the assumptions of Theorem \ref{thm.result.sel}, there exists a finite constant \(c_5 > 0\) such that for \(\eta(n)\) as defined in \eqref{lemma.def.eta.alpha},
\[
\mathbb{P} \left(\mathcal{A}_{12}\left(n \right)^c \right) < \frac{1}{2} +   c_5 \left(\beta_Z^2 + 1 + \sigma_\epsilon^2 \right)^{7/2} \frac{\left(\log n \right)^{3/2}}{n^{1/4}}  .
\]
  \end{lemma}
(It was at this point in the derivation process that we chose \(\delta(n)\) to match the rates of convergence of \(\mathbb{P} \left( \mathcal{A}_{12}(n) \right)\) and \(\mathbb{P} \left( \mathcal{F}_n \right)\) up to log terms. If \(\delta(n)\) were to vanish faster than our specified rate, \(\eta(n)\) would also go to 0 faster, so \(n\) would not have to grow large as quickly for \(\mathcal{A}_{12}(n) \) to hold, which would make the rate of convergence of \(\mathbb{P} \left( \mathcal{A}_{12}(n) \right)\) faster. But this change would require \(n\) to grow faster for \(\mathcal{F}_n\) to continue to hold with high probability, so the rate of convergence of \(\mathbb{P} \left(  \mathcal{F}_n \right)\) would be slower. Our choice of \(\delta(n)\) balances these considerations, allowing for the fastest overall rate of convergence up to log terms.)

Finally, using Lemma \ref{lem.conc}, Lemma \ref{lemma.jacob}, and the union bound, we bound the probability of the event from Lemma \ref{lemma.e.event}:
\begin{align*}
 & \mathbb{P} \left( \mathcal{A}_{12}(n) \cap  \mathcal{A}_{13} \cap S_1 \cap \mathcal{A}_3  \cap \tilde{\mathcal{E}}_1(n) \cap \mathcal{E}_2 \right)  
\\ \geq ~ & \mathbb{P} \left(   \mathcal{A}_{12}(n) \cap \mathcal{F}_n \right)
\\ \geq ~ & 1 - \mathbb{P} \left(   \mathcal{A}_{12}(n)^c \right) -  \mathbb{P} \left(  \mathcal{F}_n^c \right)
% \\ = ~ & 1 - \mathbb{P} \left( \mathcal{A}_1(\eta(n))^C \cup S_1^C \cup \mathcal{A}_3^C  \cup \tilde{\mathcal{E}}_1(n)^C \cup \mathcal{E}_2^C  \right) 
 %
%
\\  >  ~&  \frac{1}{2} -  c_5 \left(\beta_Z^2 + 1 + \sigma_\epsilon^2 \right)^{7/2} \frac{\left(\log n \right)^{3/2}}{n^{1/4}} - \frac{ 30}{n^{1/4}}
\\  \geq ~&  \frac{1}{2} -  c_{3}  \left(\beta_Z^2 + 1 + \sigma_\epsilon^2 \right)^{7/2} \frac{\left(\log n \right)^{3/2}}{n^{1/4}}
\end{align*}
for \(c_{3}  :=  c_5 + \frac{30}{2^{7/2} (\log 100)^{3/2}}\) (since \(\beta_Z > 1\) and \(n \geq 100\)).
%The choice of \(\eta(n)\) comes from work done in the proofs of Lemma \ref{lemma.bound.probs}, which follows, and Proposition \ref{lemma.prob.a1.eta}.
%
%Later we prove the following lemma individually bounding the probabilities of each of the events we are interested in:
%
%
%To prove \eqref{prob.s1} -- \eqref{prob.e.2}, we use a concentration inequality on the sample correlations. The key ingredient for \eqref{prob.a1} comes from the fact that \(\hat{R}_{1y} - \hat{R}_{2y}\), when appropriately scaled, converges to a standard normal distribution by the delta method. In Proposition \ref{lemma.prob.a1.eta}, we exploit a Berry-Esseen-type result for the delta method to bound the distribution of \(\hat{R}_{1y} - \hat{R}_{2y}\) in finite samples.

\section{Proofs of Supporting Results For Theorem \ref{thm.result.sel}}

Appendix \ref{proof.prop} contains the proof of Proposition \ref{lemma.prob.a1.eta} and Appendix \ref{sec.lemma.prob} contains the proofs of the remaining lemmas stated in the proof of Theorem \ref{thm.result.sel}. These proofs require more supporting lemmas, the proofs of which (which are mostly technical, or just algebraic manipulations) are contained in Appendix \ref{tech.lemmas}.

\subsection{Proof of Proposition \ref{lemma.prob.a1.eta}}\label{proof.prop}

Our approach will be to establish a Berry-Esseen-type finite sample bound for the delta method applied to \( \hat{R}_{1y} - \hat{R}_{2y}\) considered as a nonlinear function of certain sample moments. In particular, we will apply Theorem 2.11 in \citet{pinelis2016}, which we state here for completeness. 
\begin{theorem}[Theorem 2.11 in \citet{pinelis2016}]\label{pinelis.thm.2.9}

Let \( \mathcal{X}\) be a Hilbert space and let \(g: \mathcal{X} \to \mathbb{R}\) be a Borel-measurable functional. Suppose that 
\begin{equation}\label{2.9.smoothness}
\lVert \nabla^2 g( \boldsymbol{x}) \rVert_\text{op} \leq M_\epsilon  \qquad \forall  \boldsymbol{x} \in \mathcal{X} \text{ with } \lVert \boldsymbol{x} \rVert_2 < \epsilon \text{ for some } \epsilon \in (0,\infty) \text{ and some } M_\epsilon \in (0, \infty),
\end{equation}
where \(\lVert \cdot \rVert_\text{op} \) is the operator norm. Let \(\boldsymbol{V}_1, \boldsymbol{V}_2, \ldots, \boldsymbol{V}_{n}\) be i.i.d. zero-mean random vectors with \(\tilde{\sigma} := \left( \E \left| L(\boldsymbol{V}_1) \right|^2 \right)^{1/2} > 0\) and \( \E \left| L(\boldsymbol{V}_1) \right|^3 < \infty\), where \(L: \mathcal{X} \to \mathbb{R}\) is the linear functional that is the first derivative of \(g\) at the origin (the tangent plane). Then for all \(z \in \mathbb{R}\),
\[
\left| \mathbb{P} \left( \frac{\sqrt{n} \cdot g(\overline{\boldsymbol{V}})}{\tilde{\sigma}} \leq z \right)  - \Phi(z) \right| \leq \frac{\mathcal{C}}{\sqrt{n}}.
\]
where \(\overline{\boldsymbol{V}} := n^{-1} \sum_{i=1}^n \boldsymbol{V}_i\), \(\Phi(\cdot)\) is the distribution function of the standard normal distribution, and 
\begin{multline}\label{mathcal.c.formula}
\mathcal{C} := k_0 + k_1 \frac{ \E \left| L(\boldsymbol{V}_1) \right|^3 }{ \left( \E \left| L(\boldsymbol{V}_1) \right|^2 \right)^{3/2}}+ \left(k_{20} + k_{21}\frac{ \left( \E \left| L(\boldsymbol{V}_1) \right|^3 \right)^{1/3}}{ \left( \E \left| L(\boldsymbol{V}_1) \right|^2 \right)^{1/2}} \right)\E \left \lVert \boldsymbol{V}_1 \right \rVert^2  
\\ + \left(k_{30} + k_{31}\frac{ \left( \E \left| L(\boldsymbol{V}_1) \right|^3 \right)^{1/3}}{ \left( \E \left| L(\boldsymbol{V}_1) \right|^2 \right)^{1/2}} \right)\left( \E \left\lVert \boldsymbol{V}_1 \right \rVert^3 \right)^{2/3} 
+ k_\epsilon,
\end{multline}
where
\begin{multline}\label{mathcal.c.terms}
k_0 := 0.13925, \quad k_1 := 2.33554, 
\\ \left(k_{20}, k_{21}, k_{30}, k_{31} \right) := \frac{M_\epsilon}{2 \tilde{\sigma}}\left(2 \left( \frac{2}{\pi}\right)^{1/6}, 2 + \frac{2^{2/3}}{n^{1/6}}, \frac{(8/\pi)^{1/6}}{n^{1/3}}, \frac{2}{n^{1/2}} \right) , \quad \text{and}
\\ k_\epsilon := \min\left\{ \frac{\E \left \lVert \boldsymbol{V}_1 \right \rVert^2}{\epsilon^2 n^{1/2}}, \frac{2 \left(\E \left \lVert \boldsymbol{V}_1 \right \rVert^2 \right)^{3/2} +\E \left\lVert \boldsymbol{V}_1 \right \rVert^3/n^{1/2}}{\epsilon^3 n}\right\}.
\end{multline}

\end{theorem}

\begin{remark}

To be precise, \citeauthor{pinelis2016} do not state smoothness condition \eqref{2.9.smoothness} in Theorem 2.11, but rather the following smoothness condition: there exists a continuous linear functional \(L: \mathcal{X} \to \mathbb{R}\) such that
\[
|f(x) - L(x)| \leq \frac{M_\epsilon}{2} \lVert x \rVert_2^2 \qquad \forall x \in \mathcal{X} \text{ with } \lVert x \rVert_2 \leq \epsilon.
\]
However, at the top of page 1007, the authors state that \eqref{2.9.smoothness} is a sufficient condition for this smoothness condition to hold.

\end{remark}

We will begin by defining the relevant sample moments along with the nonlinear function \(g\) and its linear approximation \(L\) that we will make use of. Let
\begin{equation}\label{lemma.tangent.plane.V}
\boldsymbol{V}_i:= \left(X_{i1}^2 -\Sigma_{11}, X_{i2}^2 -\Sigma_{11}, y_i^2 - \Sigma_{yy}, X_{i1}y_i - \Sigma_{1y}, X_{i2}y_i - \Sigma_{1y}\right) , \qquad i \in [n]
\end{equation}
and
\[
\overline{\boldsymbol{V}} := \frac{1}{n} \sum_{i=1}^n  \boldsymbol{V}_i  = \begin{bmatrix}
n^{-1} \sum_{i=1}^n X_{i1}^2 - \Sigma_{11} \\
n^{-1}\sum_{i=1}^n X_{i2}^2 - \Sigma_{11}\\
n^{-1}\sum_{i=1}^n y_i^2 - \Sigma_{yy} \\
n^{-1}\sum_{i=1}^n X_{i1}y_i - \Sigma_{1y}\\
n^{-1}\sum_{i=1}^n X_{i2} y_i - \Sigma_{1y}
\end{bmatrix} .
\]

We show the following inequalities hold:

\begin{lemma}\label{lemma.calc.v.exp}

\begin{equation} \label{v2}
\E \left \lVert \boldsymbol{V}_1 \right \rVert^2  < 12\Sigma_{yy}^2 ,
\end{equation}
and
\begin{align}
 \E \left\lVert \boldsymbol{V}_1 \right \rVert^3   < ~ & 140\Sigma_{yy}^3.  \label{v3}
\end{align}

\end{lemma}

(The proofs of all lemmas are contained in Appendix \ref{tech.lemmas}.) Define \(g: (-\Sigma_{11}, \infty)^3 \times \mathbb{R}^2 \to \mathbb{R}\) as 
\begin{align}
g(\boldsymbol{u})  & := \frac{u_4 + \Sigma_{1y}}{\sqrt{(u_1 + \Sigma_{11})( u_3 + \Sigma_{yy})}} - \frac{u_5 + \Sigma_{1y}}{\sqrt{(u_2 + \Sigma_{11}) (u_3 + \Sigma_{yy})}} \label{lemma.3.g.def}
\\ \implies g(\overline{\boldsymbol{V}}) & = \frac{n^{-1} \sum_{i=1}^n X_{i1}y_i }{\sqrt{n^{-1} \sum_{i=1}^n X_{i1}^2 \cdot  n^{-1}  \sum_{i=1}^n y_i^2 }} - \frac{n^{-1} \sum_{i=1}^n X_{i2} y_i }{\sqrt{n^{-1} \sum_{i=1}^n X_{i2}^2 \cdot n^{-1}  \sum_{i=1}^n y_i^2}} \nonumber
\\ & =   \hat{R}_{1y} - \hat{R}_{2y} , \nonumber
\end{align}
the difference of the (uncentered) sample correlations as defined in \eqref{lemma.unc.samp.corr.def}. Let \(\boldsymbol{\mu} := \E \overline{\boldsymbol{V}} = \boldsymbol{0} \in \mathbb{R}^5\).  We have 
 \begin{equation}\label{lemma.gradient}
\nabla g(\boldsymbol{u}) = \begin{bmatrix}
\frac{u_4 + \Sigma_{1y}}{\sqrt{u_3 + \Sigma_{yy}}} \left(- \frac{1}{2} (u_1 + \Sigma_{11})^{-3/2} \right) \\
\frac{u_5 + \Sigma_{1y}}{\sqrt{u_3 + \Sigma_{yy}}} \left( \frac{1}{2} (u_2 + \Sigma_{11})^{-3/2} \right) \\
\frac{1}{2} \left(  \frac{u_5 + \Sigma_{1y}}{\sqrt{u_2 + \Sigma_{11} }} -  \frac{u_4 + \Sigma_{1y}}{\sqrt{u_1 + \Sigma_{11} }}    \right) (u_3 + \Sigma_{yy})^{-3/2} \\
 \frac{1}{\sqrt{(u_1 + \Sigma_{11})( u_3 + \Sigma_{yy})}} \\
- \frac{1}{\sqrt{(u_2 + \Sigma_{11}) (u_3 + \Sigma_{yy})}} 
\end{bmatrix} 
\implies \nabla g(\boldsymbol{\mu}) 
=  \begin{bmatrix}
- \frac{1}{2}   \frac{\Sigma_{1y}}{\sqrt{\Sigma_{yy} \Sigma_{11}^3}} \\
\frac{1}{2}   \frac{\Sigma_{1y}}{\sqrt{\Sigma_{yy}\Sigma_{11}^3}} \\
0 \\
\frac{1}{\sqrt{\Sigma_{11} \Sigma_{yy}}}\\
- \frac{1}{\sqrt{ \Sigma_{11} \Sigma_{yy}}}
\end{bmatrix}.
\end{equation}

Next, the tangent plane to \(g\) at \(\boldsymbol{0}\) is given by
\begin{align}
L(\boldsymbol{u})  & := g(\boldsymbol{0}) +  \nabla g(\boldsymbol{0})^\top \boldsymbol{u}  \nonumber
\\& = 0 -\frac{1}{2} \frac{\Sigma_{1y}}{\sqrt{\Sigma_{yy} \Sigma_{11}^3}} u_1 + \frac{1}{2} \frac{\Sigma_{1y}}{\sqrt{\Sigma_{yy} \Sigma_{11}^3}}  u_2 + \frac{1}{\sqrt{ \Sigma_{11} \Sigma_{yy}}} u_4 - \frac{1}{\sqrt{\Sigma_{11} \Sigma_{yy}}}u_5 \nonumber
\\ & =\frac{1}{2}\frac{\Sigma_{1y}}{\sqrt{\Sigma_{yy} \Sigma_{11}^3}}( u_2 - u_1 )  +  \frac{1}{\sqrt{ \Sigma_{11} \Sigma_{yy}}} \left( u_4 - u_5 \right); \label{lemma.tangent.plane}
\end{align}
in particular,
\[
L ( \boldsymbol{V}_1) = \frac{1}{2} \frac{\Sigma_{1y}}{\sqrt{\Sigma_{yy} \Sigma_{11}^3}}( X_{12}^2  - X_{11}^2)  +  \frac{1}{\sqrt{ \Sigma_{11} \Sigma_{yy}}} \left( X_{11}y_1 - X_{12}y_1 \right),
\]
so
\[
L ( \E \boldsymbol{V}_1) = \frac{1}{2}  \frac{\Sigma_{1y}}{\sqrt{\Sigma_{yy} \Sigma_{11}^3}}(\E X_{12}^2  - \E X_{11}^2)  +  \frac{1}{\sqrt{ \Sigma_{11} \Sigma_{yy}}} \left( \E (X_{11}y_1) -  \E(X_{12}y_1 )\right)= 0.
\]

Later we will prove the following result:
\begin{lemma}\label{calc.lemma.delt.meth.sigma.tilde}
Under the assumptions of Proposition \ref{lemma.prob.a1.eta},  for \(L:(-1,\infty)^3 \times \mathbb{R}^2 \to \mathbb{R}\) defined in \eqref{lemma.tangent.plane} and \(V\) defined in \eqref{lemma.tangent.plane.V},
\begin{equation}\label{lemma.def.tilde.sigma}
\tilde{\sigma} := \sqrt{\E \left| L(\boldsymbol{V}_1) \right|^2} =  \sqrt{ \left( 1 - \rho_{12} \right) \left(  -3 \rho_{1y}^2 + \rho_{1y}^2 \rho_{12} + 2 \right) } \in \left( 0, \sqrt{2} \right],
\end{equation}
where \(\rho_{12}\) is the correlation between \(\boldsymbol{X}_{\cdot 1}\) and \(\boldsymbol{X}_{\cdot 2}\), and \(\rho_{1y}\) is the correlation between \(\boldsymbol{y}\) and \(\boldsymbol{X}_{\cdot 1}\). Further, \(  \E \left| L(\boldsymbol{V}_1) \right|^3\) is finite (in particular, \(\E \left| L(\boldsymbol{V}_1) \right|^3  \leq 16 \sqrt{2}/\pi\)), and
\begin{equation}\label{lemma.bound.ratio}
\frac{ \left( \E \left| L(\boldsymbol{V}_1) \right|^3 \right)^{1/3}}{ \left( \E \left| L(\boldsymbol{V}_1) \right|^2 \right)^{1/2}} \leq \sqrt[3]{\frac{8}{\pi}}.
\end{equation}

\end{lemma}
In order to apply Theorem \ref{pinelis.thm.2.9}, we must show that \(g(\cdot)\) satisfies \eqref{2.9.smoothness}; we show this is the case in the proof of the following lemma:
\begin{lemma}\label{lemma.smoothness.2.1}
Under the assumptions of Proposition \ref{lemma.prob.a1.eta}, the functional \(g: (-1, \infty)^3 \times \mathbb{R}^2  \to \mathbb{R}\) as defined in \eqref{lemma.3.g.def} satisfies \eqref{2.9.smoothness} with \(\epsilon = \Sigma_{11}/2\) and \(M_\epsilon := 36 \Sigma_{yy}.\)
%\begin{equation}\label{def.m.epsilon}
\end{lemma}
Then Theorem \ref{pinelis.thm.2.9} yields that for all \(z \in \mathbb{R}\)
\begin{align*}
& \left| \mathbb{P} \left(  \frac{ \sqrt{n}(\hat{R}_{1y} - \hat{R}_{2y})}{\tilde{\sigma}} \leq z \right)  - \Phi(z) \right| \leq \frac{\mathcal{C}}{\sqrt{n}}
\\
\implies \qquad & \mathbb{P} \left(  \hat{R}_{1y} - \hat{R}_{2y} \leq \frac{\tilde{\sigma} z}{\sqrt{n}} \right)   \leq  \Phi(z) + \frac{\mathcal{C}}{\sqrt{n}}
\\ \iff   \qquad & \mathbb{P} \left(  \hat{R}_{1y} - \hat{R}_{2y} \leq \eta \right)   \leq  \Phi\left( \frac{ \eta \sqrt{n}}{\tilde{\sigma}}\right) + \frac{\mathcal{C}}{\sqrt{n}},
\end{align*}
where \(\eta := \tilde{\sigma} z/\sqrt{n}\) may equal any real number since \(z\) may equal any real number. We upper bound \(\mathcal{C}\) in the following result:
\begin{lemma}\label{lemma.ub.mathcal.c}
Under the assumptions of Proposition \ref{lemma.prob.a1.eta},
\begin{align*}
\mathcal{C}  < ~ &  \Sigma_{yy}^3  \left( \frac{1462.717}{ \tilde{\sigma}}   + 13.859 \right)  ,
\end{align*}
where \(\tilde{\sigma}\) is defined in \eqref{lemma.def.tilde.sigma}. 
\end{lemma}
This yields \eqref{lemma.prob.ineq.1}.

\subsection{Proofs of Lemmas Stated in Proof of Theorem \ref{thm.result.sel}}\label{sec.lemma.prob}

\begin{proof}[Proof of Lemma \ref{lemma.x1.not.leave}] Suppose (without loss of generality) that feature \(\boldsymbol{X}_{\cdot 1}\) enters the lasso path first. Define 
\[
 \lambda_2^{\text{cross}} := \frac{\boldsymbol{X}_{\cdot 1}^\top \boldsymbol{y}}{n \lVert  \boldsymbol{X}_{\cdot 1} \rVert_2 s_1}      .
 \]
The feature \(\boldsymbol{X}_{\cdot 1}\) would be removed from the active set before \(\boldsymbol{X}_{\cdot 2}\) or \(\boldsymbol{X}_{\cdot 3}\) were added to the active set under the event
\[
 \left\{ \lambda_2^{(2)} \vee \lambda_2^{(3)} <  \lambda_2^{\text{cross}} < \lambda_1 \right\} \qquad  \text{\citep{Tibshirani2013}.}
\] 
But since feature \(\boldsymbol{X}_{\cdot 1}\) entered first,
\[
\lambda_1 =   \max_j   \left\{ \frac{\left| \boldsymbol{X}_{\cdot j}^\top \boldsymbol{y} \right| }{n \lVert \boldsymbol{X}_{\cdot j} \rVert_2 } \right\}  =  \frac{\left| \boldsymbol{X}_{\cdot 1}^\top \boldsymbol{y} \right| }{n \lVert \boldsymbol{X}_{\cdot 1} \rVert_2 } = \frac{\boldsymbol{X}_{\cdot 1}^\top \boldsymbol{y}}{n \lVert  \boldsymbol{X}_{\cdot 1} \rVert_2 s_1} = \lambda_2^{\text{cross}},
\]
so the probability of this event is 0.

\end{proof}

\begin{proof}[Proof of Lemma \ref{lemma.e.event}]  We will first verify that \eqref{lasso.gen.lambda.3.exp} holds under the specified events. Then we will use \eqref{lasso.gen.lambda.3.exp} to show that \eqref{lambda.ineq.order} holds as well.

We will want to know the coefficient \( \hat{\beta}_1(\lambda) \) on \(\boldsymbol{X}_{\cdot 1}\) in the part of the lasso path before the second feature enters the active set. Conditional on \(\mathcal{A}_{12} \cap \mathcal{A}_{13} \cap S_1\), for \(\lambda \in [\lambda_2, \lambda_1]\), by the KKT conditions \eqref{lasso.gen.simplest.kkt.conds} it holds that 
\begin{align}
\lambda s_1  & =  \frac{1}{n \lVert  \boldsymbol{X}_{\cdot 1} \rVert_2} \boldsymbol{X}_{\cdot 1}^\top \left(\boldsymbol{y} - \frac{\boldsymbol{X}_{\cdot 1}}{ \lVert \boldsymbol{X}_{\cdot 1} \rVert_2} \hat{\beta}_1(\lambda) \right)  \nonumber
%\\ \iff \qquad  \lambda & =   \frac{\boldsymbol{X}_{\cdot 1}^\top\boldsymbol{y}}{n \lVert  \boldsymbol{X}_{\cdot 1} \rVert_2} - \frac{1}{n} \hat{\beta}_1(\lambda)   \nonumber
\\ \iff  \qquad \hat{\beta}_1(\lambda)  & = n\left( \frac{\boldsymbol{X}_{\cdot 1}^\top\boldsymbol{y}}{n \lVert  \boldsymbol{X}_{\cdot 1} \rVert_2} - \lambda  \right)  \nonumber
\\  & = \lVert \boldsymbol{y} \rVert_2 \hat{R}_{1y} - n \lambda \label{lass.gen.beta.1} ,
\end{align}
and \(\hat{\beta}_j(\lambda) = 0\) for \(j \neq 1\). Now we are prepared to show \eqref{lasso.gen.lambda.3.exp}. Using \eqref{lasso.gen.simplest.kkt.conds}, for \(s_2 \in \{-1, 1\}\) the knot \(\lambda_2^{(2)}\) satisfies 
\begin{align}
 0 & = \frac{-1}{n  \lVert \boldsymbol{X}_{\cdot 2} \rVert_2} \boldsymbol{X}_{\cdot 2}^\top\left(\boldsymbol{y} - \frac{\boldsymbol{X}_{\cdot 1}}{\lVert \boldsymbol{X}_{\cdot 1} \rVert_2}  \hat{\beta}_1\left(\lambda_2^{(2)} \right)\right)+ \lambda_2^{(2)} s_2  \nonumber
\\ \stackrel{(a)}{\iff} \qquad \lambda_2^{(2)} & =  \left| \frac{1}{n \lVert \boldsymbol{X}_{\cdot 2} \rVert_2} \boldsymbol{X}_{\cdot 2}^\top\left(\boldsymbol{y} - \frac{\boldsymbol{X}_{\cdot 1}}{\lVert \boldsymbol{X}_{\cdot 1} \rVert_2}  \hat{\beta}_1\left(\lambda_2^{(2)}\right)  \right)\right|    \nonumber
 \\ \iff \qquad   \lambda_2^{(2)} & =  \left| \frac{1}{n \lVert \boldsymbol{X}_{\cdot 2} \rVert_2} \boldsymbol{X}_{\cdot 2}^\top\left(\boldsymbol{y} - \frac{\boldsymbol{X}_{\cdot 1}}{\lVert \boldsymbol{X}_{\cdot 1} \rVert_2}  \left[ \lVert \boldsymbol{y} \rVert_2 \hat{R}_{1y} - n \lambda_2^{(2)} \right]  \right)\right|   \nonumber
\\ 
\iff \qquad   \lambda_2^{(2)} & =   \left| \frac{\boldsymbol{X}_{\cdot 2}^\top \boldsymbol{y}}{n  \lVert \boldsymbol{X}_{\cdot 2} \rVert_2 }  - \frac{\boldsymbol{X}_{\cdot 2}^\top \boldsymbol{X}_{\cdot 1}}{n \lVert \boldsymbol{X}_{\cdot 2} \rVert_2 \lVert \boldsymbol{X}_{\cdot 1} \rVert_2}\lVert \boldsymbol{y} \rVert_2 \hat{R}_{1y}   +  \frac{\boldsymbol{X}_{\cdot 2}^\top\boldsymbol{X}_{\cdot 1} }{ \lVert \boldsymbol{X}_{\cdot 2} \rVert_2  \lVert \boldsymbol{X}_{\cdot 1} \rVert_2} \lambda_2^{(2)} \right|  \nonumber
\\ 
\iff \qquad  \lambda_2^{(2)} & =    \left|  \frac{\lVert \boldsymbol{y} \rVert_2}{n}\left( \hat{R}_{2y}  - \hat{R}_{12} \hat{R}_{1y} \right)  + \hat{R}_{12} \lambda_2^{(2)} \right| , \nonumber
\end{align}
where (a) follows because \(s_2 \in \{-1, 1\}\) takes on the value that ensures \(\lambda_2^{(2)}\) is positive and we used \eqref{lass.gen.beta.1} and the uncentered sample correlations from \eqref{lemma.unc.samp.corr.def}.  Similarly, the knot for \(\boldsymbol{X}_{\cdot 3}\) is at \(\lambda_2^{(3)}\) satisfying
\begin{align}
\lambda_2^{(3)} & =  \left| \frac{\lVert \boldsymbol{y} \rVert_2}{n} \left( \hat{R}_{3y} -     \hat{R}_{13} \hat{R}_{1y} \right) + \hat{R}_{13} \lambda_2^{(3)} \right| , \nonumber
\end{align}
and if \(\mathcal{A}_3\) holds, we can write
 \begin{align}
 \lambda_2^{(3)} & =    \left|  \frac{\lVert \boldsymbol{y} \rVert_2}{n} \left( \hat{R}_{3y}  - \hat{R}_{13} \hat{R}_{1y}  \right)  + \hat{R}_{13} \lambda_2^{(3)} \right| \nonumber
 \\ & =      \frac{\lVert \boldsymbol{y} \rVert_2}{n} \left( \hat{R}_{3y}  - \hat{R}_{13} \hat{R}_{1y}  \right)  + \hat{R}_{13} \lambda_2^{(3)} \nonumber
 \\ \iff \qquad  
 \lambda_2^{(3)} & =  \frac{\lVert \boldsymbol{y} \rVert_2}{n} \frac{ \hat{R}_{3y}  - \hat{R}_{13} \hat{R}_{1y} }{1 -  \hat{R}_{13}}, \nonumber
 \end{align}
 which is \eqref{lasso.gen.lambda.3.exp}. 

 Now we seek to verify \eqref{lambda.ineq.order}. To see that the right side of \eqref{lambda.ineq.order} holds, note that
% also that if \(s_3 = 1\), then
  \begin{align*}
 \lambda_2^{(3)} < \lambda_1 \qquad \iff \qquad &   \frac{\lVert \boldsymbol{y} \rVert_2}{n} \frac{ \hat{R}_{3y}  - \hat{R}_{13} \hat{R}_{1y} }{1 -  \hat{R}_{13}}< \frac{ \lVert \boldsymbol{y} \rVert_2}{n} \hat{R}_{1y}
 \\  \iff \qquad &  \hat{R}_{3y} -   \hat{R}_{13}  \hat{R}_{1y}  <  \hat{R}_{1y} \left( 1 - \hat{R}_{13} \right)
  \\  \iff \qquad &  \hat{R}_{3y }<  \hat{R}_{1y} ,
 \end{align*}
 which is true under \(\mathcal{A}_{13}\). It only remains to show \( \lambda_2^{(2)} < \lambda_2^{(3)} \). Observe that we can write
 \[
 \lambda_2^{(2)} =   \max_{s_2 \in \{-1, 1\}} \left\{ \frac{n^{-1} \lVert \boldsymbol{y} \rVert_2 \left(\hat{R}_{2y} -   \hat{R}_{12}  \hat{R}_{1y} \right)}{s_2 - \hat{R}_{12}} \right\},
 \]
 so under the assumed events we have
%  We then have that \eqref{lambda.ineq.order} holds and \(\boldsymbol{X}_{\cdot 3}\) enters the active set before \(\boldsymbol{X}_{\cdot 2}\) under the event \(
\begin{align}
& \lambda_2^{(2)} < \lambda_2^{(3)} \nonumber
\\ \iff \qquad &  \max_{s_2 \in \{-1, 1\}} \left\{ \frac{n^{-1} \lVert \boldsymbol{y} \rVert_2 \left(\hat{R}_{2y} -   \hat{R}_{12}  \hat{R}_{1y} \right)}{s_2 - \hat{R}_{12}} \right\}
<
  \frac{\lVert \boldsymbol{y} \rVert_2}{n} \frac{ \hat{R}_{3y}  - \hat{R}_{13} \hat{R}_{1y} }{1 -  \hat{R}_{13}} \nonumber
 \\ \iff \qquad &  \max_{s_2 \in \{-1, 1\}} \left\{ \frac{\hat{R}_{2y} -   \hat{R}_{12}  \hat{R}_{1y}  }{s_2 - \hat{R}_{12}} \right\} 
<
 \frac{\lVert \boldsymbol{y} \rVert_2}{n} \frac{ \hat{R}_{3y}  - \hat{R}_{13} \hat{R}_{1y} }{1 -  \hat{R}_{13}} . \label{lasso.gen.cond.sel.order}
\end{align}
We can write the argument of the left side of \eqref{lasso.gen.cond.sel.order} as
\[
\frac{\hat{R}_{2y} -   \hat{R}_{12}  \hat{R}_{1y}}{s_2  - \hat{R}_{12}} 
= \frac{\left(s_2  -   \hat{R}_{12}\right)  \hat{R}_{1y} + \left(\hat{R}_{2y} - s_2\hat{R}_{1y} \right)}{s_2  - \hat{R}_{12}} 
=\hat{R}_{1y} + \frac{\hat{R}_{2y} - s_2\hat{R}_{1y}}{s_2  - \hat{R}_{12}} 
\]
and similarly the right side is
\[
\frac{ \hat{R}_{3y} -   \hat{R}_{13}  \hat{R}_{1y}}{1 - \hat{R}_{13}}
%= \frac{\left(1  -   \hat{R}_{13}\right)  \hat{R}_{1y} + \left(\hat{R}_{3y} - \hat{R}_{1y} \right)}{1  - \hat{R}_{13}} 
=\hat{R}_{1y} + \frac{\hat{R}_{3y} - \hat{R}_{1y}}{1  - \hat{R}_{13}} 
\]
%\[
%\hat{R}_{1y} + \frac{\hat{R}_{3y} - \hat{R}_{1y}}{1  - \hat{R}_{13}} 
%\]
which means that \eqref{lasso.gen.cond.sel.order}, and therefore \eqref{lambda.ineq.order}, are equivalent on the event  \(\mathcal{A}_{12} \cap \mathcal{A}_{13} \cap S_1 \cap \mathcal{A}_3 \cap S_1 \) to
\begin{align*}
& \left\{  \max_{s_2 \in \{-1, 1\}} \left\{  \frac{\hat{R}_{2y} - s_2\hat{R}_{1y}}{s_2  - \hat{R}_{12}}  \right\}
<
  \frac{\hat{R}_{3y} - \hat{R}_{1y}}{1  - \hat{R}_{13}}  \right\} 
\\ \iff \qquad & \left\{
\frac{\hat{R}_{2y} - \hat{R}_{1y}}{1  - \hat{R}_{12}} 
< 
\frac{\hat{R}_{3y} - \hat{R}_{1y}}{1  - \hat{R}_{13}} 
\right\} \cap  \left\{
\frac{\hat{R}_{2y} + \hat{R}_{1y}}{-1  - \hat{R}_{12}} 
< 
\frac{\hat{R}_{3y} - \hat{R}_{1y}}{1  - \hat{R}_{13}} 
\right\}
\\  \iff \qquad & 
\left\{
\frac{\hat{R}_{1y} - \hat{R}_{2y}}{1  - \hat{R}_{12}} 
> 
\frac{\hat{R}_{1y} - \hat{R}_{3y}}{1  - \hat{R}_{13}} 
\right\}  \cap \left\{
\frac{\hat{R}_{1y} + \hat{R}_{2y}}{1  + \hat{R}_{12}} 
> 
\frac{\hat{R}_{1y} - \hat{R}_{3y}}{1  - \hat{R}_{13}} 
\right\} 
\\ & = \mathcal{E}_1 \cap \mathcal{E}_2.
\end{align*}

\end{proof}

\begin{proof}[Proof of Lemma \ref{lem.ineq.assum}] We begin by stating a few results we will require. The proofs of these lemmas are provided in Appendix \ref{tech.lemmas}. 

\begin{remark} The inequality
\[
\max\left\{2 \left(12 + \sigma_\epsilon^2\right), 5\left(1 + \sigma_\epsilon^2 \right) \right\} > \frac{9}{4} \left(5 + \sigma_\epsilon^2\right)
\]
along with \(\beta_Z^2 < 4\) (from Statement i) yields the following inequalities from \eqref{n.large.delta.cond.max} and \eqref{c4n.conds.max} that we will make use of throughout the proof:
%\begin{lemma}\label{very.first.lemma} Under the assumptions of Theorem \ref{thm.result.sel},
% \begin{equation}\label{n.exp.size.req}
%n > \exp \left\{ \left(   \frac{100}{19 \log (100)}\sqrt{\frac{c_2}{\beta_Z^2 + 1 + \sigma_\epsilon^2}}  +  \frac{10}{19} \sqrt{ 3 + \sigma_\epsilon^2 } \right)^4 \right\},
%\end{equation}
\begin{equation}\label{n.large.delta.cond}
\frac{n}{\log n} > \frac{\beta_Z^2 + 1 + \sigma_\epsilon^2}{c_2}\max \left\{ \frac{1}{4 t_0^2 (2 + \sigma_\epsilon^2 )^2}, 2\left(12 + \sigma_\epsilon^2 \right),  5 \left(1 +  \sigma_\epsilon^2\right), \frac{9}{4}  \left( \beta_Z^2 + 1 + \sigma_\epsilon^2 \right) \right\}
%%%%% Note: to see why we can't do without any of these conditions, examine https://www.desmos.com/calculator/piz14ehzfo.
\end{equation}
and
\begin{equation}\label{c4n.conds}
 \frac{n}{\left(\log n \right)^{3/2}}  > \frac{3.61}{c_2} \left( \beta_Z^2 + 1 + \sigma_\epsilon^2 \right) .
\end{equation}

%\end{lemma}
\end{remark}

\begin{lemma}\label{lemma.delt.ineq} Under the assumptions of Theorem \ref{thm.result.sel},
\begin{equation} \label{beta.u.bound.ineq2}
\delta(n) < \min \left\{  \frac{5}{19 \left(\log n \right)^{1/4}},  \frac{1}{\sqrt{20 \left( 1 + \sigma_\epsilon^2 \right)}} ,  t_0\left(\beta_Z^2 + 1 + \sigma_\epsilon^2 \right) \right\}, 
\end{equation}
where \(\delta(n)\) is defined in \eqref{def.delta.n}. This implies that \(1 - 2 \delta(n) \sqrt{3 + \sigma_\epsilon^2} > \frac{1}{5}\). Further, the following inequalities hold:
\begin{align}
  \sqrt{ 1 + \sigma_\zeta^2(n)} - 1 & < 20\sqrt{\frac{c_2}{\beta_Z^2 + 1 + \sigma_\epsilon^2}} \frac{\delta(n)}{\log n}  ,  \quad \text{and}   \label{lemma.ineq.sig.zeta.sqrt}
\\ \sigma_\zeta^2(n)  &  < 1 ,\label{sigma.zeta.sq.l.1.lemma}
\end{align}
 where \(\sigma_\zeta^2(n)\) is defined in \eqref{lasso.thm.sig.zeta.def}.

\end{lemma}

\begin{lemma}\label{cond.beta.z.lemma}
Under the assumptions of Theorem \ref{thm.result.sel},
\[
1 < \frac{\sqrt{ 1 + \sigma_\zeta^2(n)} }{ 1 - 2 \delta(n) \sqrt{ 3 + \sigma_\epsilon^2 }} < \beta_Z  <  \frac{5}{5 -19 \left(\log n \right)^{1/4} \delta(n)}.
\]
\end{lemma}

\begin{lemma}\label{lemma.dist.bounds}

Under the assumptions of Theorem \ref{thm.result.sel}, the random variables \(y_i, X_{i1}, X_{i2}\), and \(X_{3i}\) are distributed as
\begin{equation}\label{cov.matrix.def}
\begin{pmatrix}
y_i \\
X_{i1} \\
X_{i2} \\
X_{3i} 
\end{pmatrix} \sim \mathcal{N} \left( \boldsymbol{0}, \begin{bmatrix}
\beta_Z^2 + 1 + \sigma_\epsilon^2 & \beta_Z & \beta_Z &  1 \\
\beta_Z & 1 + \sigma_\zeta^2(n) & 1 & 0 \\
\beta_Z &1 & 1 + \sigma_\zeta^2(n) & 0 \\
1 & 0 & 0 & 1
\end{bmatrix} \right), \qquad \forall i \in [n].
\end{equation}
It follows that the correlation matrix is
\begin{multline}\label{cor.matrix.def}
\begin{pmatrix}
\rho_{yy} & \rho_{1y}(n) & \rho_{2y}(n) & \rho_{3y} \\
\rho_{1y}(n) & \rho_{11} & \rho_{12}(n) & \rho_{13} \\
\rho_{2y}(n) & \rho_{12}(n) & \rho_{22} & \rho_{23} \\
\rho_{3y} & \rho_{13} & \rho_{23} & \rho_{33} 
\end{pmatrix} 
\\   =  \begin{pmatrix}
1 & \frac{\beta_Z}{\sqrt{\left(\beta_Z^2 + 1 + \sigma_\epsilon^2\right)\left(1 + \sigma_\zeta^2(n) \right)}} &  \frac{\beta_Z}{\sqrt{\left(\beta_Z^2 + 1 + \sigma_\epsilon^2\right)\left(1 + \sigma_\zeta^2(n) \right)}} & \frac{1}{\sqrt{\beta_Z^2 + 1 + \sigma_\epsilon^2}} \\
\frac{\beta_Z}{\sqrt{\left(\beta_Z^2 + 1 + \sigma_\epsilon^2\right)\left(1 + \sigma_\zeta^2(n) \right)}} & 1 & \frac{1}{1 + \sigma_\zeta^2(n)} & 0 \\
 \frac{\beta_Z}{\sqrt{\left(\beta_Z^2 + 1 + \sigma_\epsilon^2\right)\left(1 + \sigma_\zeta^2(n) \right)}} & \frac{1}{1 + \sigma_\zeta^2(n)}& 1 & 0 \\
\frac{1}{\sqrt{\beta_Z^2 + 1 + \sigma_\epsilon^2}} & 0 & 0 & 1 
\end{pmatrix} .
\end{multline}

Observe that \(\max_i \left(\boldsymbol{\Sigma}_{ii}^* \right)= \beta_Z^2 + 1 + \sigma_\epsilon^2 \) and \(\min_i \left(\boldsymbol{\Sigma}_{ii}^* \right)= 1\), where \(\boldsymbol{\Sigma}^*\) is the covariance matrix in \eqref{cov.matrix.def}. Lastly, \(\beta_Z > 1 + \sigma_\zeta^2(n)\).

\end{lemma}
From \eqref{sigma.zeta.sq.l.1.lemma} in Lemma \ref{lemma.delt.ineq} and \eqref{cor.matrix.def} in Lemma \ref{lemma.dist.bounds} we have that 
\[
\rho_{12}(n) = \frac{1}{1 + \sigma_\zeta^2(n)}  > \frac{1}{1 + 1} .
\]
Since we know from Lemma \ref{lemma.delt.ineq} that \(\delta(n) < 1/2 \), \eqref{itm:3} is verified. Next, from \eqref{n.large.delta.cond} we have
\begin{align*}
\frac{n}{\log n} & > \frac{1}{c_2} \cdot \frac{9}{4}  \left( \beta_Z^2 + 1 + \sigma_\epsilon^2 \right)^2
\\ \iff \qquad  \frac{1}{\beta_Z^2  + 1 +  \sigma_\epsilon^2} & >  9 \cdot \frac{(\beta_Z^2 + 1 + \sigma_\epsilon^2) \log n}{4 c_2 n} 
\\ & = 9 \delta^2(n)
\\ \implies \qquad  \rho_{3y} & > 3 \delta (n) .
\end{align*}
By \eqref{cor.matrix.def} we have \(\rho_{jy}(n) = \frac{\beta_Z}{\sqrt{1 + \sigma_\zeta^2(n)}} \rho_{3y}\), \(j \in [2]\). Since from Lemma \ref{cond.beta.z.lemma} we know that \(\beta_Z > \frac{\sqrt{ 1 + \sigma_\zeta^2(n)} }{ 1 - 2 \delta(n) \sqrt{ 3 + \sigma_\epsilon^2 }} > \sqrt{1 + \sigma_\zeta^2(n)}\), we have \(\frac{\beta_Z}{\sqrt{1 + \sigma_\zeta^2(n)}} > 1\). Therefore \(\rho_{1y}(n) = \rho_{2y}(n) > \rho_{3y} > 3 \delta (n) > 0\), which is \eqref{itm:1}. Next we will show \eqref{itm:7}. From Lemma \ref{lemma.delt.ineq} we have \( 1 - 2 \delta(n) \sqrt{3 + \sigma_\epsilon^2} > 0\), so
\begin{align*}
\beta_Z  ~> ~& \frac{\sqrt{ 1 + \sigma_\zeta^2(n)} }{ 1 - 2 \delta(n) \sqrt{ 3 + \sigma_\epsilon^2 }}
\\  \geq  ~& \frac{\sqrt{ 1 + \sigma_\zeta^2(n)} }{ 1 - 2 \delta(n) \sqrt{ 2 + \sigma_\zeta^2(n) + \sigma_\epsilon^2 }} 
  \\ = ~& \frac{\sqrt{ 1 + \sigma_\zeta^2(n)}\left( 1 + 2 \delta(n) \sqrt{ 
2 + \sigma_\zeta^2(n)    
 +  \sigma_\epsilon^2} \right) }{ 1 - 4 \delta^2(n) \left( 2 + \sigma_\zeta^2(n) + \sigma_\epsilon^2 \right)} 
 \\  \geq ~ & \frac{\sqrt{ 1 + \sigma_\zeta^2(n)}\left( 1 + 2 \delta(n) \sqrt{ 
1 + \sigma_\zeta^2(n)    
+  \left(1 +  \sigma_\epsilon^2\right)\left[1 - 
 4 \delta^2(n) \left( 1 + \sigma_\zeta^2(n) \right) \right] } \right) }{ 1 - 4 \delta^2(n) \left( 1 + \sigma_\zeta^2(n) \right)} 
\\  = ~& \frac{2\sqrt{ 1 + \sigma_\zeta^2(n)} + 4 \delta(n) \sqrt{ \left( 1 + \sigma_\zeta^2(n) \right)  
\left( 1 + \sigma_\zeta^2(n)    
+  \left(1 +  \sigma_\epsilon^2\right)\left[1 - 
 4 \delta^2(n) \left( 1 + \sigma_\zeta^2(n) \right) \right] \right)}}{2\left[ 1 - 4 \delta^2(n) \left( 1 + \sigma_\zeta^2(n) \right)\right]} 
\\  = ~ & \frac{-b + \sqrt{16 \delta^2(n) \left( 1 + \sigma_\zeta^2(n) \right)  
\left( 1 + \sigma_\zeta^2(n)    
+  \left(1 +  \sigma_\epsilon^2\right)\left[1 - 
 4 \delta^2(n) \left( 1 + \sigma_\zeta^2(n) \right) \right] \right)}}{2a} 
\end{align*}
where \(a = 1 - 4 \delta^2(n) \left( 1 + \sigma_\zeta^2(n) \right) \) and \(b = -2\sqrt{ 1 + \sigma_\zeta^2(n)}\), and 
\begin{align*}
& 16 \delta^2(n) \left( 1 + \sigma_\zeta^2(n) \right)  
\left( 1 + \sigma_\zeta^2(n)    
+  \left(1 +  \sigma_\epsilon^2\right)\left[1 - 
 4 \delta^2(n) \left( 1 + \sigma_\zeta^2(n) \right) \right] \right)
 \\ = & 16 \delta^2(n) \left( 1 + \sigma_\zeta^2(n) \right)  
\left( 1 + \sigma_\zeta^2(n)    
+  1 +  \sigma_\epsilon^2
 - \left(1 +  \sigma_\epsilon^2\right) \cdot 4 \delta^2(n) \left( 1 + \sigma_\zeta^2(n) \right)  \right)
\\ = & 4 \left( 1 + \sigma_\zeta^2(n) \right)  \left( 4 \delta^2(n) \left( 1 + \sigma_\zeta^2(n) \right)   +  4 \delta^2(n)\left(1 +  \sigma_\epsilon^2\right)  - 4 \delta^2(n)\left(1 +  \sigma_\epsilon^2\right) \cdot 4 \delta^2(n) \left( 1 + \sigma_\zeta^2(n) \right)  \right)
\\ = & 4 \left( 1 + \sigma_\zeta^2(n) \right)  \left(1 -  1 + 4 \delta^2(n) \left( 1 + \sigma_\zeta^2(n) \right)   +  4 \delta^2(n)\left(1 +  \sigma_\epsilon^2\right) \left[1 - 4 \delta^2(n) \left( 1 + \sigma_\zeta^2(n) \right) \right]  \right)
\\ = & 4 \left( 1 + \sigma_\zeta^2(n) \right)  \left(1 -  \left[1 - 4 \delta^2(n) \left( 1 + \sigma_\zeta^2(n) \right) \right]\left[1 - 4 \delta^2(n)\left(1 +  \sigma_\epsilon^2\right)\right]  \right)
\\ = & 4 \left( 1 + \sigma_\zeta^2(n) \right) - 4 \left[1 - 4 \delta^2(n) \left( 1 + \sigma_\zeta^2(n) \right) \right] \left(1 + \sigma_\zeta^2(n)\right) \left[1 - 4 \delta^2(n)\left(1 +  \sigma_\epsilon^2\right)\right] 
\\& = b^2 - 4 ac
\end{align*}
where \(c = \left(1 + \sigma_\zeta^2(n)\right) \left[1 - 4 \delta^2(n)\left(1 +  \sigma_\epsilon^2\right)\right] \). Since from Lemma \ref{cond.beta.z.lemma} we have \(\beta_Z > \frac{\sqrt{ 1 + \sigma_\zeta^2(n)} }{ 1 - 2 \delta(n) \sqrt{ 3 + \sigma_\epsilon^2 }}\), it follows that \(\beta_Z  > \frac{\sqrt{ 1 + \sigma_\zeta^2(n)} }{ 1 - 2 \delta(n) \sqrt{ 3 + \sigma_\epsilon^2 }} \geq  \frac{-b + \sqrt{b^2 - 4ac}}{2a}\). Then we have
 \begin{align}
 & \beta_Z \geq  \frac{-b + \sqrt{b^2 - 4ac}}{2a} \nonumber
 \\ \implies \qquad  & a\beta_Z^2 + b \beta_Z  + c \geq 0 \nonumber
\\ \iff  \qquad & \beta_Z^2 \left[ 1 - 4 \delta^2(n) \left( 1 + \sigma_\zeta^2(n) \right) \right] - 2\beta_Z \sqrt{ 1 + \sigma_\zeta^2(n)}  + \left(1 + \sigma_\zeta^2(n)\right) \left[1 - 4 \delta^2(n)\left(1 +  \sigma_\epsilon^2\right)\right]  \geq 0  \nonumber
\\ \iff  \qquad & \beta_Z^2 +1 + \sigma_\zeta^2(n) - 2\beta_Z \sqrt{ 1 + \sigma_\zeta^2(n)}    \geq 4 \delta^2(n)\left(\beta_Z^2  + 1 +  \sigma_\epsilon^2\right) \left( 1 + \sigma_\zeta^2(n) \right) \nonumber
\\ \iff  \qquad & \beta_Z - \sqrt{ 1 + \sigma_\zeta^2(n)} \geq 2 \delta(n) \sqrt{\left(\beta_Z^2  + 1 +  \sigma_\epsilon^2\right) \left( 1 + \sigma_\zeta^2(n) \right)} \nonumber
\\ \iff  \qquad &  \frac{\beta_Z}{\sqrt{\left(\beta_Z^2  + 1 +  \sigma_\epsilon^2\right) \left( 1 + \sigma_\zeta^2(n) \right)}} -  \frac{1}{\sqrt{\beta_Z^2  + 1 +  \sigma_\epsilon^2 }}   \geq 2 \delta(n)   \nonumber
\\ \iff \qquad &  \rho_{1y}(n) -  \rho_{3y}  \geq 2\delta (n) ,  \nonumber
%\label{itm:7}
\end{align}
yielding \eqref{itm:7}. Next we will show \eqref{itm:8a}. From \eqref{beta.u.bound.ineq2} in Lemma \ref{lemma.delt.ineq} we have \( 1 - \frac{19}{5} \left(\log n \right)^{1/4} \delta(n) > 0 \). Therefore we have
\begin{align*}
 \beta_Z ~ < ~& \frac{5}{5 -19 \left(\log n \right)^{1/4} \delta(n)}
 \\ < ~& \frac{\sqrt{ 1 + \sigma_\zeta^2(n)}}{  1 - \frac{19}{5} \left(\log n \right)^{1/4} \delta(n) \sqrt{ 1 + \sigma_\zeta^2(n)}}
\\  = ~& \frac{\sqrt{ 1 + \sigma_\zeta^2(n)}\left( 1 +  \frac{19}{5} \left(\log n \right)^{1/4} \delta(n) \sqrt{ 
1 + \sigma_\zeta^2(n)    
} \right) }{\left[  1 - \frac{19}{5} \left(\log n \right)^{1/4} \delta(n) \sqrt{ 1 + \sigma_\zeta^2(n) } \right] \left[  1 + \frac{19}{5} \left(\log n \right)^{1/4} \delta(n) \sqrt{ 1 + \sigma_\zeta^2(n) }\right]} 
 \end{align*}
\begin{align*}  \leq ~ &  \sqrt{ 1 + \sigma_\zeta^2(n)}\Bigg( 1 +  \frac{19}{5} \left(\log n \right)^{1/4} \delta(n) \bigg(
1 + \sigma_\zeta^2(n)    
\\ & +  \left(1 +  \sigma_\epsilon^2\right)\left[1 - 
 3.8 \left[ \left(\log n \right)^{1/4} \right]^2 \delta^2(n) \left( 1 + \sigma_\zeta^2(n) \right) \right] \bigg)^{1/2} \Bigg) 
 \\ & \left. \middle/ \left[ 1 - 3.8 \left[ \left(\log n \right)^{1/4} \right]^2 \delta^2(n) \left( 1 + \sigma_\zeta^2(n) \right) \right] \right.
 \\ = ~&   \Bigg( 2\sqrt{ 1 + \sigma_\zeta^2(n)} + 2 \delta(n) \frac{19}{5} \left(\log n \right)^{1/4} \bigg(  \left( 1 + \sigma_\zeta^2(n) \right)  
\bigg( 1 + \sigma_\zeta^2(n) 
\\ &  
+  \left(1 +  \sigma_\epsilon^2\right)\left[1 - 
 3.8 \left[ \left(\log n \right)^{1/4} \right]^2 \delta^2(n) \left( 1 + \sigma_\zeta^2(n) \right) \right] \bigg) \bigg)^{1/2} \Bigg)
 \\ &  \left.  \middle/ \left(2\left[ 1 - 3.8 \left[ \left(\log n \right)^{1/4} \right]^2 \delta^2(n)  \left( 1 + \sigma_\zeta^2(n) \right) \right] \right) \right.
\\  = ~& \Bigg(-\tilde{b} + \bigg[ 15.2 \left[ \left(\log n \right)^{1/4} \right]^2 \delta^2(n) \left( 1 + \sigma_\zeta^2(n) \right)  
\bigg( 1 + \sigma_\zeta^2(n)    
\\ & \left. +  \left(1 +  \sigma_\epsilon^2\right)\left[1 - 
 3.8 \left[ \left(\log n \right)^{1/4} \right]^2 \delta^2(n) \left( 1 + \sigma_\zeta^2(n) \right) \right] \bigg) \bigg] ^{1/2} \Bigg) \middle/ \left( 2\tilde{a} \right) \right.
\end{align*}
where \(\tilde{a} = 1 - 3.8 \left[ \left(\log n \right)^{1/4} \right]^2 \delta^2(n)  \left( 1 + \sigma_\zeta^2(n) \right) \) and \(\tilde{b} = -2\sqrt{ 1 + \sigma_\zeta^2(n)}\), and
\begin{align*}
& 15.2 \left[ \left(\log n \right)^{1/4} \right]^2 \delta^2(n) \left( 1 + \sigma_\zeta^2(n) \right)  
\bigg( 1 + \sigma_\zeta^2(n)    
\\ & +  \left(1 +  \sigma_\epsilon^2\right)\left[1 - 
 3.8 \left[ \left(\log n \right)^{1/4} \right]^2 \delta^2(n) \left( 1 + \sigma_\zeta^2(n) \right) \right] \bigg)
\\ = ~&  15.2 \left[ \left(\log n \right)^{1/4} \right]^2 \delta^2(n) \left( 1 + \sigma_\zeta^2(n) \right)  
\bigg( 1 + \sigma_\zeta^2(n)    
+  1 +  \sigma_\epsilon^2
\\ & - \left(1 +  \sigma_\epsilon^2\right) \cdot 3.8 \left[ \left(\log n \right)^{1/4} \right]^2 \delta^2(n) \left( 1 + \sigma_\zeta^2(n) \right)  \bigg)
\\   = ~&  4 \left( 1 + \sigma_\zeta^2(n) \right)  \bigg( 3.8 \left[ \left(\log n \right)^{1/4} \right]^2 \delta^2(n) \left( 1 + \sigma_\zeta^2(n) \right)   +  3.8 \left[ \left(\log n \right)^{1/4} \right]^2 \delta^2(n)\left(1 +  \sigma_\epsilon^2\right)  
 \\ & - 3.8 \left[ \left(\log n \right)^{1/4} \right]^2 \delta^2(n)\left(1 +  \sigma_\epsilon^2\right) \cdot 3.8 \left[ \left(\log n \right)^{1/4} \right]^2 \delta^2(n) \left( 1 + \sigma_\zeta^2(n) \right)  \bigg)
\\  = ~& 4 \left( 1 + \sigma_\zeta^2(n) \right)  \bigg(1 -  1 + 3.8 \left[ \left(\log n \right)^{1/4} \right]^2 \delta^2(n) \left( 1 + \sigma_\zeta^2(n) \right)   
\\ & +  3.8 \left[ \left(\log n \right)^{1/4} \right]^2 \delta^2(n)\left(1 +  \sigma_\epsilon^2\right) \left[1 - 3.8 \left[ \left(\log n \right)^{1/4} \right]^2 \delta^2(n) \left( 1 + \sigma_\zeta^2(n) \right) \right]  \bigg)
\\ = ~& 4 \left( 1 + \sigma_\zeta^2(n) \right)  \bigg(1 -  \left[1 - 3.8 \left[ \left(\log n \right)^{1/4} \right]^2 \delta^2(n) \left( 1 + \sigma_\zeta^2(n) \right) \right]
\\ & \left[1 - 3.8 \left[ \left(\log n \right)^{1/4} \right]^2 \delta^2(n)\left(1 +  \sigma_\epsilon^2\right)\right]  \bigg)
\end{align*}
\begin{align*}  = ~& 4 \left( 1 + \sigma_\zeta^2(n) \right) 
\\ & - 4 \left[ 1 - 3.8 \left[ \left(\log n \right)^{1/4} \right]^2 \delta^2(n)  \left( 1 + \sigma_\zeta^2(n) \right) \right]  \left(1 + \sigma_\zeta^2(n) \right)\left[1 - 3.8 \left[ \left(\log n \right)^{1/4} \right]^2 \delta^2(n) \left( 1 +  \sigma_\epsilon^2\right) \right]
\\  = ~& \tilde{b}^2 - 4 \tilde{a}\tilde{c}
\end{align*}
where \(\tilde{c} = \left(1 + \sigma_\zeta^2(n) \right)\left[1 - 3.8 \left[ \left(\log n \right)^{1/4} \right]^2 \delta^2(n) \left( 1 +  \sigma_\epsilon^2\right) \right] \). Since from Lemma \ref{cond.beta.z.lemma} we have \(\beta_Z < \frac{\sqrt{ 1 + \sigma_\zeta^2(n)}}{  1 - \frac{19}{5} \left(\log n \right)^{1/4} \delta(n) \sqrt{ 1 + \sigma_\zeta^2(n)}}\), it follows that \( \beta_Z <  \frac{-\tilde{b} + \sqrt{ \tilde{b}^2 - 4\tilde{a} \tilde{c}}}{2\tilde{a}} \). Then
 \begin{align}
 & \beta_Z <  \frac{-\tilde{b} + \sqrt{ \tilde{b}^2 - 4\tilde{a} \tilde{c}}}{2\tilde{a}} \qquad  \implies \qquad  \tilde{a}\beta_Z^2 + \tilde{b}\beta_Z + \tilde{c} < 0 \nonumber
\\ \iff  \qquad & \beta_Z^2\left[1 - 3.8 \left[ \left(\log n \right)^{1/4} \right]^2 \delta^2(n)  \left( 1 + \sigma_\zeta^2(n) \right) \right]  \nonumber
\\ & - 2\beta_Z \sqrt{ 1 + \sigma_\zeta^2(n)} + \left(1 + \sigma_\zeta^2(n) \right)\left[1 - 3.8 \left[ \left(\log n \right)^{1/4} \right]^2 \delta^2(n) \left( 1 +  \sigma_\epsilon^2\right) \right]   < 0 \nonumber
\\  \iff  \qquad & \beta_Z^2 +1 + \sigma_\zeta^2(n) - 2\beta_Z \sqrt{ 1 + \sigma_\zeta^2(n)}    < 3.8 \left[ \left(\log n \right)^{1/4} \right]^2 \delta^2(n) \left(\beta_Z^2  + 1 +  \sigma_\epsilon^2\right) \left( 1 + \sigma_\zeta^2(n) \right)  \nonumber
\\  \stackrel{(a)}{\iff}  \qquad & \beta_Z - \sqrt{ 1 + \sigma_\zeta^2(n)} < \frac{19}{5} \left(\log n \right)^{1/4} \delta(n) \sqrt{ \left(\beta_Z^2  + 1 +  \sigma_\epsilon^2\right) \left( 1 + \sigma_\zeta^2(n) \right)} 
%\label{cond.beta.z.big.7}
%
\\ \iff  \qquad &  \frac{\beta_Z}{\sqrt{\left(\beta_Z^2  + 1 +  \sigma_\epsilon^2\right) \left( 1 + \sigma_\zeta^2(n) \right)}} -  \frac{1}{\sqrt{\beta_Z^2  + 1 +  \sigma_\epsilon^2 }}   < \frac{19}{5} \left(\log n \right)^{1/4} \delta(n) \nonumber
\\ \iff \qquad &  \rho_{1y}(n) -  \rho_{3y}  < \frac{19}{5} \left(\log n \right)^{1/4} \delta(n) , \nonumber
%  \label{itm:8a}
\end{align}
where \((a)\) follows because \(\beta_Z \geq \sqrt{1 + \sigma_\zeta^2(n)} \) from Lemma \ref{cond.beta.z.lemma}, yielding \eqref{itm:8a}. 

\end{proof}

\begin{proof}[Proof of Lemma \ref{lem.conc}] We will make use of the following concentration inequality:

\begin{lemma}[\textbf{Lemma D.3 from \citet{pmlr-v80-sun18c}}]\label{lemma.ncvx.cor}

Let \(\boldsymbol{X} = (\boldsymbol{X}_{\cdot 1}, \boldsymbol{X}_{\cdot 2}, \ldots, \boldsymbol{X}_{\cdot d})^\top\) be a zero-mean sub-Gaussian random vector with covariance \(\boldsymbol{\Sigma}^*\) with \((i,j)^{\text{th}}\) element \(\Sigma_{ij}^*\). (That is, each \(X_{\cdot i}/\Sigma_{ii}^*\) is sub-Gaussian with variance proxy 1.) Let \( \left\{\boldsymbol{X}^{(k)} \right\}_{k=1}^n\) be \(n\) i.i.d. samples from \(\boldsymbol{X}\). Let \(\boldsymbol{W}^2\) be a diagonal matrix with diagonal elements of \(\boldsymbol{\Sigma}^*\), and let \(\boldsymbol{C}^* := \boldsymbol{W}^{-1} \boldsymbol{\Sigma}^*\boldsymbol{W}^{-1}\) be the correlation matrix. Let \(\rho_{ij}\) be the \((i,j)^{\text{th}}\) element of \(\boldsymbol{C}^*\).

Consider the corresponding uncentered estimators: let \(\hat{\boldsymbol{\Sigma}} := n^{-1} \sum_{k=1}^n \boldsymbol{X}^{(k)} {\boldsymbol{X}^{(k)}}^\top\) denote the sample covariance and \(\hat{\boldsymbol{C}} := \hat{\boldsymbol{W}}^{-1} \hat{\boldsymbol{\Sigma}} \hat{\boldsymbol{W}}^{-1}\) denote the uncentered sample correlation matrix, where \(\hat{\boldsymbol{W}}^2\) is the diagonal matrix with diagonal elements of \(\hat{\boldsymbol{\Sigma}}\). Let \(\hat{\Sigma}_{ij}\) be the \((i,j)^{\text{th}}\) element of \(\hat{\boldsymbol{\Sigma}}\), and let \(\hat{R}_{ij}\) be the \((i,j)^{\text{th}}\) element of \(\hat{\boldsymbol{C}}\).

By Lemma D.1 in \citet{pmlr-v80-sun18c}, there exist constants \(t_0 \in (0,1]\) and 
\[
\tilde{c}_1 \in \left(0, \frac{e-1}{2e^2 \max_i \left\{\boldsymbol{\Sigma}_{ii}^* \right\}} \right]
\]
 such that for all \(t\) with \(0 \leq t \leq t_0\) the sample covariance \(\hat{\boldsymbol{\Sigma}}\) with \((i,j)^{\text{th}}\) element \(\hat{\Sigma}_{ij}\) satisfies
\[
\mathbb{P} \left( \left|  \hat{\Sigma}_{ij} - \Sigma_{ij}^* \right| \geq t\right) \leq 8 \exp \left\{ -\tilde{c}_1 nt^2 \right\} .
\]
Define 
\[
\tilde{c}_2 := \min \left\{ \frac{1}{4} \tilde{c}_1 \min \left(\boldsymbol{\Sigma}_{ii}^* \right)^2, \frac{1}{6} \right\} .
\]
 Then for any \(\delta \in \left[0, \min \left\{\frac{1}{2}, t_0 \max_{i} \left( \boldsymbol{\Sigma}_{ii}^* \right) \right\} \right] \) and for any \(i,j \in [d], i \neq j\),
%\begin{equation}\label{lemma.ncvx.cor.eqn}
\[
\mathbb{P} \left( \left|  \hat{R}_{ij} - \rho_{ij} \right| > \delta \right) \leq 6 \exp \left\{- \tilde{c}_2 n \delta^2 \right\}.
\]
%\end{equation}

\end{lemma}
From \eqref{beta.u.bound.ineq2} we have 
\begin{align*}
0 < \delta(n) & <  \min\left\{ \frac{1}{\sqrt{20 \left( 1 + \sigma_\epsilon^2 \right)}},  t_0 (\beta_Z^2 + 1 + \sigma_\epsilon^2)\right\} 
 \leq \min\left\{\frac{1}{2}, t_0\max_{i} \left( \boldsymbol{\Sigma}_{ii}^* \right) \right\},
\end{align*}
so we can apply Lemma \ref{lemma.ncvx.cor} using \(\delta = \delta(n)\). Then
\begin{align}
& \mathbb{P} \bigg(  \left\{ \left| \hat{R}_{1y} - \rho_{1y}(n) \right| \leq \delta(n) \right\}  \cap \left\{ \left| \hat{R}_{2y} - \rho_{2y}(n) \right| \leq \delta(n)  \right\} \cap \left\{   \left| \hat{R}_{3y} - \rho_{3y} \right| \leq \delta(n)   \right\} \nonumber
\\ & \cap \left\{   \left| \hat{R}_{12} - \rho_{12}(n) \right| \leq \delta(n)   \right\}  \cap \left\{   \left| \hat{R}_{13}  \right| \leq \delta(n)   \right\}  \bigg) \nonumber
\\ = ~ &1 -   \mathbb{P} \bigg(  \left\{ \left| \hat{R}_{1y} - \rho_{1y}(n) \right| > \delta(n) \right\}  \cup \left\{ \left| \hat{R}_{2y} - \rho_{2y}(n) \right| > \delta(n)  \right\} \cup \left\{   \left| \hat{R}_{3y} - \rho_{3y} \right| > \delta(n)   \right\} \nonumber
\\ & \cup \left\{   \left| \hat{R}_{12} - \rho_{12}(n) \right| > \delta(n)   \right\}  \cup \left\{   \left| \hat{R}_{13}  \right| > \delta(n)   \right\}  \bigg)  \nonumber
\\ \geq ~ &1 -   \mathbb{P} \left( \left| \hat{R}_{1y} - \rho_{1y}(n) \right| > \delta(n) \right) +  \mathbb{P} \left(   \left| \hat{R}_{2y} - \rho_{2y}(n) \right| > \delta(n)  \right)  +  \mathbb{P} \left( \left| \hat{R}_{3y} - \rho_{3y} \right| > \delta(n)   \right) \nonumber
\\ & +  \mathbb{P} \left(  \left| \hat{R}_{12} - \rho_{12}(n) \right| > \delta(n)   \right)  +  \mathbb{P} \left( \left| \hat{R}_{13}  \right| > \delta(n)    \right) \nonumber
\\ \geq ~ & 1 - 30 \exp \left\{- \tilde{c}_2 n \delta^2(n) \right\}. \label{prob.first.exp}
\end{align}
Consider the expression \(\tilde{c}_2 n \delta^2(n)\). From Lemma \ref{lemma.dist.bounds} we see that \( \min_i \left(\boldsymbol{\Sigma}_{ii}^* \right) = 1\) and \(\max_i \left(\boldsymbol{\Sigma}_{ii}^* \right)= \beta_Z^2 + 1 + \sigma_\epsilon^2 \), so
\[
\tilde{c}_1 \in \left(0, \frac{e-1}{2e^2(\beta_Z^2 + 1 + \sigma_\epsilon^2)} \right]
\]
and
\[
\tilde{c}_2 = \min \left\{ \frac{\tilde{c}_1}{4}  , \frac{1}{6} \right\}  = \frac{\tilde{c}_1}{4} \in   \left(0, \frac{e-1}{8e^2\left(\beta_Z^2 + 1 + \sigma_\epsilon^2 \right)} \right]
\]
for all \(n \in \mathbb{N}\). In general we will be interested in how changes in \(\beta_Z\) and \(\sigma_\epsilon^2\) affect our results, so rather than treating these as constants, we will define
\[
c_1 := \tilde{c}_1 \cdot \left(\beta_Z^2 + 1 + \sigma_\epsilon^2 \right) \in \left(0, \frac{e-1}{2e^2} \right]
\]
and
\begin{equation}\label{defn.c2}
c_2 := \frac{c_1}{4} \in \left(0, \frac{e-1}{8e^2} \right].
\end{equation}
Then we have 
\[
 \exp \left\{- \tilde{c}_2 n \delta^2(n) \right\} =  \exp \left\{- \frac{c_2}{\beta_Z^2 + 1 + \sigma_\epsilon^2} \cdot n \cdot \frac{(\beta_Z^2 + 1 + \sigma_\epsilon^2) \log n}{4 c_2 n} \right\}
%\\ = & \exp \left\{- \frac{1}{4  }\log n  \right\} \nonumber
 =   \frac{1}{n^{1/4}}. % \label{lem.a5.exp.id}
\]
The result follows from substituting this into \eqref{prob.first.exp}. % \eqref{lem.a5.exp.id}.

\end{proof}

\begin{proof}[Proof of Lemma \ref{lemma.bound.probs}] We will prove the results one at a time.
\begin{itemize}

\item First we will show that \(\mathcal{F}_n \subset \mathcal{A}_{13} \). Note that
\begin{align*}
\mathcal{A}_{13} & = \left\{ \hat{R}_{1y}  - \hat{R}_{3y}   > 0  \right\} 
\\  & = \left\{  \hat{R}_{1y}  - \rho_{1y}(n) - (\hat{R}_{3y}  - \rho_{3y}) + \rho_{1y}(n) - \rho_{3y}  > 0  \right\}  
 \\ & \stackrel{(a)}{\supseteq}   \left\{ \hat{R}_{1y}  - \rho_{1y}(n) - (\hat{R}_{3y}  - \rho_{3y}) + 2 \delta(n)  > 0  \right\}
  \\ & \supseteq \left\{ \left| \hat{R}_{1y}  - \rho_{1y}(n) \right| + \left| \hat{R}_{3y}  - \rho_{3y} \right|  < 2 \delta(n) \right\}
 \\ & \supset  \left\{ \left| \hat{R}_{1y}  - \rho_{1y}(n) \right| < \delta(n) \right\} \cap \left\{ \left| \hat{R}_{3y}  - \rho_{3y} \right|   <  \delta(n) \right\} 
 \\ & \supseteq \mathcal{F}_n,
%\\ & \stackrel{(b)}{\subseteq} \frac{30}{n^{1/4}},
\end{align*}
%\begin{align*}
%\mathcal{A}_{13}^c & = \left\{ \hat{R}_{1y}  - \hat{R}_{3y}   \leq 0  \right\} 
%\\  & = \left\{  \hat{R}_{1y}  - \rho_{1y}(n) - (\hat{R}_{3y}  - \rho_{3y}) + \rho_{1y}(n) - \rho_{3y}  \leq 0  \right\}  
% \\ & \stackrel{(a)}{\subset}   \left\{ \hat{R}_{1y}  - \rho_{1y}(n) - (\hat{R}_{3y}  - \rho_{3y}) + 2 \delta(n)  \leq 0  \right\}
%  \\ & \subseteq \left\{ \left| \hat{R}_{1y}  - \rho_{1y}(n) \right| + \left| \hat{R}_{3y}  - \rho_{3y} \right|   \geq 2 \delta(n) \right\}
% \\ & \subseteq  \left\{  \left\{ \left| \hat{R}_{1y}  - \rho_{1y}(n) \right| \geq \delta(n) \right\} \cup \left\{ \left| \hat{R}_{3y}  - \rho_{3y} \right|   \geq  \delta(n) \right\} \right\}
% \\ & \subseteq \mathcal{F}_n^c,
%%\\ & \stackrel{(b)}{\subseteq} \frac{30}{n^{1/4}},
%\end{align*}
%\begin{align*}
% \mathbb{P} \left( \hat{R}_{1y}  - \hat{R}_{3y}   \leq 0  \right)  & =  \mathbb{P} \left( \hat{R}_{1y}  - \rho_{1y}(n) - (\hat{R}_{3y}  - \rho_{3y}) + \rho_{1y}(n) - \rho_{3y}  \leq 0  \right)  
% \\ & \stackrel{(a)}{\leq}  \mathbb{P} \left( \hat{R}_{1y}  - \rho_{1y}(n) - (\hat{R}_{3y}  - \rho_{3y}) + 2 \delta(n)  \leq 0  \right)  
%  \\ & \leq  \mathbb{P} \left( \left| \hat{R}_{1y}  - \rho_{1y}(n) \right| + \left| \hat{R}_{3y}  - \rho_{3y} \right|   \geq 2 \delta(n) \right)  
% \\ & \leq  \mathbb{P} \left(  \left\{ \left| \hat{R}_{1y}  - \rho_{1y}(n) \right| \geq \delta(n) \right\} \cup \left\{ \left| \hat{R}_{3y}  - \rho_{3y} \right|   \geq  \delta(n) \right\} \right)  
%\\ & \stackrel{(b)}{\leq} \frac{30}{n^{1/4}},
%\end{align*}
where \((a)\) follows from \eqref{itm:7}.

\item Next we will show that \(\mathcal{F}_n \subset \mathcal{S}_1\).
\begin{align*}
S_1   = ~ &  \left\{ \hat{R}_{1y} > 0 \right\} \cap  \left\{ \hat{R}_{2y}  > 0  \right\} \cap  \left\{ \hat{R}_{3y}  > 0 \right\} \cap  \left\{ \hat{R}_{12} > 0 \right\} 
\\  \supset  ~ &  \left\{ |\hat{R}_{1y} - \rho_{1y}(n) | < \rho_{1y}(n) \right\}\cap \left\{  |\hat{R}_{2y} - \rho_{2y}(n) |  < \rho_{2y}(n)  \right\}
\\ &  \cap \left\{ |\hat{R}_{3y} - \rho_{3y}|  < \rho_{3y} \right\}\cap \left\{  |\hat{R}_{12} - \rho_{12}(n)| < \rho_{12}(n) \right\} 
\\ \stackrel{(b)}{\supseteq} ~ &  \left\{ |\hat{R}_{1y} - \rho_{1y}(n) | < \delta(n) \right\}\cap \left\{  |\hat{R}_{2y} - \rho_{2y}(n) |  < \delta(n)  \right\}
\\ &  \cap \left\{ |\hat{R}_{3y} - \rho_{3y}|  < \delta(n) \right\}\cap \left\{  |\hat{R}_{12} - \rho_{12}(n)| < \delta(n) \right\} 
\\ \supseteq ~ & \mathcal{F}_n,
\end{align*}
%\begin{align*}
%S_1^c   = ~ &  \left\{ \hat{R}_{1y} \leq 0 \right\} \cup  \left\{ \hat{R}_{2y}  \leq 0  \right\} \cup  \left\{ \hat{R}_{3y}  \leq 0 \right\} \cup  \left\{ \hat{R}_{12} \leq 0 \right\} 
%\\  \subset  ~ &  \left\{ |\hat{R}_{1y} - \rho_{1y}(n) | \geq \rho_{1y}(n) \right\}\cup \left\{  |\hat{R}_{2y} - \rho_{2y}(n) |  \geq \rho_{2y}(n)  \right\}
%\\ &  \cup \left\{ |\hat{R}_{3y} - \rho_{3y}|  \geq \rho_{3y} \right\}\cup \left\{  |\hat{R}_{12} - \rho_{12}(n)| \geq \rho_{12}(n) \right\} 
%\\ \stackrel{(b)}{\subseteq} ~ &  \left\{ |\hat{R}_{1y} - \rho_{1y}(n) | \geq \delta(n) \right\}\cup \left\{  |\hat{R}_{2y} - \rho_{2y}(n) |  \geq \delta(n)  \right\}
%\\ &  \cup \left\{ |\hat{R}_{3y} - \rho_{3y}|  \geq \delta(n) \right\}\cup \left\{  |\hat{R}_{12} - \rho_{12}(n)| \geq \delta(n) \right\} 
%\\ \subseteq ~ & \mathcal{F}_n^c,
%\end{align*}
where \((b)\) follows from by \eqref{itm:3} and \eqref{itm:1}.

\item Next we show that \(\mathcal{F}_n \subset \mathcal{A}_3\).
\begin{align*}
\mathcal{A}_3 =  ~ & \left\{   \hat{R}_{13}  \hat{R}_{1y} -  \hat{R}_{3y}  \leq 0 \right\}
\\ \supset  ~ & \left\{ | \hat{R}_{13} || \hat{R}_{1y} |  -  (\hat{R}_{3y} - \rho_{3y} ) \leq \rho_{3y} \right\}
\\ \supseteq  ~ & \left\{    | \hat{R}_{13}|  +  (\rho_{3y} - \hat{R}_{3y}  ) \leq \rho_{3y} \right\}
\\ \supseteq  ~ &    \left\{ | \hat{R}_{13}| \leq \rho_{3y}/2 \right\} \cap  \left\{ |\hat{R}_{3y} - \rho_{3y} | \leq \rho_{3y}/2 \right\} 
\\ \stackrel{(c)}{\supset}  ~ &    \left\{ | \hat{R}_{13}| < \delta(n) \right\} \cap  \left\{ |\hat{R}_{3y} - \rho_{3y} | < \delta(n) \right\} 
\\ \supseteq  ~ & \mathcal{F}_n,
\end{align*}
%\begin{align*}
%\mathcal{A}_3^C =  ~ & \left\{   \hat{R}_{13}  \hat{R}_{1y} -  \hat{R}_{3y}  > 0 \right\}
%\\ \subset  ~ & \left\{ | \hat{R}_{13} || \hat{R}_{1y} |  -  (\hat{R}_{3y} - \rho_{3y} ) > \rho_{3y} \right\}
%\\ \subseteq  ~ & \left\{    | \hat{R}_{13}|  +  (\rho_{3y} - \hat{R}_{3y}  ) > \rho_{3y} \right\}
%\\ \subseteq  ~ &    \left\{ | \hat{R}_{13}| > \rho_{3y}/2 \right\} \cup  \left\{ |\hat{R}_{3y} - \rho_{3y} | > \rho_{3y}/2 \right\} 
%\\ \stackrel{(c)}{\subseteq}  ~ &    \left\{ | \hat{R}_{13}| > \delta(n) \right\} \cup  \left\{ |\hat{R}_{3y} - \rho_{3y} | > \delta(n) \right\} 
%\\ \subseteq  ~ & \mathcal{F}_n^c,
%\end{align*}
%\begin{align*}
%\mathbb{P}(\mathcal{A}_3^C) =  ~ & \mathbb{P} \left(  \hat{R}_{3y} -   \hat{R}_{13}  \hat{R}_{1y}  < 0 \right)
%\\ \leq  ~ & \mathbb{P} \left(  \hat{R}_{3y} -   | \hat{R}_{13} || \hat{R}_{1y} | < 0 \right)
%\\ \leq  ~ & \mathbb{P} \left(  \hat{R}_{3y} -  | \hat{R}_{13}|  < 0 \right)
%\\ \vdots
%\\ =  ~ & \mathbb{P} \left(  \left( \hat{R}_{3y} - \rho_{3y} \right) -   \left(\hat{R}_{13} - \rho_{13} \right) \hat{R}_{1y} + \rho_{3y} - \rho_{13}  \hat{R}_{1y}   < 0 \right)
%\end{align*}
where \((c)\) follows from \eqref{itm:1}.
%By \eqref{lemma.ncvx.cor.eqn} (and the fact that \(\rho_{13} = 0\)),
%\begin{align*}
%& \mathbb{P} \left( \left| \hat{R}_{3y} - \rho_{3y} \right| \leq \delta(n) \cap \left| \hat{R}_{13} - \rho_{13}  \right| \leq \delta(n) \cap \left| \hat{R}_{1y} - \rho_{1y}(n) \right| \leq \delta(n) \right)
%\\ & = 1 -  \mathbb{P} \left( \left| \hat{R}_{3y} - \rho_{3y} \right| > \delta(n) \cup \left| \hat{R}_{13}  \right| > \delta(n) \cup \left| \hat{R}_{1y} - \rho_{1y}(n) \right| > \delta(n) \right)
%\\ & \geq 1 -  \mathbb{P} \left( \left| \hat{R}_{3y} - \rho_{3y} \right| > \delta(n) \right) - \mathbb{P} \left(  \left| \hat{R}_{13}  \right| > \delta(n) \right) - \mathbb{P} \left(  \left| \hat{R}_{1y} - \rho_{1y}(n) \right| > \delta(n) \right)
%\\ & \geq 1 - 18 \exp \left\{- \frac{c_2}{2 + \sigma_\epsilon^2} n \delta(n)^2 \right\},
%\end{align*}
%so with probability greater than or equal to \(1 - 18 \exp \left\{- \frac{c_2}{2 + \sigma_\epsilon^2} n \delta(n)^2 \right\} \) we have
%\begin{align*}
%\hat{R}_{3y} -    \hat{R}_{13}  \hat{R}_{1y} & = \left( \hat{R}_{3y} - \rho_{3y} \right) -   \left(\hat{R}_{13} - \rho_{13} \right) \hat{R}_{1y} + \rho_{3y} - \rho_{13}  \hat{R}_{1y} 
%\\  &  \geq -\delta(n) - \delta(n) \left| \hat{R}_{1y} \right| + \rho_{3y}  
%\\ & \geq -2\delta(n) + \rho_{3y} 
%\\ & \geq 0
%\end{align*}
%where the last step follows from \eqref{itm:1}. Therefore \eqref{lem.a5.exp.id} yields
%\[
%\mathbb{P}(\mathcal{A}_3^C)  \leq  18 \exp \left\{- \frac{c_2}{2 + \sigma_\epsilon^2} n \delta(n)^2 \right\} = \frac{18}{n^{1/4}}.
%\]

\item Next we will show \(\mathcal{F}_n \subseteq \tilde{\mathcal{E}}_1(n) \). We want to show that
\[
\mathcal{F}_n  \subseteq \left\{ \frac{\eta(n)}{1 - \hat{R}_{12}} > \frac{\hat{R}_{1y} - \hat{R}_{3y}}{1 - \hat{R}_{13}} \right\},
\]
or, equivalently,
\[
\mathcal{F}_n \subseteq \left\{ (\hat{R}_{1y} - \hat{R}_{3y})(1 - \hat{R}_{12}) + \eta(n) \left( \hat{R}_{13} - 1 \right) < 0 \right\}
\]
for \(\eta(n)\) defined in \eqref{lemma.def.eta.alpha}. Define
 \begin{equation}\label{lasso.theory.def.tilde.delta}
\tilde{\delta}(n) := 1 - \rho_{12}(n)=   \frac{\sigma_\zeta^2(n)}{1 + \sigma_\zeta^2(n)}   = \frac{10}{\sqrt{n \log n} + 10} ,
\end{equation}
where \(\sigma_\zeta^2(n) = 10/\sqrt{n \log n} \) as in \eqref{lasso.thm.sig.zeta.def} and the expression for \(\rho_{12}(n)\) is calculated in Lemma \ref{lemma.dist.bounds}. 
%By Lemma \ref{lem.conc},
%\begin{align}
% \mathbb{P} \bigg( & \left\{ \left| \hat{R}_{1y} - \rho_{1y}(n) \right| \leq \delta(n) \right\}  \cap \left\{   \left| \hat{R}_{3y} - \rho_{3y} \right| \leq \delta(n)   \right\} \nonumber
%\\ & \cap \left\{   \left| \hat{R}_{12} - \rho_{12}(n) \right| \leq \delta(n)   \right\}  \cap \left\{   \left| \hat{R}_{13}  \right| \leq \delta(n)   \right\}  \bigg) \geq 1 -   \frac{30}{n^{1/4}}. \label{lasso.gen.bound.e.2.new}
%\end{align}
Observe that on \(\mathcal{F}_n\) we have
\begin{align*}
& (\hat{R}_{1y} - \hat{R}_{3y})(1 - \hat{R}_{12}) + \eta(n) \left( \hat{R}_{13} - 1 \right)
\\ =  \quad&   (\hat{R}_{1y} - \hat{R}_{3y})(1 - \rho_{12}(n) - [ \hat{R}_{12} - \rho_{12}(n)]) 
+\eta(n) \left( \hat{R}_{13} - 1 \right)
 \\  \stackrel{(*)}{\leq}  \quad &  \left| \hat{R}_{1y} - \hat{R}_{3y}\right| \left[\tilde{\delta}(n) + \delta(n) \right] 
 +\eta(n) \left[ \delta(n)   - 1 \right]
\\  \stackrel{(d)}{\leq}  \quad&   \left(  \left|\hat{R}_{1y} - \rho_{1y}(n)   \right| + \left| \hat{R}_{3y} - \rho_{3y}  \right| +\rho_{1y}(n) - \rho_{3y} \right)  \left[\tilde{\delta}(n) + \delta(n) \right] 
\\&
 - \left[1 - \delta(n) \right] \eta(n)
\\ \text{using \eqref{itm:8a}} \quad   \stackrel{(*)}{\leq} \quad & \left(2 + \frac{19}{5} \left(\log n \right)^{1/4} \right) \delta(n)  [\tilde{\delta}(n) + \delta(n)] -   \left[1 - \delta(n) \right]\eta(n)
\\   \stackrel{(e)}{<}   \quad &   \left(2 + \frac{19}{5} \left(\log n \right)^{1/4} \right) \delta(n)  [\tilde{\delta}(n) + \delta(n)]   -  \frac{1}{2}  \eta(n)
\\ \stackrel{(f)}{=} \quad  & 0
\end{align*}
where \((d)\) follows from the triangle inequality, \((e)\) follows from \(\delta(n) < 1/2\) from Lemma \ref{lemma.delt.ineq}, \((f)\) follows because from the definition of \(\eta(n)\) in \eqref{lemma.def.eta.alpha}
\begin{align*}
&\eta(n) =  2\left(2 + \frac{19}{5} \left(\log n \right)^{1/4} \right) \delta(n)[\delta(n) + \tilde{\delta}(n)]
\\ \iff \qquad & \left(2 + \frac{19}{5} \left(\log n \right)^{1/4} \right)\delta(n)[\delta(n) + \tilde{\delta}(n)]    -  \frac{1}{2}  \eta(n) = 0  ,
\end{align*}
and the steps labeled with \((*)\) use the fact that we are on \(\mathcal{F}_n\).

\item Finally we will show \(\mathcal{F}_n \subseteq \mathcal{E}_2\). We want to show that
%\begin{align*}
%\mathbb{P}\left( \mathcal{E}_2 \right) & = \mathbb{P}\left( 
\[
\mathcal{F}_n \subseteq \left\{ \frac{\hat{R}_{1y} + \hat{R}_{2y}}{1  + \hat{R}_{12}} 
> 
\frac{\hat{R}_{1y} - \hat{R}_{3y}}{1  - \hat{R}_{13}}  \right\}
\]
%\right) 
% \\ & =   \mathbb{P}\left( 
or, equivalently,
\[
\mathcal{F}_n \subseteq \left\{  \left(\hat{R}_{1y} - \hat{R}_{3y} \right) \left(1  + \hat{R}_{12} \right) 
- \left(\hat{R}_{1y} + \hat{R}_{2y}\right)\left(1  - \hat{R}_{13}\right)   
< 0 \right\}.
\]
%\right).
%\end{align*}
%holds with high probability. By Lemma \ref{lem.conc},
%\begin{align}
% \mathbb{P} \bigg( & \left\{ \left| \hat{R}_{1y} - \rho_{1y}(n) \right| \leq \delta(n) \right\}  \cap   \left\{ \left| \hat{R}_{2y} - \rho_{2y}(n) \right| \leq \delta(n) \right\} \nonumber
%\\ &  \cap  \left\{   \left| \hat{R}_{3y} - \rho_{3y} \right| \leq \delta(n)   \right\} \ \cap \left\{   \left| \hat{R}_{13}  \right| \leq \delta(n)   \right\}  \bigg) \geq 1 -   \frac{30}{n^{1/4}}. \nonumber
%%\label{lasso.gen.bound.e.2.new}
%\end{align}
On \(\mathcal{F}_n\),
\begin{align*}
& \left(\hat{R}_{1y} - \hat{R}_{3y} \right) \left(1  + \hat{R}_{12} \right) 
- \left(\hat{R}_{1y} + \hat{R}_{2y}\right)\left(1  - \hat{R}_{13}\right)   
\\ = ~ &  - \left( \hat{R}_{3y}  - \rho_{3y} \right) - \rho_{3y} +  \left( \hat{R}_{1y} - \rho_{1y}(n) \right)\hat{R}_{12}   -  \left( \hat{R}_{3y} - \rho_{3y} \right)   \hat{R}_{12} 
  \\ & + \left( \rho_{1y}(n) -\rho_{3y} \right) \hat{R}_{12} 
 - \left( \hat{R}_{2y}  - \rho_{2y}(n) \right) - \rho_{2y}(n) + \left( \hat{R}_{1y} + \hat{R}_{2y} \right) \hat{R}_{13}
 \\ \stackrel{(*)}{\leq} ~ &  \delta(n) - \rho_{3y} + 2 \delta(n) \left| \hat{R}_{12} \right|  + \left( \rho_{1y}(n) - \rho_{3y} \right)\left| \hat{R}_{12} \right| 
+ \delta(n)  - \rho_{2y}(n) + 2 \left| \hat{R}_{13} \right|
 \\ \stackrel{(*)}{\leq} ~ &  2\delta(n) - \rho_{3y} +  2 \delta(n) +  \rho_{1y}(n) - \rho_{3y}
 - \rho_{2y}(n) 
 + 2  \delta(n)
 \\ = ~ & 6 \delta(n) - 2\rho_{3y} 
 \\ \stackrel{(g)}{<} ~ & 0,
%
% \\ \vdots
%\\  \leq~ &
% 2\left(\hat{R}_{1y} - \hat{R}_{3y} \right)
%- \left(\hat{R}_{1y} + \hat{R}_{2y}\right)\left(1  - \hat{R}_{13}
%\right)
%\\  \leq~ & 2\left(\hat{R}_{1y} - \hat{R}_{3y} \right)
%- \left(\hat{R}_{1y}  + \hat{R}_{2y} \right)\left(1 - \delta(n) \right)   
%\\  = ~ & \hat{R}_{1y} \left(1  + \delta(n) \right)  - \hat{R}_{2y} \left(1 - \delta(n) \right)  - 2 \hat{R}_{3y}
%\\  = ~ & \rho_{1y}(n) \left(1  + \delta(n) \right) 
% - \rho_{1y}(n) \left(1 - \delta(n) \right)  
% - 2 \rho_{3y} + \left(\hat{R}_{1y} 
% - \rho_{1y}(n) \right) \left(1  + \delta(n) \right)
%\\ & - \left(\hat{R}_{2y} - \rho_{1y}(n) \right) \left(1 - \delta(n) \right) - 2 \left(\hat{R}_{3y} - \rho_{3y} \right)
%\\  \leq ~ & 2 \delta(n) \rho_{1y}(n)
% - 2 \rho_{3y}
%  + \delta(n) \left(1  + \delta(n) \right)
% +  \delta(n) \left(1 - \delta(n) \right) + 2 \delta(n)
%\\  = ~ & 2 \delta(n) \rho_{1y}(n) 
% - 2 \rho_{3y}
%  + 4 \delta(n)
%%
%\\  \leq ~ & 2 \delta(n) \cdot 1
% - 2 \rho_{3y}
%+ 4 \delta(n)
%\\ <  ~&  0
\end{align*}
where we used \(\rho_{1y}(n) = \rho_{2y}(n)\), \((g)\) follows from \eqref{itm:1}, and the steps where we used the fact that we are on \(\mathcal{F}_n\) are labeled with \((*)\). %Therefore applying \eqref{lem.a5.exp.id} 
%\[
%\mathbb{P}\left( \mathcal{E}_2^C  \right) \leq 24 \exp \left\{- \frac{c_2}{2 + \sigma_\epsilon^2} n \delta(n)^2 \right\} = \frac{24}{n^{1/4}}.
%\]

%\[
%\hat{R}_{1y} - \hat{R}_{3y} \geq \rho_{1y}(n) - \delta(n) - \rho_{3y}(n) - \delta(n) \geq 0
%\]
%by \eqref{itm:7}. 

\end{itemize}

\end{proof}

\begin{proof}[Proof of Lemma \ref{lemma.jacob}] We will apply Proposition \ref{lemma.prob.a1.eta}. We see from Lemma \ref{lemma.dist.bounds} that all of the required assumptions on the covariance matrix for Proposition \ref{lemma.prob.a1.eta} are satisfied. Note that \(\tilde{\sigma}\) as defined in \eqref{lemma.def.tilde.sigma} varies with \(n\) in the setting of Theorem \ref{thm.result.sel}; in particular, substituting in the quantities from Lemma \ref{lemma.dist.bounds} into \eqref{lemma.def.tilde.sigma} yields
\begin{equation}\label{lemma.def.tilde.sigma.n}
\tilde{\sigma}(n) := \sqrt{ (1 - \rho_{12}(n)) \left( - 3 \rho_{1y}^2(n) + \rho_{1y}^2(n) \rho_{12}(n) + 2 \right) }.
\end{equation}
It only remains to substitute this expression, \(\eta(n)\) from \eqref{lemma.def.eta.alpha}, and the correlation and covariance expressions from Lemma \ref{lemma.dist.bounds} into the conclusion of Proposition \ref{lemma.prob.a1.eta}. Some of these manipulations are tedious and we defer them to Appendix \ref{tech.lemmas}.
\begin{lemma}\label{lemma.bound.alpha}
If \(n \geq 100\), there exists a finite constant \(c_6 >0\) such that 
\begin{equation}\label{lemma.bound.alpha.alpha}
\Phi\left( \frac{  \sqrt{n}\eta(n)}{\tilde{\sigma}(n)}\right)   < \frac{1}{2} + c_6 \frac{\left(\beta_Z^2  + 1 +  \sigma_\epsilon^2\right)^{3/2}}{\left(1 + \sigma_\epsilon^2 \right)^{1/2}} \cdot \frac{\left(\log n\right)^{3/2}}{n^{1/4}}  
\end{equation}
and
\begin{equation}\label{lemma.bound.alpha.sec}
 \frac{1}{ \tilde{\sigma}(n)} \leq   \sqrt{\frac{\beta_Z^2  + 1 +  \sigma_\epsilon^2 }{10(1 + \sigma_\epsilon^2)}}  \left( n \log n \right)^{1/4},
\end{equation}
where \(\tilde{\sigma}(n)\) is defined in \eqref{lemma.def.tilde.sigma.n}.
\end{lemma}
Using this, we have
 \begin{align*}
 & \Phi \left(  \frac{\eta(n) \sqrt{n}}{\tilde{\sigma}(n)} \right) +   \left(    \frac{1463}{ \tilde{\sigma}(n)} + 14 \right)  \frac{\Sigma_{yy}^3}{n^{1/2}} 
 \\ < ~  & \frac{1}{2} + c_6 \frac{\left(\beta_Z^2  + 1 +  \sigma_\epsilon^2\right)^{3/2}}{\left(1 + \sigma_\epsilon^2 \right)^{1/2}} \cdot \frac{\left(\log n\right)^{3/2}}{n^{1/4}} 
 \\ & +   \left(    \frac{1463}{\sqrt{10}} \sqrt{\frac{\beta_Z^2  + 1 +  \sigma_\epsilon^2 }{1 + \sigma_\epsilon^2}}  \left( n \log n \right)^{1/4} + 14 \right)  \frac{\left( \beta_Z^2 + 1 + \sigma_\epsilon^2\right)^3}{n^{1/2}} 
  \\ < ~ & \frac{1}{2} +  c_5 \left(\beta_Z^2 + 1 + \sigma_\epsilon^2 \right)^{7/2} \frac{\left(\log n \right)^{3/2}}{n^{1/4}} 
 \end{align*}
 for \(c_5 := c_6/2^2 + 1463/\sqrt{10} + 14/\sqrt{2} \).
 \end{proof}

\section{Other Results}\label{app.other.res}

\subsection{Proofs of Corollary \ref{cor.stab.sel2} and Proposition \ref{cor.gen.stab.sel}}\label{other.other.proofs}

\begin{proof}[Proof of Corollary \ref{cor.stab.sel2}] By linearity of expectation, in the \citet{meinshausen-2010} stability selection algorithm
\[
\E \left[   \binom{n}{\lfloor n/2 \rfloor}^{-1} \sum_{b=1}^{\binom{n}{\lfloor n/2 \rfloor}}  \mathbbm{1} \left\{  j \in \hat{S}^\lambda \left( A_b \right) \right\} \right]   = \mathbb{P} \left(  j \in \hat{S}_{\lfloor n/2 \rfloor}^\lambda \left( A_b \right) \right) \qquad \forall j \in [3].
\]
Note that the \citet{shah_samworth_2012} estimator has the same expectation. Therefore we only need to bound these selection probabilities, which we can do with Theorem \ref{thm.result.sel}. First we will upper-bound the probability of the event \(\mathcal{E}_j := \left\{  j \in \hat{S}_{\lfloor n/2 \rfloor}^\lambda \left( A_b \right) \right\}\)  for \(j \in [2]\) on any one lasso fit from a subsample of stability selection for any \(\lambda\) between the second and third knots of the lasso path. (Lemma \ref{lemma.x1.not.leave} assures us that there will be two features in the selected set at this point.) We have
\[
1   =    \mathbb{P}(\{1 \text{ is selected first, then } 3 \}) +  \mathbb{P}(\{2 \text{ is selected first, then } 3 \})    +  \mathbb{P}(\{ \text{else}  \})  
\]
so it follows from Theorem \ref{thm.result.sel} that
\begin{align*}
\mathbb{P}(\mathcal{E}_1) & \leq 1 - \mathbb{P}(\{2 \text{ is selected first, then } 3 \})  
\\ & \leq \frac{1}{2} +  c_{3}  \left(\beta_Z^2 + 1 + \sigma_\epsilon^2 \right)^{7/2} \frac{\left(\log \lfloor n/2 \rfloor \right)^{3/2}}{\lfloor n/2 \rfloor^{1/4}}
\\ & \leq  \frac{1}{2} + 3^4 \cdot  c_{3}  \left(\beta_Z^2 + 1 + \sigma_\epsilon^2 \right)^{7/2} \frac{\left(\log n \right)^{3/2}}{n^{1/4}}
\end{align*}
where we used the fact that \(\lfloor n/2 \rfloor \geq n/3\) for all \(n \geq 100\). By exchangeability, the same is true of \(\boldsymbol{X}_{\cdot 2}\). Define
\begin{equation}\label{def.c14}
c_4 := 81 c_{3}  .
\end{equation}

Next we will lower-bound the probability that \(3 \in \hat{S}^{\lambda}(A_b)\). Between the second and third knots of the lasso path, by Theorem \ref{thm.result.sel} we have
\begin{align*}
\mathbb{P}( 3 \in \hat{S}_{\lfloor n/2 \rfloor}^\lambda \left( A_b \right) ) & \geq  \mathbb{P}(\{1 \text{ is selected first, then } 3 \})  +  \mathbb{P}(\{2 \text{ is selected first, then } 3 \}) 
\\ & \geq 2 \left( \frac{1}{2} +  c_{3}  \left(\beta_Z^2 + 1 + \sigma_\epsilon^2 \right)^{7/2} \frac{\left(\log \lfloor n/2 \rfloor \right)^{3/2}}{\lfloor n/2 \rfloor^{1/4}}   \right)
\\ & \geq 1 - 2 c_4 \left(\beta_Z^2 + 1 + \sigma_\epsilon^2 \right)^{7/2} \frac{\left(\log n \right)^{3/2}}{n^{1/4}} .
\end{align*}
\end{proof}
\begin{remark}
Observe that \( \frac{1}{2} + c_4 \left(\beta_Z^2 + 1 + \sigma_\epsilon^2 \right)^{7/2} \frac{\left(\log n \right)^{3/2}}{n^{1/4}}\) also upper-bounds the probability that \(j \in \hat{S}_{\lfloor n/2 \rfloor}^{\lambda}(A_b), j \in [2]\) for any \(\lambda\) between the first and second knots of the lasso path (that is, for a selection procedure that selects the first feature to enter the lasso path), since
\[
 \mathbb{P}(\{1 \text{ is selected first} \}) \leq 1 - \mathbb{P}(\{2 \text{ is selected first, then } 3 \})  .
\]
\end{remark}

\begin{proof}[Proof of Proposition \ref{cor.gen.stab.sel}] Let
\[
\mathcal{E}_k := \left\{  C_k \cap \hat{S}(A_b)  \neq \emptyset \right\}, \qquad k \in [2].
% \mathbbm{1} \left\{ \mathcal{G}(j) \cap \hat{S}(A_b)  \neq \emptyset \right\}
\]
By linearity of expectation, using subsamples as in the algorithm proposed by \citet{meinshausen-2010},
\begin{align*}
\E \left[ \hat{\Theta}_{B} \left(C_k \right)  \right]  = \E \left[   \frac{1}{B} \sum_{b=1}^B  \mathbbm{1} \left\{ \mathcal{E}_k \right\} \right] 
 = \mathbb{P} \left( \mathcal{E}_k \right), \qquad k \in [2].
\end{align*}
The same is true for complementary pairs subsampling as proposed by \citet{shah_samworth_2012}. Therefore it suffices to show \( \mathbb{P} \left( \mathcal{E}_1 \right) \geq \mathbb{P} \left( \mathcal{E}_2 \right) \) for \(\lambda\) between the first and third knots of the lasso path. By Lemma \ref{lemma.x1.not.leave}, two features are selected by the lasso between the second and third knots of the lasso path almost surely. Then \( \mathbb{P} \left( \mathcal{E}_1 \right)  = 1\) by the pigeonhole principle, and the result follows.

To show that the result holds between the first and second knots of the lasso path (that is, when only one feature is selected), it suffices to show that feature \(1\) or \(2\) is selected first by the lasso with high probability. From Theorem \ref{thm.result.sel} we have
\begin{align}
& \mathbb{P}(\{\boldsymbol{X}_{\cdot 1} \text{ is the first feature to enter the lasso path}\})  \nonumber
\\  \geq ~ &  \mathbb{P}(\{1 \text{ is selected first, then } 3 \}) \label{cor.gen.stab.sel.rmk}
\\  \geq  ~ &\frac{1}{2} -  c_{3}  \left(\beta_Z^2 + 1 + \sigma_\epsilon^2 \right)^{7/2} \frac{\left(\log \lfloor n/2 \rfloor \right)^{3/2}}{\lfloor n/2 \rfloor^{1/4}}  \nonumber
\\  \geq  ~ & \frac{1}{2} - 3^4 \cdot  c_{3}  \left(\beta_Z^2 + 1 + \sigma_\epsilon^2 \right)^{7/2} \frac{\left(\log n \right)^{3/2}}{n^{1/4}}
 , \nonumber
\end{align}
where we used \(\lfloor n/2 \rfloor \geq n/3\) for all \(n \geq 100\). By exchangeability, the same is true of \(\boldsymbol{X}_{\cdot 2}\). Since these events are disjoint, between the first and second knots of the lasso path we have
\begin{align*}
\mathbb{P}( \mathcal{E}_1)  &  \geq 1 - 2 c_4 \left(\beta_Z^2 + 1 + \sigma_\epsilon^2 \right)^{7/2} \frac{\left(\log n \right)^{3/2}}{n^{1/4}},
\end{align*}
where \(c_4\) is defined in \eqref{def.c14}. Next, for any \(\lambda\) between the first and second knots of the lasso path, note that
\[
\mathbb{P}(\mathcal{E}_2 ) = \mathbb{P}(\mathcal{E}_1^c) \leq 2 c_4 \left(\beta_Z^2 + 1 + \sigma_\epsilon^2 \right)^{7/2} \frac{\left(\log n \right)^{3/2}}{n^{1/4}} .
\]
So in this regime,
\begin{align*}
\E \left[ \hat{\Theta}_{B} \left( C_1 \right)  \right]  & \geq  1 - 2 c_4 \left(\beta_Z^2 + 1 + \sigma_\epsilon^2 \right)^{7/2} \frac{\left(\log n \right)^{3/2}}{n^{1/4}}  
\\ & \stackrel{(*)}{\geq}  2 c_4 \left(\beta_Z^2 + 1 + \sigma_\epsilon^2 \right)^{7/2} \frac{\left(\log n \right)^{3/2}}{n^{1/4}}
\\ &  \geq \E \left[ \hat{\Theta}_{B} \left( C_2 \right)  \right] 
\end{align*}
holds, with \((*)\) holding for \(n\) sufficiently large.

\end{proof}

\begin{remark}

The result also holds (trivially) for any \(\lambda\) where the size of the selected set is 3. This result is loose in the sense that in line \eqref{cor.gen.stab.sel.rmk} we only need the probability of \(\boldsymbol{X}_{\cdot 1}\) entering the lasso path first, but examining the proof of Theorem \ref{thm.result.sel} (in particular, Lemmas \ref{lemma.bound.probs} and \ref{lem.conc}), we see that this event would have the same rate of convergence under our theory, so we just apply Theorem \ref{thm.result.sel} rather than working out a separate result. 

\end{remark}

\subsection{Proofs of Proposition \ref{prop.gen.pred.risk}, Corollary \ref{cor.thm.1.risk}, and Proposition \ref{proxies.risk.different.weight}}\label{risk.proofs}

Before proving these results, we state some lemmas that we will need. (The proofs are provided in Appendix \ref{tech.lemmas}.)

\begin{lemma}\label{dist.cond.one.sub.e.mse.lem}

Assume the setup of \eqref{thm.alt.spec.x.gen}, \eqref{thm.alt.spec.gen}, and \eqref{thm.alt.spec.y.gen} with only one weak signal feature \( \boldsymbol{X}_{\cdot q + 1} \) (that is, \(p = q + 1\)). Then the following identities hold:

\begin{align}
\E \left[ \boldsymbol{y}^\top \boldsymbol{y} \right] & = n(\beta_Z^2 + \beta_{q+1}^2 + \sigma_\epsilon^2), \label{dist.cond.one.sub.e.mse.yTy}
\\ \hat{\beta}_j \mid \boldsymbol{X}_{\cdot j} & \sim \mathcal{N} \left( \frac{\beta_Z}{1 + \sigma_{\zeta j}^2} ,  \left( \frac{\beta_Z^2  \sigma_{\zeta j}^2}{1 + \sigma_{\zeta j}^2}   +   \beta_{q+1}^2  +  \sigma_\epsilon^2   \right)  \frac{1}{ \boldsymbol{X}_{\cdot j}^\top  \boldsymbol{X}_{\cdot j} }\right) , \qquad j \in [q],  \label{dist.cond.one.sub.e.mse.var.alpha.mid.xj}
\\ \hat{\beta}_{q + 1} \mid \boldsymbol{X}_{\cdot q + 1} & \sim \mathcal{N} \left( \beta_{q+1} ,  \frac{\beta_Z^2 + \sigma_\epsilon^2}{ \boldsymbol{X}_{\cdot q +1}^\top  \boldsymbol{X}_{\cdot q + 1}}  \right) ,\label{dist.cond.one.sub.e.mse.var.alpha.mid.x3}
\\  \hat{\beta}_Z \mid \boldsymbol{Z} & \sim \mathcal{N} \left( \beta_Z ,    \frac{\beta_{q+1}^2 + \sigma_\epsilon^2}{ \boldsymbol{Z}^\top  \boldsymbol{Z}} \right), \label{dist.cond.one.sub.e.mse.dist.alpha.mid.z}
\\  \E \left[ \boldsymbol{X}_{\cdot j}^\top \boldsymbol{X}_{\cdot j} \right] & = n(1 + \sigma_{\zeta j}^2), \qquad j \in [q], \label{dist.cond.one.sub.e.mse.xTx}
%
%\\  \E \left[ \boldsymbol{X}_{\cdot j}^\top \boldsymbol{X}_{\cdot k} \right] & = n \qquad j, k \in [q], j \neq k \label{dist.cond.one.sub.e.mse.xTdiffx}
%%
\\   \E \left[ \boldsymbol{X}_{\cdot j}^\top \boldsymbol{y}  \right]  &=  \beta_Z n, \qquad j \in [q],   \label{dist.cond.one.sub.e.mse.xTy}
\\   \E \left[ \boldsymbol{X}_{\cdot q + 1}^\top \boldsymbol{y}  \right]  &= \beta_{q+1} n,  \label{dist.cond.one.sub.e.mse.xTy.3}
\\  \E \left[   \frac{1}{ \boldsymbol{X}_{\cdot j}^\top  \boldsymbol{X}_{\cdot j} }  \right] & = \frac{1}{(n-2)(1+ \sigma_{\zeta j}^2)}   , \qquad j \in [q],\label{dist.cond.one.sub.e.mse.xTx.inv}
\\  \E \left[   \frac{1}{ \boldsymbol{X}_{\cdot q + 1}^\top  \boldsymbol{X}_{\cdot q + 1} }  \right] & = \frac{1}{n-2}   ,\label{dist.cond.one.sub.e.mse.xTx.inv.3}
\\ \E \left[ \hat{\beta}_j^2  \right] & =  \frac{1}{1+ \sigma_{\zeta j}^2} \left[ \left(  \frac{\beta_Z^2 \sigma_{\zeta j}^2}{1 + \sigma_{\zeta j}^2}  + \beta_{q+1}^2 + \sigma_\epsilon^2 \right) \cdot \frac{1}{n-2} +  \frac{\beta_Z^2}{1 + \sigma_{\zeta j}^2} \right],  \qquad j \in [q],  \label{dist.cond.one.sub.e.mse.ex.alpha.squared.mid.xj}
\\ \E \left[ \hat{\beta}_{q + 1}^2  \right] & = \frac{\beta_Z^2 + \sigma_\epsilon^2}{ n-2 } + \beta_{q+1}^2 , \qquad \text{and} \label{dist.cond.one.sub.e.mse.ex.alpha.squared.mid.x3}
\\ \E \left[ \hat{\beta}_Z^2  \right] & = \frac{\beta_{q+1}^2 + \sigma_\epsilon^2}{ n - 2}   + \beta_Z^2  .\label{dist.cond.one.sub.e.mse.alpha.sq.z} 
\end{align}

\end{lemma}

\begin{lemma}\label{pred.risk.lemma.x1}
Assume the setup of \eqref{thm.alt.spec.x.gen}, \eqref{thm.alt.spec.gen}, and \eqref{thm.alt.spec.y.gen}.
\begin{enumerate}[(i)]
\item (Prediction risk of selecting \(\boldsymbol{X}_{\cdot j}, j \in [q]\).) 
\[
R(j) =   \frac{n - 1 }{n-2} \left(  \frac{\beta_Z^2 \sigma_{\zeta j}^2}{1+ \sigma_{\zeta j}^2} +  \sum_{j' = q+1}^p \beta_{j'}^2 + \sigma_\epsilon^2 \right), \qquad j \in [q].
\]

\item (Prediction risk of selecting \(\boldsymbol{X}_{\cdot j}, j \in \{q+1, \ldots, p\}\).) 

\[
R(q + 1) = \frac{n-1}{n-2}  \left( \beta_Z^2  +  \sum_{j' \in \{q+1, \ldots, p\} \setminus j} \beta_{j'}^2 + \sigma_\epsilon^2  \right)    , \qquad j \in \{q+1, \ldots, p\}.
\]

\item If for some \(j \in [q]\) \(\boldsymbol{X}_{\cdot j} = \boldsymbol{Z}\) (that is, if \(\sigma_{\zeta j}^2 = 0\)), 
 
 \begin{equation}\label{dist.cond.z.e.mse.result}
%\E \left[ \frac{1}{n} \left\lVert \tilde{\boldsymbol{y}} - \hat{\boldsymbol{y}} \right\rVert_2^2 \right] = 2 + \frac{n+1}{n} \sigma_\epsilon^2  
% \mathcal{E}_{\text{ideal}} := \E \left[ \frac{1}{n} \left\lVert \tilde{\boldsymbol{y}} - \hat{\boldsymbol{y}} \right\rVert_2^2 \right] = \frac{n-1 }{n-2} \cdot \sigma_\epsilon^2
R(j) =   \mathcal{E}_{\text{ideal}} :=     \frac{n-1}{n-2} \left(  \sum_{j' = q+1}^p \beta_{j'}^2 + \sigma_\epsilon^2  \right).
%\frac{n-1}{n-2}  \left(\beta_{q+1}^2 + \sigma_\epsilon^2 \right).
 \end{equation}

\end{enumerate}

\end{lemma}

\begin{lemma}\label{lem.opt.weighting.all.positive}

Assume the setup of Proposition \ref{proxies.risk.different.weight}. The prediction risk for arbitrary weights \((w_1, \ldots, w_q) \in \Delta^{q-1}\), as defined in \eqref{pred.risk.synth.proxy.def}, is
\begin{equation}\label{proxies.risk.different.weight.pred.risk}
%\mathcal{E}_{\text{ideal}} +   \frac{n - 1 }{n-2} \cdot \frac{\beta_Z^2   \sum_{j=1}^q w_j^2 \sigma_{\zeta j}^2}{1+   \sum_{j=1}^q w_j^2 \sigma_{\zeta j}^2},
 \frac{n - 1 }{n-2} \left(  \frac{\beta_Z^2   \sum_{j=1}^q w_j^2 \sigma_{\zeta j}^2 }{1+ \sum_{j=1}^q w_j^2 \sigma_{\zeta j}^2} + \sum_{j=q+1}^p \beta_j^2 +  \sigma_\epsilon^2  \right).
\end{equation}

\end{lemma}

We are now prepared to prove our main results from the paper.

\begin{proof}[Proof of Proposition \ref{prop.gen.pred.risk}] By Lemma \ref{pred.risk.lemma.x1}(i) the risk of an arbitrary \(\boldsymbol{X}_{\cdot j}\), \(j \in [q]\) is
\[
R(j) =  \frac{n - 1 }{n-2} \left(  \frac{\beta_Z^2 \sigma_{\zeta j}^2}{1+ \sigma_{\zeta j}^2} + \sum_{k = q+1}^p \beta_k^2   +  \sigma_\epsilon^2  \right) \qquad \forall j \in [q],
\]
%Notice that
%\[
%\boldsymbol{y} = \beta_Z \boldsymbol{Z} + \sum_{k = q+1}^p \beta_k \boldsymbol{X}_{\cdot k} + \boldsymbol{\epsilon} =   \beta_Z \boldsymbol{Z} + \beta_{q +1} \boldsymbol{X}_{\cdot , q+1} + \tilde{\boldsymbol{\epsilon}},
%\]
%where \(\tilde{\boldsymbol{\epsilon}} = \boldsymbol{\epsilon} +  \sum_{k = q+2}^p \beta_k \boldsymbol{X}_{\cdot k} \sim \mathcal{N}_n\left(\boldsymbol{0}, \left( \sigma_\epsilon^2 + \sum_{k = q+2}^p \beta_k^2  \right) \boldsymbol{I}_n \right)   \). Then 
%Similarly, for an arbitrary \(j \in \{q+ 1, \ldots, p\}\), 
%\[
%\boldsymbol{y} =   \beta_Z \boldsymbol{Z} + \beta_{j} \boldsymbol{X}_{\cdot , j} + \tilde{\boldsymbol{\epsilon}},
%\]
%where \(\tilde{\boldsymbol{\epsilon}} = \boldsymbol{\epsilon} +  \sum_{k \in \{q+1, \ldots, p\} \setminus j} \beta_k \boldsymbol{X}_{\cdot k} \sim \mathcal{N}_n\left(\boldsymbol{0}, \left( \sigma_\epsilon^2 +  \sum_{k \in \{q+1, \ldots, p\} \setminus j}\beta_k^2  \right) \boldsymbol{I}_n \right)   \). Then 
and by Lemma \ref{pred.risk.lemma.x1}(ii), the risk for an arbitrary \(j \in \{q+ 1, \ldots, p\}\) is
\[
R(j) = \frac{n-1}{n-2}  \left( \beta_Z^2  + \sigma_\epsilon^2  +  \sum_{k \in \{q+1, \ldots, p\} \setminus j}\beta_k^2   \right) .
\]
Then for any \(j \in [q]\) and any \(k \in \{q+1, \ldots, p\}\),
\begin{align*}
R(j)  & <  R(k)
\\ \iff \qquad  \frac{\beta_Z^2 \sigma_{\zeta j}^2}{1+ \sigma_{\zeta j}^2} + \sum_{k' = q+1}^p \beta_{k'}^2   +  \sigma_\epsilon^2   & < \beta_Z^2  + \sigma_\epsilon^2  +  \sum_{k' \in \{q+1, \ldots, p\} \setminus k}\beta_{k'}^2 
\\ \iff \qquad  \frac{ \beta_Z^2}{1+ \sigma_{\zeta j}^2}     & >  \beta_{k}^2  
\\\ \iff \qquad   \frac{\beta_Z^2}{\beta_{k}^2 }    & >  1+ \sigma_{\zeta j}^2.
\end{align*}
\end{proof}

\begin{proof}[Proof of Corollary \ref{cor.thm.1.risk}] If we can show that under the assumptions of Theorem \ref{thm.result.sel} \(\beta_Z^2    >1+ \sigma_\zeta^2 (n) \), then \(R(1) < R(3)\) is immediate by Proposition \ref{prop.gen.pred.risk} (and \(R(2) = R(1)\) comes from exchangeability). By Lemma \ref{cond.beta.z.lemma},
\[
\beta_Z^2 >    \frac{ 1 + \sigma_\zeta^2(n)}{ \left( 1 - 2 \delta(n) \sqrt{ 3 + \sigma_\epsilon^2 }\right)^2}.
\]
Since we have \(1 - 2 \delta(n) \sqrt{3 + \sigma_\epsilon^2} > 0\) from Lemma \ref{lemma.delt.ineq} and clearly \(1 - 2 \delta(n) \sqrt{ 3 + \sigma_\epsilon^2 } < 1\) (because \(\delta(n) > 0\)), the result is proven.
\end{proof}

\begin{proof}[Proof of Proposition \ref{proxies.risk.different.weight}]
\begin{enumerate}[(i)]

\item

 Let \(\boldsymbol{\sigma}_{\zeta}^2 := (\sigma_{\zeta 1}^2, \ldots, \sigma_{\zeta q}^2)^\top\), and denote \(\boldsymbol{w}^2 := (w_1^2, \ldots, w_q^2)^\top\). Consider the optimization problem to find the optimal weights for the expression of prediction risk from \eqref{proxies.risk.different.weight.pred.risk} (with no restriction on whether some weights equal 0):
\[
\begin{aligned}
\boldsymbol{w}^* : = \  & {\arg \min}
& &  \frac{  \left( \boldsymbol{w}^2 \right)^\top \boldsymbol{\sigma}_{\zeta}^2 }{1+   \left( \boldsymbol{w}^2 \right)^\top \boldsymbol{\sigma}_{\zeta}^2 } \\
& \text{subject to}
& & \boldsymbol{1}^\top \boldsymbol{w} - 1 = 0 , \\
& & & w_j \geq 0 &  \forall j \in [q].
\end{aligned}
\]
This can be simplified because \(t \mapsto \frac{t}{1 +t}\) is monotonically increasing for \(t \geq 0\), so it is sufficient to minimize \(  \left( \boldsymbol{w}^2 \right)^\top \boldsymbol{\sigma}_{\zeta}^2\). For any feasible choice of \(\boldsymbol{w}\), let
\[
K^2 := \left( \boldsymbol{w}^2 \right)^\top \boldsymbol{\sigma}_{\zeta}^2  = \frac{w_1^2}{1/\sigma_{\zeta 1}^2} + \ldots +  \frac{w_q^2}{1/\sigma_{\zeta q}^2} .
\]
Considered as a function of \(\boldsymbol{w}\), this is the equation for an ellipsoid with semi-axes of length \(K/\sigma_{\zeta 1}, \ldots, K/\sigma_{\zeta q}\) centered at the origin. Our optimization problem is to find the smallest \(K^2\) such that this ellipsoid intersects with the hyperplane described by \( \boldsymbol{1}^\top \boldsymbol{w} = 1\) in the orthant with \(w_j \geq 0\) for all \(j\). We can find this optimal \(\boldsymbol{w}^*\) using the method of Lagrange multipliers. The Lagrangian is
\[
\mathcal{L}(\boldsymbol{w}, \lambda, \boldsymbol{\mu}) :=   \left( \boldsymbol{w}^2 \right)^\top \boldsymbol{\sigma}_{\zeta}^2   - \lambda \boldsymbol{1}^\top \boldsymbol{w}.
% - \boldsymbol{\mu}^\top \boldsymbol{w} .
\]
%Now consider the first-order KKT conditions
We have
\[
\nabla_{w_j} \mathcal{L}(\boldsymbol{w}^*, \lambda,  \boldsymbol{\mu})  =  2 w_j^* \sigma_{\zeta j}^2 - \lambda = 0 , \qquad j \in [q],
\]
which yields \(w_j^* = \lambda/( 2  \sigma_{\zeta j}^2)\). Finally, the constraint \(\sum_j w_j^* = 1\) leads to
\begin{align*}
\lambda \sum_j \frac{1}{2 \sigma_{\zeta j}^2} & = 1
\\ \iff \qquad \lambda & = 2 \left( \sum_j \frac{1}{ \sigma_{\zeta j}^2}  \right)^{-1}
\\ \implies \qquad w_j^* & =  \left. \frac{1}{ \sigma_{\zeta j}^2} \middle/ \sum_{j'} \frac{1}{ \sigma_{\zeta j'}^2} \right. , \qquad j \in [q].
\end{align*}

%\[
%\vdots
%\]

\item

By Lemma \ref{lem.opt.weighting.all.positive}, the prediction risk of these weights is
\begin{align*}
&  \frac{n - 1 }{n-2} \left(  \frac{\beta_Z^2   \sum_{j=1}^q   \left( \frac{1}{ \sigma_{\zeta j}^2} \middle/ \sum_{j'} \frac{1}{ \sigma_{\zeta j'}^2} \right)^2 \sigma_{\zeta j}^2 }{1+ \sum_{j=1}^q \left( \frac{1}{ \sigma_{\zeta j}^2} \middle/ \sum_{j'} \frac{1}{ \sigma_{\zeta j'}^2} \right)^2 \sigma_{\zeta j}^2} + \sum_{j=q+1}^p \beta_j^2 +  \sigma_\epsilon^2  \right)
\\ = ~  &  \frac{n - 1 }{n-2} \left(  \frac{\beta_Z^2   }{1+  \sum_{j=1}^q   \frac{1}{ \sigma_{\zeta j}^2} }  + \sum_{j=q+1}^p \beta_j^2 +  \sigma_\epsilon^2   \right).
\end{align*}

\item 
%For an arbitrary \(j \in \{q+ 1, \ldots, p\}\), we have
%\[
%\boldsymbol{y} =   \beta_Z \boldsymbol{Z} + \beta_{j} \boldsymbol{X}_{\cdot , j} + \tilde{\boldsymbol{\epsilon}},
%\]
%where \(\tilde{\boldsymbol{\epsilon}} = \boldsymbol{\epsilon} +  \sum_{k \in \{q+1, \ldots, p\} \setminus j} \beta_k \boldsymbol{X}_{\cdot k} \sim \mathcal{N}_n\left(\boldsymbol{0}, \left( \sigma_\epsilon^2 +  \sum_{k \in \{q+1, \ldots, p\} \setminus j}\beta_k^2  \right) \boldsymbol{I}_n \right)   \). Then 
By Lemma \ref{pred.risk.lemma.x1}(ii), the prediction risk for an arbitrary \(k \in \{q+ 1, \ldots, p\}\) is
\[
R(k) = \frac{n-1}{n-2}  \left( \beta_Z^2  + \sigma_\epsilon^2  +  \sum_{j \in \{q+1, \ldots, p\} \setminus k}\beta_{j}^2   \right) .
\]
Then for any \(k \in \{q+ 1, \ldots, p\}\), using the result from part (ii) we have
\begin{align*}
R([q]; \boldsymbol{w}^*)  & <  R(k)
\\ \iff \qquad  \frac{\beta_Z^2   }{1+  \sum_{j=1}^q   \frac{1}{ \sigma_{\zeta j}^2} }  + \sum_{j=q+1}^p \beta_j^2
& <  \beta_Z^2    +  \sum_{j \in \{q+1, \ldots, p\} \setminus k}\beta_j^2  
\\ \iff \qquad    \beta_k^2
& <  \beta_Z^2     \left(1  -  \frac{1  }{1+  \sum_{j=1}^q   \frac{1}{ \sigma_{\zeta j}^2} } \right)
\\ \iff \qquad   \frac{ \beta_Z^2 }{ \beta_k^2}
& >    \frac{1+  \sum_{j=1}^q   \frac{1}{ \sigma_{\zeta j}^2} }{ \sum_{j=1}^q   \frac{1}{ \sigma_{\zeta j}^2}    } .
\end{align*}

\end{enumerate}
\end{proof}

\subsection{Statement of Theorem \ref{ss.thm.2} and Outline of Proofs of Theorems \ref{ss.thm.1} and \ref{ss.thm.2}}\label{proof.ss.thm.1}

We begin this section by stating our theorem generalizing Theorem 2, Equation 7, and Equation 8 of \citet{shah_samworth_2012} to our setting with clusters of features.

\begin{theorem}\label{ss.thm.2}

Recall the setup of Theorem \ref{ss.thm.1}.

\begin{enumerate}[(i)]

\item (Generalization of Theorem 2 of \citealt{shah_samworth_2012}.) Define \(\hat{S}^{\Lambda; \mathcal{C}} (A)\) for any set \(A \subseteq [n]\) as \(\hat{S}_{|A|}^{\Lambda; \mathcal{C}}  \) applied to the observations in \(A\). For any \(C \in \mathcal{C}\), define the simultaneous selection proportion
%\begin{equation}\label{theta.tilde.exp}
\[
\tilde{\Theta}_B(C) := \frac{1}{B} \sum_{b=1}^B  \mathbbm{1} \left\{  C \in \hat{S}^{\Lambda ; \mathcal{C}} \left( A_b \right) \right\}   \mathbbm{1} \left\{  C \in \hat{S}^{\Lambda ; \mathcal{C}} \left( \overline{A}_b \right) \right\}  .
\]
%\end{equation}
Suppose \(\tilde{\Theta}_B(C_k) \) has a unimodal distribution for each \(C_k \in L_\theta\). Then for any \(\tau \in \{1/2 + 1/B, 1/2 + 3/(2B), \ldots, 1\} \) there exists a function \(C( \cdot, \cdot)\) such that
%\[
%\tau \in \{1/2 + 1/B, 1/2 + 3/(2B), \ldots, 1\} \cap \left( \min\left\{\frac{3}{4} \theta^2 + \frac{1}{4B} + \frac{1}{2},  \theta^2 + \frac{1}{2} \right\}, 1 \right] ,
%\]
\[
\E \left| \hat{S}_{n, \tau}^{\text{CSS}; \Lambda, \mathcal{C}} \cap L_\theta  \right|    \leq \theta  \cdot C(\tau, B) \cdot  \E \left|  \hat{S}_{\lfloor n/2 \rfloor}^{\Lambda; \mathcal{C}} \cap L_\theta \right| ,
\]
where, when \(\theta \leq 1/\sqrt{3}\),
%\begin{equation}\label{c.tau.b.def}
\[
C(\tau, B) = \begin{cases}
 \frac{1}{2[2 \tau - 1 - 1/(2B)]}, &   \tau  \in \left( \min\left\{\frac{3}{4} \theta^2 + \frac{1}{4B} + \frac{1}{2},  \theta^2 + \frac{1}{2} \right\}, \frac{3}{4} \right],
 \\ \frac{4 \left[1 -  \tau  + 1/(2B) \right]}{1 + 1/B} ,& \tau  \in \left(\frac{3}{4}, 1 \right].
\end{cases}
\]
%\end{equation}

\item (Generalization of Equation 7 of \citealt{shah_samworth_2012}.) Recall the definition of \(r\)-concavity from Definitions 3 and 4 \citet{shah_samworth_2012}. Suppose \(\tilde{\Theta}_B(C_k) \) has an \(r\)-concave distribution for each \(C_k \in L_\theta\). If \(\tau \in \{1/2 + 1/B, 1/2 + 3/(2B), \ldots, 1\} \), then there exists a function \(D\) such that
\[
\E \left| \hat{S}_{n, \tau}^{\text{CSS}; \Lambda, \mathcal{C}} \cap L_\theta  \right|    \leq D \left( \theta^2, 2\tau - 1, B, r \right)  \left| L_\theta \right|. 
% \theta \E \left|  \hat{S}_{\lfloor n/2 \rfloor}^\lambda \cap L_\theta \right| ,
\]

\item (Generalization of Equation 8 of \citealt{shah_samworth_2012}.) Suppose \(\tilde{\Theta}_B(C_k) \) has an \(r_1\)-concave distribution for each \(C_k \in L_\theta\), and likewise every \(\hat{\Theta}_B(C_k) \) is \(r_2\)-concave. If \(\tau \in \{1/2 + 1/B, 1/2 + 3/(2B), \ldots, 1\} \), then
\[
\E \left| \hat{S}_{n, \tau}^{\text{CSS}; \Lambda, \mathcal{C}} \cap L_\theta  \right|    \leq \min \left\{D \left( \theta^2, 2\tau - 1, B, r_1 \right), D \left( \theta, \tau, 2B, r_2 \right) \right\}  \left| L_\theta \right|
% \theta \E \left|  \hat{S}_{\lfloor n/2 \rfloor}^\lambda \cap L_\theta \right| ,
\]
for all \(\tau \in (\theta, 1]\), where \(D\) is defined in the same way as in part (ii) with the convention \(\tilde{D}(\cdot, t, \cdot, \cdot) = 1\) for \(t \leq 0\).

\end{enumerate}

\end{theorem}

Theorems \ref{ss.thm.1} and \ref{ss.thm.2} follow from identical proofs to the corresponding results in \citet{shah_samworth_2012} after some minor swaps, so we omit these proofs rather than duplicating them. For instance, to arrive at a proof for our Theorem \ref{ss.thm.1}, take the proofs for \citeauthor{shah_samworth_2012}'s Lemma 1 and Theorem 1 and replace everywhere \(k\) with \(C_k\) (for some \(C_k \in \mathcal{C}\)), replace \(p\) with \(K\), replace \(p_{k, \lfloor n/2 \rfloor}\) with \(p_{C_k, \lfloor n/2 \rfloor, \Lambda}\), replace \(\hat{\Pi}_B( k)\) with our \(\hat{\Theta}_B(C_k)\), replace \(\tilde{\Pi}_B( k)\) with our \(\tilde{\Theta}_B(C_k)\), replace \(\hat{S}_{\lfloor n/2 \rfloor}\) with \(\hat{S}_{\lfloor n/2 \rfloor}^{\Lambda; \mathcal{C}}\), and replace \(\hat{S}_{n, \tau}^{\text{CPSS}}\) with \(\hat{S}_{n, \tau}^{\text{CSS}; \Lambda, \mathcal{C}}\). The results in Theorem \ref{ss.thm.2} follow in the same way.

\begin{remark}

 \citeauthor{shah_samworth_2012} demonstrate empirically that the assumptions of Theorem \ref{ss.thm.2}(iii) with \(r_1 = -1/2\) and \(r_2 = -1/4\) are reasonable when \((C_1, \ldots, C_p ) = (\{1\}, \ldots, \{p\})\), but we do not investigate the reasonableness of this assumption in our more general setting.
 
 \end{remark}

\section{Proofs of Technical Lemmas}\label{tech.lemmas}

\begin{proof}[Proof of Lemma \ref{lemma.calc.v.exp}] Using that \(X_{11}\) and \(y_1\) are jointly Gaussian, we have that \(\E \left[  X_{11}^2y_1^2\right]  = \Sigma_{11} \Sigma_{yy} + 2 \Sigma_{1y}^2\).
Then
\begin{align}
\E \left \lVert \boldsymbol{V}_1 \right \rVert^2  
 & =\E \left \lVert \left(X_{11}^2 -\Sigma_{11}, X_{12}^2 -\Sigma_{11}, y_1^2 - \Sigma_{yy}, X_{11}y_1 - \Sigma_{1y}, X_{12}y_1 - \Sigma_{1y}\right) \right \rVert^2  \nonumber
\\ & = \E \left[ \left(X_{11}^2 - \Sigma_{11} \right)^2 + \left( X_{12}^2 - \Sigma_{11} \right)^2 + \left( y_1^2 - \Sigma_{yy} \right)^2 + \left( X_{11}y_1 - \Sigma_{1y}\right)^2 + \left( X_{12}y_1 - \Sigma_{1y}\right)^2 \right]  \nonumber
\\ & = \Var \left(X_{11}^2 \right) +\Var \left( X_{12}^2 \right) + \Var \left( y_1^2  \right) \nonumber
\\ & + \E \left[  X_{11}^2y_1^2 - \Sigma_{1y} X_{11}y_1 + \Sigma_{1y}^2 +  X_{12}^2y_1^2 - \Sigma_{1y}X_{12}y_1 +  \Sigma_{1y}^2 \right]  \nonumber
\\ & = \Sigma_{11}^2 \Var \left(  \frac{ X_{11}^2}{\Sigma_{11}}  \right) + \Sigma_{11}^2 \Var \left(\frac{X_{12}^2}{\Sigma_{11}}  \right) + \Sigma_{yy}^2\Var \left(  \frac{y_1^2}{\Sigma_{yy}} \right) + 2\E \left[  X_{11}^2y_1^2\right] - 2\Sigma_{1y} \E \left[X_{11}y_1 \right] + 2\Sigma_{1y}^2\nonumber
\\ & = 2 \Sigma_{11}^2  + 2 \Sigma_{11}^2  + 2\Sigma_{yy}^2 + 2 \left(\Sigma_{11} \Sigma_{yy} + 2 \Sigma_{1y}^2\right)  - 2\Sigma_{1y}^2   + 2\Sigma_{1y}^2  \nonumber
\\ & = 4 \Sigma_{11}^2   + 2\Sigma_{yy}^2 + 2 \Sigma_{11} \Sigma_{yy} + 4 \Sigma_{1y}^2, \nonumber  
\end{align}
 where we used that \(X_{11}^2/ \Sigma_{11}\), \(X_{12}^2/\Sigma_{11}\), and \(y_1^2/\Sigma_{yy}\) are \(\chi_1^2\) random variables and the exchangeability of \(X_{11}\) and \(X_{12}\). Finally, using that \(\Sigma_{yy} > \max \left\{ \Sigma_{11}, \Sigma_{1y} \right\}\), we have \(\E \left \lVert \boldsymbol{V}_1 \right \rVert^2   < 12\Sigma_{yy}^2\), which is \eqref{v2}. Further,
 \begin{align}
 \E \left\lVert \boldsymbol{V}_1 \right \rVert^3 
 = &  \E \left \lVert \left(X_{11}^2 -1, X_{12}^2 -1, y_1^2 - 1, X_{11}y_1 - \Sigma_{1y}, X_{12}y_1 - \Sigma_{1y}\right) \right \rVert^3  \nonumber
 \\ = &   \E \left[ \left|X_{11}^2 -1\right|^3 + \left| X_{12}^2 -1\right|^3 + \left| y_1^2 - 1 \right|^3 + \left| X_{11}y_1 - \Sigma_{1y}\right|^3 + \left| X_{12}y_1 - \Sigma_{1y}\right|^3 \right]  \nonumber
\\ \leq &  \E \left[ \left(X_{11}^2  + 1\right)^3 + \left( X_{12}^2 + 1 \right)^3 + \left( y_1^2  + 1 \right)^3 + \left( \left|X_{11}y_1 \right| + \Sigma_{1y} \right) ^3 +\left( \left| X_{12}y_1 \right| + \Sigma_{1y}\right)^3 \right]  \nonumber
\\ = & 2 \E \left[ \left(X_{11}^2  + 1\right)^3 \right] + \E \left( y_1^2  + 1 \right)^3 + 2 \E \left[ \left( \left|X_{11}y_1 \right| + \Sigma_{1y} \right) ^3 \right]  \nonumber
\\ = & 2 \E \left[ X_{11}^6 + 3 X_{11}^4 + 3 X_{11}^2  + 1 \right] + \E \left( y_1^6 + 3y_1^4 +3y_1^2  + 1 \right)   \nonumber
\\ & + 2 \E \left[  \left|X_{11}y_1 \right|^3 + 3 \Sigma_{1y} \left|X_{11}y_1 \right|^2 + 3 \Sigma_{1y}^2 \left|X_{11}y_1 \right| + \Sigma_{1y}^3 \right]  \nonumber
\\ \leq & 2 \E \left[ X_{11}^6 + 3 X_{11}^4 + 3 X_{11}^2  + 1 \right] + \E \left( y_1^6 + 3y_1^4 +3y_1^2  + 1 \right)   \nonumber
\\ & + 2 \sqrt{\E  \left| X_{11}\right|^6 \E \left|y_1 \right|^6  } + 6 \Sigma_{1y} \sqrt{\E  \left| X_{11}\right|^4 \E \left|y_1 \right|^4  } + 6 \Sigma_{1y}^2 \sqrt{\E  \left| X_{11}\right|^2 \E \left|y_1 \right|^2  } + 2\Sigma_{1y}^3  \nonumber
\\  = &  \E \left[ 2X_{11}^6 + 6 X_{11}^4 +6 X_{11}^2  + y_1^6 + 3y_1^4 +3y_1^2  \right] + 3   \nonumber
\\ & + 2 \sqrt{\E   X_{11}^6 \E y_1^6  } + 6 \Sigma_{1y} \sqrt{\E   X_{11}^4 \E y_1^4  } + 6 \Sigma_{1y}^2 \sqrt{\E  X_{11}^2 \E y_1^2  } + 2\Sigma_{1y}^3  \nonumber
\\  = &   2 \cdot 15 \cdot \Sigma_{11}^3 + 6 \cdot 3 \cdot \Sigma_{11}^2 +6 \cdot \Sigma_{11}  +  15 \cdot \Sigma_{yy}^3 + 3 \cdot 3 \cdot \Sigma_{yy}^2   \nonumber
 +3\Sigma_{yy}  + 3  
 \\ & + 2 \sqrt{15 \cdot \Sigma_{11}^3 \cdot 15 \cdot \Sigma_{yy}^3  } 
+ 6 \Sigma_{1y} \sqrt{3 \cdot \Sigma_{11}^2  \cdot 3 \cdot \Sigma_{yy}^2  }  + 6 \Sigma_{1y}^2 \sqrt{\Sigma_{11} \Var (y_1)  } + 2\Sigma_{1y}^3  \nonumber
\\  = &   30\Sigma_{11}^3 + 18\Sigma_{11}^2  +6 \Sigma_{11} +  15\Sigma_{yy}^3 +9\Sigma_{yy}^2 +3\Sigma_{yy}   + 3  + 30 \Sigma_{11}^{3/2} \Sigma_{yy}^{3/2} \nonumber
\\ & + 18\Sigma_{1y}  \Sigma_{11}\Sigma_{yy}  \nonumber
 + 6 \Sigma_{1y}^2 \sqrt{ \Sigma_{11} \Sigma_{yy} }+ 2\Sigma_{1y}^3  ,\nonumber 
 \end{align}
where we used the triangle inequality, the Cauchy-Schwarz inequality and the exchangeability of \(X_{11}\) and \(X_{12}\). Again, since \(\Sigma_{yy} > \max \left\{ \Sigma_{11}, \Sigma_{1y} \right\}\), this yields \(  \E \left\lVert \boldsymbol{V}_1 \right \rVert^3  < 140 \Sigma_{yy}^3\) which is \eqref{v3}. 
\end{proof}
\begin{proof}[Proof of Lemma \ref{calc.lemma.delt.meth.sigma.tilde}] Note that
\begin{align*}
L ( \boldsymbol{V}_1)  & = \frac{1}{2}  \frac{\Sigma_{1y}}{\sqrt{\Sigma_{yy} \Sigma_{11}^3}}( X_{12}^2  - X_{11}^2)  +  \frac{1}{\sqrt{ \Sigma_{11} \Sigma_{yy}}} \left( X_{11}y_1 - X_{12}y_1 \right)
\\ & = \left(X_{12} - X_{11} \right) \left(\frac{1}{2}  \frac{\Sigma_{1y}}{\sqrt{\Sigma_{yy} \Sigma_{11}^3}}( X_{12}  + X_{11}) -  \frac{1}{\sqrt{ \Sigma_{11} \Sigma_{yy}}} y_1 \right),
\end{align*}
so
\begin{align}
\tilde{\sigma}^2  := &  \E \left| L(\boldsymbol{V}_1) \right|^2 \nonumber
=   \E \left[  \left(X_{12} - X_{11} \right)^2 \left(\frac{1}{2}  \frac{\Sigma_{1y}}{\sqrt{\Sigma_{yy} \Sigma_{11}^3}}( X_{12}  + X_{11}) -  \frac{1}{\sqrt{ \Sigma_{11} \Sigma_{yy}}} y_1 \right)^2 \right].  \nonumber
%\label{lemma.delt.meth.sigma.tilde}
\end{align}

We will use the fact that \(X_{12} - X_{11}\) and \(\frac{1}{2}  \frac{\Sigma_{1y}}{\sqrt{\Sigma_{yy} \Sigma_{11}^3}}( X_{12}  + X_{11}) -  \frac{1}{\sqrt{ \Sigma_{11} \Sigma_{yy}}} y_1\) are mean zero jointly Gaussian random variables. Note that they are independent:
\begin{align*}
& \Cov \left( X_{12} - X_{11},  \frac{1}{2}  \frac{\Sigma_{1y}}{\sqrt{\Sigma_{yy} \Sigma_{11}^3}}( X_{12}  + X_{11}) -  \frac{1}{\sqrt{ \Sigma_{11} \Sigma_{yy}}} y_1 \right)  
\\ = ~ &   \Cov \left(X_{12},  \frac{1}{2}  \frac{\Sigma_{1y}}{\sqrt{\Sigma_{yy} \Sigma_{11}^3}}( X_{12}  + X_{11})   \right) - \Cov \left( X_{12},   \frac{1}{\sqrt{ \Sigma_{11} \Sigma_{yy}}}y_1 \right)
\\ & - \Cov\left(X_{11}, \frac{1}{2}  \frac{\Sigma_{1y}}{\sqrt{\Sigma_{yy} \Sigma_{11}^3}}( X_{12}  + X_{11})   \right) + \Cov \left( X_{11},  \frac{1}{\sqrt{ \Sigma_{11} \Sigma_{yy}}} y_1 \right)
\\ = ~&   \frac{1}{2}  \frac{\Sigma_{1y}}{\sqrt{\Sigma_{yy} \Sigma_{11}^3}}\left[ \Cov \left(X_{12},  X_{12} + X_{11}   \right) 
 - \Cov\left(X_{11}, X_{12} + X_{11}  \right) \right] - \frac{1}{\sqrt{ \Sigma_{11} \Sigma_{yy}}} \cdot \Sigma_{1y}
 \\ &  + \frac{1}{\sqrt{ \Sigma_{11} \Sigma_{yy}}} \cdot \Sigma_{1y}
\\ = ~& 0
\end{align*}
(where we used the exchangeability of \(X_{11}\) and \(X_{12}\)). Therefore 
\begin{align}
\tilde{\sigma}^2  & =\E \left[  \left(X_{12} - X_{11} \right)^2 \left(\frac{1}{2}  \frac{\Sigma_{1y}}{\sqrt{\Sigma_{yy} \Sigma_{11}^3}}( X_{12}  + X_{11}) -  \frac{1}{\sqrt{ \Sigma_{11} \Sigma_{yy}}} y_1 \right)^2 \right] \nonumber
\\ & =\E \left[  \left(X_{12} - X_{11} \right)^2 \right] \E \left[ \left(\frac{1}{2}  \frac{\Sigma_{1y}}{\sqrt{\Sigma_{yy} \Sigma_{11}^3}}( X_{12}  + X_{11}) -  \frac{1}{\sqrt{ \Sigma_{11} \Sigma_{yy}}} y_1 \right)^2 \right]\nonumber
 \\ & =  \Var \left[ X_{12} - X_{11}\right] \Var \left[ \frac{1}{2}  \frac{\Sigma_{1y}}{\sqrt{\Sigma_{yy} \Sigma_{11}^3}}( X_{12}  + X_{11}) -  \frac{1}{\sqrt{ \Sigma_{11} \Sigma_{yy}}} y_1   \right].\label{lemma.def.tilde.sigma.pf.var}
 \end{align}
 We have
\begin{align}
\Var(X_{12} - X_{11}) = &  \Sigma_{11} + \Sigma_{11} - 2 \Sigma_{12} = 2 (\Sigma_{11} - \Sigma_{12}), \nonumber
\\ \Var \left( \frac{1}{2}  \frac{\Sigma_{1y}}{\sqrt{\Sigma_{yy} \Sigma_{11}^3}}( X_{12}  + X_{11}) -  \frac{1}{\sqrt{ \Sigma_{11} \Sigma_{yy}}} y_1 \right)  = &  \frac{1}{4}  \frac{\Sigma_{1y}^2}{\Sigma_{yy}\Sigma_{11}^3} \Var (X_{12} + X_{11}) + \frac{1}{ \Sigma_{11} \Sigma_{yy}} \Var(y_1) \nonumber
\\ &  - 2 \cdot \frac{1}{2}  \frac{\Sigma_{1y}}{\Sigma_{yy}\Sigma_{11}^2} \Cov(X_{12} + X_{11}, y_1) \nonumber
\\  = &   \frac{1}{4}  \frac{\Sigma_{1y}^2}{\Sigma_{yy}\Sigma_{11}^3} \left(\Sigma_{11} + \Sigma_{11} + 2\Sigma_{12} \right) + \frac{1}{\Sigma_{11}} \nonumber
\\ & -   \frac{\Sigma_{1y}}{\Sigma_{yy}\Sigma_{11}^2}  \left(\Sigma_{1y} + \Sigma_{1y} \right) \nonumber
\\  = &   -\frac{3}{2}  \frac{\Sigma_{1y}^2}{\Sigma_{yy}\Sigma_{11}^2}  +  \frac{1}{2}  \frac{\Sigma_{1y}^2\Sigma_{12}}{\Sigma_{yy}\Sigma_{11}^3} + \frac{1}{\Sigma_{11}} ,\nonumber
\end{align}
so
\begin{align*}
\tilde{\sigma}^2   & =     2 (\Sigma_{11} - \Sigma_{12}) \left(  -\frac{3}{2}  \frac{\Sigma_{1y}^2}{\Sigma_{yy}\Sigma_{11}^2}  +  \frac{1}{2}  \frac{\Sigma_{1y}^2\Sigma_{12}}{\Sigma_{yy}\Sigma_{11}^3} + \frac{1}{\Sigma_{11}}\right) 
\\ & =     2 \left( \frac{\Sigma_{11}}{\Sigma_{11}} - \frac{\Sigma_{12}}{\Sigma_{11}}\right) \left(  -\frac{3}{2}  \left(\frac{\Sigma_{1y}}{\sqrt{\Sigma_{yy}\Sigma_{11}}}\right)^2 +  \frac{1}{2}  \left(\frac{\Sigma_{1y}}{\sqrt{\Sigma_{yy} \Sigma_{11}}} \right)^2 \frac{\Sigma_{12}}{\Sigma_{11}} + 1 \right)
\\ & =     \left( 1 - \rho_{12} \right) \left(  -3 \rho_{1y}^2 + \rho_{1y}^2 \rho_{12} + 2 \right) , 
\end{align*}
which yields the expression in \eqref{lemma.def.tilde.sigma}. Finally, to see that \(\tilde{\sigma}^2 \leq 2\), note that 
\begin{align}
\tilde{\sigma}^2 & =     \left( 1 - \rho_{12} \right) \left(  -3 \rho_{1y}^2 + \rho_{1y}^2 \rho_{12} + 2 \right) \nonumber
\\ &  =   2 \left( 1 - \rho_{12} \right) \left( -\rho_{1y}^2 \left[  \frac{3 - \rho_{12}}{2} \right] + 1 \right) \nonumber \nonumber
\\ & \leq 2 \cdot 1 \cdot \left( 0 + 1\right) = 2. \label{tilde.sigma.ub}
\end{align}
This verifies \eqref{lemma.def.tilde.sigma}. Next we will consider \( \E \left| L(\boldsymbol{V}_1) \right|^3\). We will begin by showing that \eqref{lemma.bound.ratio} holds. First, observe that (again using the independence of \(X_{12} - X_{11}\) and \( \frac{1}{2}  \frac{\Sigma_{1y}}{\sqrt{\Sigma_{yy} \Sigma_{11}^3}}( X_{12}  + X_{11}) -  \frac{1}{\sqrt{ \Sigma_{11} \Sigma_{yy}}} y_1\))
\begin{align}
\E \left| L ( \boldsymbol{V}_1)  \right|^3 & =\E \left| (X_{12} - X_{11}) \left( \frac{1}{2}  \frac{\Sigma_{1y}}{\sqrt{\Sigma_{yy} \Sigma_{11}^3}}( X_{12}  + X_{11}) -  \frac{1}{\sqrt{ \Sigma_{11} \Sigma_{yy}}} y_1 \right)\right|^3 \nonumber
\\ & =\E  \left| X_{12} - X_{11} \right|^3 \E \left| \frac{1}{2}  \frac{\Sigma_{1y}}{\sqrt{\Sigma_{yy} \Sigma_{11}^3}}( X_{12}  + X_{11}) -  \frac{1}{\sqrt{ \Sigma_{11} \Sigma_{yy}}} y_1 \right|^3  \nonumber
%\\ \text{(by Cauchy-Schwarz)} \quad & \leq  \sqrt{ \E \left|  X_{12}  -  X_{11}\right|^6 \E \left| \frac{1}{2} \Sigma_{1y}  \left( X_{12}  +  X_{11}  \right)  - y_1 \right|^6} \nonumber
\\ & = \sqrt{\frac{8}{\pi}} \left[ \Var \left( X_{12} - X_{11}\right)\right]^{3/2} \cdot \sqrt{\frac{8}{\pi}}  \left[ \Var \left( \frac{1}{2}  \frac{\Sigma_{1y}}{\sqrt{\Sigma_{yy} \Sigma_{11}^3}}( X_{12}  + X_{11}) -  \frac{1}{\sqrt{ \Sigma_{11} \Sigma_{yy}}} y_1\right)\right]^{3/2} .\label{lemma.delt.meth.sigma.tilde.cubed.a.var}
\end{align}
Using \eqref{lemma.delt.meth.sigma.tilde.cubed.a.var} and \eqref{lemma.def.tilde.sigma.pf.var}, we have
\begin{align*}
\frac{ \left( \E \left| L(\boldsymbol{V}_1) \right|^3 \right)^{1/3}}{ \left( \E \left| L(\boldsymbol{V}_1) \right|^2 \right)^{1/2}} 
& \leq \frac{ \left( \frac{8}{\pi}  \left[ \Var \left( X_{12} - X_{11}\right)\right]^{3/2}   \left[ \Var \left( \frac{1}{2}  \frac{\Sigma_{1y}}{\sqrt{\Sigma_{yy} \Sigma_{11}^3}}( X_{12}  + X_{11}) -  \frac{1}{\sqrt{ \Sigma_{11} \Sigma_{yy}}} y_1\right)\right]^{3/2}  \right)^{1/3}}{ \left(\Var \left[ X_{12} - X_{11}\right] \cdot\Var \left[ \frac{1}{2}  \frac{\Sigma_{1y}}{\sqrt{\Sigma_{yy} \Sigma_{11}^3}}( X_{12}  + X_{11}) -  \frac{1}{\sqrt{ \Sigma_{11} \Sigma_{yy}}} y_1  \right] \right)^{1/2}} 
%\\ & = \left( \frac{8}{\pi} \right)^{1/3} \cdot \frac{   \sqrt{ \Var \left[  X_{12}  -  X_{11}\right]   \cdot \Var \left[ \frac{1}{2}  \frac{\Sigma_{1y}}{\sqrt{\Sigma_{yy} \Sigma_{11}^3}}( X_{12}  + X_{11}) -  \frac{1}{\sqrt{ \Sigma_{11} \Sigma_{yy}}} y_1 \right]}  }{ \sqrt{\Var \left[ X_{12} - X_{11}\right] \cdot \Var \left[ \frac{1}{2}  \frac{\Sigma_{1y}}{\sqrt{\Sigma_{yy} \Sigma_{11}^3}}( X_{12}  + X_{11}) -  \frac{1}{\sqrt{ \Sigma_{11} \Sigma_{yy}}} y_1  \right] }} 
\\ & = \left( \frac{8}{\pi} \right)^{1/3}.
\end{align*}
This verifies \eqref{lemma.bound.ratio}. Since from \eqref{tilde.sigma.ub} we have \(\left( \E \left| L(\boldsymbol{V}_1) \right|^2 \right)^{1/2} = \tilde{\sigma} \leq \sqrt{2}\), it follows that 

\[
\left( \E \left| L(\boldsymbol{V}_1) \right|^3 \right)^{1/3} \leq \sqrt{2} \cdot \left( \frac{8}{\pi} \right)^{1/3} \quad \implies \quad \E \left| L(\boldsymbol{V}_1) \right|^3 \leq \frac{8}{\pi} \cdot 2^{3/2} = \frac{16 \sqrt{2}}{\pi},
\]

verifying that \( \E \left| L(\boldsymbol{V}_1) \right|^3 \) is finite.

\end{proof}

\begin{proof}[Proof of Lemma \ref{lemma.smoothness.2.1}] For notational ease, let \(\tilde{u}_j := u_j + \Sigma_{11}\) for \(j \in \{1,2\}\), let \(\tilde{u}_3 := u_3 + \Sigma_{yy}\), and let \(\tilde{u}_j := u_j + \Sigma_{1y}\) for \(j \in \{4,5\}\). Then differentiating \eqref{lemma.gradient} yields
\begin{equation}\label{lemma.hessian}
 \nabla^2 g(\boldsymbol{u})
 =  \begin{bmatrix}
% 1st row below
\frac{3\tilde{u}_4}{4\sqrt{\tilde{u}_3\tilde{u}_1^{5}}}    
& 0 
&  \frac{\tilde{u}_4}{4\sqrt{\tilde{u}_1^{3} \tilde{u}_3^{3}}}  
& \frac{-1}{2\sqrt{\tilde{u}_3\tilde{u}_1^{3} }}  
& 0 \\
% 2nd row below
0 
& - \frac{3\tilde{u}_5}{4\sqrt{\tilde{u}_3\tilde{u}_2^{5}}}  
& -\frac{\tilde{u}_5}{4\sqrt{\tilde{u}_2^3\tilde{u}_3^3}}  
& 0 
& \frac{1}{2\sqrt{\tilde{u}_3\tilde{u}_2^{3} }}  \\
% 3rd row below
 \frac{\tilde{u}_4}{4\sqrt{\tilde{u}_1^{3} \tilde{u}_3^{3}}}   
 &  -\frac{\tilde{u}_5}{4\sqrt{\tilde{u}_2^3\tilde{u}_3^3}}  
 &- \frac{3}{4 \sqrt{\tilde{u}_3^5}} \left(  \frac{\tilde{u}_5}{\sqrt{\tilde{u}_2  }} -  \frac{\tilde{u}_4}{\sqrt{\tilde{u}_1 }}    \right) 
 &    -\frac{1}{2\sqrt{\tilde{u}_1\tilde{u}_3^3}}  
 & \frac{1}{2\sqrt{\tilde{u}_2\tilde{u}_3^3 }}    \\
 % 4th row below
\frac{-1}{2\sqrt{\tilde{u}_3\tilde{u}_1^{3} }}  
& 0  
&  -\frac{1}{2\sqrt{\tilde{u}_1\tilde{u}_3^3}} 
& 0 
& 0 \\
% 5th row below
0 
&  \frac{1}{2\sqrt{\tilde{u}_3\tilde{u}_2^{3} }} 
&   \frac{1}{2\sqrt{\tilde{u}_2\tilde{u}_3^3 }}   
& 0 
& 0
\end{bmatrix}   .
\end{equation}
Let \(\epsilon := \Sigma_{11}/2\). Then \eqref{2.9.smoothness} is satisfied if for all \(\boldsymbol{u} \in \mathcal{X}_{\epsilon}\) where \\ \( \mathcal{X}_{\epsilon} := \left\{ \boldsymbol{u} \in  (-\Sigma_{11}, \infty)^3 \times \mathbb{R}^2  : \lVert \boldsymbol{u} \rVert_2 \leq \epsilon \right\}\), \(\lVert \nabla^2 g(\boldsymbol{u}) \rVert_{\text{op}} \leq M_{\epsilon}\) for some \(M_{\epsilon} \in (0, \infty)\). Since for any real-valued square matrix \(\boldsymbol{A} \in \mathbb{R}^{k \times k}\) it holds that \( \lVert \boldsymbol{A} \rVert_{\text{op}} \leq \lVert \boldsymbol{A} \rVert_F \leq \sqrt{\tilde{\ell}} \max_{i,j} \left| \boldsymbol{A}_{ij} \right| \) (where \(\tilde{\ell}\) is the number of nonzero entries in \(\boldsymbol{A}\)), it suffices to bound \(\max_{u \in  \mathcal{X}_{\epsilon}} \left\{ \max_{i,j} \left| \begin{pmatrix} \nabla^2 g(\boldsymbol{u}) \end{pmatrix}_{ij} \right| \right\}\).
Note that in \(\mathcal{X}_\epsilon\) we have
\begin{align*}
\tilde{u}_j = u_j + \Sigma_{11}& \in \left(\Sigma_{11}/2, 3 \Sigma_{11}/2 \right), \qquad j \in \{1, 2\},
\\ \tilde{u}_3 = u_3 + \Sigma_{yy} & \in \left(\Sigma_{yy} - \Sigma_{11}/2, \Sigma_{yy} + \Sigma_{11}/2 \right) ,   \qquad \text{and}
\\ \tilde{u}_j = u_j + \Sigma_{1y}& \in  \left(\Sigma_{1y} - \Sigma_{11}/2, \Sigma_{1y} + \Sigma_{11}/2 \right), \qquad j \in \{4, 5\}.
\end{align*}
Since (using \( \Sigma_{11} < \Sigma_{1y} < \Sigma_{yy} \)) 
\begin{align*}
1/2 &  \leq \min \left\{  \Sigma_{11}/2,   \Sigma_{yy}/2 , \Sigma_{1y}/2 \right\} 
\\ & =\min \left\{   \Sigma_{11}/2,   \Sigma_{yy} -  \Sigma_{yy}/2 , \Sigma_{1y} -  \Sigma_{1y}/2 \right\}
%\\ \vdots
\\ & \leq \min \left\{ \Sigma_{11}/2,   \Sigma_{yy} - \Sigma_{11}/2 , \Sigma_{1y} - \Sigma_{11}/2 \right\} 
\end{align*}
 and similarly \(\Sigma_{yy} + \Sigma_{11}/2 > \max \left\{ \Sigma_{1y} + \Sigma_{11}/2 , 3\Sigma_{11}/2 \right\}\), we can bound all of the \(\tilde{u}_j\) in \(\mathcal{X}_\epsilon\) using
\[
1/2 \leq \tilde{u}_j \leq \Sigma_{yy} + \Sigma_{11}/2 \qquad \forall j \in [5].
\]
Therefore in \( \mathcal{X}_{\epsilon} \) the absolute values of the \((1,1)\) and \((2,2)\) terms of \eqref{lemma.hessian} can be upper-bounded by
\[
\frac{3}{4}\frac{\Sigma_{yy} + \Sigma_{11}/2}{\sqrt{\left( 1/2 \right)^6}} =   \frac{3 \cdot 8}{4} \left( \Sigma_{yy} + \Sigma_{11}/2\right) = 6\Sigma_{yy} + 3\Sigma_{11}<  9\Sigma_{yy},
\]
the absolute values of the \((1,3)\) and \((2,3)\) terms of \eqref{lemma.hessian} (and their symmetric counterparts, the \((3,1)\) and \((3,2)\) terms) can be upper-bounded by
\[
\frac{1}{4}\frac{\Sigma_{yy} + \Sigma_{11}/2}{\sqrt{\left( 1/2 \right)^6}} =  \frac{8}{4} \left( \Sigma_{yy} + \Sigma_{11}/2 \right) = 2\Sigma_{yy} + \Sigma_{11} < 3\Sigma_{yy} <   9\Sigma_{yy},
\]
and the absolute values of the \((1, 4)\), \((2,5)\), \((3, 4)\), and \((3,5)\) terms of \eqref{lemma.hessian} and their symmetric counterparts can be upper-bounded by
\[
\frac{1}{2}\frac{1}{\sqrt{\left(1/2 \right)^4}} =  \frac{4}{2}< 9\Sigma_{yy}.
\]
Finally, the absolute value of the center \((3,3)\) term can be upper-bounded by
\begin{align*}
\left|- \frac{3}{4 \sqrt{\tilde{u}_3^5}} \left(  \frac{\tilde{u}_5}{\sqrt{\tilde{u}_2  }} -  \frac{\tilde{u}_4}{\sqrt{\tilde{u}_1 }}    \right) \right| &  \leq \frac{3}{4 \sqrt{\left( 1/2 \right)^5}} \left( \frac{\Sigma_{yy} + \Sigma_{11}/2}{\sqrt{ 1/2 }} - \frac{1/2}{\sqrt{\Sigma_{yy} + \Sigma_{11}/2  }} \right)
\\ &  \leq \frac{3 \left( \Sigma_{yy} + \Sigma_{11}/2 \right)}{4 \sqrt{\left( 1/2 \right)^6}}
\\ &  = \frac{3 \cdot 8 }{4 }\left( \Sigma_{yy} + \Sigma_{11}/2 \right)
\\ &  =6\Sigma_{yy} + 3\Sigma_{11} 
\\ & <  9\Sigma_{yy} .
\end{align*}
Therefore for the 15 non-zero terms in \(\nabla^2 g(\boldsymbol{u})\) we have \( \max_{i,j}  \left| \begin{pmatrix}\nabla^2 g(\boldsymbol{u}) \end{pmatrix}_{ij}  \right| < 9 \Sigma_{yy} \), so for all \(\boldsymbol{u} \in \mathcal{X}_{\epsilon}\),
\[
\lVert \nabla^2 g(\boldsymbol{u}) \rVert_{\text{op}}  \leq \lVert \nabla^2 g(\boldsymbol{u}) \rVert_F \leq \sqrt{15} \max_{i,j}  \left| \begin{pmatrix}\nabla^2 g(\boldsymbol{u}) \end{pmatrix}_{ij}  \right|  <  9  \sqrt{15}\Sigma_{yy} < 36 \Sigma_{yy}  = M_{\epsilon}. 
\]
\end{proof}
\begin{proof}[Proof of Lemma \ref{lemma.ub.mathcal.c}] Substituting \eqref{v2}, \eqref{v3}, and \(\epsilon = \Sigma_{11}/2\) into the expression for \(k_\epsilon\) in \eqref{mathcal.c.terms} yields
\begin{align*}
k_\epsilon & < \min\left\{ \frac{12 \Sigma_{yy}^2 }{\Sigma_{11}^2 n^{1/2}/4}, \frac{2 \left( 12 \Sigma_{yy}^2 \right)^{3/2} + 140 \Sigma_{yy}^3 /n^{1/2}}{\Sigma_{11}^3 n/8}\right\}
\\ & = \min\left\{  \frac{12 \Sigma_{yy}^2 }{\Sigma_{11}^2 n^{1/2}/4}, \frac{384\sqrt{3} \Sigma_{yy}^3 + 1120 \Sigma_{yy}^3 /n^{1/2}}{\Sigma_{11}^3 n}\right\}
\\ & \leq \frac{384\sqrt{3} \Sigma_{yy}^3 + 1120 \Sigma_{yy}^3 /n^{1/2}}{\Sigma_{11}^3 n}.
\end{align*}
Substituting this along with the other expressions from \eqref{mathcal.c.terms} into \eqref{mathcal.c.formula} (also using \eqref{v2}, \eqref{v3}, and \eqref{lemma.bound.ratio}) yields
\begin{align}
\mathcal{C} \leq ~  &  k_0 + k_1 \left( \sqrt[3]{\frac{8}{\pi}} \right)^3 + \left(k_{20} + k_{21}  \sqrt[3]{\frac{8}{\pi}} \right)\cdot 12 \Sigma_{yy}^2  + \left(k_{30} + k_{31}  \sqrt[3]{\frac{8}{\pi}} \right)\left( 140 \Sigma_{yy}^3 \right)^{2/3}  + k_\epsilon \nonumber
\\  < ~  &  0.13925 + \frac{8}{\pi} \cdot 2.33554  +  \frac{M_\epsilon}{2 \tilde{\sigma}}\left( 2\left(\frac{2}{\pi}\right)^{1/6} +\left[2 + \frac{2^{2/3}}{n^{1/6}} \right]  \left(\frac{8}{\pi}\right)^{1/3}  \right) \cdot 12 \Sigma_{yy}^2 \nonumber
\\ & +  \frac{M_\epsilon}{ 2\tilde{\sigma}}\left(  \frac{(8/\pi)^{1/6}}{n^{1/3}} + \frac{2}{n^{1/2}} \left(\frac{8}{\pi}\right)^{1/3} \right)\left( 140 \Sigma_{yy}^3 \right)^{2/3}   +  \frac{384\sqrt{3} \Sigma_{yy}^3 + 1120 \Sigma_{yy}^3 /n^{1/2}}{\Sigma_{11}^3 n}\nonumber
\\  < ~  & 6.087  +  \frac{M_\epsilon}{ \tilde{\sigma}} \Sigma_{yy}^2 \bigg[  12 \left(\frac{2}{\pi}\right)^{1/6} + 12\left[1 + \frac{1}{2^{1/3}n^{1/6}} \right]  \left(\frac{8}{\pi}\right)^{1/3}   \nonumber
\\ & +    \frac{\left( 140 \right)^{2/3}}{n^{1/3} \pi^{1/6} \sqrt{2}} + \frac{\left( 140 \right)^{2/3}}{n^{1/2}} \left(\frac{8}{\pi}\right)^{1/3} \bigg]  +  \frac{384\sqrt{3} + 1120  /n^{1/2}}{ n} \cdot  \frac{\Sigma_{yy}^3}{\Sigma_{11}^3}\nonumber
\\  < ~  & 6.087  +  \frac{M_\epsilon}{ \tilde{\sigma}} \Sigma_{yy}^2 \left( 27.517 + \frac{13.007}{n^{1/6}}  + \frac{15.754}{n^{1/3}} + \frac{36.819}{n^{1/2}}  \nonumber
  \right)  +  \frac{\Sigma_{yy}^3}{\Sigma_{11}^3} \left( \frac{384\sqrt{3}}{n} + \frac{1120}{ n^{3/2}} \right) . \nonumber
%%
%\\ \vdots \nonumber
%%
%\\  <  ~  & 6.087 +  2\frac{M_\epsilon}{ \tilde{\sigma}}  \Sigma_{yy}^2  \Bigg(12 \left(\frac{2}{\pi}\right)^{1/6} + 12  \left[1 + \frac{1}{2^{1/3} \cdot n^{1/6}} \right]  \left(\frac{8}{\pi}\right)^{1/3}  \nonumber
%\\ &  +  \left( 140 \right)^{2/3}  \cdot \frac{(8/\pi)^{1/6}}{2 n^{1/3}} + \left( 140 \right)^{2/3} \cdot \frac{\left(\frac{8}{\pi}\right)^{1/3}}{n^{1/2}}  \Bigg) + \frac{384\sqrt{3} \Sigma_{yy}^3 + 1120 \Sigma_{yy}^3 /n^{1/2}}{\Sigma_{11}^3 n} \nonumber
%%
%\\  <  ~  & 6.087 +  2\frac{M_\epsilon}{ \tilde{\sigma}}  \Sigma_{yy}^2  \left(27.517 +  \frac{13.007}{ n^{1/6}}    +  \frac{15.754}{n^{1/3}} + \frac{36.819}{n^{1/2}}  \right) + \frac{\Sigma_{yy}^3}{\Sigma_{11}^3} \left( \frac{384\sqrt{3}  }{ n} + \frac{1120 }{n^{3/2}} \right) \nonumber
\end{align}
Also, note that for \(n  \geq 100\),
\begin{align*}
& 27.517 + \frac{13.007}{n^{1/6}} + \frac{15.754}{n^{1/3}} +  \frac{36.819}{n^{1/2}}
\\ \leq ~ &   27.517 + \frac{13.007}{100^{1/6}} + \frac{15.754}{100^{1/3}} +  \frac{36.819}{100^{1/2}}
\\ < ~ & 40.631,
\end{align*}
and
\[
\frac{384\sqrt{3}  }{ n} + \frac{1120 }{n^{3/2}}  \leq  \frac{384\sqrt{3}  }{ n} + \frac{1120 }{n \cdot 100^{1/2}} < \frac{777.108}{n}.
\]
Using this along with \(M_\epsilon = 36 \Sigma_{yy}\) from Lemma \ref{lemma.smoothness.2.1} and \(\Sigma_{11} \geq 1\) yields 
\begin{align*}
\mathcal{C}    <  ~  & 6.087 +   \frac{36 \Sigma_{yy}}{ \tilde{\sigma}}  \Sigma_{yy}^2 \cdot 40.631 + \Sigma_{yy}^3 \cdot \frac{777.108}{n} \nonumber
\\ < ~ &  6.087 + 1462.717 \cdot \frac{ \Sigma_{yy}^3 }{ \tilde{\sigma}}   +  \Sigma_{yy}^3  \frac{777.108}{n}
\\ \stackrel{(a)}{<} ~ & 6.087 + 1462.717 \cdot \frac{ \Sigma_{yy}^3 }{ \tilde{\sigma}}   +  7.772 \cdot \Sigma_{yy}^3
\\ \stackrel{(b)}{<} ~ &  1462.717 \cdot \frac{ \Sigma_{yy}^3 }{ \tilde{\sigma}}   +  13.859 \cdot \Sigma_{yy}^3,
\end{align*}
where \((a)\) follows because since \(n \geq 100\) we have
 \[
\Sigma_{yy}^3 \cdot   \frac{777.108}{n} \leq  \Sigma_{yy}^3 \cdot \frac{777.108}{100} 
 \]
and \((b)\) follows because \(\Sigma_{yy} > 1\).

\end{proof}

\begin{proof}[Proof of Lemma \ref{lemma.delt.ineq}] First we will show the inequalities from \eqref{beta.u.bound.ineq2}. By \eqref{c4n.conds},
\begin{align}
 \frac{n}{\left(\log n \right)^{3/2}} & > \frac{361}{25}  \frac{(\beta_Z^2 + 1 + \sigma_\epsilon^2)}{4 c_2 }  \nonumber
\\ \iff \qquad \frac{25}{361 \left(\log n \right)^{1/2}} & >  \frac{(\beta_Z^2 + 1 + \sigma_\epsilon^2) \log n}{4 c_2 n}  \nonumber
\\ \iff \qquad \sqrt{\frac{(\beta_Z^2 + 1 + \sigma_\epsilon^2) \log n}{4 c_2 n}} & < \frac{5}{19 \left(\log n \right)^{1/4}}   \nonumber
\\ \iff \qquad  \delta(n) & <    \frac{5}{19 \left(\log n \right)^{1/4}} .\nonumber
\end{align} 
Next, from \eqref{n.large.delta.cond} we have
 \begin{align*}
\frac{n}{ \log n} & > \frac{\beta_Z^2 + 1 + \sigma_\epsilon^2}{c_2}  \cdot 5 \left(1 +  \sigma_\epsilon^2\right)
\\ \iff \qquad \frac{1}{20 \left( 1 + \sigma_\epsilon^2 \right)} & > \frac{  \log n}{4c_2n} \cdot (\beta_Z^2 + 1 + \sigma_\epsilon^2)
 \\ \iff \qquad   \frac{1}{\sqrt{20 \left( 1 + \sigma_\epsilon^2 \right)}} & >  \delta(n) .
 \end{align*}
 Note that this yields
\begin{align*}
\delta(n) & <  \frac{1}{\sqrt{20 \left( 1 + \sigma_\epsilon^2 \right)}} <  \frac{2}{5\sqrt{3 + \sigma_\epsilon^2}}
\\ \iff \qquad  2 \delta(n) \sqrt{3 + \sigma_\epsilon^2} & < \frac{4}{5}.
\end{align*}
Lastly, we have by \eqref{n.large.delta.cond}
\begin{align}
&   \frac{n}{\log n} >  \frac{1}{c_2} \frac{\beta_Z^2 + 1 + \sigma_\epsilon^2}{4 t_0^2 (2 + \sigma_\epsilon^2 )^2}  \nonumber
%\\ \stackrel{(a)}{\implies} \qquad &   \frac{n}{\log n} >  2\left(\beta_Z^2 + 1 + \sigma_\epsilon^2\right)\left(12 + \sigma_\epsilon^2 \right) \nonumber
\\ \iff \qquad &   \sqrt{\frac{(\beta_Z^2 + 1 + \sigma_\epsilon^2) \log n}{4 c_2 n}} <  t_0 (2 + \sigma_\epsilon^2 ) \nonumber
\\ \iff \qquad & \delta(n) <   t_0 (2 + \sigma_\epsilon^2) \nonumber
\\ \implies \qquad & \delta(n) <   t_0 \left( \beta_Z^2 + 1 + \sigma_\epsilon^2 \right), \nonumber
\end{align} 
where we used that \(\beta_Z > 1\), so \eqref{beta.u.bound.ineq2} is verified. Now we will show that the remaining inequalities hold. From the definition of \(\delta(n)\) in \eqref{def.delta.n} we have that 
\[
\delta(n) =  \frac{ \log(n) \sigma_\zeta^2(n) }{20} \cdot \sqrt{\frac{\beta_Z^2 + 1 + \sigma_\epsilon^2}{c_2}} \qquad \iff \qquad \sigma_\zeta^2(n)=  20\sqrt{\frac{c_2}{\beta_Z^2 + 1 + \sigma_\epsilon^2}}  \frac{\delta(n)}{\log n}  .
\]
Since \(1 + \sigma_\zeta^2(n) > 1\) for all \(n \in \mathbb{N}\) and  \(x > \sqrt{x}\) for all \(x > 1\),
\[
 \sqrt{ 1 + \sigma_\zeta^2(n)} - 1 < 1 + \sigma_\zeta^2(n) - 1 =  \sigma_\zeta^2(n) =    20\sqrt{\frac{c_2}{\beta_Z^2 + 1 + \sigma_\epsilon^2}}  \frac{\delta(n)}{\log n} ,
\]
which is \eqref{lemma.ineq.sig.zeta.sqrt}. Finally, using \(n \geq 100\),
\[
\sigma_\zeta^2(n) = \frac{10}{\sqrt{n \log n}} \leq   \frac{10}{\sqrt{100 \log (100)}} < 1.
\]
 This verifies \eqref{sigma.zeta.sq.l.1.lemma}.

\end{proof}

\begin{proof}[Proof of Lemma \ref{cond.beta.z.lemma}] First, note that 
\[
1 < \frac{\sqrt{ 1 + \sigma_\zeta^2(n)} }{ 1 - 2 \delta(n) \sqrt{ 3 + \sigma_\epsilon^2 }} 
\]
because \(1 - 2 \delta(n) \sqrt{ 3 + \sigma_\epsilon^2 } > 0\) due to \eqref{beta.u.bound.ineq2}, \(1 - 2 \delta(n) \sqrt{ 3 + \sigma_\epsilon^2 }  < 1\) since \(\delta(n) > 0\), and \(\sqrt{ 1 + \sigma_\zeta^2(n)} > 1\) since \( \sigma_\zeta^2(n) > 0\). Next we will show the inequality on the right. Due to the assumption that \(\beta_Z \in I(n)\) from \eqref{cond.beta.z} (and using, in particular, \(\beta_Z > 1\)) we have
\begin{align}
% 2 + \sigma_\epsilon^2 & > \left( \beta_Z - 1 \right)^2 \left( \beta_Z^2 + 1 + \sigma_\epsilon^2 \right) \nonumber
% \\ \vdots \nonumber
% \\ & \stackrel{(a)}{>} \left( \beta_Z - 1 \right)^2 \frac{c_2}{3.61}   \frac{n}{\left(\log n \right)^{3/2}} \nonumber
% \\ \iff \qquad   \left( \beta_Z - 1 \right)^2   & < \left( 2 + \sigma_\epsilon^2 \right)  \frac{3.61} {c_2}  \frac{\left(\log n \right)^{3/2}}{n}  
\beta_Z  &  < 1 +   \frac{19}{10}\sqrt{\frac{ 2 + \sigma_\epsilon^2 }{c_2}}   \frac{\left(\log n \right)^{3/4}}{n^{1/2}} \nonumber
\\ \implies \qquad    \frac{n^{1/2}}{\left(\log n \right)^{3/4}}  & <\frac{19}{10}\sqrt{\frac{ 2 + \sigma_\epsilon^2 }{c_2}}    \frac{1}{  \beta_Z - 1 } \nonumber % \label{beta.z.max.condition}
\\ & <\frac{19}{10}\sqrt{\frac{ \beta_Z^2 + 1 + \sigma_\epsilon^2 }{c_2}}    \frac{\beta_Z}{  \beta_Z - 1 } \nonumber
%\\ \implies \qquad  \frac{n^{1/2}}{\left(\log n \right)^{3/4}} & <   \frac{19}{5} \sqrt{\frac{\beta_Z^2 + 1 + \sigma_\epsilon^2 }{4 c_2 }}  \frac{\beta_Z} {\beta_Z - 1} \nonumber
\\ \implies \qquad  10\left(\frac{\beta_Z - 1}{\beta_Z} \right) & < 19 \frac{\left(\log n \right)^{3/4}}{n^{1/2}}  \sqrt{\frac{\beta_Z^2 + 1 + \sigma_\epsilon^2 }{ c_2 }}  \nonumber
\\ \iff \qquad  5\left(1 - \frac{1}{\beta_Z} \right) & < 19 \frac{\left(\log n \right)^{3/4}}{n^{1/2}}  \sqrt{\frac{\beta_Z^2 + 1 + \sigma_\epsilon^2 }{4 c_2 }}  \nonumber
\\ & =  19 \left(\log n \right)^{1/4} \delta(n) \nonumber
%\\ \iff \qquad  5 -   19 \left(\log n \right)^{1/4} \delta(n)  & <   \frac{5}{\beta_Z} \nonumber
%\\ \iff \qquad  \beta_Z & <  \frac{5}{5 -19 \left(\log n \right)^{3/4}/n^{1/2}  \sqrt{\frac{\beta_Z^2 + 1 + \sigma_\epsilon^2 }{4 c_2 }} } \nonumber
%\\ \iff \qquad  \beta_Z & <  \frac{5}{5 -19 \left(\log n \right)^{1/4}  \sqrt{\frac{(\beta_Z^2 + 1 + \sigma_\epsilon^2) \log n}{4 c_2 n}} } \nonumber
\\ \iff \qquad \beta_Z & <  \frac{5}{5 -19 \left(\log n \right)^{1/4} \delta(n)}, \nonumber
\end{align}
%where \((a)\) follows from \(n \geq 100\)
%\[
%\beta_Z^2 + 1 + \sigma_\epsilon^2< \frac{c_2}{3.61}   \frac{n}{\left(\log n \right)^{3/2}}
% \]
% from \eqref{c4n.conds} 
where the last step is permissible because \(5 -19 \left(\log n \right)^{1/4} \delta(n) > 0\) due to \eqref{beta.u.bound.ineq2}. 
%\[
%\vdots
%\]
%We will use this to show that
%\[
%\beta_Z >  \frac{\sqrt{ 1 + \sigma_\zeta^2(n)} }{ 1 - 2 \delta(n) \sqrt{ 3 + \sigma_\epsilon^2 }} .
%\]
Finally we will show the middle inequality. From Lemma \ref{lemma.delt.ineq} we have \(2 \delta(n) \sqrt{ 3 + \sigma_\epsilon^2 } < \frac{4}{5} \). Using the inequality  \(1/(1-t) \leq 1 + 5 t\) valid for \(t \in [0, 4/5]\), we have
\begin{align}
 & \frac{\sqrt{ 1 + \sigma_\zeta^2(n)} }{ 1 - 2 \delta(n) \sqrt{ 3 + \sigma_\epsilon^2 }}  \nonumber
 \\ \leq ~ & \sqrt{ 1 + \sigma_\zeta^2(n)}  \left( 1 + 5 \cdot 2 \delta(n) \sqrt{ 3 + \sigma_\epsilon^2 } \right) \nonumber
 \\  = ~ & \sqrt{ 1 + \sigma_\zeta^2(n)}  \left( 1 + 5 \cdot 2 \sqrt{\frac{(\beta_Z^2 + 1 + \sigma_\epsilon^2) \log n}{4 c_2 n}}\sqrt{ 3 + \sigma_\epsilon^2 } \right)  \nonumber
\\  \stackrel{(a)}{\leq} ~ & \left(1 +  \frac{1}{2}\sigma_\zeta^2(n) \right)  \left( 1 + 5 \cdot  \sqrt{\frac{\left(  \beta_Z^2 + 1 + \sigma_\epsilon^2 \right) \log n}{ c_2 n}}\sqrt{ 3 + \sigma_\epsilon^2 } \right) \nonumber
\\  = ~ & \left(1 + \frac{5}{\sqrt{n \log n}} \right)  \left( 1 + h(\beta_Z, \sigma_\epsilon^2)\sqrt{ \frac{\log n}{n}} \right) \nonumber
\\  = ~ & 1 + \frac{5}{\sqrt{n \log n}} +  h(\beta_Z, \sigma_\epsilon^2)\sqrt{ \frac{\log n}{n}}  +  \frac{5 h(\beta_Z, \sigma_\epsilon^2)}{n}, \label{sec.lemma.ineq}
\end{align}
where in \((a)\) we used the inequality \(\sqrt{1 + t} \leq 1 + \frac{1}{2} t\), valid for \(t  \geq 0\) and
\begin{equation}\label{def.k.sig.ep.squared}
h(\beta_Z, \sigma_\epsilon^2) :=  5  \sqrt{\frac{\left(  \beta_Z^2 + 1 + \sigma_\epsilon^2 \right) \left( 3 + \sigma_\epsilon^2  \right)}{ c_2 }}.
\end{equation}
The following lemma allows us to bound this expression.
\begin{lemma}\label{lem.a1.req}
Under the assumptions of Theorem \ref{thm.result.sel},
\[
\frac{5}{\sqrt{n \log n}} +  h(\beta_Z, \sigma_\epsilon^2)\sqrt{ \frac{\log n}{n}}  +  \frac{5 h(\beta_Z, \sigma_\epsilon^2)}{n} <  \frac{100}{\sqrt{n \log n}}
\]
for \(h(\beta_Z, \sigma_\epsilon^2)\) defined in \eqref{def.k.sig.ep.squared}.
\end{lemma}
\begin{proof} Provided later in Appendix \ref{tech.lemmas}.
 \end{proof}
Finally, the assumption that \(\beta_Z \in I(n)\) from \eqref{cond.beta.z} yields
\begin{align*}
 \beta_Z  & > 1 + \frac{100}{\sqrt{n \log n}}
\\ & \stackrel{(b)}{>} 1 + \frac{5}{\sqrt{n \log n}} +  h(\beta_Z, \sigma_\epsilon^2)\sqrt{ \frac{\log n}{n}}  +  \frac{5 h(\beta_Z, \sigma_\epsilon^2)}{n},
\\ & \stackrel{(c)}{\geq}  \frac{\sqrt{ 1 + \sigma_\zeta^2(n)} }{ 1 - 2 \delta(n) \sqrt{ 3 + \sigma_\epsilon^2 }} ,
\end{align*}
where in \((b)\) we used Lemma \ref{lem.a1.req} and in \((c)\) we used \eqref{sec.lemma.ineq}.

\end{proof}

\begin{proof}[Proof of Lemma \ref{lemma.dist.bounds}] The calculation of the covariance and correlation matrices is trivial. Note that 
\[
 \beta_Z^2 + 1 + \sigma_\epsilon^2 > 1 + 1 +  \sigma_{\epsilon}^2   > 1 + \sigma_{\zeta}^2(n)
 \]
since \(\beta_Z > 1\) and by \eqref{sigma.zeta.sq.l.1.lemma} \(\sigma_{\zeta}^2(n) < 1\). So \(\max \left(\boldsymbol{\Sigma}_{ii}^* \right)= \beta_Z^2 + 1 + \sigma_\epsilon^2 = \Sigma_{yy}\). By inspection we see \(\min \left(\boldsymbol{\Sigma}_{ii}^* \right)= 1\). Finally, we will show that \( \beta_Z  - \left( 1 + \sigma_\zeta^2(n) \right)  > 0\). We have
\begin{align*}
 \beta_Z - \left(  1 + \sigma_\zeta^2(n) \right)   &  \stackrel{(a)}{>}  \frac{\sqrt{1 + \sigma_\zeta^2(n)}}{1 - 2 \delta(n) \sqrt{3 + \sigma_\epsilon^2}} - \left(  1 + \sigma_\zeta^2(n) \right) 
\\ &  =\sqrt{1 + \sigma_\zeta^2(n)}  \left( \frac{1 }{1 - 2 \delta(n) \sqrt{3 + \sigma_\epsilon^2}}  - \sqrt{  1 + \sigma_\zeta^2(n)} \right)
\\ & \stackrel{(b)}{>} \frac{1 }{1 - 2 \delta(n) \sqrt{3 + \sigma_\epsilon^2}}  - \left( 1 +  20\sqrt{\frac{c_2}{\beta_Z^2 + 1 + \sigma_\epsilon^2}}  \frac{\delta(n)}{\log n}   \right)
\\ & =2 \delta(n) \left( \frac{\sqrt{3 + \sigma_\epsilon^2}}{1 - 2 \delta(n) \sqrt{3 + \sigma_\epsilon^2}}  -   10\sqrt{\frac{c_2}{\beta_Z^2 + 1 + \sigma_\epsilon^2}}  \frac{1}{\log n}  \right)
\\ & \stackrel{(c)}{>} 0,
\end{align*} 
where \((a)\) uses Lemma \ref{cond.beta.z.lemma}, \((b)\) follows from \eqref{lemma.ineq.sig.zeta.sqrt} and \(\sigma_\zeta^2(n) > 0\), and \((c)\) 
comes from
\begin{align*}
n & \geq 100
\\ \implies \qquad     n & > \exp \left\{10 \sqrt{\frac{e-1}{48e^2}} \right\} 
\\ \implies \qquad     \log n & > 10 \sqrt{\frac{e-1}{8e^2(\beta_Z^2 + 1 + \sigma_\epsilon^2)(3 + \sigma_\epsilon^2)}} 
\\ \stackrel{(d)}{\implies} \qquad     \log n & > \frac{1 - 2 \delta(n) \sqrt{3 + \sigma_\epsilon^2}}{\sqrt{3 + \sigma_\epsilon^2}} \cdot  10 \sqrt{\frac{c_2}{\beta_Z^2 + 1 + \sigma_\epsilon^2}} 
\\ \iff \qquad   \frac{\sqrt{3 + \sigma_\epsilon^2}}{1 - 2 \delta(n) \sqrt{3 + \sigma_\epsilon^2}} & > 10\sqrt{\frac{c_2}{\beta_Z^2 + 1 + \sigma_\epsilon^2}}  \frac{1}{\log n} 
\end{align*}
where \((d)\) follows from \(c_2 < (e-1)/(8e^2)\) and we used that \(\beta_Z > 1\).

\end{proof}

\begin{proof}[Proof of Lemma \ref{lemma.bound.alpha}] To prove the first result we will use one more lemma.
\begin{lemma}\label{calc.lemma.delt.meth.sigma.tilde.2}
Under the assumptions of Theorem \ref{thm.result.sel}, 
\[
2\tilde{\delta}(n) \cdot \frac{1 + \sigma_\epsilon^2}{\beta_Z^2  + 1 +  \sigma_\epsilon^2}  < \tilde{\sigma}^2(n) \leq 2,
\]
where \(\tilde{\delta}(n)\) is defined in \eqref{lasso.theory.def.tilde.delta} and \(\tilde{\sigma}(n)\) is defined in \eqref{lemma.def.tilde.sigma.n}.
\end{lemma}
\begin{proof}Provided later in Appendix \ref{tech.lemmas}.
\end{proof}

First we will show that \eqref{lemma.bound.alpha.alpha} holds. Using the definitions of \(\delta(n)\) from \eqref{def.delta.n}, \(\eta(n)\) from \eqref{lemma.def.eta.alpha}, and \(\tilde{\delta}(n)\) from \eqref{lasso.theory.def.tilde.delta}, we have that the argument of \(\Phi(\cdot)\) in \eqref{lemma.bound.alpha.alpha} is
\begin{align}
& \frac{\sqrt{n} \eta(n)}{\tilde{\sigma}(n)} \nonumber
\\ = ~ & \frac{2\left(2+ \frac{19}{5} \left(\log n \right)^{1/4} \right) \sqrt{n} \delta(n) [\delta(n) + \tilde{\delta}(n)]}{\tilde{\sigma}(n)}  \nonumber
\\   \stackrel{(a)}{<}  ~ & \frac{2\left(2+ \frac{19}{5}  \left(\log n \right)^{1/4} \right) \sqrt{\beta_Z^2  + 1 +  \sigma_\epsilon^2}  \sqrt{n} \delta(n) [\delta(n) + \tilde{\delta}(n)]}{\sqrt{2\left( 1 + \sigma_\epsilon^2\right)\tilde{\delta}(n)} }  \nonumber
\\   = ~ & \frac{2 \left(2+ \frac{19}{5} \left(\log n \right)^{1/4} \right) \sqrt{\beta_Z^2  + 1 +  \sigma_\epsilon^2} \cdot \sqrt{\frac{(\beta_Z^2 + 1 + \sigma_\epsilon^2) \log n}{4 c_2}} \left[\sqrt{\frac{(\beta_Z^2 + 1 + \sigma_\epsilon^2) \log n}{4 c_2 n}} + \frac{10}{\sqrt{n \log n} + 10} \right]}{\sqrt{2\left( 1 + \sigma_\epsilon^2\right)\cdot \frac{10}{\sqrt{n \log n} + 10}}}  \nonumber
\\   =~  & \left(2+ \frac{19}{5} \left(\log n \right)^{1/4} \right) \left(\beta_Z^2  + 1 +  \sigma_\epsilon^2\right) \sqrt{\frac{4}{80 c_2\left(1 + \sigma_\epsilon^2\right)}}  \nonumber
 \frac{\sqrt{\log n} \left[\sqrt{\frac{(\beta_Z^2 + 1 + \sigma_\epsilon^2) \log n}{4 c_2 n}} + \frac{1}{\sqrt{n}} \frac{10}{\sqrt{ \log n} + 10/\sqrt{n}} \right]}{\sqrt{ \frac{1}{\sqrt{n \log n} + 10}}}   \nonumber
\\   \stackrel{(b)}{<}  ~ &2 \cdot \frac{19}{5} \left(\log n \right)^{1/4} \cdot  \frac{\beta_Z^2  + 1 +  \sigma_\epsilon^2}{\sqrt{20 c_2\left(1 + \sigma_\epsilon^2\right)}}  \nonumber
 \cdot 
 \frac{\sqrt{\log n} \left[\sqrt{\frac{(\beta_Z^2 + 1 + \sigma_\epsilon^2) \log n}{4 c_2 n}} + 2 \sqrt{\frac{\log n}{n}} \right]}{\sqrt{ \frac{1}{\sqrt{n \log n} + 10}}} \nonumber
\\   = ~ &   \frac{38}{5} \left(\log n \right)^{1/4} \cdot \frac{\beta_Z^2  + 1 +  \sigma_\epsilon^2}{2\sqrt{5 c_2\left(1 + \sigma_\epsilon^2\right)}} \log n  \sqrt{ \frac{ \left(n \log n\right)^{1/2} + 10 }{n}}
 \left[\frac{\sqrt{\beta_Z^2 + 1 + \sigma_\epsilon^2 }}{2\sqrt{ c_2 }} + 2 \right]  \nonumber
\\   \stackrel{(c)}{<} ~ &  \frac{19}{5} \left(\log n \right)^{1/4} \cdot \frac{\beta_Z^2  + 1 +  \sigma_\epsilon^2}{ \sqrt{5 c_2(1 + \sigma_\epsilon^2)}}  \cdot  \log n  \sqrt{ \frac{2 \left(n \log n\right)^{1/2} }{n}}
\cdot \frac{\sqrt{\beta_Z^2 + 1 + \sigma_\epsilon^2 }}{\sqrt{ c_2 }}    \nonumber
\\ = ~ & \frac{19 \sqrt{2}}{5} \left(\log n \right)^{1/4} \cdot  \frac{\left( \beta_Z^2  + 1 +  \sigma_\epsilon^2\right)^{3/2}}{ c_2\sqrt{5 (1 + \sigma_\epsilon^2)}}  \cdot  \frac{ \left(\log n\right)^{5/4}}{n^{1/4}}  
  \nonumber
\\ = ~ &  k(\beta_z, \sigma_\epsilon^2)  \cdot  \frac{ \left(\log n\right)^{3/2}}{n^{1/4}}  ,  \nonumber
\end{align}
where \((a)\) follows from Lemma \ref{calc.lemma.delt.meth.sigma.tilde.2}, \((b)\) follows from
\[
\frac{10}{\sqrt{ \log n} + 10/\sqrt{n}} < 2 \sqrt{\log n}
\]
and
\[
2 < \frac{19}{5} \left(\log n \right)^{1/4}
\]
for all \(n \geq 100\), \((c)\) follows from \(10 < \sqrt{n \log n}\) for \(n \geq 100\) and
\[
\frac{\sqrt{\beta_Z^2 + 1 + \sigma_\epsilon^2 }}{2\sqrt{ c_2 }} > \frac{\sqrt{2 }}{2} \sqrt{\frac{8e^2}{e-1}} > 2,
\]
and 
\[
k(\beta_z, \sigma_\epsilon^2) :=  \frac{19\sqrt{2}}{5c_2 \sqrt{5}}  \frac{\left(\beta_Z^2  + 1 +  \sigma_\epsilon^2\right)^{3/2}}{\left(1 + \sigma_\epsilon^2 \right)^{1/2}} .  
\]
So we have that 
\[
\Phi \left( k(\beta_z, \sigma_\epsilon^2)  \cdot  \frac{ \left(\log n\right)^{3/2}}{n^{1/4}} \right)
\]
is a valid upper bound for \(\Phi \left(  \eta(n) \sqrt{n}/\tilde{\sigma}(n)\right)\). The first order Taylor expansion of \(\Phi( k(\beta_z, \sigma_\epsilon^2) x)  \) centered at 0 is 
\begin{align*}
&  \Phi(0)   + \frac{d}{dx}  \left[ \Phi \left( k(\beta_z, \sigma_\epsilon^2)  \cdot  x \right)  \right]_{x = 0} \cdot x + R(x) 
\\  = ~ & \frac{1}{2} +k(\beta_z, \sigma_\epsilon^2)  \phi \left(  0  \right) x + R(x)
\\  = ~ &  \frac{1}{2} +  \frac{1}{\sqrt{2 \pi}} k(\beta_z, \sigma_\epsilon^2)   x + R(x)
\end{align*}
where \(\phi(x) = \frac{1}{\sqrt{2 \pi}} e^{-\frac{x^2}{2}}\) is the pdf of a standard Gaussian random variable and \(R(x) = \Omega(x^2)\). Since \(\Phi''(x) = \phi'(x) \leq 0\) for all \(x \geq 0\), \(\Phi( k(\beta_z, \sigma_\epsilon^2) x) \) is concave for all \(x \geq 0\), so the first order Taylor series is an overestimate for \(\Phi( k(\beta_z, \sigma_\epsilon^2) x) \) when \(x \geq 0\). Therefore \(R(x) \leq 0\) for all \(x \geq 0\), so for all \(n \geq 100\),
\begin{align*}
\Phi \left(  \frac{\eta(n) \sqrt{n}}{\tilde{\sigma}(n)}\right) & < \Phi \left( k(\beta_z, \sigma_\epsilon^2) \frac{ \left(\log n\right)^{3/2}}{n^{1/4}} \right)  
\\ & \leq \frac{1}{2} +    \frac{1}{\sqrt{2 \pi}} k(\beta_z, \sigma_\epsilon^2)   \frac{ \left(\log n\right)^{3/2}}{n^{1/4}} 
\\  & < \frac{1}{2} +  c_6  \frac{\left(\beta_Z^2  + 1 +  \sigma_\epsilon^2\right)^{3/2}}{\left(1 + \sigma_\epsilon^2 \right)^{1/2}} \cdot \frac{\left(\log n\right)^{3/2}}{n^{1/4}} 
\end{align*}
for 
\[
c_6 :=   \frac{19}{5c_2 \sqrt{5 \pi}},
\]
which is \eqref{lemma.bound.alpha.alpha}.

We conclude by showing that \eqref{lemma.bound.alpha.sec} holds. Using the definition of \(\tilde{\delta}(n)\) in \eqref{lasso.theory.def.tilde.delta}, we have
\begin{align}
 \frac{1}{ \tilde{\sigma}(n)} & \stackrel{(d)}{<}   \frac{\sqrt{\beta_Z^2  + 1 +  \sigma_\epsilon^2 }}{\sqrt{2(1 + \sigma_\epsilon^2)}} \frac{1}{ \sqrt{\tilde{\delta}(n)} } \nonumber
  \\ & = \sqrt{\frac{\beta_Z^2  + 1 +  \sigma_\epsilon^2 }{2(1 + \sigma_\epsilon^2)}} \frac{\left( \sqrt{n \log n} + 10\right)^{1/2}}{ \sqrt{10}} \nonumber
    \\ & \stackrel{(e)}{<} \sqrt{\frac{\beta_Z^2  + 1 +  \sigma_\epsilon^2 }{2(1 + \sigma_\epsilon^2)}} \frac{\sqrt{2} \left( n \log n \right)^{1/4}}{ \sqrt{10 }} \nonumber
\end{align}
where \((d)\) follows from Lemma \ref{calc.lemma.delt.meth.sigma.tilde.2} and \((e)\) follows from \(10 < \sqrt{n \log n}\) for \(n \geq 100\). This yields \eqref{lemma.bound.alpha.sec}.

\end{proof}

\begin{proof}[Proof of Lemma \ref{dist.cond.one.sub.e.mse.lem}] We establish the identities one at at time.
\begin{itemize}

\item Note that 
\begin{align*}
& \frac{1}{\sqrt{\beta_Z^2 + \beta_{q+1}^2 + \sigma_\epsilon^2}} \boldsymbol{y}  \sim \mathcal{N}(0, \boldsymbol{I}_n ) 
\\ \implies \qquad  & \E \left(  \left[ \frac{1}{\sqrt{\beta_Z^2 + \beta_{q+1}^2 + \sigma_\epsilon^2}} \boldsymbol{y} \right]^\top \left[\frac{1}{\sqrt{\beta_Z^2 + \beta_{q+1}^2 + \sigma_\epsilon^2}} \boldsymbol{y}\right] \right)   = n .
\end{align*}

This establishes \eqref{dist.cond.one.sub.e.mse.yTy}.

\item

We have for \(j \in [q]\)
\[
 \begin{pmatrix}
 \boldsymbol{Z} \\
 \boldsymbol{X}_{\cdot j}
 \end{pmatrix} = \begin{pmatrix}
 \boldsymbol{Z} \\
 \boldsymbol{Z}
 \end{pmatrix} + \begin{pmatrix}
 \boldsymbol{0} \\
 \boldsymbol{\zeta}_j
 \end{pmatrix} \sim \mathcal{N} \left( \begin{pmatrix}
 \boldsymbol{0} \\
 \boldsymbol{0} 
 \end{pmatrix}, \begin{pmatrix}
 \boldsymbol{I}_n & \boldsymbol{I}_n \\
 \boldsymbol{I}_n & \boldsymbol{I}_n + \sigma_{\zeta j}^2 \boldsymbol{I}_n
 \end{pmatrix} \right).
\]
Then
\begin{align}
 \boldsymbol{Z} \mid  \boldsymbol{X}_{\cdot j} & \sim \mathcal{N} \left(   \boldsymbol{I}_n \left( \boldsymbol{I}_n + \sigma_{\zeta j}^2 \boldsymbol{I}_n \right)^{-1} \boldsymbol{X}_{\cdot j} , \boldsymbol{I}_n - \boldsymbol{I}_n(\boldsymbol{I}_n + \sigma_{\zeta j}^2 \boldsymbol{I}_n)^{-1}\boldsymbol{I}_n  \right) \nonumber
\\ & = \mathcal{N} \left( \frac{\boldsymbol{X}_{\cdot j} }{1 + \sigma_{\zeta j}^2},  \left[ 1 - \frac{1}{1+ \sigma_{\zeta j}^2} \right] \boldsymbol{I}_n \right) \nonumber
\\ & =  \mathcal{N} \left( \frac{\boldsymbol{X}_{\cdot j} }{1 + \sigma_{\zeta j}^2}, \frac{\sigma_{\zeta j}^2}{1+ \sigma_{\zeta j}^2} \boldsymbol{I}_n \right).\label{dist.cond.one.sub.e.mse.z.mid.xj}
\end{align}

Next,

\begin{align*}
\E\left[ \hat{\beta}_j  \mid \boldsymbol{X}_{\cdot j} \right] & =  \E\left[ ( \boldsymbol{X}_{\cdot j}^\top  \boldsymbol{X}_{\cdot j} )^{-1}  \boldsymbol{X}_{\cdot j}^\top \left(  \beta_Z \boldsymbol{Z} + \beta_{q+1} \boldsymbol{X}_{\cdot q+1} + \boldsymbol{\epsilon} \right) \mid \boldsymbol{X}_{\cdot j} \right]  
\\ & = ( \boldsymbol{X}_{\cdot j}^\top  \boldsymbol{X}_{\cdot j} )^{-1}  \boldsymbol{X}_{\cdot j}^\top   \left( \beta_Z \E\left[ \boldsymbol{Z} \mid \boldsymbol{X}_{\cdot j} \right] +  \beta_{q+1} \E\left[  \boldsymbol{X}_{\cdot q+1} \mid \boldsymbol{X}_{\cdot j} \right]     +  \E\left[  \boldsymbol{\epsilon} \mid \boldsymbol{X}_{\cdot j} \right]   \right)
\\ & = \frac{\beta_Z}{1 + \sigma_{\zeta j}^2} ( \boldsymbol{X}_{\cdot j}^\top  \boldsymbol{X}_{\cdot j} )^{-1}  \boldsymbol{X}_{\cdot j}^\top   \boldsymbol{X}_{\cdot j} + \boldsymbol{0}
\\ & =   \frac{\beta_Z}{1 + \sigma_{\zeta j}^2}
\end{align*}

where we used \eqref{dist.cond.one.sub.e.mse.z.mid.xj} and \(\E\left[ \boldsymbol{\epsilon}\mid \boldsymbol{X}_{\cdot j} \right] =0\) because \(\E\left[ \boldsymbol{\epsilon}\right] =0\) and \(\boldsymbol{\epsilon}\) and \(\boldsymbol{X}_{\cdot j} \) are independent by assumption, and similarly \( \E\left[  \boldsymbol{X}_{\cdot q+1} \mid \boldsymbol{X}_{\cdot j} \right] =0\). Next,

%This verifies \eqref{dist.cond.one.sub.e.mse.alpha.mid.xj}.

\begin{align*}
\Var\left[ \hat{\beta}_j  \mid \boldsymbol{X}_{\cdot j} \right] & =  \Var\left[ ( \boldsymbol{X}_{\cdot j}^\top  \boldsymbol{X}_{\cdot j} )^{-1}  \boldsymbol{X}_{\cdot j}^\top \left(\beta_Z \boldsymbol{Z}  + \beta_{q+1} \boldsymbol{X}_{\cdot q+1} + \boldsymbol{\epsilon} \right)\mid \boldsymbol{X}_{\cdot j} \right] 
\\ & = \frac{1}{( \boldsymbol{X}_{\cdot j}^\top  \boldsymbol{X}_{\cdot j} )^{2}}  \boldsymbol{X}_{\cdot j}^\top  \Var\left[  \beta_Z \boldsymbol{Z}  +   \beta_{q+1} \boldsymbol{X}_{\cdot q+1}  + \boldsymbol{\epsilon} \mid \boldsymbol{X}_{\cdot j} \right]  \boldsymbol{X}_{\cdot j}
\\ & = \frac{1}{( \boldsymbol{X}_{\cdot j}^\top  \boldsymbol{X}_{\cdot j} )^{2}}  \boldsymbol{X}_{\cdot j}^\top  ( \beta_Z^2 \Var\left[  \boldsymbol{Z}  \mid \boldsymbol{X}_{\cdot j} \right]   +   \beta_{q+1}^2 \Var\left[    \boldsymbol{X}_{\cdot q+1}  \mid \boldsymbol{X}_{\cdot j} \right]   
\\ & +   \Var\left[   \boldsymbol{\epsilon} \mid \boldsymbol{X}_{\cdot j} \right] )  \boldsymbol{X}_{\cdot j}
\\ & = \frac{1}{( \boldsymbol{X}_{\cdot j}^\top  \boldsymbol{X}_{\cdot j} )^{2}}  \boldsymbol{X}_{\cdot j}^\top  \left( \frac{\beta_Z^2  \sigma_{\zeta j}^2}{1 + \sigma_{\zeta j}^2}   +   \beta_{q+1}^2  +  \sigma_\epsilon^2 \right) \boldsymbol{I}_n  \boldsymbol{X}_{\cdot j}
\\ & = \left( \frac{\beta_Z^2  \sigma_{\zeta j}^2}{1 + \sigma_{\zeta j}^2}   +   \beta_{q+1}^2  +  \sigma_\epsilon^2   \right) \frac{ \boldsymbol{X}_{\cdot j}^\top    \boldsymbol{X}_{\cdot j}}{( \boldsymbol{X}_{\cdot j}^\top  \boldsymbol{X}_{\cdot j} )^{2}} 
\\ & = \left( \frac{\beta_Z^2  \sigma_{\zeta j}^2}{1 + \sigma_{\zeta j}^2}   +   \beta_{q+1}^2  +  \sigma_\epsilon^2   \right) \frac{1}{ \boldsymbol{X}_{\cdot j}^\top  \boldsymbol{X}_{\cdot j}} 
\end{align*}

where we used \eqref{dist.cond.one.sub.e.mse.z.mid.xj}. Note also that \(\hat{\beta}\) conditioned on \(\boldsymbol{X}_{\cdot j}\) is a Gaussian random variable because \(\boldsymbol{Z}\mid \boldsymbol{X}_{\cdot j}\) and \(\boldsymbol{\epsilon}\mid \boldsymbol{X}_{\cdot j}\) are independent Gaussian random variables. Putting this together we have
\[
\hat{\beta}_j \mid \boldsymbol{X}_{\cdot j} \sim \mathcal{N} \left( \frac{\beta_Z}{1 + \sigma_{\zeta j}^2} ,  \left( \frac{\beta_Z^2  \sigma_{\zeta j}^2}{1 + \sigma_{\zeta j}^2}   +   \beta_{q+1}^2  +  \sigma_\epsilon^2   \right)  \frac{1}{ \boldsymbol{X}_{\cdot j}^\top  \boldsymbol{X}_{\cdot j} }\right).  
\]

This proves \eqref{dist.cond.one.sub.e.mse.var.alpha.mid.xj}.

\item 

\begin{align*}
\E\left[ \hat{\beta}_{q+1}  \mid \boldsymbol{X}_{\cdot q + 1} \right] & =  \E\left[ ( \boldsymbol{X}_{\cdot q + 1}^\top  \boldsymbol{X}_{\cdot q + 1} )^{-1}  \boldsymbol{X}_{\cdot q + 1}^\top \left(  \beta_Z \boldsymbol{Z} + \beta_{q+1} \boldsymbol{X}_{\cdot q+1} + \boldsymbol{\epsilon} \right) \mid \boldsymbol{X}_{\cdot q + 1} \right]  
\\ & =\beta_{q+1}  ( \boldsymbol{X}_{\cdot q + 1}^\top  \boldsymbol{X}_{\cdot q + 1} )^{-1}  \boldsymbol{X}_{\cdot q + 1}^\top     \boldsymbol{X}_{\cdot q+1}
\\ & = \beta_{q+1},
\end{align*}
where we used the fact that \(\E\left[ \boldsymbol{\epsilon}\mid \boldsymbol{X}_{\cdot q + 1} \right] =0\) because \(\E\left[ \boldsymbol{\epsilon}\right] =0\) and \(\boldsymbol{\epsilon}\) and \(\boldsymbol{X}_{\cdot q + 1} \) are independent by assumption, and similarly \( \E\left[  \boldsymbol{Z} \mid \boldsymbol{X}_{\cdot q + 1} \right] =0\). Next,
\begin{align*}
\Var\left[ \hat{\beta}_{q+1}  \mid \boldsymbol{X}_{\cdot q + 1} \right] & =  \Var\left[ ( \boldsymbol{X}_{\cdot q + 1}^\top  \boldsymbol{X}_{\cdot q + 1} )^{-1}  \boldsymbol{X}_{\cdot q + 1}^\top \left(\beta_Z \boldsymbol{Z}  + \beta_{q+1} \boldsymbol{X}_{\cdot q+1} + \boldsymbol{\epsilon} \right)\mid \boldsymbol{X}_{\cdot q + 1} \right] 
\\ & = \frac{1}{( \boldsymbol{X}_{\cdot q + 1}^\top  \boldsymbol{X}_{\cdot q + 1} )^{2}}  \boldsymbol{X}_{\cdot q + 1}^\top  \Var\left[  \beta_Z \boldsymbol{Z}  +   \beta_{q+1} \boldsymbol{X}_{\cdot q+1}  + \boldsymbol{\epsilon} \mid \boldsymbol{X}_{\cdot q + 1} \right]  \boldsymbol{X}_{\cdot q + 1}
\\ & = \frac{1}{( \boldsymbol{X}_{\cdot q + 1}^\top  \boldsymbol{X}_{\cdot q + 1} )^{2}}  \boldsymbol{X}_{\cdot q + 1}^\top  ( \beta_Z^2 \Var\left[  \boldsymbol{Z}  \mid \boldsymbol{X}_{\cdot q + 1} \right]   +   \beta_{q+1}^2 \Var\left[    \boldsymbol{X}_{\cdot q+1}  \mid \boldsymbol{X}_{\cdot q + 1} \right]   
\\ & +   \Var\left[   \boldsymbol{\epsilon} \mid \boldsymbol{X}_{\cdot q + 1} \right] )  \boldsymbol{X}_{\cdot q + 1}
\\ & = \frac{1}{( \boldsymbol{X}_{\cdot q + 1}^\top  \boldsymbol{X}_{\cdot q + 1} )^{2}}  \boldsymbol{X}_{\cdot q + 1}^\top  \left( \beta_Z^2 +  \sigma_\epsilon^2 \right) \boldsymbol{I}_n  \boldsymbol{X}_{\cdot q + 1}
\\ & = \left( \beta_Z^2  +  \sigma_\epsilon^2   \right) \frac{ \boldsymbol{X}_{\cdot q + 1}^\top    \boldsymbol{X}_{\cdot q + 1}}{( \boldsymbol{X}_{\cdot q + 1}^\top  \boldsymbol{X}_{\cdot q + 1} )^{2}} 
\\ & =  \frac{\beta_Z^2  + \sigma_\epsilon^2}{ \boldsymbol{X}_{\cdot q + 1}^\top  \boldsymbol{X}_{\cdot q + 1}} .
\end{align*}
Note also that \(\hat{\beta}_{q+1}\) conditioned on \(\boldsymbol{X}_{\cdot q + 1}\) is a Gaussian random variable because \(\boldsymbol{Z}\mid \boldsymbol{X}_{\cdot q + 1}\), \(\boldsymbol{X}_{\cdot q + 1}\mid \boldsymbol{X}_{\cdot q + 1}\), and \(\boldsymbol{\epsilon}\mid \boldsymbol{X}_{\cdot q + 1}\) are independent Gaussian random variables. Putting this together we have
\[
\hat{\beta}_{q+1} \mid \boldsymbol{X}_{\cdot q + 1} \sim \mathcal{N} \left(\beta_{q+1} ,   \frac{\beta_Z^2  + \sigma_\epsilon^2}{ \boldsymbol{X}_{\cdot q + 1}^\top  \boldsymbol{X}_{\cdot q + 1}}  \right).  
\]

This shows \eqref{dist.cond.one.sub.e.mse.var.alpha.mid.x3}.

\item 

\begin{align*}
\E\left[ \hat{\beta}_Z  \mid \boldsymbol{Z} \right] & =  \E\left[ ( \boldsymbol{Z}^\top  \boldsymbol{Z} )^{-1}  \boldsymbol{Z}^\top \left(  \beta_Z \boldsymbol{Z} + \beta_{q+1} \boldsymbol{X}_{\cdot q+1} + \boldsymbol{\epsilon} \right) \mid \boldsymbol{Z} \right]  
\\ & =\beta_Z  ( \boldsymbol{Z}^\top  \boldsymbol{Z} )^{-1}  \boldsymbol{Z}^\top \boldsymbol{Z}
\\ & = \beta_Z,
\end{align*}
where we used the fact that \(\E\left[ \boldsymbol{\epsilon}\mid \boldsymbol{Z} \right] =0\) because \(\E\left[ \boldsymbol{\epsilon}\right] =0\) and \(\boldsymbol{\epsilon}\) and \(\boldsymbol{Z} \) are independent by assumption, and similarly \( \E\left[  \boldsymbol{X}_{\cdot q+1} \mid \boldsymbol{Z} \right] =0\). Next,
\begin{align*}
\Var\left[ \hat{\beta}_Z  \mid \boldsymbol{Z} \right] & =  \Var\left[ ( \boldsymbol{Z}^\top  \boldsymbol{Z} )^{-1}  \boldsymbol{Z}^\top \left(\beta_Z \boldsymbol{Z}  + \beta_{q+1} \boldsymbol{X}_{\cdot q+1} + \boldsymbol{\epsilon} \right)\mid \boldsymbol{Z} \right] 
\\ & = \frac{1}{( \boldsymbol{Z}^\top  \boldsymbol{Z} )^{2}}  \boldsymbol{Z}^\top  \Var\left[  \beta_Z \boldsymbol{Z}  +   \beta_{q+1} \boldsymbol{X}_{\cdot q+1}  + \boldsymbol{\epsilon} \mid \boldsymbol{Z} \right]  \boldsymbol{Z}
\\ & = \frac{1}{( \boldsymbol{Z}^\top  \boldsymbol{Z} )^{2}}  \boldsymbol{Z}^\top  ( \beta_Z^2 \Var\left[  \boldsymbol{Z}  \mid \boldsymbol{Z} \right]   +   \beta_{q+1}^2 \Var\left[    \boldsymbol{X}_{\cdot q+1}  \mid \boldsymbol{Z} \right]   
\\ & +   \Var\left[   \boldsymbol{\epsilon} \mid \boldsymbol{Z} \right] )  \boldsymbol{Z}
\\ & = \frac{1}{( \boldsymbol{Z}^\top  \boldsymbol{Z} )^{2}}  \boldsymbol{Z}^\top  \left( \beta_{q+1}^2 +  \sigma_\epsilon^2 \right) \boldsymbol{I}_n  \boldsymbol{Z}
\\ & = \left( \beta_{q+1}^2  +  \sigma_\epsilon^2   \right) \frac{ \boldsymbol{Z}^\top    \boldsymbol{Z}}{( \boldsymbol{Z}^\top  \boldsymbol{Z} )^{2}} 
\\ & =  \frac{\beta_{q+1}^2  + \sigma_\epsilon^2}{ \boldsymbol{Z}^\top  \boldsymbol{Z}} .
\end{align*}
Note also that \(\hat{\beta}_Z\) conditioned on \(\boldsymbol{Z}\) is a Gaussian random variable because \(\boldsymbol{Z}\mid \boldsymbol{Z}\), \(\boldsymbol{X}_{q+1} \mid \boldsymbol{Z}\), and \(\boldsymbol{\epsilon}\mid \boldsymbol{Z}\) are independent Gaussian random variables. Putting this together we have
\[
\hat{\beta}_Z \mid \boldsymbol{Z} \sim \mathcal{N} \left(\beta_{Z} ,   \frac{\beta_{q+1}^2  + \sigma_\epsilon^2}{ \boldsymbol{Z}^\top  \boldsymbol{Z}}  \right).  
\]

This proves \eqref{dist.cond.one.sub.e.mse.dist.alpha.mid.z}.

\item 

\[
 \boldsymbol{X}_{\cdot j} \sim \mathcal{N} \left( 0, (1 + \sigma_{\zeta j}^2) \boldsymbol{I}_n \right) \implies \frac{1}{\sqrt{1 + \sigma_{\zeta j}^2}} \boldsymbol{X}_{\cdot j} \sim \mathcal{N}(0,1)
 \]
 
 This verifies \eqref{dist.cond.one.sub.e.mse.xTx}.

 \item  For any \(j \in [q]\),

\[
 \E \left[ \boldsymbol{X}_{\cdot j}^\top \boldsymbol{y}  \right] =  \E \left[ \left(  \boldsymbol{Z} + \boldsymbol{\zeta}_j \right) ^\top \left(  \beta_Z \boldsymbol{Z}  +  \beta_{q+1} \boldsymbol{X}_{\cdot q +1} + \boldsymbol{\epsilon} \right) \right] =  \E \left[ \beta_Z  \boldsymbol{Z}^\top\boldsymbol{Z} \right] = \beta_Z n.
\]
 
This establishes \eqref{dist.cond.one.sub.e.mse.xTy}.

 \item

\[
 \E \left[ \boldsymbol{X}_{\cdot q+1}^\top \boldsymbol{y}  \right] =  \E \left[  \boldsymbol{X}_{\cdot q+1}^\top  \left(  \beta_Z \boldsymbol{Z}  +  \beta_{q+1} \boldsymbol{X}_{\cdot q +1} + \boldsymbol{\epsilon}  \right) \right] = \beta_{q+1}  n.
\]
 
This proves \eqref{dist.cond.one.sub.e.mse.xTy.3}.

\item Since \( \frac{1}{\sqrt{1 + \sigma_{\zeta j}^2}} \boldsymbol{X}_{\cdot j} \sim \mathcal{N}(0, \boldsymbol{I}_n)\), for any \(j \in [q]\) the random variable

\[
\frac{1+ \sigma_{\zeta j}^2}{ \boldsymbol{X}_{\cdot j}^\top  \boldsymbol{X}_{\cdot j}}
\]
has an inverse \(\chi^2\) distribution with \(n\) degrees of freedom. Therefore

\[
\E \left[ \frac{1+ \sigma_{\zeta j}^2}{ \boldsymbol{X}_{\cdot j}^\top  \boldsymbol{X}_{\cdot j}} \right] = \frac{1}{n-2} \implies  \E \left[   \frac{1}{ \boldsymbol{X}_{\cdot j}^\top  \boldsymbol{X}_{\cdot j} }  \right] = \frac{1}{(n-2)(1+ \sigma_{\zeta j}^2)},
\]
verifying \eqref{dist.cond.one.sub.e.mse.xTx.inv}.

\item

Since \(  \boldsymbol{X}_{\cdot q + 1} \sim \mathcal{N}(0, \boldsymbol{I}_n)\), the random variable
\[
\frac{1}{ \boldsymbol{X}_{\cdot q+1j}^\top  \boldsymbol{X}_{\cdot q + 1}}
\]
has an inverse \(\chi^2\) distribution with \(n\) degrees of freedom, so
\[
\E \left[\frac{1}{ \boldsymbol{X}_{\cdot q+1}^\top  \boldsymbol{X}_{\cdot q + 1}} \right] = \frac{1}{n-2}.
\]
This shows \eqref{dist.cond.one.sub.e.mse.xTx.inv.3}.

\item To calculate \(\E \left[  \hat{\beta}_j^2  \right]\) for any \(j \in [q]\), we will use \eqref{dist.cond.one.sub.e.mse.var.alpha.mid.xj} and \eqref{dist.cond.one.sub.e.mse.xTx.inv}:
\begin{align*}
%\hat{\beta} \mid \boldsymbol{X}_{\cdot j} & \sim \mathcal{N} \left( \frac{1}{1 + \sigma_{\zeta j}^2} ,  \left(  \frac{\sigma_{\zeta j}^2}{1 + \sigma_{\zeta j}^2}  +  \sigma_\epsilon^2 \right)  \frac{1}{ \boldsymbol{X}^\top  \boldsymbol{X}_{\cdot j} }\right)
%\\ \implies \qquad 
\E \left[ \hat{\beta}_j^2  \right] & = \E \left[ \E \left[ \hat{\beta}_j^2 \mid \boldsymbol{X}_{\cdot j}  \right]  \right] 
\\ & =  \E \left[\Var \left[ \hat{\beta}_j \mid \boldsymbol{X}_{\cdot j}  \right]  +  \E \left[ \hat{\beta}_j \mid \boldsymbol{X}_{\cdot j}  \right]^2  \right]
\\ & = \E \left[   \left( \frac{\beta_Z^2  \sigma_{\zeta j}^2}{1 + \sigma_{\zeta j}^2}   +   \beta_{q+1}^2  +  \sigma_\epsilon^2   \right)  \frac{1}{ \boldsymbol{X}_{\cdot j}^\top  \boldsymbol{X}_{\cdot j} }  + \left( \frac{\beta_Z}{1 + \sigma_{\zeta j}^2}  \right)^2 \right]
\\ & =  \frac{1}{1 + \sigma_{\zeta j}^2}  \left[   \left( \frac{\beta_Z^2  \sigma_{\zeta j}^2}{1 + \sigma_{\zeta j}^2}   +   \beta_{q+1}^2  +  \sigma_\epsilon^2   \right)  \frac{1}{n-2}  +  \frac{\beta_Z^2}{1 + \sigma_{\zeta j}^2}  \right].
\end{align*}

This proves \eqref{dist.cond.one.sub.e.mse.ex.alpha.squared.mid.xj}.

\item

Using \eqref{dist.cond.one.sub.e.mse.var.alpha.mid.x3} and \eqref{dist.cond.one.sub.e.mse.xTx.inv.3} we have
\begin{align*}
%\hat{\beta} \mid \boldsymbol{X}_{\cdot j} & \sim \mathcal{N} \left( \frac{1}{1 + \sigma_{\zeta j}^2} ,  \left(  \frac{\sigma_{\zeta j}^2}{1 + \sigma_{\zeta j}^2}  +  \sigma_\epsilon^2 \right)  \frac{1}{ \boldsymbol{X}^\top  \boldsymbol{X}_{\cdot j} }\right)
%\\ \implies \qquad 
\E \left[ \hat{\beta}_{q+1}^2  \right] & = \E \left[ \E \left[ \hat{\beta}_{q+1}^2 \mid \boldsymbol{X}_{\cdot q+1}  \right] \right]
\\  & = \E \left[ \Var \left[ \hat{\beta}_{q+1} \mid \boldsymbol{X}_{\cdot q+1}  \right]  +  \E \left[ \hat{\beta}_{q+1} \mid \boldsymbol{X}_{\cdot q+1}  \right]^2 \right]
\\ & = \E \left[    \frac{\beta_Z^2 + \sigma_\epsilon^2}{ \boldsymbol{X}_{\cdot q+1}^\top  \boldsymbol{X}_{\cdot q+1} }  + \beta_{q+1}^2 \right]
\\ & =   \frac{\beta_Z^2 + \sigma_\epsilon^2}{ n - 2 }  + \beta_{q+1}^2 
\end{align*}

This verifies \eqref{dist.cond.one.sub.e.mse.ex.alpha.squared.mid.x3}.

\item

Using \eqref{dist.cond.one.sub.e.mse.dist.alpha.mid.z} we have
\begin{align*}
\E \left[ \hat{\beta}_Z^2 \right] & =  \E \left[  \Var \left[ \hat{\beta}_Z \mid \boldsymbol{Z}  \right]  +  \E \left[ \hat{\beta}_Z \mid \boldsymbol{Z}  \right]^2 \right]
\\ & = \E \left[    \frac{\beta_{q+1}^2 + \sigma_\epsilon^2}{ \boldsymbol{Z}^\top  \boldsymbol{Z} }  + \beta_{Z}^2 \right]
\\ & =   \frac{\beta_{q+1}^2 + \sigma_\epsilon^2}{ n - 2 }  + \beta_{Z}^2 ,
\end{align*}

which is \eqref{dist.cond.one.sub.e.mse.alpha.sq.z}.

\end{itemize}

\end{proof}

\begin{proof}[Proof of Lemma \ref{pred.risk.lemma.x1}]
\begin{enumerate}[(i)]

\item First consider the case where there is only one directly observed feature (that is, \(p = q + 1\)). For any \(j \in [q]\),

\begin{align}
& \E \left[ \frac{1}{n} \left\lVert \tilde{\boldsymbol{y}} -\hat{\beta}_{j} \tilde{\boldsymbol{X}}_{\cdot j}  \right\rVert_2^2 \right]  \nonumber
\\ = ~ & \frac{1}{n} \E \left[  \left( \tilde{\boldsymbol{y}} - \hat{\beta}_{j} \tilde{\boldsymbol{X}}_{\cdot j} \right)^\top   \left( \tilde{\boldsymbol{y}} - \hat{\beta}_{j} \tilde{\boldsymbol{X}}_{\cdot j}\right)\right]  \nonumber
%\\   = ~ &  \frac{1}{n} \E \left[ \tilde{\boldsymbol{y}}^\top \tilde{\boldsymbol{y}} - 2\hat{\beta}_{j} \tilde{\boldsymbol{X}}_{\cdot j} ^\top\tilde{\boldsymbol{y}} +  \hat{\beta}_{j}^2 \tilde{\boldsymbol{X}}_{\cdot j}^\top\tilde{\boldsymbol{X}}_{\cdot j}  \right]  \nonumber
\\  \stackrel{(a)}{=}  ~ & \frac{1}{n} \E \left[ \tilde{\boldsymbol{y}}^\top \tilde{\boldsymbol{y}} \right] -  
 \frac{2}{n} \E \left[  \hat{\beta}_{j} \right] \E \left[ \tilde{\boldsymbol{X}}_{\cdot j}^\top\tilde{\boldsymbol{y}}  \right]  +  \frac{1}{n} \E \left[  \hat{\beta}_{j}^2 \right] \E \left[ \tilde{\boldsymbol{X}}_{\cdot j}^\top\tilde{\boldsymbol{X}}_{\cdot j}  \right]  \nonumber
 \\  \stackrel{(b)}{=} ~ & \beta_Z^2 + \beta_{q+1}^2 + \sigma_\epsilon^2 -  \frac{2}{n}  \E \left[  \E \left( \hat{\beta}_{j}  \mid \boldsymbol{X}_{\cdot j} \right) \right] \beta_Z  n \nonumber
 \\ & + \frac{1}{n} \cdot \frac{1}{1+ \sigma_{\zeta j}^2} \left[ \left(  \frac{\beta_Z^2 \sigma_{\zeta j}^2}{1 + \sigma_{\zeta j}^2} + \beta_{q+1}^2  +  \sigma_\epsilon^2 \right) \cdot \frac{1}{n-2} +  \frac{\beta_Z^2}{1 + \sigma_{\zeta j}^2} \right]   \nonumber
 \cdot n(1 + \sigma_{\zeta j}^2)    \nonumber
 \\  \stackrel{(c)}{=}  ~ & \beta_Z^2 + \beta_{q+1}^2 + \sigma_\epsilon^2 -  \frac{2 \beta_Z^2}{1 + \sigma_{\zeta j}^2} +   \left(  \frac{\beta_Z^2 \sigma_{\zeta j}^2}{1 + \sigma_{\zeta j}^2} + \beta_{q+1}^2 +  \sigma_\epsilon^2 \right) \cdot \frac{1}{n-2}   \nonumber
+  \frac{\beta_Z^2}{1 + \sigma_{\zeta j}^2}  \nonumber
 \\ =  ~ &  \frac{n-1}{n-2}\left(\beta_{q+1}^2 + \sigma_\epsilon^2\right) + \beta_Z^2  -  \frac{ \beta_Z^2}{1 + \sigma_{\zeta j}^2} +    \frac{\beta_Z^2 \sigma_{\zeta j}^2}{1 + \sigma_{\zeta j}^2} \cdot \frac{1}{n-2}   \nonumber
  \\ =  ~ &  \frac{n-1}{n-2}\left(\beta_{q+1}^2 + \sigma_\epsilon^2\right) +  \frac{ \beta_Z^2 \sigma_{\zeta j}^2}{1 + \sigma_{\zeta j}^2} +    \frac{\beta_Z^2 \sigma_{\zeta j}^2}{1 + \sigma_{\zeta j}^2} \cdot \frac{1}{n-2}   \nonumber
\\ = ~ & \frac{n - 1 }{n-2} \left(  \frac{\beta_Z^2 \sigma_{\zeta j}^2}{1+ \sigma_{\zeta j}^2} + \beta_{q+1}^2   +  \sigma_\epsilon^2  \right) ,\label{dist.cond.one.sub.e.mse.result.intmd}
\end{align}
where \((a)\) follows because \(\tilde{\boldsymbol{X}}_{\cdot j}\) and \(\tilde{\boldsymbol{y}}\) are independent from \(\boldsymbol{X}_{\cdot j}\) and \(\boldsymbol{y}\), so \(\hat{\beta}_{j} = (\boldsymbol{X}_{\cdot j}^\top\boldsymbol{X}_{\cdot j})^{-1}\boldsymbol{X}_{\cdot j}^\top\boldsymbol{y} \) is independent from \(\tilde{\boldsymbol{X}}_{\cdot j}\) and \(\tilde{\boldsymbol{y}}\), \((b)\) follows from \eqref{dist.cond.one.sub.e.mse.yTy}, \eqref{dist.cond.one.sub.e.mse.xTx}, \eqref{dist.cond.one.sub.e.mse.xTy}, and \eqref{dist.cond.one.sub.e.mse.ex.alpha.squared.mid.xj}, and \((c)\) follows from \eqref{dist.cond.one.sub.e.mse.var.alpha.mid.xj}.

Now suppose we have an arbitrary number of directly observed features (that is, an arbitrary \(p > q\)). Then for any \(j \in [q]\),
\begin{align*}
\boldsymbol{y}  & = \beta_Z \boldsymbol{Z} + \sum_{j = q+1}^p \beta_j \boldsymbol{X}_{\cdot j} + \boldsymbol{\epsilon} = \beta_Z \boldsymbol{Z} +  \beta_{q+1}\boldsymbol{X}_{\cdot q+ 1} + \boldsymbol{\tilde{\epsilon}}
\end{align*}
where
\[
\boldsymbol{\tilde{\epsilon}} := \sum_{j' = q+2}^p \beta_{j'} \boldsymbol{X}_{\cdot j'} + \boldsymbol{\epsilon}
\]
is independent of \( \boldsymbol{Z}\) and \(\boldsymbol{X}_{\cdot q+ 1} \). So we can use \eqref{dist.cond.one.sub.e.mse.result.intmd} and we have
\begin{align*}
\E \left[ \frac{1}{n} \left\lVert \tilde{\boldsymbol{y}} -\hat{\beta}_{j} \tilde{\boldsymbol{X}}_{\cdot j}  \right\rVert_2^2 \right] & =  \frac{n - 1 }{n-2} \left(  \frac{\beta_Z^2 \sigma_{\zeta j}^2}{1+ \sigma_{\zeta j}^2} + \beta_{q+1}^2   + \Var \left(\boldsymbol{\tilde{\epsilon}} \right) \right)
\\ & =  \frac{n - 1 }{n-2} \left(  \frac{\beta_Z^2 \sigma_{\zeta j}^2}{1+ \sigma_{\zeta j}^2} + \beta_{q+1}^2   + \sum_{j' = q+2}^p \beta_{j'}^2 + \sigma_\epsilon^2 \right).
\end{align*}

\item Again, we start by considering the case where \(p = q + 1\).
 \begin{align} 
& \E \left[ \frac{1}{n} \left\lVert \tilde{\boldsymbol{y}} -  \hat{\beta}_{q + 1} \tilde{\boldsymbol{X}}_{\cdot q + 1} \right\rVert_2^2 \right]  \nonumber
\\ = ~ & \frac{1}{n} \E \left[  \left( \tilde{\boldsymbol{y}} - \hat{\beta}_{q + 1} \tilde{\boldsymbol{X}}_{\cdot q + 1} \right)^\top   \left( \tilde{\boldsymbol{y}} - \hat{\beta}_{q + 1} \tilde{\boldsymbol{X}}_{\cdot q + 1}\right)\right]  \nonumber
%\\   = ~ &  \frac{1}{n} \E \left[ \tilde{\boldsymbol{y}}^\top \tilde{\boldsymbol{y}} - 2\hat{\beta}_{q + 1} \tilde{\boldsymbol{X}}_{\cdot q + 1} ^\top\tilde{\boldsymbol{y}} +  \hat{\beta}_{q + 1}^2 \tilde{\boldsymbol{X}}_{\cdot q + 1}^\top\tilde{\boldsymbol{X}}_{\cdot q + 1}  \right]  \nonumber
\\  \stackrel{(a)}{=}  ~ & \frac{1}{n} \E \left[ \tilde{\boldsymbol{y}}^\top \tilde{\boldsymbol{y}} \right] -  
 \frac{2}{n} \E \left[  \hat{\beta}_{q + 1} \right] \E \left[ \tilde{\boldsymbol{X}}_{\cdot q + 1}^\top\tilde{\boldsymbol{y}}  \right]  \nonumber
+  \frac{1}{n} \E \left[  \hat{\beta}_{q + 1}^2 \right] \E \left[ \tilde{\boldsymbol{X}}_{\cdot q + 1}^\top\tilde{\boldsymbol{X}}_{\cdot q + 1}  \right]  \nonumber
 \\  \stackrel{(b)}{=} ~ & \beta_Z^2 + \beta_{q+1}^2 + \sigma_\epsilon^2 -  \frac{2}{n}  \E \left[  \E \left( \hat{\beta}_{q + 1}  \mid \boldsymbol{X}_{\cdot q + 1} \right) \right]  \cdot \beta_{q+1}   \cdot n  \nonumber
 + \frac{1}{n} \left( \frac{\beta_Z^2 + \sigma_\epsilon^2}{ n-2 } + \beta_{q+1}^2 \right) \cdot n    \nonumber
  \\  \stackrel{(c)}{=} ~ & \beta_Z^2 + \beta_{q+1}^2 + \sigma_\epsilon^2 - 2 \beta_{q+1}^2  + \frac{\beta_Z^2 + \sigma_\epsilon^2}{ n-2 } + \beta_{q+1}^2     \nonumber
%\\  =  ~ & \left( \frac{n-1}{n-2} \right) \beta_Z^2  + \sigma_\epsilon^2   + \frac{ \sigma_\epsilon^2 }{ n-2 } \nonumber
\\  =  ~ &  \frac{n-1}{n-2}  \left( \beta_Z^2  + \sigma_\epsilon^2    \right)    ,\label{dist.cond.one.sub.e.mse.result.intmd.3}
\end{align}

where \((a)\) follows because \(\tilde{\boldsymbol{X}}_{q + 1}\) and \(\tilde{\boldsymbol{y}}\) are independent from \(\hat{\beta}_{q + 1}\), \((b)\) follows from \eqref{dist.cond.one.sub.e.mse.yTy}, \eqref{dist.cond.one.sub.e.mse.xTy.3}, \eqref{dist.cond.one.sub.e.mse.ex.alpha.squared.mid.x3}, and the fact that \(\E \left[ \tilde{\boldsymbol{X}}_{\cdot q + 1}^\top\tilde{\boldsymbol{X}}_{\cdot q + 1} \right] = 1\), and \((c)\) follows from \eqref{dist.cond.one.sub.e.mse.var.alpha.mid.x3}.

Now for an arbitrary \(p > q\), for any \(j \in \{q + 1, \ldots, p\}\) we have
\begin{align*}
\boldsymbol{y}  & = \beta_Z \boldsymbol{Z} + \sum_{j = q+1}^p \beta_j \boldsymbol{X}_{\cdot j} + \boldsymbol{\epsilon} = \beta_Z \boldsymbol{Z} +  \beta_{j}\boldsymbol{X}_{\cdot j} + \boldsymbol{\tilde{\epsilon}}
\end{align*}
where
\[
\boldsymbol{\tilde{\epsilon}} := \sum_{j' \in \{q+1, \ldots, p\} \setminus j} \beta_{j'} \boldsymbol{X}_{\cdot j'} + \boldsymbol{\epsilon}
\]
is independent of \( \boldsymbol{Z}\) and \(\boldsymbol{X}_{\cdot j} \). So we can use \eqref{dist.cond.one.sub.e.mse.result.intmd.3} and we have
\begin{align*}
 \E \left[ \frac{1}{n} \left\lVert \tilde{\boldsymbol{y}} -  \hat{\beta}_{j} \tilde{\boldsymbol{X}}_{\cdot q + 1} \right\rVert_2^2 \right] & = \frac{n-1}{n-2}  \left( \beta_Z^2  +  \Var \left(\boldsymbol{\tilde{\epsilon}} \right)   \right)   
 \\ & = \frac{n-1}{n-2}  \left( \beta_Z^2  +  \sum_{j' \in \{q+1, \ldots, p\} \setminus j} \beta_{j'}^2 + \sigma_\epsilon^2  \right)   .
\end{align*}

\item First we handle the case with only one directly observed feature \(\boldsymbol{X}_{\cdot q +1}\), as in the previous parts.
\begin{align}
& \E \left[ \frac{1}{n} \left\lVert \tilde{\boldsymbol{y}} -  \hat{\beta}_Z \boldsymbol{\tilde{Z}}  \right\rVert_2^2 \right] \nonumber
%\\ = ~ & \frac{1}{n} \E \left[  \left( \tilde{\boldsymbol{y}} - \hat{\beta}_Z \boldsymbol{\tilde{Z}} \right)^\top   \left( \tilde{\boldsymbol{y}} - \hat{\beta}_Z \boldsymbol{\tilde{Z}}\right)\right] \nonumber
\\ = ~ & \frac{1}{n} \E \left[  \tilde{\boldsymbol{y}}^\top\tilde{\boldsymbol{y}}   - 2 \hat{\beta}_Z \boldsymbol{\tilde{Z}}^\top\tilde{\boldsymbol{y}}     +\hat{\beta}_Z^2 \boldsymbol{\tilde{Z}}^\top\boldsymbol{\tilde{Z}}  \right] \nonumber
\\  \stackrel{(a)}{=} ~ & \frac{1}{n} \Bigg(  n(\beta_Z^2 + \beta_{q+1}^2 + \sigma_\epsilon^2)   - 2 \E \left[ \hat{\beta}_Z  \right] \E \left[ \boldsymbol{\tilde{Z}}^\top  \E \left(\tilde{\boldsymbol{y}} \mid \boldsymbol{\tilde{Z}} \right)   \right]     \nonumber
 +\E \left[\hat{\beta}_Z^2 \right] \E \left[ \boldsymbol{\tilde{Z}}^\top\boldsymbol{\tilde{Z}}   \right] \Bigg) \nonumber
\\  \stackrel{(b)}{=} ~ & \frac{1}{n} \Bigg(  n(\beta_Z^2 + \beta_{q+1}^2 + \sigma_\epsilon^2)   - 2 \E \left[ \E \left( \hat{\beta}_Z  \mid \boldsymbol{Z} \right) \right]  \beta_Z \E \left[ \boldsymbol{\tilde{Z}}^\top \boldsymbol{\tilde{Z}}  \right]    \nonumber
  + n \left( \frac{\beta_{q+1}^2 + \sigma_\epsilon^2}{ n - 2}   + \beta_Z^2  \right) \Bigg) \nonumber
\\  \stackrel{(c)}{=} ~ & \frac{1}{n} \left(  n(\beta_Z^2 + \beta_{q+1}^2 + \sigma_\epsilon^2)   - 2 \beta_Z^2  n   + n \left( \frac{\beta_{q+1}^2 + \sigma_\epsilon^2}{ n - 2}   + \beta_Z^2  \right)  \right) \nonumber
\\  = ~ &  \beta_Z^2 + \beta_{q+1}^2 + \sigma_\epsilon^2   - 2\beta_Z^2      +  \frac{\beta_{q+1}^2 + \sigma_\epsilon^2}{ n-2}   + \beta_Z^2 \nonumber
\\  = ~ &   \frac{n-1}{n-2} \left( \beta_{q+1}^2 + \sigma_\epsilon^2 \right), \label{z.risk.intmd}
\end{align}
where in \((a)\) we used \eqref{dist.cond.one.sub.e.mse.yTy} and the independence of \(\hat{\beta}_Z\) from \(\boldsymbol{\tilde{Z}}\) and \(\tilde{\boldsymbol{y}}\), in \((b)\) we used \(\E\left[ \boldsymbol{Z}^\top \boldsymbol{Z} \right] =n\) and \eqref{dist.cond.one.sub.e.mse.alpha.sq.z}, and in \((c)\) we used \eqref{dist.cond.one.sub.e.mse.dist.alpha.mid.z} and \(\E\left[ \boldsymbol{Z}^\top \boldsymbol{Z} \right] =n\). Now we make this more general. We have
\[
\boldsymbol{y}  = \beta_Z \boldsymbol{Z} + \sum_{j = q+1}^p \beta_j \boldsymbol{X}_{\cdot j} + \boldsymbol{\epsilon} = \beta_Z \boldsymbol{Z} +  \beta_{q+1} \boldsymbol{X}_{\cdot q +1} + \boldsymbol{\tilde{\epsilon}}
\]
where
\[
\boldsymbol{\tilde{\epsilon}} :=   \sum_{j = q+2}^p \beta_j \boldsymbol{X}_{\cdot j} + \boldsymbol{\epsilon}  \sim \mathcal{N} \left( \boldsymbol{0}, \left(  \sum_{j = q+2}^p \beta_j^2 + \sigma_\epsilon^2 \right) \boldsymbol{I}_n \right).
\]
Substituting into \eqref{z.risk.intmd}, we see that the prediction risk of \(\boldsymbol{Z}\) is
\[
\frac{n-1}{n-2} \left( \beta_{q+1}^2 + \left(  \sum_{j = 2}^p \beta_j^2 + \sigma_\epsilon^2 \right) \right) = \frac{n-1}{n-2} \left(  \sum_{j = q+1}^p \beta_j^2 + \sigma_\epsilon^2  \right) .
\]

\end{enumerate}

\end{proof}

\begin{proof}[Proof of Lemma \ref{lem.opt.weighting.all.positive}] Note that
\[
\sum_{j=1}^q w_q \boldsymbol{X}_{\cdot j}= \boldsymbol{Z} + \sum_{j=1}^q w_q \boldsymbol{\zeta}_j \stackrel{d}{=} \boldsymbol{Z} + \tilde{\boldsymbol{\zeta}},
\]
where \(\stackrel{d}{=} \) denotes equality in distribution and \(\tilde{\boldsymbol{\zeta}} \sim \mathcal{N}\left(0, \sum_{j=1}^q w_j^2 \sigma_{\zeta j}^2 \boldsymbol{I}_n \right)\) is independent of \(\boldsymbol{Z}\). 
%Since
%\[
%\boldsymbol{y} = \beta_Z \boldsymbol{Z} + \sum_{j=q+1}^p \beta_j \boldsymbol{X}_{\cdot j} + \boldsymbol{\epsilon} =  \beta_Z \boldsymbol{Z} + \beta_{q + 1} \boldsymbol{X}_{\cdot q+1}   + \boldsymbol{\tilde{\epsilon}}
%\] 
%where 
%\[
%\boldsymbol{\tilde{\epsilon}} := \sum_{j=q+2}^p \beta_j \boldsymbol{X}_{\cdot j} + \boldsymbol{\epsilon} \sim \mathcal{N} \left( \boldsymbol{0}, \left( \sum_{j=q+2}^p  \beta_j^2 + \sigma_\epsilon^2 \right) \boldsymbol{I}_n \right) ,
%\]
% from
 Then the result follows from Lemma \ref{pred.risk.lemma.x1}(i).
\end{proof}

\begin{proof}[Proof of Lemma \ref{lem.a1.req}] Using
\[ 
\frac{\log n}{n} < \frac{c_2}{\beta_Z^2 + 1 + \sigma_\epsilon^2} \cdot \frac{1}{2\left(12 + \sigma_\epsilon^2 \right)}
\]
from \eqref{n.large.delta.cond}, we have
\begin{align*}
h(\beta_Z, \sigma_\epsilon^2)\sqrt{ \frac{\log n}{n}}  & <  h(\beta_Z, \sigma_\epsilon^2)\sqrt{  \frac{c_2}{\beta_Z^2 + 1 + \sigma_\epsilon^2} \cdot \frac{1}{2\left(12 + \sigma_\epsilon^2 \right)}}
\\ & = 5 \sqrt{  \frac{ 3 + \sigma_\epsilon^2}{2\left(12 + \sigma_\epsilon^2 \right)}}
\\ & < \frac{5 \sqrt{2}}{2}.
\end{align*}
Similarly,
\begin{align*}
\frac{5 h(\beta_Z, \sigma_\epsilon^2)}{n} & = 5 h(\beta_Z, \sigma_\epsilon^2)\sqrt{ \frac{\log n}{n}}  \cdot \frac{1}{\sqrt{n \log n}} < \frac{25 \sqrt{2}}{2} \cdot \frac{1}{\sqrt{n \log n}} .
\end{align*}
This yields
\begin{align*}
\frac{5}{\sqrt{n \log n}} +  h(\beta_Z, \sigma_\epsilon^2)\sqrt{ \frac{\log n}{n}}  +  \frac{5 h(\beta_Z, \sigma_\epsilon^2)}{n}  & <  \frac{5}{\sqrt{n \log n}} + \frac{5 \sqrt{2}}{2} +    \frac{25\sqrt{2}}{2\sqrt{n \log n}} 
\\ & <  \frac{5}{\sqrt{n \log n}}  \left( 1 +\frac{\sqrt{2}}{2}\sqrt{100 \log 100} +  \frac{5 \sqrt{2}}{2} \right)
\\ & <  \frac{100}{\sqrt{n \log n}}.
\end{align*}

\end{proof}

\begin{proof}[Proof of Lemma \ref{calc.lemma.delt.meth.sigma.tilde.2}] Note that by Lemma \ref{lemma.dist.bounds} the assumed covariance matrix structure in Proposition \ref{lemma.prob.a1.eta} holds. Also, the definition of \(\tilde{\sigma}\) in \eqref{lemma.def.tilde.sigma} matches the definition in \eqref{lemma.def.tilde.sigma.n}. Therefore the assumptions of Lemma \ref{calc.lemma.delt.meth.sigma.tilde} are satisfied, and the right side of the inequality follows since in Lemma \ref{calc.lemma.delt.meth.sigma.tilde} we show that \(\tilde{\sigma}^2 \leq 2\). To see that the left side of the inequality holds, note that from Lemma \ref{lemma.dist.bounds} we have
\[
\rho_{1y}^2(n) = \frac{\beta_Z^2}{\left(\beta_Z^2  + 1 +  \sigma_\epsilon^2\right) \left( 1 + \sigma_\zeta^2(n)  \right)}  =r\rho_{12}(n)
\]
where
\[
r :=  \frac{\beta_Z^2}{\beta_Z^2  + 1 +  \sigma_\epsilon^2} \in (0,1).
\]
So
\begin{align*}
\tilde{\sigma}^2 & =     2 (1 - \rho_{12}(n)) \left( - \frac{3}{2} \rho_{1y}^2(n) + \frac{1}{2} \rho_{1y}^2(n)  \rho_{12}(n) + 1\right)
\\ &  =   2 (1 - \rho_{12}(n)) \left(    -\rho_{1y}^2(n) \left[ \frac{3 -  \rho_{12}(n)}{2}\right] + 1 \right)
\\ & = 2 (1 - \rho_{12}(n)) \left(   - r\rho_{12}(n)\left[ \frac{1 - \rho_{12}(n) }{2} + 1\right] + 1   \right) 
\\ & = 2 (1 - \rho_{12}(n)) \left( r\left[ \frac{1 - \rho_{12}(n) }{2} + 1\right] - r\rho_{12}(n)\left[ \frac{1 - \rho_{12}(n) }{2} + 1\right] -  r\frac{1 - \rho_{12}(n) }{2}   + 1 - r  \right) 
\\ & = 2 (1 - \rho_{12}(n))\left( r\left[ \frac{1 - \rho_{12}(n) }{2} + 1\right]\left[1  - \rho_{12}(n) \right]  - r\frac{1 - \rho_{12}(n) }{2} + 1 - r \right) 
\\ & = 2 (1 - \rho_{12}(n))\left( r\left[1  - \rho_{12}(n) \right] \left[ \frac{1 - \rho_{12}(n) }{2} + 1  - \frac{1 }{2} \right] + 1 - r \right) 
\\ & >  2(1 - \rho_{12}(n))(1 - r)
\\ & = 2(1 - \rho_{12}(n)) \cdot \frac{1 + \sigma_\epsilon^2}{\beta_Z^2  + 1 +  \sigma_\epsilon^2}.
\end{align*}
\end{proof}

%%\bibliographystyle{abbrvnat}
%\bibliography{mybib2fin}

\end{document}